\pdfoutput=1
\documentclass[english,11pt]{article}
\usepackage[T1]{fontenc}
\usepackage[latin9]{inputenc}
\usepackage{geometry}
\usepackage{xcolor}
\usepackage{float}
\usepackage{amsmath}
\usepackage{amsthm}
\usepackage{amssymb}

\makeatletter

\floatstyle{ruled}
\newfloat{algorithm}{tbp}{loa}
\providecommand{\algorithmname}{Algorithm}
\floatname{algorithm}{\protect\algorithmname}
\usepackage{tikz}
\usetikzlibrary{calc}

\numberwithin{equation}{section}
\numberwithin{figure}{section}
\theoremstyle{plain}
\newtheorem{thm}{\protect\theoremname}
\theoremstyle{remark}
\newtheorem{rem}[thm]{\protect\remarkname}
\theoremstyle{definition}
\newtheorem{defn}[thm]{\protect\definitionname}
\theoremstyle{plain}
\newtheorem{lem}[thm]{\protect\lemmaname}
\theoremstyle{plain}
\newtheorem{fact}[thm]{\protect\factname}
\theoremstyle{remark}
\newtheorem*{rem*}{\protect\remarkname}
\theoremstyle{plain}
\newtheorem{cor}[thm]{\protect\corollaryname}

\pdfoutput=1

\usepackage{hyperref}
\hypersetup{
    colorlinks,
    allcolors=blue
}

\usepackage[lined,boxed,ruled,norelsize,algo2e,linesnumbered,noend]{algorithm2e}
\usepackage{caption}
\usepackage{amsfonts}

\usepackage{float}
\newcommand{\LineComment}[1]{\tcc{#1}}
\newcommand{\State}[0]{} % NO OP so I dont have to delecte all \State
\newcommand{\Comment}[1]{\tcp*[h]{#1}} % redefine the old comment command

\usepackage{lmodern}

\usepackage{thmtools}
\usepackage{thm-restate}
\usepackage[bottom]{footmisc} %%% Zhao add this

\newcommand{\algstore}[1]{\newcounter{algline#1}\setcounter{algline#1}{\value{AlgoLine}}}

\newcommand{\algrestore}[1]{\setcounter{AlgoLine}{\value{algline#1}}}

\usepackage{fancyvrb}

\makeatother

\usepackage{babel}
\providecommand{\corollaryname}{Corollary}
\providecommand{\definitionname}{Definition}
\providecommand{\factname}{Fact}
\providecommand{\lemmaname}{Lemma}
\providecommand{\remarkname}{Remark}
\providecommand{\theoremname}{Theorem}

\begin{document}
\pagenumbering{gobble}
\title{Solving Tall Dense Linear Programs\\
in Nearly Linear Time}
\author{Jan van den Brand\thanks{\texttt{janvdb@kth.se}. KTH Royal Institute of Technology.}\\
\and Yin Tat Lee\thanks{\texttt{yintat@uw.edu}. University of Washington and MSR Redmond.}\\
\and Aaron Sidford\thanks{\texttt{sidford@stanford.edu}. Stanford University.}\\
\and Zhao Song\thanks{\texttt{zhaos@ias.edu}. Princeton University and Institute for Advanced Study.}\\
}
\date{}

\maketitle
In this paper we provide an $\widetilde{O}(nd+d^{3})$ time randomized
algorithm for solving linear programs with $d$ variables and $n$
constraints with high probability. To obtain this result we provide
a robust, primal-dual $\widetilde{O}(\sqrt{d})$-iteration interior
point method inspired by the methods of Lee and Sidford (2014, 2019)
and show how to efficiently implement this method using new data-structures
based on heavy-hitters, the Johnson--Lindenstrauss lemma, and inverse
maintenance. Interestingly, we obtain this running time without using
fast matrix multiplication and consequently, barring a major advance
in linear system solving, our running time is near optimal for solving
dense linear programs among algorithms that don't use fast matrix
multiplication.

\thispagestyle{empty}

\textcolor{purple}{}
\global\long\def\defeq{\stackrel{\mathrm{{\scriptscriptstyle def}}}{=}}%
\textcolor{purple}{}
\global\long\def\norm#1{\|#1\|}%
\textcolor{purple}{}
\global\long\def\normFull#1{\left\Vert #1\right\Vert }%
\textcolor{purple}{}
\global\long\def\R{\mathbb{R}}%
\textcolor{purple}{}
\global\long\def\Rn{\mathbb{R}^{n}}%
\textcolor{purple}{}
\global\long\def\tr{\mathrm{Tr}}%
\textcolor{purple}{}
\global\long\def\diag{\mathrm{diag}}%
\textcolor{purple}{}
\global\long\def\mdiag{\mathbf{Diag}}%
\textcolor{purple}{}
\global\long\def\ma{\mathbf{A}}%
\textcolor{purple}{}
\global\long\def\mb{\mathbf{B}}%
\textcolor{purple}{}
\global\long\def\mn{\mathbf{N}}%
\textcolor{purple}{}
\global\long\def\mx{\mathbf{X}}%
\textcolor{purple}{}
\global\long\def\my{\mathbf{Y}}%
\textcolor{purple}{}
\global\long\def\ms{\mathbf{S}}%
\textcolor{purple}{}
\global\long\def\mi{\mathbf{I}}%
\textcolor{purple}{}
\global\long\def\mv{\mathbf{V}}%
\textcolor{purple}{}
\global\long\def\mw{\mathbf{W}}%
\textcolor{purple}{}
\global\long\def\mproj{\mathbf{P}}%
\textcolor{purple}{}
\global\long\def\mq{\mathbf{Q}}%
\textcolor{purple}{}
\global\long\def\mz{\mathbf{Z}}%
\textcolor{purple}{}
\global\long\def\md{\mathbf{D}}%
\textcolor{purple}{}
\global\long\def\mh{\mathbf{H}}%
\textcolor{purple}{}
\global\long\def\mf{\mathbf{F}}%
\textcolor{purple}{}
\global\long\def\mr{\mathbf{R}}%
\textcolor{purple}{}
\global\long\def\mm{\mathbf{M}}%
\textcolor{purple}{}
\global\long\def\mDelta{\Delta}%
\textcolor{purple}{}
\global\long\def\mLambda{\boldsymbol{\Lambda}}%
\textcolor{purple}{}
\global\long\def\mSigma{\boldsymbol{\Sigma}}%
\textcolor{purple}{}
\global\long\def\oSigma{\overline{\mSigma}}%
\textcolor{purple}{}
\global\long\def\me{\mathbf{E}}%
\textcolor{purple}{}
\global\long\def\mj{\mathbf{J}}%
\textcolor{purple}{}
\global\long\def\new{\mathrm{(new)}}%
\textcolor{purple}{}
\global\long\def\mPhi{\mathbf{\Phi}}%
\textcolor{purple}{}
\global\long\def\tilde#1{\mathrm{\widetilde{#1}}}%
\textcolor{purple}{}
\global\long\def\vones{1}%
\textcolor{purple}{}
\global\long\def\vzero{0}%
\textcolor{purple}{}
\global\long\def\mzero{\boldsymbol{0}}%
\textcolor{purple}{}
\global\long\def\rank{\mathrm{rank}}%
\textcolor{purple}{}
\global\long\def\Rnd{\mathbb{R}^{n\times d}}%
\textcolor{purple}{}
\global\long\def\Rnn{\mathbb{R}^{n\times n}}%
\textcolor{purple}{}
\global\long\def\omx{\overline{\mx}}%
\textcolor{purple}{}
\global\long\def\oms{\overline{\ms}}%
\textcolor{purple}{}
\global\long\def\os{\overline{s}}%
\textcolor{purple}{}
\global\long\def\tw{\widetilde{w}}%
\textcolor{purple}{}
\global\long\def\ox{\overline{x}}%
\textcolor{purple}{}
\global\long\def\ow{\overline{w}}%
\textcolor{purple}{}
\global\long\def\ov{\overline{v}}%
\textcolor{purple}{}
\global\long\def\og{\overline{g}}%
\textcolor{purple}{}
\global\long\def\otau{\overline{\tau}}%
\textcolor{purple}{}
\global\long\def\omw{\overline{\mw}}%
\textcolor{purple}{}
\global\long\def\omw{\overline{\mw}}%
\textcolor{purple}{}
\global\long\def\tmh{\mathbf{\widetilde{H}}}%
\textcolor{purple}{}
\global\long\def\tmw{\mathbf{\widetilde{W}}}%
\textcolor{purple}{}
\global\long\def\median{\mathrm{median}}%
\textcolor{purple}{}
\global\long\def\tsig{\widetilde{\sigma}}%
\textcolor{purple}{}
\global\long\def\tSig{\tilde{\mSigma}}%
\textcolor{purple}{}
\global\long\def\osigma{\overline{\sigma}}%
\textcolor{purple}{}
\global\long\def\grad{\nabla}%
\textcolor{purple}{}
\global\long\def\sign{\mathrm{sign}}%
\textcolor{purple}{}
\global\long\def\otilde{\mathrm{\widetilde{O}}}%
\textcolor{purple}{}
\global\long\def\mT{\mathbf{T}}%
\textcolor{purple}{}
\global\long\def\capprox{\mathbf{T}}%
\textcolor{purple}{}
\global\long\def\cnorm{{\color{purple}{\normalcolor C_{\mathrm{norm}}}}}%
\textcolor{purple}{}
\global\long\def\mixedNorm#1#2{\norm{#1}_{#2+\infty}}%
\textcolor{purple}{}
\global\long\def\mixedNormFull#1#2{\mathrm{\normFull{#1}_{#2+\infty}}}%
\textcolor{purple}{}
\global\long\def\argmax{\mathop{\mathrm{argmax}}}%
\textcolor{purple}{}
\global\long\def\argmin{\mathop{\mathrm{argmin}}}%
\textcolor{purple}{}
\global\long\def\poly{\mathrm{poly}}%
\textcolor{purple}{}
\global\long\def\nnz{\mathrm{nnz}}%
\textcolor{purple}{}
\global\long\def\mc{\mathbf{C}}%
\textcolor{purple}{}
\global\long\def\alg{\mathrm{(alg)}}%
\textcolor{purple}{}
\global\long\def\ttau{\widetilde{\tau}}%
\textcolor{purple}{}
\global\long\def\E{\mathbb{E}}%
\textcolor{purple}{}
\global\long\def\P{\mathbb{P}}%
\textcolor{purple}{}
\global\long\def\t{^{(t)}}%
\textcolor{purple}{}
\global\long\def\k{^{(k)}}%
\textcolor{purple}{}
\global\long\def\old{^{(\mathrm{old})}}%
\textcolor{purple}{}
\global\long\def\safe{_{(\mathrm{safe})}}%
\textcolor{purple}{}
\global\long\def\N{\mathbb{N}}%
\textcolor{purple}{}
\global\long\def\otilde{\tilde O}%
\textcolor{purple}{}
\global\long\def\mg{\mathbf{G}}%
\textcolor{purple}{}
\global\long\def\gap{\mathop{\mathrm{gap}}}%
\global\long\def\taustandard{\tau_{\mathrm{std}}}%
\global\long\def\tauweight{\tau_{\mathrm{weight}}}%
\global\long\def\tauls{\tau_{\mathrm{\mathrm{ls}}}}%
\global\long\def\taureg{\tau_{\mathrm{\mathrm{reg}}}}%
\textcolor{purple}{}
\global\long\def\sketch{\textsc{Sketch}}%
\textcolor{purple}{}
\global\long\def\hitter{\textsc{Decode}}%
\textcolor{purple}{}
\global\long\def\JL{\textsc{JL}}%
\textcolor{purple}{}
\global\long\def\indicVec#1{1_{#1}}%
\textcolor{purple}{}
\global\long\def\power{6}%
\textcolor{purple}{}
\global\long\def\polylog{\mathrm{polylog}\,}%

\pagebreak\pagenumbering{arabic}
\setcounter{page}{2}

\tableofcontents{}

\listof{algorithm}{List of Algorithms}

\pagebreak
\section{Introduction}

Given $\ma\in\R^{n\times d}$, $b\in\R^{d}$, and $c\in\R^{n}$ solving
a linear program $(P)$ and its dual $(D)$:
\begin{equation}
(P)=\min_{x\in\R_{\geq0}^{n}:\ma^{\top}x=b}c^{\top}x\text{ and }(D)=\max_{y\in\R^{d}:\ma y\geq c}b^{\top}y\,.\label{eq:primal_dual}
\end{equation}
is one of the most fundamental and well-studied problems in computer
science and optimization. Developing faster algorithms for (\ref{eq:primal_dual})
has been the subject of decades of extensive research and the pursuit
of faster linear programming methods has lead to numerous algorithmic
advances and the advent of fundamental optimization techniques, e.g.
simplex methods \cite{dantzig1951maximization}, ellipsoid methods
\cite{khachiyan1980polynomial}, and interior-point methods (IPMs)
\cite{karmarkar1984new}.

The current fastest algorithms for solving (\ref{eq:primal_dual})
are the IPMs of Lee and Sidford \cite{lee2015efficient} and Cohen,
Lee, and Song \cite{cohen2019solving}. These results build on a long
line of work on IPMs \cite{karmarkar1984new,renegar1988polynomial,vaidya90parallel,vaidya89convexSet,vaidya1993technique,Nesterov1994,anstreicher96},
fast matrix multiplication \cite{s69,cw82,s86,DCoppersmithW90,s10,williams2012matrixmult,gall2014powers,GallU18},
and linear system solvers \cite{clarkson2013low,meng2013low,nelson2012osnap,li2012iterative,cohen2015uniform}.
The first, \cite{lee2015efficient} combined an $\otilde(\sqrt{d})$
iteration IPM from \cite{leeS14,lsJournal19} with new techniques
for \emph{inverse maintenance,} i.e. maintaining an approximate inverse
of a slowly changing matrix, to solve (\ref{eq:primal_dual}) in time
$\otilde(\nnz(\ma)\sqrt{d}+d^{2.5})$ with high probability. For sufficiently
large values of $\nnz(\ma)$ or $n$, this is the fastest known running
time for solving linear programming up to polylogarithmic factors.

The second, \cite{cohen2019solving}, developed a stable and robust
version of the IPM of \cite{renegar1988polynomial} (using techniques
from \cite{leeS14,lsJournal19}) and combined it with novel randomization,
data structure, and rectangular matrix multiplication \cite{GallU18}
techniques to solve (\ref{eq:primal_dual}) in time $\otilde(\max\{n,d\}^{\omega})$
with high probability where $\omega<2.373$ is the current best known
matrix multiplication constant \cite{williams2012matrixmult,gall2014powers}.
When $n=\tilde{\Theta}(d)$, this running time matches that of the
best known linear system solvers for solving $n\times n$ linear systems
and therefore is the best possible barring a major linear system solving
advance.

Though, these results constitute substantial advances in algorithmic
techniques for linear programming, the running times of \cite{lee2015efficient}
and \cite{cohen2019solving} are incomparable and neither yield a
nearly linear running time when $n$ grows polynomially with $d$,
i.e. when $d=n^{\delta}$ for any $\delta>0$ neither \cite{lee2015efficient}
or \cite{cohen2019solving} yields a nearly linear running time. Consequently,
it has remained a fundamental open problem to determine whether or
not it is possible to solve high-dimensional linear programs to high-precision
in nearly linear time for any polynomial ratio of $n,$$d$, and $\nnz(\ma)$.
Though achieving such a nearly linear runtime is known in the simpler
setting of linear regression \cite{clarkson2013low,meng2013low,nelson2012osnap,li2012iterative,cohen2015uniform,Cohen16},
achieving analogous results for solving such tall linear programs
has been elusive. 

In this paper we provide the first such nearly linear time algorithm
for linear programming. We provide an algorithm which solves (\ref{eq:primal_dual})
with high probability in time $\otilde(nd+d^{3})$. Whenever $\ma$
is dense, i.e. $\nnz(\ma)=\Omega(nd)$, and sufficiently tall, i.e.
$n=\Omega(d^{2})$, this constitutes a nearly linear running time.
In contrast to previous state-of-the-art IPMs for linear programming
\cite{lee2015efficient,cohen2019solving} our algorithm doesn't need
to use fast matrix multiplication and thereby matches the best known
running time for regression on dense matrices which does not use fast
matrix multiplication.

To achieve this result, we introduce several techniques that we believe
are independent interest. First, we consider the IPM of \cite{lsJournal19}
and develop an efficiently implementable robust primal-dual variant
of it in the style of \cite{cohen2019solving,LeeSZ19,br2019deterministicArxiv}
which only requires solving $\otilde(\sqrt{d})$ linear systems. Prior
to \cite{leeS14,lsJournal19}, obtaining such an efficient $\otilde(\sqrt{d})$-iteration
IPM was a major open problem in the theory of IPMs. We believe our
primal dual method and its analysis is simpler than that of \cite{lsJournal19};
by developing a primal-dual method we eliminate the need for explicit
$\ell_{p}$-Lewis weight computation for $p\neq2$ as in \cite{lsJournal19}
and instead work with the simpler $\ell_{2}$-Lewis weights or leverage
scores. Further, we show that this primal dual method is highly stable
and can be implemented efficiently given only multiplicative estimates
of the variables, Hessians, and leverage scores.

With this IPM in hand, the problem of achieving our desired running
time reduces to implementing this IPM efficiently. This problem is
that of maintaining multiplicative approximations to vectors (the
current iterates), leverage scores (a measure of importance of the
rows under local rescaling), and the inverse of matrix (the system
one needs to solve to take a step of the IPM) under small perturbations.
While variants of vector maintenance have been considered recently
\cite{cohen2019solving,LeeSZ19,br2019deterministicArxiv} and inverse
maintenance is well-studied historically \cite{karmarkar1984new,nesterov1989self,vaidya89convexSet,nesterov1991acceleration,leeS14,lee2015efficient,cohen2019solving,LeeSZ19,lsJournal19,br2019deterministicArxiv},
none of these methods can be immediately applied in our setting where
we cannot afford to pay too much in terms of $n$ each iteration.

Our second contribution is to show that these data-structure problems
can be solved efficiently. A key technique we use to overcome these
issues is heavy-hitters sketching. We show that it is possible to
apply a heavy hitters sketch (in particular \cite{knpw11,p13}) to
the iterates of the method such that we can efficiently find changes
in the coordinates. This involves carefully sketching groups of updates
and dynamically modifying the induced data-structure. These sketches
only work against non-adaptive adversaries, and therefore care is
need to ensure that the sketch is used only to save time and not affect
the progression of the overall IPM in a way that break this non-adaptive
assumption. To achieve this we use the sketches to propose short-lists
of possible changes which we then filter to ensure that the output
of the datastructure is deterministic (up to a low probability failure
event). Coupling this technique with known Johnson-Lindestrauss sketches
yields our leverage score maintenance data structure and adapting
and simplifying previous inverse maintenance techniques yields our
inverse maintenance data structure. We believe this technique of sketching
the central path is powerful and may find further applications.

\subsection{Our Results \label{sub:our_results}}

The main result of this paper is the following theorem for solving
(\ref{eq:primal_dual}). This algorithm's running time is nearly linear
whenever the LP is sufficiently dense and tall, i.e. $\nnz(\ma)=\tilde{\Omega}(nd)$
and $n=\tilde{\Omega}(d^{2})$. This is the first polynomial time
algorithm with a nearly linear running time for high-dimensional instances
(i.e. when $d$ can grow polynomially with $n$). This algorithm does
not use fast matrix multiplication (FMM) and consequently its running
time matches the best known running time for checking whether there
even exists $x$ such that $\ma^{\top}x=b$ for dense $\ma$ without
using FMM.

\begin{restatable}[Main Result]{thm}{mainresult}\label{thm:mainresult}
There is an algorithm (Algorithm~\ref{alg:master}) which given any
linear program of the form (\ref{eq:primal_dual}) for non-degenerate\footnote{We assume throughout that $\ma$ is \emph{non-degenerate} meaning
it has full-column rank and no-zero rows. This assumption can be avoided
by preprocessing $\ma$ to check for zero rows and adding a tiny amount
of noise to the matrix to make it full-column rank. There are other
natural ways to remove this assumption, see e.g. \cite{lsJournal19}.} $\ma\in\R^{n\times d}$ and $\delta\in[0,1]$ computes a point $x\in\R_{\geq0}^{n}$
such that
\begin{align*}
c^{\top}x & \leq\min_{\ma^{\top}x=b,x\geq0}c^{\top}x+\delta\cdot\|c\|_{2}\cdot R\text{ and }\\
\|\ma^{\top}x-b\|_{2} & \leq\delta\cdot\Big(\|\ma\|_{F}\cdot R+\|b\|_{2}\Big)
\end{align*}
where $R$ is the diameter of the polytope in $\ell_{2}$ norm, i.e.
$\|x\|_{2}\leq R$ for all $x\in\R_{\geq0}^{n}$ with $\ma^{\top}x=b$.
Further, the expected running time of the method is $O((nd+d^{3})\log^{O(1)}n\log\frac{n}{\delta})$.

\end{restatable}
\begin{rem}
See \cite{renegar1988polynomial,leeS14} on the discussion on converting
such an approximate solution to an exact solution. For integral $\ma,b,c$,
it suffices to pick $\delta=2^{-O(L)}$ to get an exact solution where
$L=\log(1+d_{\max}+\|c\|_{\infty}+\|b\|_{\infty})$ is the bit complexity
and $d_{\max}$ is the largest absolute value of the determinant of
a square sub-matrix of $\ma$. For many combinatorial problems $L=O(\log(n+\|b\|_{\infty}+\|c\|_{\infty}))$.
\end{rem}

Beyond this result (Section~\ref{sec:lp_algorithms}) we believe
our robust primal-dual $\otilde(\sqrt{d})$-iteration IPM (Section~\ref{sec:primal_dual_path})
and our data structures for maintaining multiplicative approximations
to vectors (Section~\ref{sec:vec_maintenance}), leverage scores
(Section~\ref{sec:leverage_maintenance}), and the inverse of matrices
(Section~\ref{sec:inverse_maintenance}) are of independent interest.

\subsection{Previous Work}

Linear programming has been the subject of extensive research for
decades and it is impossible to completely cover to this impressive
line of work in this short introduction. Here we cover results particularly
relevant to our approach. For more detailed coverage of prior-work
see, e.g. \cite{Nesterov1994,ye2011interior}.\textbf{\vspace{8pt}}\\
\textbf{IPMs}: The first proof of a polynomial time IPM was due to
Karmarkar in \cite{karmarkar1984new}. After multiple running time
improvements \cite{karmarkar1984new,nesterov1989self,vaidya89convexSet,nesterov1991acceleration}
the current fastest IPMs are the aforementioned results of \cite{lsJournal19}
and \cite{cohen2019solving}. Beyond these results, excitingly the
work of \cite{cohen2019solving} was recently extended to obtain comparable
running time improvements for solving arbitrary empirical risk minimization
problems (ERM) \cite{LeeSZ19} and was recently simplified and de-randomized
by van den Brand \cite{br2019deterministicArxiv}. These works consider
variants of the vector maintenance problem and our work is inspired
in part by them. Our IPM leverages the barrier from \cite{lsJournal19}
in a new way, that enables the application of robustness techniques
from \cite{lsJournal19,cohen2019solving,br2019deterministicArxiv}
and new techniques for handling approximately feasible points. \textbf{\vspace{8pt}}\\
\textbf{Heavy Hitters and Sketching}: Sketching is a well-studied
problem with a broad range of applications. Johnson-Lindenstrauss
sketches were used extensively in previous IPMs \cite{lsJournal19},
but only for the restricted application of computing leverage scores.
Further sketching techniques have been used for the purpose of dimension
reduction and sampling in other optimization contexts, e.g. solving
linear systems \cite{li2012iterative,clarkson2013low,nelson2012osnap,meng2013low,cohen2015uniform,Cohen16}
and certain forms of $\ell_{p}$-regression \cite{CohenP15}. Our
methods make use of $\ell_{2}$-heavy hitters sketches \cite{cm04,nnw14,ccf02,knpw11,p13,ch09,lnnt16},
in particular the $\ell_{2}$-sketches of \cite{knpw11,p13}, to decrease
iteration costs. We are unaware of these sketches being used previously
to obtain provable improvements for solving offline (as opposed to
online, streaming, or dynamic optimization problems) variants of linear
programming previously. We believe this is due in part to the difficulty
of using these methods with non-oblivious adversaries and consequently
we hope the techniques we use to overcome these issues may be of further
use. Also, note that the definition of $\ell_{2}$-heavy hitters problem
(Lemma~\ref{lem:ell_2_heavy_hitter}) that we use is equivalent to
$\ell_{\infty}/\ell_{2}$-sparse recovery problem. Though the $\ell_{2}/\ell_{2}$-sparse
recovery \cite{d06,crt06,glps10,price13,h16,ns19,song19} is a more
standard task in compressed sensing, we are unaware of how to efficiently
use its guarantees for our applications.\textbf{\vspace{8pt}}\\
\textbf{Leverage Scores and Lewis Weights}: Leverage scores \cite{spielmanS08sparsRes,mahoney11survey,li2012iterative,cohen2015uniform}
(and more broadly Lewis weights \cite{Lewis1978,bourgain1989approximation,CohenP15,lsJournal19})
are fundamental notions of importance of rows of a matrix with numerous
applications. In this paper we introduce a natural online problem
for maintaining multiplicative approximations to leverage scores and
we show how to solve this problem efficiently. Though we are unaware
of this problem being studied previously, we know that in the special
case of leverage scores induced by graph problems, which are known
as effective resistances, there are dynamic algorithms for maintaining
them, e.g \cite{DurfeeGGP19}. However, these algorithms seem tailored
to graph structure and it is unclear how to apply them in our setting.\textbf{
}Further, there are streaming algorithms for variants of this problem
\cite{AhnGM13,KapralovLMMS17,knst19arxiv}, however their running
time is too large for our purposes.\textbf{ \vspace{8pt}}\\
\textbf{Inverse Maintenance}: This problem has been studied extensively
and are a key component in obtaining efficient linear programming
running times with previous IPMs \cite{karmarkar1984new,nesterov1989self,vaidya89convexSet,nesterov1991acceleration,leeS14,lee2015efficient,lsJournal19,cohen2019solving,LeeSZ19,br2019deterministicArxiv,sidford15,lee16,song19}
and other optimization methods \cite{lee2015convex,adil2019iterative}.
Outside the area of optimization, this problem is also known as Dynamic
Matrix Inverse \cite{Sankowski04,BrandNS19}. Our method for solving
inverse maintenance is closely related to these results with modifications
needed to fully take advantage of our leverage score maintenance,
obtain $\poly(d)$ (as opposed to $\poly(n)$ runtimes), ensure that
randomness used to make the algorithm succeed doesn't affect the input
to the data structure, and ultimately produce solvers that are correct
in expectation.

\section{Overview of Approach}

\label{sec:approach}

We prove Theorem~\ref{thm:mainresult} (Section~\ref{sec:lp_algorithms})
in two steps. First we provide a new robust, primal-dual, $\tilde{O}(\sqrt{d})$
iteration IPM inspired by the LS-barriers of \cite{leeS14,lsJournal19}
and central path robustness techniques of \cite{leeS14,cohen2019solving,LeeSZ19,lsJournal19,br2019deterministicArxiv}
(See Section~\ref{sec:primal_dual_path}). Then, we show how to implement
this IPM efficiently using new data-structures based on heavy hitters,
the Johnson-Lindenstrauss lemma, and inverse maintenance (Section~\ref{sec:vec_maintenance},
Section~\ref{sec:leverage_maintenance}, and Section~\ref{sec:inverse_maintenance}).
Here we provide an overview of tools, techniques, and approach to
executing these two steps.

\subsection{A Robust Primal-Dual $\protect\otilde(\sqrt{d})$ IPM}

\label{sec:overview_IPM}

Here we provide an overview of our approach to deriving our robust,
primal-dual, $\tilde{O}(\sqrt{d})$ iteration IPM. We assume throughout
this section (and the bulk of the paper) that $\ma$ is \emph{non-degenerate},
meaning it is full-column rank and no non-zero rows. By standard techniques,
in Section~\ref{sub:main_algorithm} we efficiently reduce solving
(\ref{eq:primal_dual}) in general to solving an instance of where
this assumption holds.

We first give a quick introduction to primal-dual path IPMs, review
the central path from \cite{lsJournal19} and from it derive the central
path that our method is based on, and explain our new IPM.

\paragraph{Primal-Dual Path Following IPMs:}

Our algorithm for solving the linear programs given by (\ref{eq:primal_dual})
is rooted in classic primal-dual path-following IPMs. Primal-dual
IPMs, maintain a \emph{primal feasible point}, $x\in\R_{\geq0}^{n}$
with $\ma^{\top}x=b$, a \emph{dual feasible point} $y\in\R^{d}$
with $s=\ma y-c\geq0$, and attempt to decrease the \emph{duality
gap}
\[
\gap(x,y)\defeq c^{\top}x-b^{\top}y=(\ma y+s)^{\top}x-(\ma^{\top}x)^{\top}y=s^{\top}x\,.
\]
Note that $\gap(x,y)$ upper bounds the error of $x$ and $y$, i.e.
$\gap(x,y)\leq\epsilon$ implies that $x$ and $y$ are each optimal
up to an additive $\epsilon$ in objective function, and therefore
to solve a linear program it suffices to decrease the duality gap
for primal and dual feasible points.

Primal-dual path-following IPMs carefully trade-off decreasing the
duality gap (which corresponds to objective function progress) with
staying away from the inequality constraints, i.e. $x\geq0$ and $s\geq0$
(in order to ensure it is easier to make progress). Formally, they
consider a (\emph{weighted}) \emph{central path} defined as the unique
set of $(x_{\mu},y_{\mu},s_{\mu})\in\R_{\geq0}^{n}\times\R^{d}\times\R_{\geq0}^{n}$
for $\mu>0$ that satisfy
\begin{eqnarray}
\mx_{\mu}\ms_{\mu}\vones & = & \mu\cdot\tauweight(x_{\mu},s_{\mu}),\label{eq:regls_path}\\
\ma^{\top}x_{\mu} & = & b,\nonumber \\
\ma y_{\mu}+s_{\mu} & = & c,\nonumber 
\end{eqnarray}
where $\mx_{\mu}\defeq\mdiag(x_{\mu})$, $\ms_{\mu}\defeq\mdiag(s_{\mu})$,
and $\tauweight:\R_{\geq0}^{n}\times\R_{\geq0}^{n}\rightarrow\R$
is a \emph{weight function}. These methods maintain primal and dual
feasible points and take Newton steps on the above non-linear inequalities
to maintain feasible points that are more central, i.e. have (\ref{eq:regls_path})
closer to holding, for decreasing $\mu$. Since many properties of
the central path can be defined in terms of just $x_{\mu}$ and $s_{\mu}$
we often describe methods using only these quantities and we adopt
the following notation.
\begin{defn}[Feasible Point]
\label{def:feasible_point} We say that $(x,s)\in\R_{\geq0}^{n}\times\R_{\geq0}^{n}$
is a \emph{feasible point }if there exists $y\in\R^{d}$ with $\ma^{\top}x=b$
and $\ma y+s=c$.
\end{defn}

Now, perhaps the most widely-used and simple weight function is $\tauweight(x,s)\leftarrow\taustandard(x,s)\defeq\vones$,
i.e. the all-ones vector. Central path points for this weight function
are the solution to the following
\begin{align*}
x_{\mu} & =\argmin_{x\in\R_{\geq0}^{n}:\ma^{\top}x=b}c^{\top}x-\mu\sum_{i\in[n]}\log(x_{i})\text{ and }\\
(y_{\mu},s_{\mu}) & =\argmax_{(y,s)\in\R^{d}\times\R_{\geq0}^{n}:\ma^{\top}y+s=c}b^{\top}y-\mu\sum_{i\in[n]}\log(s_{i})\,,
\end{align*}
i.e. optimization problems trading off the objective function with
logarithmic barriers on the inequality constraints. There are numerous
methods for following this central path \cite{renegar1988polynomial,mehrota92,gonzaga1992path,Nesterov1994,YTM94}.
By starting with nearly-central feasible points for large $\mu$ and
iteratively finding nearly-central feasible points for small $\mu$,
they can compute $\epsilon$-approximate solutions to (\ref{eq:primal_dual})
in $\otilde(\sqrt{n})$-iterations. For these methods, each iteration
(or Newton step) consists of nearly-linear time computation and solving
one linear system in the matrix $\ma^{\top}\mx\ms^{-1}\ma\in\R^{d\times d}$
for $\mx=\mdiag(x)\in\R^{n\times n}$ and $\ms=\mdiag(s)\in\R^{n\times n}$.

The first such $\otilde(\sqrt{n})$-iteration IPM with this iteration
cost was established by Renegar in 1988 \cite{renegar1988polynomial}
and no faster-converging method with the same iteration cost was developed
until the work of Lee and Sidford in 2014 \cite{lsJournal19,leeS14}.
It had been known since seminal work of Nesterov and Nemirovski in
1994 \cite{Nesterov1994} that there is a weight function that yields
a $\otilde(\sqrt{d})$-iteration primal-dual path-following IPM, however
obtaining any $\otilde(\sqrt{d})$-iteration IPM that can be implemented
with iterations nearly as efficient as those of \cite{renegar1988polynomial}
was a long-standing open problem.

\paragraph{The Lewis Weight Barrier and Beyond}

The work of \cite{lsJournal19,leeS14} addressed this open problem
in IPM theory and provided an efficient $\otilde(\sqrt{d})$-iteration
IPM by providing new weight functions and new tools for following
the central path they induce. In particular, \cite{lsJournal19} introduced
the \emph{Lewis weight barrier} which induces the central path in
which for some $p>0$
\[
\tauweight(x,s)\leftarrow\tauls(x,s)\defeq\sigma^{(p)}(\ms^{-1}\ma)
\]
where for any $\mb\in\R^{n\times d}$, $\sigma^{(p)}(\mb)$ are the
\emph{$\ell_{p}$-Lewis weights of the rows of $\mb$ }\cite{Lewis1978},
a fundamental and natural measure of importance of rows with respect
to the $\ell_{p}$-norm \cite{bourgain1989approximation,CohenP15,lsJournal19}.
In the case of non-degenerate $\mb$, they are defined recursively
as the vector $w\in\R_{>0}^{n}$ which satisfies 
\[
w=\diag(\mw^{(1/2)-(1/p)}\mb(\mb^{\top}\mw^{1-(2/p)}\mb)^{-1}\mb^{\top}\mw^{(1/2)-(1/p)})
\]
where $\mw=\mdiag(w)$. In the special case when $p=2$, 
\[
w=\sigma(\mb)\defeq\diag(\mb(\mb^{\top}\mb)^{\dagger}\mb^{\top})
\]
is known as the \emph{leverage scores }of the rows of $\ma$ and is
a fundamental object for dimension reduction and solving linear systems
\cite{spielmanS08sparsRes,mahoney11survey,li2012iterative,nw14,w14,cohen2015uniform,swz19}.

The work of \cite{lsJournal19} formally showed that there is an $\otilde(\sqrt{d})$-iteration
primal-dual IPM which uses $\tauweight(x,s)\leftarrow\tauls(x,s)$
when $p=\Omega(\log n)$. This choice of $p$ is motivated by drew
geometric connections between Lewis weights and ellipsoidal approximations
of polytopes \cite{lsJournal19}. Further, \cite{lsJournal19} showed
how to modify this IPM to have iterations of comparable cost to the
methods which use $\taustandard(x,s)$. This was achieved by leveraging
and modifying efficient algorithms for approximately computing Lewis
weights and leverage scores and developing techniques for dealing
with the noise such approximate computation induces.

To obtain the results of this paper we further simplify \cite{lsJournal19}
and provide more robust methods for following related central paths.
Our first observation is that the central path induced by $\tauls$
can be re-written more concisely in terms of only leverage scores.
Note that (\ref{eq:regls_path}) and the definition of Lewis weights
imply that there is $w_{\mu}$ with
\begin{align*}
\mx_{\mu}\ms_{\mu}\vones & =\mu\cdot w_{\mu}\text{ and }w_{\mu}=\sigma(\mw_{\mu}^{(1/2)-(1/p)}\ms_{\mu}^{-1}\ma)
\end{align*}
for $\mw_{\mu}\defeq\mdiag(w_{\mu})$. Substituting the first equation
into the second, gives more compactly that for $\alpha=1/p$
\[
\mx_{\mu}\ms_{\mu}\vones=\mu\cdot\mdiag(\sigma(\ms_{\mu}^{-1/2-\alpha}\mx_{\mu}^{1/2-\alpha}\ma))\,.
\]

Consequently, rather than defining centrality in terms of Lewis weights,
we would get the same central path as the one induced by $\tauls(x,s)$
by letting $\tauweight(x,s)\leftarrow\sigma(\ms_{\mu}^{1/2-\alpha}\mx_{\mu}^{1/2-\alpha}\ma)$,
i.e. defining it in terms of leverage scores. In other words, the
optimality of the central path conditions forces $\mx_{\mu}\ms_{\mu}\vones$
to be a type of Lewis weight if we carefully define centrality in
terms of leverage scores. Though $\tauweight(x,s)\leftarrow\sigma(\ms^{1/2-\alpha}\mx^{1/2-\alpha}\ma)$
induces the same central path as $\tauls(x,s)$, these weight functions
can be different outside the central path and thereby lead to slightly
different algorithms if only approximate centrality is preserved.
Further, with the current state-of-the-art theory, leverage scores
are simpler to compute and approximate, as Lewis weight computation
is often reduced to leverage score computation \cite{CohenP15,lsJournal19}.

Formally, in this paper we consider the following regularized-variant
of this centrality measure:
\begin{defn}[Weight Function]
\label{weight_function} Throughout this paper we let $\alpha\defeq1/(4\log(4n/d))$
and for all $x,s\in\R_{>0}^{n}$ let $\taureg(x,s)\defeq\sigma(\ms^{-1/2-\alpha}\mx^{1/2-\alpha}\ma)+\frac{d}{n}\vones$
where $\mx=\mdiag(x)$ and $\ms=\mdiag(s)$.
\end{defn}

This centrality measure is the same as the Lewis weight barrier except
that we add a multiple of the all-ones vector, $\frac{d}{n}\vones$,
to simplify our analysis. Further, this allows us to pick $\alpha\defeq1/(4\log(4n/d))$
as opposed to $\alpha=1/\Omega(\log n)$ due to the extra stability
it provides. Since we use this weight function throughout the paper
we overload notation and let $\tau(x,s)\defeq\taureg(x,s)$ and $\tau(\mb)\defeq\sigma(\mb)+\frac{d}{n}\vones$.

\paragraph{Our Robust Primal-Dual Method:}

We obtain our results by proving that there is an efficient primal-dual
path-following IPM based on $\taureg(x,s)$. We believe that our analysis
is slightly simpler than \cite{lsJournal19} due to its specification
in terms of leverage scores, rather than the more general Lewis weights,
but remark that the core ingredients of its analysis are similar.
Formally, we provide Newton-method type steps that allow us to control
centrality with respect to this measure and increase $\mu$ fast enough
that this yields an $\otilde(\sqrt{d})$-iteration method.

Beyond providing a simplified $\otilde(\sqrt{d})$-iteration IPM,
we leverage this analysis to provide a robust method. Critical to
the development of recent IPMs is that it is possible to design efficient
primal-dual $\otilde(\sqrt{n})$-iteration IPMs that take steps using
only crude, multiplicative approximations to $x$ and $s$ \cite{cohen2019solving,LeeSZ19,br2019deterministicArxiv}.
These papers consider the standard central path but measure centrality
using potential functions introduced in \cite{leeS14,lsJournal19}.
This robustness allows these papers to efficiently implement steps
by only needing to change smaller amounts of coordinates.

Similarly, we show how to apply these approximate centrality measurement
techniques to the central path induced by $\taureg(x,s)$. We show
that it suffices to maintain multiplicative approximations to the
current iterate ($x,s\in\R_{\ge0}^{n}$), the regularized leverage
scores ($\sigma(\ms^{-1/2-\alpha}\mx^{1/2-\alpha}\ma)+\frac{d}{n}\vones$),
and the inverse of the local Hessian ($(\ma^{\top}\ms^{-1}\mx\ma)^{-1}$)
to maintain approximate centrality with respect to $\taureg(x,s)$.
Interestingly, to do this we slightly modify the type of steps we
take in our IPM. Rather than taking standard Newton steps, we slightly
change the steps sizes on $x$ and $s$ to account for the effect
of $\taureg(x,s)$ (see Definition~\ref{def:newton_step}). Further,
approximation of the Hessian causes the $x$ iterates to be infeasible,
but in Section~\ref{sec:maintaining_infeasibility} we discuss how
to modify the steps to control this infeasibility and still prove
the desired theorem.

The main guarantees of this new IPM are given by Theorem~\ref{thm:path_following}
(Section~\ref{sub:ipm_restated}) and proven in Section~\ref{sec:primal_dual_path}
and Section~\ref{sec:maintaining_infeasibility}. This theorem formalizes
the above discussion and quantifies how much the iterates, leverage
scores, and Hessian can change in each iteration. These bounds are
key to obtaining an efficient method that can maintain multiplicative
approximations to these quantities.

\subsection{Heavy Hitters, Congestion Detection, and Sketching the Central Path}

It has long been known that the changes to the central path are essentially
sparse or low-rank on average. In Renegar's IPM \cite{renegar1988polynomial},
the multiplicative change in $x$ and $s$ per iteration are bounded
in $\ell_{2}$. Consequently, the matrices for which linear systems
are solved in Renegar's method do not change too quickly. In fact,
this phenomenon holds for IPMs more broadly, and from the earliest
work on polynomial time IPMs \cite{karmarkar1984new}, to the most
recent fastest methods \cite{leeS14,lee2015efficient,LeeSZ19,lsJournal19,br2019deterministicArxiv},
and varied work in between \cite{nesterov1989self,vaidya1989speeding,nesterov1991acceleration}
this bounded change in IPM iterates has been leveraged to obtain faster
linear programming algorithms.

Our IPM also enjoys a variety of stability properties on its iterates.
We show that the multiplicative change in the iterates are bounded
in a norm induced by $\taureg(x,s)$ and therefore are also bounded
in $\ell_{2}$. This is known to imply that changes to $\ma^{\top}\ms^{-1}\mx\ma\in\R^{d\times d}$
can be bounded over the course of the algorithm and we further show
that this implies that the changes in the weight function, $\taureg(x,s)$,
can also be bounded. These facts, combined with the robustness properties
of our IPM imply that to obtain an efficient linear programming algorithm
it suffices to maintain multiplicative approximations to the following
three quantities (1) the vectors $x,s\in\R_{\geq0}^{n}$, (2) the
regularized leverage scores $\taureg(x,s)$, and (3) the Hessian inverse
$(\ma^{\top}\ms^{-1}\mx\ma)^{-1}\in\R^{d\times d}$ under bounds on
how quickly these quantities change.

We treat each of these problems as a self contained data-structure
problem, the first we call the \emph{vector maintenance problem},
the second we call the \emph{leverage score maintenance problem},
and the third has been previously studied (albeit different variants)
and is called the \emph{inverse maintenance problem}. For each problem
we build efficient solutions by combining techniques form the sketching
literature (e.g. heavy hitters sketches and Johnson-Lindenstrauss
sketches) and careful potential functions and tools for dealing with
sparse and low rank approximations. (Further work is also needed to
maintain the gradient of a potential used to measure proximity to
the central path (Section~\ref{sec:grad_maintenance}) and this is
discussed in the appendix.)

In the remainder of this overview, we briefly survey how we solve
each of these problems. A common issue that needs to addressed in
solving each problem is that of hiding randomness and dealing with
adversarial input. Each data-structure uses sampling and sketching
techniques to improve running times. While these techniques are powerful
and succeed with high probability, they only work against an \emph{oblivious
adversary}, i.e. one which provides input that does not depend on
the randomness of the data structure. However, the output of our data-structures
are used to take steps along the central path and provide the next
input, so care needs to be taken to argue that the output of the data
structure, as used by the method, doesn't somehow leak information
about the randomness of the sketches and samples into the next input.
A key contribution of our work is showing how to overcome this issue
in each case.

\paragraph{Vector Maintenance: }

In the vector maintenance problem, we receive two online sequences
of vectors $h^{(1)},h^{(2)},...\in\R^{d}$ and $g^{(1)},g^{(2)},...\in\R^{n}$
and must maintain the sum $y^{(t)}:\defeq\sum_{i\in[t]}\mg^{(i)}\ma h^{(i)}$
for a fixed matrix $\ma\in\R^{n\times d}$. The naive way of solving
this would just compute $\mg^{(t)}\ma h^{(t)}$ in iteration $t$
and add it to the previous result. Unfortunately, this takes $O(\nnz(A))$
time per iteration, which is too slow for our purposes. Luckily we
do not have to maintain this sum exactly. Motivated by the robustness
of our IPM, it is enough to maintain a multiplicative approximation
$\tilde{y}^{(t)}\approx_{\epsilon}y^{(t)}$ of the sum. An exact definition
of this problem, together with our upper bounds, can be found in Section~\ref{sub:vector_restated}
and the formal proof of these results can be found in Section~\ref{sec:vec_maintenance}.

We now outline how vector maintenance problem can be solved. For simplicity
assume we already have an accurate approximation $\tilde{y}^{(t-1)}\approx_{\epsilon/2}y^{(t-1)}$.
Then we only have to change the entries $i$ of $\tilde{y}^{(t-1)}$
where $y_{i}^{(t)}\not\approx_{\epsilon/2}y_{i}^{(t-1)}$, because
for all other $i$ we have $\tilde{y}_{i}^{(t-1)}\approx_{\epsilon}y_{i}^{(t)}$.
This means we must detect the large entries of $y^{(t)}-y^{(t-1)}=\mg^{(t)}\ma h^{(t)}$,
which can be done via heavy hitter techniques. From past research
on heavy hitters, we know one can construct a small, sparse, random
matrix $\Phi\in\R^{k\times n}$ with $k\ll n$, such that known the
much smaller vector $x\in\R^{k}$ with $x=\Phi y$ allows for a quick
reconstruction of the large entries of $y$. Thus for our task, we
maintain the product $\Phi\mg^{(t)}\ma$, which can be done quickly
if $g^{(t)}$ does not change in too many entries compared to $g^{(t-1)}$.
Then we can reconstruct the large entries of $\mg^{(t)}\ma h^{(t)}$
by computing $(\Phi\mg^{(t)}\ma)h^{(t)}$. Note that this product
can be computed much faster than $\mg^{(t)}\ma h^{(t)}$, because
$\Phi\mg^{(t)}\ma$ is a $k\times d$ matrix and $k\ll n$.

One issue we must overcome, is that the output of our data-structure
must not leak any information about $\Phi$. This matrix is randomly
constructed and the large entries of $y$ can only be reconstructed
from $x=\Phi y$, if the vector $y$ is independent from the randomness
in $\Phi$. Thus if the output of the data-structure depends on $\Phi$
and the next future input $h^{(t)}$ depends on the previous output,
then the required independence can no longer be guaranteed. We overcome
this problem by computing any entry $y_{i}^{(t)}$ exactly, whenever
the heavy hitter techniques detect a large change in said entry. As
we now know the value of said entry exactly, we can compare it to
the previous result and verify, if the entry did indeed change by
some large amount. This allows us to define the output of our data-structure
in a deterministic way, e.g. maintain $\tilde{y}_{i}^{(t)}\approx_{\epsilon}y_{i}^{(t)}$
by setting $\tilde{y}_{i}^{(t)}=y_{i}^{(t)}$ whenever $\tilde{y}_{i}^{(t-1)}\not\approx_{\epsilon}y_{i}^{(t)}$.
(The exact deterministic definition we use is slightly more complex,
but this is the high-level idea.)

So far we only explained how we can detect large changes in $y^{(t)}$
that occur within a single iteration. However, it could be that some
entry changes slowly over several iterations. It is easy to extend
our heavy hitter technique to also detect these slower changes. The
idea is that we not just detect changes within one iteration (i.e.
$y_{i}^{(t)}\not\approx_{\epsilon/2}y_{i}^{(t-1)}$) but also changes
within any power of two (i.e. $y_{i}^{(t)}\not\approx_{\epsilon/2}y_{i}^{(t-2^{i})}$for
$i=1,...,\log t$). To make sure that this does not become too slow,
we only check every $2^{i}$ iterations if there was a large change
within the past $2^{i}$ iterations. One can prove that this is enough
to also detect slowly changing entries.

\paragraph{Leverage Score Maintenance:}

In the leverage score maintenance problem, we must maintain an approximation
of the regularized leverage scores of a matrix of the form $\mg^{1/2}\ma\in\R^{n\times d}$,
for a slowly changing diagonal matrix $\mg\in\R^{n\times n}$. That
means we are interested in a data-structure that maintains a vector
$\tilde{\tau}\in\R^{n}$ with $\tilde{\tau}_{i}\approx_{\epsilon}\tau_{i}:=(\mg^{1/2}\ma(\ma^{\top}\mg\ma)^{-1}\ma^{\top}\mg^{1/2})_{i,i}+d/n$.
The high-level idea for how to solve this task is identical to the
previously outlined vector-maintenance: we try to detect for which
$i$ the value $\tau_{i}$ changed significantly, and then update
$\tilde{\tau}_{i}$ so that the vector stays a valid approximation.
We show that detecting these indices $i$ can be reduced to the previous
vector maintenance. Here we simply outline this reduction, a formal
description of our result can be found in Section~\ref{sub:ls_restated}
and its proof and analysis is in Section~\ref{sec:leverage_maintenance}.

Consider the case where we change from $\mg$ to some $\mg'$. We
want to detect indices $i$ where the $i$th leverage score changed
significantly. Define $\mm\defeq(\ma^{\top}\mg\ma)^{-1}\ma^{\top}\mg^{1/2}-(\ma^{\top}\mg'\ma)^{-1}\ma^{\top}\mg'^{1/2}\in\R^{d\times n}$
and note that 
\begin{align*}
(\mg^{1/2}\ma(\ma^{\top}\mg\ma)^{-1}\ma^{\top}\mg^{1/2})_{i,i} & =\|e_{i}^{\top}\mg^{1/2}\ma(\ma^{\top}\mg\ma)^{-1}\ma^{\top}\mg^{1/2}\|_{2}^{2},
\end{align*}
so a large change in the $i$th leverage score, when changing $\mg\in\R^{n\times n}$
to $\mg'\in\R^{n\times n}$, results in a large norm of $e_{i}^{\top}\mg^{1/2}\ma\mm$,
if $\mg_{i,i}$ and $\mg'_{i,i}$ are roughly the same. (As $\mg$
is slowly changing there are not too many $i$ where $\mg_{i,i}$
and $\mg'_{i,i}$ significantly. So we can just compute the $i$th
leverage score exactly for these entries to check if the $i$th score
changed significantly.) By multiplying this term with a Johnson-Lindenstrauss
matrix $\mj$ the task of detecting a large leverage score change,
becomes the task of detecting rows of $\mg^{1/2}\ma\mm\mj$ for which
the row-norm changed significantly. Given that $\mj$ only needs some
$O(\log n)$ columns to yield a good approximation of these norms,
we know that any large change in the $i$th leverage score, must result
in some index $j$ where $|(\mg^{1/2}\ma\mm\mj)_{i,j}|$ must be large.
Detecting these large entries can be done in the same way as in the
vector-maintenance problem by considering each column of $\mm\mj$
as a vector.

To make sure that no information about the random matrix $\mj$ is
leaked, we use the same technique previously outlined in the vector-maintenance
paragraph. That is, after detecting a set $I\subset[n]$ of indices
$i$ for which the leverage score might have changed significantly,
we compute the $i$th leverage score to verify the large change and
set $\tilde{\tau}_{i}$ to be this computed leverage score, if the
change was large enough. Unlike the vector case however, the $i$th
leverage score is not computed in a deterministic way (as this would
be prohibitively expensive). Instead we use another random Johanson-Lindenstrauss
matrix $\mj'$, so the output $\tilde{\tau}$ is actually defined
w.r.t the input and this new matrix $\mj'$. By using a fresh independent
Johnson-Lindenstrauss $\mj'$ to verify changes to leverage scores,
this data structure works against an adaptive adversary.

\paragraph{Inverse Maintenance:}

In the inverse maintenance problem, we maintain a spectral sparsifier
of $\ma^{\top}\mw\ma\in\R^{d\times d}$ and its inverse, for a slowly
changing diagonal matrix $\mw\in\R^{n\times n}$. Using the leverage
score data-structure, we can assume an approximation to the leverage
scores of $\ma^{\top}\mw\ma$ is given. Hence, we can sample $\otilde(d)$
many rows of $\ma$ to form a spectral sparsifier. This allows us
to get a spectral sparsifier and its inverse in $\otilde(d^{\omega})$.
To speed up the runtime, we follow the idea in \cite{lee2015efficient},
which resamples the row only if the leverage score changed too much.
This makes sure the sampled matrix is slowly changing in $\ell_{0}$
sense and hence we can try to apply the lazy low-rank update idea
in \cite{cohen2019solving} to update the inverse in $\otilde(d^{2})$
time. Unfortunately, this algorithm works only against oblivious adversary
and the sampled matrix is changing too fast in $\ell_{2}$ sense,
which is required for the leverage score maintenance.

Fortunately, we note that we do not need a sparsifier that satisfies
both conditions at all time. Therefore, our final data structure has
two ways to output the inverse of the sparsifier. The sparsifier that
works only against oblivious adversary is used in implementing the
Newton steps and the sparsifier that is slowly changing is used to
compute the sketch $\mm\mj$ used in the leverage score maintenance
problem mentioned above.

For the Newton steps, we only need to make sure the sparsifier does
not leak the randomness since the input $\mw$ depends on the Newton
step. Since we only need to solve linear systems of the form $\ma^{\top}\mw\ma x=b\in\R^{d}$,
we can handle this problem by adding an appropriate noise to the output
$x$. This makes sure the randomness we use in this data structure
does not leak when the output $x$ is used. This idea is also used
\cite{lee2015efficient}, but extra care is needed to remove the $\nnz(\ma)$
per step cost in their algorithm.

For computing the sketch $\mm\mj$, we do not need to worry about
leakage of randomness, but only need to make sure it is slowly changing
in $\ell_{2}$ sense. Instead of using \cite{cohen2019solving} as
a black-box, we show how to combing the idea of resampling in \cite{lee2015efficient}
and the low-rank update in \cite{cohen2019solving}. This gives us
an alternate smoother scheme that satisfies the requirement for slowly
changing.

\section{Preliminaries\label{sec:prelim}}

Here we discuss varied notation we use throughout the paper. We adopt
similar notation to \cite{leeS14} and some of the explanations here
are copied directly form this work.\textbf{\vspace{10pt}}\\
\textbf{Matrices}: We call a matrix $\ma$\emph{ non-degenerate} if
it has full column-rank and no zero rows.\textbf{ }We call symmetric
matrix $\mb\in\Rnn$ positive semidefinite (PSD) if $x^{\top}\mb x\geq0$
for all $x\in\Rn$ and positive definite (PD) if $x^{\top}\mb x>0$
for all $x\in\Rn$.\textbf{\vspace{10pt}}\\
\textbf{Matrix Operations:} For symmetric matrices $\ma,\mb\in\Rnn$
we write $\ma\preceq\mb$ to indicate that $x^{\top}\ma x\leq x^{\top}\mb x$
for all $x\in\Rn$ and define $\prec$, $\preceq$, and $\succeq$
analogously. For $\ma,\mb\in\R^{n\times m}$, we let $\ma\circ\mb$
denote the Schur product, i.e. $[\ma\circ\mb]_{i,j}\defeq\ma_{i,j}\cdot\mb_{i,j}$
for all $i\in[n]$ and $j\in[m]$. We use $\nnz(\ma)$ to denote the
number of nonzero entries in $\ma$.\textbf{\vspace{10pt}}\\
\textbf{Diagonals:} For $\ma\in\R^{n\times n}$ we define $\diag(\ma)\in\R^{n}$
with $\diag(\ma)_{i}=\ma_{ii}$ for all $i\in[n]$ and for $x\in\R^{n}$
we define $\mdiag(x)\in\R^{n\times n}$ as the diagonal matrix with
$\diag(\mdiag(x)(x))=x$. We often use upper case to denote a vectors
associated diagonal matrix and in particular let $\mx\defeq\mdiag(x)$,
$\ms\defeq\mdiag(s)$, $\mw\defeq\mdiag(w)$, $\mT\defeq\mdiag(\tau)$,
$\omx\defeq\mdiag(\ox)$, $\oms\defeq\mdiag(\overline{s})$, $\omw\defeq\mdiag(\ow)$,
$\mx_{t}\defeq\mdiag(x_{t})$, $\ms_{t}=\mdiag(s_{t})$, $\mw_{t}=\mdiag(w_{t})$,
and $\mT_{t}\defeq\mdiag(\tau_{t})$.\textbf{\vspace{10pt}}\\
\textbf{Fundamental Matrices}: For any non-degenerate matrix $\ma\in\R^{n\times d}$
we let $\mproj(\ma)\defeq\ma(\ma^{\top}\ma)^{-1}\ma^{\top}$ denote
the orthogonal projection matrix onto $\ma$'s image. Further, we
let $\sigma(\ma)\defeq\diag(\mproj(\ma))$ denote $\ma$'s \emph{leverage
scores} and we let \emph{$\tau(\ma)\defeq\sigma(\ma)+\frac{d}{n}\vones$}
denote its regularized leverage scores. Further, we define $\mSigma(\ma)\defeq\mdiag(\sigma(\ma))$,
$\mT(\ma)\defeq\mdiag(\tau(\ma))$, $\mproj^{(2)}(\ma)\defeq\mproj(\ma)\circ\mproj(\ma)$
(where $\circ$ denotes entrywise product), and $\mLambda(\ma)\defeq\mSigma(\ma)-\mproj^{(2)}(\ma)$.\textbf{\vspace{10pt}}\\
\textbf{Approximations}: We use $x\approx_{\epsilon}y$ to denote
that $\exp(-\epsilon)y\leq x\leq\exp(\epsilon)y$ and $\ma\approx_{\epsilon}\mb$
to denote that $\exp(-\epsilon)\mb\preceq\ma\preceq\exp(\epsilon)\mb$.\textbf{\vspace{10pt}}\\
\textbf{Norms}: For PD $\ma\in\Rnn$ we let $\|\cdot\|_{\ma}$ denote
the norm where $\norm x_{\ma}^{2}\defeq x^{\top}\ma x$ for all $x\in\Rn$.
For positive $w\in\R_{>0}^{n}$ we let $\|\cdot\|_{w}$ denote the
norm where $\norm x_{w}^{2}\defeq\sum_{i\in[n]}w_{i}x_{i}^{2}$ for
all $x\in\Rn$. For any norm $\|\cdot\|$ and matrix $\mm$, its induced
operator norm of $\mm$ is defined by $\norm{\mm}=\sup_{\|x\|=1}\norm{\mm x}$.\textbf{\vspace{10pt}}\\
\textbf{Time and Probability:} We use $\otilde(\cdot)$ to hide factors
polylogarithmic in $n$ and $d$. We say an algorithm has a property
``with high probability (w.h.p.) in $n$'' if it holds with probability
at least $1-1/O(\poly(n))$ for any polynomial by choice of the constants
in the runtime of the algorithm.\textbf{\vspace{10pt}}\\
\textbf{Misc:} We let $[z]\defeq\{1,2,..,z\}$. We let $\vones_{n},\vzero_{n}\in\R^{n}$
denote the all-one and all-zero vectors, $\mzero_{n},\mi_{n}\in\R^{n\times n}$
denote the all zero and identity matrices, and drop subscripts when
the dimensions are clear. We let $\indicVec i$ denote the indicator
vector for coordinate $i$, i.e. the $i$-th basis vector.

\section{Linear Programming Algorithm\label{sec:lp_algorithms}}

Here we prove the main result of this paper that there is a $\tilde{O}(nd+d^{3})$
time algorithm for solving linear programs. This theorem is restated
below for convenience:

\mainresult*

The proof for Theorem~\ref{thm:mainresult} uses four intermediate
results, which we formally state in the next Sections~\ref{sub:ipm_restated},
to \ref{sub:inverse_restated}. Each of these intermediate results
is self-contained and analyzed in its own section, Sections~\ref{sec:primal_dual_path}
to \ref{sec:inverse_maintenance} correspondingly. In this section
we show how these results can be combined to obtain Theorem~\ref{thm:mainresult}.
The first result is a new, improved IPM as outlined in Section~\ref{sec:overview_IPM}.
The exact statement is given in Section~\ref{sub:ipm_restated} and
its analysis can be found in Section~\ref{sec:primal_dual_path}
and Section~\ref{sec:maintaining_infeasibility}. Here we give a
rough summary to motivate the other three results used by our linear
programming algorithm.

Our IPM is robust in the sense, that it makes progress, even if we
maintain the primal dual solution pair $(x,s)$ only approximately.
Additionally, the linear system that is solved in each iteration,
allows for spectral approximations. More accurately, it is enough
to maintain a spectral approximation of an inverse of a matrix of
the form $\ma^{\top}\mw\ma$ for some diagonal matrix $\mw$. For
this robust IPM we also require to compute approximate leverage scores,
which allows the IPM to converge in just $\tilde{O}(\sqrt{d})$ iterations.
These properties of our IPM motivate three new data-structures:

(i) In Section~\ref{sub:vector_restated}, we present a data-structure
that can maintain an approximation of the primal dual solution pair
$(x,s)$ efficiently. More formally, this data-structure maintains
an approximation of the sum $\sum_{k\in[t]}\mw^{(k)}\ma h^{(k)}$
for diagonal matrices $\mw^{(k)}\in\R^{n\times n}$ and vectors $h^{(k)}\in\R^{d}$.
The data-structure is proven and analyzed in Section \ref{sec:vec_maintenance}.

(ii) In Section~\ref{sub:ls_restated} we present a data-structure
that can main approximate leverage scores, required by the IPM. Its
correctness is proven in Section~\ref{sec:leverage_maintenance}.

(iii) The last requirement of the IPM is for us to maintain a spectral
approximation of $(\ma^{\top}\mw\ma)^{-1}$ for some diagonal matrix
$\mw$. We present a data-structure that can maintain this inverse
approximately in $\otilde(d^{\omega-\frac{1}{2}}+d^{2})$ amortized
time per step, when the matrix $\mw$ changes slowly w.r.t $\ell_{2}$-norm
and if we have estimates of the leverage scores of $\mw^{1/2}\ma$.
The exact result is stated in Section~\ref{sub:inverse_restated}
and proven in Section \ref{sec:inverse_maintenance}.

With this we have all tools available for proving the main result
Theorem~\ref{thm:mainresult}. In Section~\ref{sub:main_algorithm}
we show how to combine all these tools to obtain the fast linear programming
algorithm.

\subsection{Interior Point Method\label{sub:ipm_restated}}

In Section~\ref{sec:primal_dual_path} we derive and analyze the
core subroutine of our primal-dual robust $\otilde(\sqrt{d})$-iteration
IPM (Algorithm~\ref{alg:pathfollowing}). This subroutine takes an
approximate central point for parameter $\mu^{\textrm{(init)}}$ and
in $\otilde(\sqrt{d}\log(\mu^{\textrm{(target)}}/\mu^{\textrm{(init)}}))$
iterations outputs an approximate central path point for any given
parameter $\mu^{\textrm{(target)}}$. Here we simply state the routine
(Algorithm~\ref{alg:pathfollowing}) and the main theorem regarding
its performance (Theorem~\ref{thm:path_following}). We defer the
full motivation of the method and its analysis (i.e. the proof of
Theorem~\ref{thm:path_following}) to Section~\ref{sec:primal_dual_path}.
In the remainder of this section we argue how with the appropriate
data-structures, this theorem implies our main result.

\begin{algorithm2e}[h]

\caption{Path Following (Theorem~\ref{thm:path_following}) \label{alg:pathfollowing}}

\SetKwProg{Proc}{procedure}{}{}

\Proc{$\textsc{Centering}$$(x^{\textrm{\ensuremath{\mathrm{(init)}}}}\in\R_{>0}^{n},s^{\textrm{\ensuremath{\mathrm{(init)}}}}\in\R_{>0}^{n},\mu^{\textrm{\ensuremath{\mathrm{(init)}}}}>0,\mu^{\textrm{\ensuremath{\mathrm{(target)}}}}>0,\epsilon>0)$}{

\State $\alpha\leftarrow1/(4\log(4n/d))$, $\lambda\leftarrow\frac{2}{\epsilon}\log(\frac{2^{16}n\sqrt{d}}{\alpha^{2}})$,
$\gamma\leftarrow\min(\frac{\epsilon}{4},\frac{\alpha}{50\lambda})$

\State $\mu\leftarrow\mu^{\textrm{(init)}}$, $x\leftarrow x^{\textrm{(init)}}$,
$s\leftarrow s^{\textrm{(init)}}$

\While{$\textsc{\ensuremath{\mathsf{true}}}$}{

\State Pick any $\ox,\os\in\Rn_{>0}$ such that $\ox\approx_{\epsilon}x$,
$\os\approx_{\epsilon}s$\label{line:pf_1}

\State Find $\ov$ such that $\|\ov-v\|_{\infty}\leq\gamma$ where
$v=\mu\cdot\tau(x,s)/w$\label{line:pf_2}

\State Let $\Phi(v)\defeq\sum_{i=1}^{n}\exp(\lambda(v_{i}-1))+\exp(-\lambda(v_{i}-1))$
for all $v\in\R^{n}$

\lIf{$\mu=\mu^{\textrm{\ensuremath{\mathrm{(target)}}}}$ and $\Phi(\ov)\leq\frac{2^{16}n\sqrt{d}}{\alpha^{2}}$}{\State \textbf{break}}

\State $h=\gamma\nabla\Phi(\overline{v})^{\flat(\otau)}$ for $\otau\approx_{\epsilon}\tau(x,s)$
\tcp*{ See Lemma~\ref{lem:smoothing:helper} for definition of $(\cdot){}^{\flat}$}

\State Pick any $\mh\in\R^{d\times d}$ with $\mh\approx_{(c\epsilon)/(d^{1/4}\log^{3}n)}\ma^{\top}\oms^{-1}\omx\ma$
and $\E[\mh]=\ma^{\top}\oms^{-1}\omx\ma$ for small constant $c>0$.\label{line:approx_H}

\State Let $\mq=\overline{\ms}^{-1/2}\overline{\mx}^{1/2}\ma\mh^{-1}\ma^{\top}\omx^{1/2}\oms^{-1/2}$,
$\omw=\omx\oms$

\State \textbf{$x\leftarrow x+(1+2\alpha)\omx\omw^{-1/2}(\mi-\mq)\omw^{1/2}h$}\label{line:pf_3}\label{line:x_formula}

\State $s\leftarrow s+(1-2\alpha)\oms\omw^{-1/2}\mq\omw^{1/2}h$\label{line:pf_4}\label{line:s_formula}

\State Pick any $\ox^{\textrm{(new)}},\os^{\textrm{(new)}},\overline{\tau}^{\textrm{(new)}}\in\Rn_{>0}$
with $\ox^{\textrm{(new)}}\approx_{\epsilon}x$, $\os^{\textrm{(new)}}\approx_{\epsilon}s$,
$\overline{\tau}^{\textrm{(new)}}\approx_{1}\tau(x,s)$

\State $\delta_{\lambda}\leftarrow\textsc{MaintainFeasibility}(x,s,\overline{\tau}^{\textrm{(new)}})$
\tcp*{Algorithm \ref{alg:maintain_INfeasibility}}\label{line:pf_maintainFeasibility}

\State $x\leftarrow x+\omx^{\text{(new)}}\left(\oms^{\text{(new)}}\right){}^{-1}\ma\delta_{\lambda}$\label{line:pf_move_x_improve_feasible}

\lIf{$\mu>\mu^{\textrm{\ensuremath{\mathrm{(target)}}}}$}{\State $\mu\leftarrow\max\{\mu^{\textrm{(target)}},(1-\frac{\gamma\alpha}{2^{15}\sqrt{d}})\mu\}$}

\lElseIf{$\mu<\mu^{\textrm{\ensuremath{\mathrm{(target)}}}}$}{\State $\mu\leftarrow\min\{\mu^{\textrm{(target)}},(1+\frac{\gamma\alpha}{2^{15}\sqrt{d}})\mu\}$}

}

\State \Return $(x,s)$

}

\end{algorithm2e}
\begin{thm}
\label{thm:path_following} There exists a constant $\zeta>0$ such
that, given $x^{\textrm{\ensuremath{\mathrm{(init)}}}},s^{\textrm{\ensuremath{\mathrm{(init)}}}}\in\Rn_{>0}$,
$\mu^{\textrm{\ensuremath{\mathrm{(init)}}}}>0$, $\mu^{\textrm{\ensuremath{\mathrm{(target)}}}}>0$,
and $\epsilon\in(0,\alpha/16000)$ with $x^{\textrm{\ensuremath{\mathrm{(init)}}}}s^{\textrm{\ensuremath{\mathrm{(init)}}}}\approx_{2\epsilon}\mu^{\textrm{\ensuremath{\mathrm{(init)}}}}\cdot\tau(x^{\textrm{\ensuremath{\mathrm{(init)}}}},s^{\textrm{\ensuremath{\mathrm{(init)}}}})$
and
\begin{align*}
\frac{1}{\mu^{\mathrm{(init)}}}\|\ma & x^{\textrm{\ensuremath{\mathrm{(init)}}}}-b\|_{(\ma^{\top}\mx^{\mathrm{(init)}}\ms^{\mathrm{(init)}-1}\ma)^{-1}}^{2}\le\frac{\zeta\epsilon^{2}}{\log^{\power}n}
\end{align*}
Algorithm \ref{alg:pathfollowing} outputs $(x^{\mathrm{(final)}},s^{\mathrm{(final)}})$
such that
\begin{align*}
x^{\mathrm{(final)}}s^{\mathrm{(final)}}\approx_{\epsilon}\mu^{\mathrm{(target)}}\cdot\tau(x^{\mathrm{(final)}},s^{\mathrm{(final)}}) & \qquad\text{and}\\
\frac{1}{\mu^{\mathrm{(target)}}}\|\ma x^{\textrm{\ensuremath{\mathrm{(final)}}}}-b\|_{(\ma^{\top}\mx^{\mathrm{(final)}}\ms^{\mathrm{(final)}-1}\ma)^{-1}}^{2}\le\frac{\zeta\epsilon^{2}}{\sqrt{d}\log^{\power}n} & \qquad\\
\text{in }O\left(\sqrt{d}\log(n)\cdot\left(\frac{1}{\epsilon\alpha}\cdot\log\left(\frac{\mu^{\textrm{\ensuremath{\mathrm{(target)}}}}}{\mu^{\textrm{\ensuremath{\mathrm{(init)}}}}}\right)+\frac{1}{\alpha^{3}}\right)\right) & \text{ iterations.}
\end{align*}
 Furthermore, throughout Algorithm~\ref{alg:pathfollowing}, we have
\begin{itemize}
\item $xs\approx_{4\epsilon}\mu\cdot\tau(x,s)$ for some $\mu$ where $(x,s)$
is immediate points in the algorithms
\item $\|\mx^{-1}\delta_{x}\|_{\tau+\infty}\leq\frac{\epsilon}{2}$, $\|\ms^{-1}\delta_{s}\|_{\tau+\infty}\leq\frac{\epsilon}{2}$,
and $\|\diag(\tau)^{-1}\delta_{\tau}\|_{\tau+\infty}\leq2\epsilon$
where $\delta_{x}$, $\delta_{s}$ and $\delta_{\tau}$ is the change
of $x$, $s$ and $\tau$ in one iteration and $\norm x_{\tau+\infty}\defeq\norm x_{\infty}+\cnorm\norm x_{\tau}$
for $\cnorm\defeq\frac{10}{\alpha}$ (See Section~\ref{sec:weight_changes}).
\end{itemize}
\end{thm}

\begin{rem}
We will take $\epsilon=1/\poly\log(n)$.
\end{rem}

Note that the complexity of Theorem \ref{thm:path_following} depends
on the cost of implementing Lines \ref{line:pf_1}, \ref{line:pf_2},
\ref{line:pf_3}, and \ref{line:pf_4} of Algorithm \ref{alg:pathfollowing},
as well as the cost of Algorithm \ref{alg:maintain_INfeasibility}.
Here Lines \ref{line:pf_1}, \ref{line:pf_3} and \ref{line:pf_4}
ask us to maintain an approximation of the primal dual solution pair
$(x,s)$. A data-structure for this task is presented in Section \ref{sub:vector_restated}.
Additionally, to compute Lines \ref{line:pf_3} and \ref{line:pf_4},
we must have access to an approximate inverse of $\ma^{\top}\oms^{-1}\omx\ma$
(see Line \ref{line:approx_H}). The task of maintaining this inverse
will be performed by the data-structure presented in Section \ref{sub:inverse_restated}.
At last, consider Line~\ref{line:pf_2}. To implement this line,
we must have an approximation of the leverage scores $\tau(x,s)$.
In Section~\ref{sub:ls_restated}, we present a data-structure that
can efficiently maintain such an approximation.

To help us analyze the cost of Algorithm~\ref{alg:maintain_INfeasibility},
we prove the following in Section~\ref{sec:maintaining_infeasibility}. 

\begin{restatable}[Maintain Feasibility]{thm}{maintainFeasibility}\label{thm:feasibilityComplexity}

The additional amortized cost of calling $\textsc{MaintainFeasibility}$
in Line~\ref{line:pf_maintainFeasibility} of Algorithm~\ref{alg:pathfollowing}
is $\otilde(nd^{0.5}+d^{2.5}/\epsilon^{2})$ per call, plus the cost
of querying $\otilde(n/\sqrt{d}+d^{1.5}/\epsilon^{2})$ entries of
$x$ and $s$ (assuming $x$, $s$ are given implicitly, e.g. via
some data structure).

\end{restatable}

\subsection{Vector Data Structure\label{sub:vector_restated}}

Consider an online sequence of $n\times n$ diagonal matrices $\mg^{(1)},\mg^{(2)},...,\mg^{(T)}\in\R^{n\times n}$
and vectors $h^{(1)},...,h^{(T)}\in\R^{d}$, $\delta^{(1)},...,\delta^{(T)}\in\R^{n}$
and define $y^{(t+1)}:=\sum_{k=1}^{t}\mg^{(k)}\ma h^{(k)}+\delta^{(k)}$.
In this subsection we describe a data-structure that can efficiently
maintain an approximation $\bar{y}^{(t)}\approx_{\varepsilon}y^{(t)}$,
when the relative changes $\|(\my^{(k)})^{-1}\mg^{(k)}\ma h^{(k)}\|_{2}$
and $\|(\my^{(k)})^{-1}\delta^{(k)}\|_{2}$ are small. This is motivated
by the following requirement of our IPM: we must maintain a multiplicative
approximation of a sequence of vectors $x^{(t)}$, $s^{(t)}\in\R^{n}$
(see Line \ref{line:pf_1} of Algorithm \ref{alg:pathfollowing}),
where $x^{(k+1)}=x^{(k)}+\delta_{x}^{(k)}$, $s^{(k+1)}=s^{(k)}+\delta_{s}^{(k)}$
and the terms $\delta_{x}^{(k)}$ and $\delta_{s}^{(k)}$are roughly
of the form (see Lines \ref{line:pf_3} and \ref{line:pf_4} of Algorithm
\ref{alg:pathfollowing}):
\begin{align*}
\delta_{x}^{(k)} & =(1+2\alpha)\omx^{(k)}\left(\omw^{(k)}\right)^{-1/2}(\mi-\mq^{(k)})\left(\omw^{(k)}\right)^{1/2}v^{(k)},\\
\delta_{s}^{(k)} & =(1-2\alpha)\oms^{(k)}\left(\omw^{(k)}\right)^{-1/2}\mq^{(k)}\left(\omw^{(k)}\right)^{1/2}v^{(k)}.
\end{align*}
To maintain an approximation of $x^{(t)}$, we can then use the data-structure
for maintaining an approximation of $y^{(t)}$ by choosing
\begin{align*}
\mg^{(k)} & =(1+2\alpha)\omx^{(k)}\left(\omw^{(k)}\right)^{-1/2},\\
h^{(k)} & =-\mq^{(k)}\left(\omw^{(k)}\right)^{1/2}v^{(k)},\\
\delta^{(k)} & =\left(\omw^{(k)}\right)^{1/2}v^{(k)}.
\end{align*}

Likewise, we can maintain an approximation of $s^{(t)}$ by a slightly
different choice of parameters. The exact result, which we prove in
Section~\ref{sec:vec_maintenance}, is the following Theorem~\ref{thm:product_sum_datastructure}:

\begin{restatable}[Vector Maintenance]{thm}{productsumdatastructure}\label{thm:product_sum_datastructure}There
exists a Monte-Carlo data-structure (Algorithm \ref{alg:product_sum_datastructure}),
that works against an adaptive adversary, with the following procedures:
\begin{itemize}
\item \textsc{Initialize($\ma,g,x^{(0)},\epsilon$):} Given matrix $\ma\in\R^{n\times d}$,
scaling $g\in\R^{n}$, initial vector $x^{(0)}$, and target accuracy
$\epsilon\in(0,1/10)$, the data-structure preprocesses in $O(\nnz(\ma)\log^{5}n)$
time.
\item \textsc{Scale($i,u$):} Given $i\in[n]$ and $u\in\R$ sets $g_{i}=u$
in $O(d\log^{5}n)$ amortized time.
\item \textsc{Query($h^{(t)},\delta\t$):} Let $g\t\in\R^{n}$ be the scale
vector $g\in\R^{n}$ during $t$-th call to \textsc{Query} and let
$h\t\in\R^{d},\delta\t\in\R^{n}$ be the vectors given during that
query. Define
\[
x\t=x^{(0)}+\sum_{k\in[t]}\mg\k\ma h\k+\sum_{k\in[t]}\delta\k.
\]
Then, w.h.p. in $n$ the data-structure outputs a vector $y\in\R^{n}$
such that $y\approx_{\epsilon}x\t$. Furthermore, the total cost over
$T$ steps is
\[
O\left(T\left(n\log n+\sum_{k\in[T]}\left(\left\Vert (\mx\k)^{-1}\mg\k\ma h\k\right\Vert _{2}^{2}+\|(\mx\k)^{-1}\delta\k\|_{2}^{2}\right)\cdot\varepsilon^{-2}\cdot d\log^{6}n\right)\right).
\]
\item \textsc{ComputeExact($i$):} Output $x_{i}^{(t)}\in\R^{n}$ exactly
in amortized time $O(d\log n)$.
\end{itemize}
\end{restatable}

\subsection{Leverage Score Maintenance\label{sub:ls_restated}}

The IPM of Theorem \ref{thm:path_following}, requires approximate
leverage scores of some matrix of the form $\mg\ma$, where $\mg$
is a diagonal matrix (see Line~\ref{line:pf_2} of Algorithm~\ref{alg:pathfollowing},
where $\mg=(\omx/\oms)^{1/2}$) . Here the matrix $\mg$ changes slowly
from one iteration of the IPM to the next one, which allows us to
create a data-structure that can maintain the scores more efficiently
than recomputing them from scratch every time $\mg$ changes. In Section~\ref{sec:leverage_maintenance}
we prove the following result for maintaining leverage scores:

\begin{restatable}[Leverage Score Maintenance]{thm}{leveragescoredatastructure}\label{thm:leverage_score_datastructure}
There exists a Monte-Carlo data-structure (Algorithm \ref{alg:leverage_score_datastructure_1}),
that works against an adaptive adversary, with the following procedures:
\begin{itemize}
\item \textsc{Initialize($\ma,g,\epsilon$):} Given matrix $\ma\in\R^{n\times d}$,
scaling $g\in\R^{n}$ and target accuracy $\epsilon>0$, the data-structure
preprocesses in $O(nd\epsilon^{-2}\log^{4}n)$ time.
\item \textsc{Scale($i,u$):} Given $i\in[n]$ and $u\in\R$ sets $g_{i}=u$
in $O(d\epsilon^{-2}\log^{5}n)$ time.
\item \textsc{Query($\Psi\t,\Psi\t\safe$):} Let $g\t$ be the vector $g$
during $t$-th call to $\textsc{Query}$, assume $g^{(t)}\approx_{1/16}g^{(t-1)}$
and define $\mh^{(t)}=\ma^{\top}(\mg\t)^{2}\ma$. Given random input-matrices
$\Psi\t\in\R^{d\times d}$ and $\Psi\t\safe\in\R^{d\times d}$ such
that 
\[
\Psi\t\approx_{\epsilon/(24\log n)}(\mh\t)^{-1},\Psi\t\safe\approx_{\epsilon/(24\log n)}(\mh\t)^{-1}.
\]
and any randomness used to generate $\Psi\t\safe$ is independent
of the randomness used to generate $\Psi\t$, w.h.p. in $n$ the data-structure
outputs a vector $\ttau\in\Rn$ independent of $\Psi^{(1)},...,\Psi\t$
such that $\ttau_{i}\approx_{\epsilon}\tau_{i}(\mg\t\ma)$ for all
$i\in[n]$. Furthermore, the total cost over $T$ steps is 
\[
O\left(\left(\sum_{t\in[T]}\|\mg^{(t)}\ma\Psi^{(t)}\ma^{\top}\mg^{(t)}-\mg^{(t-1)}\ma\Psi^{(t-1)}\ma^{\top}\mg^{(t-1)}\|_{F}\right)^{2}\cdot\text{\ensuremath{\epsilon^{-4}}}n\log^{7}n+T\left(T_{\Psi}+\epsilon^{-2}d^{2}\log^{3}n\right)\right)
\]
where $T_{\Psi}$ is the time required to multiply a vector with $\Psi^{(t)}$
(i.e. in case it is given implicitly via a data structure).
\end{itemize}
\end{restatable}

\subsection{Inverse Maintenance\label{sub:inverse_restated}}

For the IPM we must approximately maintain the inverse $(\ma^{\top}\mw\ma)^{-1}$
where $\ma\in\R^{n\times d}$ undergoes changes to the diagonal matrix
$\mw$ (see Line \ref{line:approx_H} of Algorithm \ref{alg:pathfollowing}
where $\mw=\oms^{-1}\omx$). By using estimates of the leverage scores
of $\mw^{1/2}\ma$ (as maintained by Theorem \ref{thm:leverage_score_datastructure},
Section \ref{sub:ls_restated}), we are able to maintain the inverse
in amortized $\otilde(d^{\omega-\frac{1}{2}}+d^{2})$ time per step,
even for $n\gg d$. The exact result is stated as Theorem \ref{thm:inverse_main}
and proven in Section \ref{sec:inverse_maintenance}.

\begin{restatable}[Inverse Maintenance]{thm}{inversemain}\label{thm:inverse_main}
Given a full rank matrix $\ma\in\R^{n\times d}$ with $n\geq d$ and
error tolerance $\epsilon\in(0,1/10)$, there is a data structure
that approximately solves a sequence of linear systems $\ma^{\top}\tmw\ma y=b\in\R^{d}$
for positive diagonal matrices $\tmw\in\R^{n\times n}$ through the
following operation:
\begin{itemize}
\item \textsc{Initialize($\ma,\tw,\ttau,\epsilon$):} Given matrix $\ma\in\R^{n\times d}$,
scaling $\tw\in\R_{>0}^{n}$, shifted leverage score estimates $\ttau\in\R_{>0}^{n}$,
and accuracy $\epsilon\in(0,1/10)$, the data-structure preprocesses
in $O(d^{\omega})$ time.
\item $\textsc{Update}(\tw,\ttau)$: Output a matrix $\Psi\in\R^{d\times d}$
and vector $\tw^{\alg}$ where $\Psi^{-1}$ is close to $\ma^{\top}\tmw\ma\in\R^{d\times d}$
and $\tw^{\alg}$ is close to $\tw$.
\item $\textsc{Solve}(b,\ow,\delta)$: Input $\ow\approx_{1}\tw$ and $\delta>0$,
output $y=\Psi b\in\R^{d}$ for some random matrix $\Psi^{-1}\in\R^{d\times d}$
that is close to $\ma^{\top}\omw\ma\in\R^{d\times d}$.
\end{itemize}
Let $\tau(w)\defeq\tau(\mw\ma)$. Suppose that all estimate shifted
leverage scores $\ttau_{i}\in(1\pm\frac{1}{16\left\lceil \log d\right\rceil })\tau(w)_{i}$
for $i\in[n]$\footnote{Recall that $\tau(w)_{i}=(\sqrt{\mw}\ma(\ma^{\top}\mw\ma)^{-1}\ma^{\top}\sqrt{\mw})_{i,i}+\frac{d}{n},\forall i\in[n]$}
and that there is a sequence $w^{(0)},w^{(1)},\cdots,w^{(K)}\in\Rn$
such that the $w^{(k)}$ satisfy
\begin{equation}
\frac{1}{\epsilon^{2}}\|(\mw^{(k)})^{-1}(w^{(k+1)}-w^{(k)})\|_{\tau(w^{(k)})}^{2}+\|(\mathbf{T}(w^{(k)}))^{-1}(\tau(w^{(k+1)})-\tau(w^{(k)}))\|_{\tau(w^{(k)})}^{2}\leq O(1)\label{eq:input_ass}
\end{equation}
for $k=0,1,\cdots,K-1$ with $K=n^{O(1)}$. Further assume that the
update sequence $\tw^{(0)},\tw^{(1)},\cdots,\tw^{(K)}\in\Rn$ satisfies
$\tw^{(k)}\approx_{\epsilon/(16\log d)}w^{(k)}$ for all $k$ and
the $\tw^{(k)}$ are independent to the output of $\textsc{Update}$
and $\textsc{Solve}$. (The input $b^{(k)}$ can depend on the previous
output of the data structure.) Then, we have the following:
\begin{itemize}
\item The amortized time per call of $\textsc{Update}$ is $O(n+\epsilon^{-2}\cdot(d^{\omega-\frac{1}{2}}+d^{2})\cdot\log^{3/2}(n))$.
\item The time per call of $\textsc{Solve}$ is $O(n+\ensuremath{\delta^{-2}\cdot}d^{2}\cdot\log^{2}(n/\delta))$.
\item $\textsc{Update}$ outputs some $\Psi$, $\tw^{\alg}$ where $\Psi^{-1}=\ma^{\top}\widetilde{\mw}^{\alg}\ma\approx_{\epsilon}\ma^{\top}\widetilde{\mw}\ma$
with probability $1-1/\poly(n)$ and $\tw^{\alg}\approx_{\epsilon}\tw$.
\item $\textsc{Solve}$ outputs some $y=\Psi b$ where $\Psi^{-1}\approx_{\delta}\ma^{\top}\omw\ma$
with probability $1-1/\poly(n)$ and $\E[\Psi b]=(\ma^{\top}\omw\ma)^{-1}b$.
\end{itemize}
\end{restatable}

In general, Theorem \ref{thm:inverse_main} does not work against
adaptive adversaries, i.e. the input $w$ and $\tilde{\tau}$ to the
$\textsc{Update}$ procedure is not allowed to depend on previous
outputs. In Section \ref{sec:inverse_maintenance} we show, that this
algorithm can be improved such that the input $w$ and $\tilde{\tau}$
is allowed to depend on the output of $\textsc{Solve}$. However,
the input is still not allowed to depend on the output of $\textsc{Update}$.

\begin{restatable}{lem}{totalPmovement}\label{lem:totalPmovement}
Theorem~\ref{thm:inverse_main} holds even if the input $w$ and
$\ttau$ of the algorithm depends on the output of $\textsc{Solve}$.
Furthermore, we have
\[
\E\left[\sum_{k\in[K-1]}\Big\|\sqrt{\tmw^{\alg(k+1)}}\ma\Psi^{(k+1)}\ma^{\top}\sqrt{\tmw^{\alg(k+1)}}-\sqrt{\tmw^{\alg(k)}}\ma\Psi^{(k)}\ma^{\top}\sqrt{\tmw^{\alg(k)}}\Big\|_{F}\right]\leq16K\log^{5/2}n
\]
where $\Psi^{(k)}\in\R^{d\times d}$, $\tw^{\alg(k)}$ is the output
of the $k$-th step of $\textsc{Update}(\tw^{(k)},\ttau^{(k)})$.

\end{restatable}

\subsection{Linear Programming Algorithm\label{sub:main_algorithm}}

Here we show how to combine the tools from Section \ref{sub:ipm_restated}
to \ref{sub:inverse_restated} to obtain a linear program solver that
runs in $\tilde{O}(nd+d^{3})$ time. First, we give a brief summary
of our linear programming algorithm, Algorithm \ref{alg:master}.
The algorithm consists of two phases. In the first phase we construct
a good initial feasible solution, and in the second we move along
the central path towards the optimal solution. 

The construction of the initial point works as follows: via a simple
transformation (stated below as Theorem~\ref{thm:infeasible_reduction}),
we obtain a feasible solution pair $(x,s)$ where both $x$ and $s$
are close to the all $1$ vector and hence good enough as a point
close to the central path of the standard log barrier function. However,
we need to find a point such that $xs\approx_{\epsilon}\mu\cdot(\sigma(\ms^{-1/2-\alpha}\mx^{1/2-\alpha}\ma)+\frac{d}{n}\vones)$.
By picking $x=1$ and $\mu=1$, the initial slack $s$ needs to satisfy
$s\approx_{\epsilon}\sigma(\ms^{-\frac{1}{2}-\alpha}\ma)+\frac{d}{n}\vones$,
which (up to the additive $\frac{d}{n}\vones$) is exactly the condition
for $\ell_{p}$ Lewis weight with $p=\frac{1}{1+\alpha}$. Cohen and
Peng showed that such a vector $s$ can be found efficiently as long
as $p\in(0,4)$ \cite{CohenP15}. We note that such $s$ might not
satisfy $\ma y+s=c$. That is why we define $c^{(\text{tmp})}:=\ma y+s$
for $y=0$, so that $(x,s)$ is a feasible solution pair for the cost
vector $c^{(\text{tmp})}$. In the first phase of Algorithm \ref{alg:master},
we move the point $(x,s)$ along the central path of the temporary
cost vector $c^{(\text{tmp})}$ and bring the points to a location
where we can switch the cost $c^{(\text{tmp})}$ to $c$ without violating
the centrality conditions. This is how we obtain our feasible starting
point for the cost vector $c$. In the subsequent second phase, we
move along the path for the cost $c$ until it is close to the optimal
solution.

Moving along these two paths of the first and second phase is performed
via the IPM of Algorithm~\ref{alg:pathfollowing} (Section~\ref{sub:ipm_restated},
Theorem~\ref{thm:path_following}). Note that Algorithm~\ref{alg:pathfollowing}
does not specify in Line~\ref{line:pf_1} how to obtain the approximate
solution pair $(\ox,\os)$, so we must implement this step on our
own. Likewise, we must specify how to efficiently compute the steps
in Lines~\ref{line:pf_3} and~\ref{line:pf_4}. These implementations
can be found in the second part of Algorithm~\ref{alg:master}. The
high-level idea is to use the data-structures presented in Section
\ref{sub:vector_restated} to \ref{sub:inverse_restated}.

To illustrate, consider Line~\ref{line:pf_4} of Algorithm \ref{alg:pathfollowing},
which computes 
\begin{align*}
s & \leftarrow s+(1-2\alpha)\oms\omw^{-1/2}\mq\omw^{1/2}h
\end{align*}
for $\omw=\omx\oms$ and $\mq=\overline{\ms}^{-1/2}\overline{\mx}^{1/2}\ma\mh^{-1}\ma^{\top}\omx^{1/2}\oms^{-1/2}$
for $\mh\approx_{\epsilon}\ma^{\top}\oms^{-1}\omx\ma$. We split this
task into three parts: (i) compute $r:=\ma^{\top}\omx^{1/2}\oms^{-1/2}\omw^{1/2}h$,
(ii) compute $v:=\mh^{-1}r$, and (iii) compute $(1-2\alpha)\oms\omw^{-1/2}\overline{\ms}^{-1/2}\overline{\mx}^{1/2}\ma v=(1-2\alpha)\ma v$.
Part (i), the vector $r$, can be maintained efficiently, because
we maintain the approximate solutions $\ox$, $\os$ (thus also $\ow$)
and vector $h$, (by Algorithm~\ref{alg:gradient_maintenance}, Theorem~\ref{thm:gradient_maintenance},
of Section~\ref{sec:grad_maintenance}) in such a way, that per iteration
only few entries change on average. Part (ii) is solved by the inverse
maintenance data-structure of Section~\ref{sub:inverse_restated}
(Theorem~\ref{thm:inverse_main}). The last part (iii) is solved
implicitly by the data-structure of Section \ref{sub:vector_restated}
(Theorem \ref{thm:product_sum_datastructure}), which is also used
to obtain the approximate solutions $\ox,\os$ in Line \ref{line:pf_1}
of Algorithm \ref{alg:pathfollowing}. We additionally run the data-structure
of Section \ref{sub:ls_restated} (Theorem \ref{thm:leverage_score_datastructure})
in parallel, to maintain an approximation of the leverage scores,
which allows us to find the approximation $\ov$ required in Line
\ref{line:pf_2}. These modifications to Algorithm~\ref{alg:pathfollowing}
are given in the second part of Algorithm~\ref{alg:master}.

\begin{algorithm2e}[!t]

\caption{LP Algorithm (based on Algorithm \ref{alg:pathfollowing})}\label{alg:master}

\SetKwProg{Globals}{global variables}{}{}

\SetKwProg{Proc}{procedure}{}{}

\Globals{}{

\State $\ma\in\R^{n\times d}$, $\mu>0$

\State $D_{\textsc{Inverse}}$\tcp*{Inverse Maintenance data structure
(Theorem~\ref{thm:inverse_main})}

\State $D_{\textsc{MatVec}}^{(x)},D_{\textsc{MatVec}}^{(s)}$\tcp*{Matrix
Vector Maintenance data structure (Theorem~\ref{thm:product_sum_datastructure})}

\State $D_{\textsc{Leverage}}$\tcp*{Leverage Score Maintenance
data structure (Theorem~\ref{thm:leverage_score_datastructure})}

\State $D_{\text{\textsc{Gradient}}}$ \tcp*{Gradient Maintenance
data structure (Theorem~\ref{thm:gradient_maintenance})}

\State $\overline{\tau}\in\Rn$\tcp*{$\overline{\tau}_{\gamma/8}\approx(\sigma(\ms^{-1/2-\alpha}\mx^{1/2-\alpha}\ma)+\frac{d}{n}\vones)$}

\vspace{0.2 in}

\LineComment{The following parameters are the same as in Theorem
\ref{thm:path_following}}

\State $\alpha\defeq1/(4\log(4n/d))$\tcp*{Maintain $\mx s\approx_{\epsilon}\mu\cdot(\sigma(\ms^{-1/2-\alpha}\mx^{1/2-\alpha}\ma)+\frac{d}{n}\vones)$}

\State $\epsilon\defeq\frac{\alpha}{16000}$\tcp*{Distance to the
central path}

\State $\lambda\defeq\frac{2}{\epsilon}\log(\frac{2^{16}n\sqrt{d}}{\alpha^{2}}),\gamma\defeq\min(\epsilon/4,\frac{\alpha}{50\lambda})$\tcp*{Potential
parameter $\lambda$ and step size $\gamma$}

}

\vspace{0.4 in}

\Proc{$\textsc{Solve}(\ma\in\R^{n\times d},b\in\Rn,c\in\R^{d},\delta>0)$}{

\LineComment{\textbf{Initialize data structures and initial primal
dual points $(x,s)$}}

\State Modify the LP and obtain an initial $x$, $y$ and $s$ by
Lemma~\ref{thm:infeasible_reduction} to accuracy $\delta/8n^{2}$

\vspace{0.2 in}

\LineComment{ For notational simplicity, we use $\ma,b,c,n,d$ for
the modified LP induced by Lemma~\ref{thm:infeasible_reduction}
in the remainder of the code.}

\State Use Theorem \ref{thm:Lewis_weight_compute} to compute $s$
such that $s\approx_{\epsilon}\sigma(\ms^{-1/2-\alpha}\ma)+\frac{d}{n}\vones$\label{line:find_lewis}

\State $\tau\leftarrow s,\mu\leftarrow1$\tcp*{Since $x=1$ (Lemma
\ref{thm:infeasible_reduction}), $\mx s\approx_{\epsilon}\mu\tau$}

\State $D_{\textsc{Inverse}}.\textsc{Initialize}(\ma,\ms^{-1-2\alpha}x^{1-2\alpha},\tau,\gamma/512\log n)$\label{line:d_inverse_init}

\State $D_{\textsc{Leverage}}.\textsc{Initialize}(\ma,\ms^{-1-2\alpha}x^{1-2\alpha},\gamma/8)$\label{line:sigma_init}

\State $D_{\textsc{MatVec}}^{(x)}.\textsc{Initialize}(\ma,\ms^{-1}x,x,\gamma/8)$,
$D_{\textsc{MatVec}}^{(s)}.\textsc{Initialize}(\ma,1,s,\gamma/8)$

\State $D_{\textsc{Gradient}}.\textsc{Initialize}(\ma,\mu\tau/(xs),\tau,x,\gamma)$
\tcp*{Theorem \ref{thm:gradient_maintenance}}

\vspace{0.2in}

\LineComment{\textbf{Find the ``center'' of the linear program
with $(\ma,b,\ma y+s)$}}

\State $c^{\text{(tmp)}}\leftarrow s$\tcp*{This is same as $c^{\text{(tmp)}}=\ma y+s$
for $y=0$}

\State $x,s,\overline{\tau},\mu\leftarrow\textsc{Centering}(x,s,\tau,\mu,\Theta(n^{2}\sqrt{d}/(\gamma\alpha^{2}))$\tcp*{This
keeps $\ma y+s$ unchanged.}\label{line:first_centering}

\vspace{0.2in}

\LineComment{\textbf{Switch central path from }$c^{\text{(tmp)}}$
to $c$}

\State $s^{\new}\leftarrow s+c-c^{\text{(tmp)}}$ \label{line:switch_c}\tcp*{This
is the same as $s\leftarrow c-(\ma y+s)$ for current $y$}

\State $D_{\textsc{MatVec}}^{(s)}.\textsc{Initialize}(\ma,1,s^{\new},\gamma/8)$

\vspace{0.2in}

\LineComment{\textbf{Solve the linear program with $(\ma,b,c)$}}

\State $x,s,\overline{\tau},\mu\leftarrow\textsc{Centering}(x,s,\overline{\tau},\mu,\delta^{2}/(8^{3}n^{4}d))$\label{line:second_centering}

\vspace{0.2in}

\LineComment{\textbf{Output the solution}}

\State Decrease $\frac{1}{\mu}\|\ma^{\top}x-b\|_{(\ma^{\top}\mx\ms^{-1}\ma)^{-1}}^{2}$
to small enough $O(\delta/n^{2})$ via Corollary~\ref{cor:tiny_phi_b}
\label{line:make_near_feasible}

\State Return an approximate solution of the original linear program
according to Lemma \ref{thm:infeasible_reduction}

}

\algstore{masteralg}

\end{algorithm2e}

\begin{algorithm2e}[!t]

\caption{LP Algorithm (based on Algorithm \ref{alg:pathfollowing}),
Continuation of Algorithm \ref{alg:master}.}\label{alg:master_2}

\algrestore{masteralg}

\SetKwProg{Proc}{procedure}{}{}

\Proc{$\text{\textsc{Centering}}$$(x^{\textrm{\ensuremath{\mathrm{(init)}}}}\in\R_{>0}^{n},s^{\textrm{\ensuremath{\mathrm{(init)}}}}\in\R_{>0}^{n},\tau^{\textrm{\ensuremath{\mathrm{(init)}}}}>0,\mu^{\textrm{\ensuremath{\mathrm{(init)}}}}>0,\mu^{\textrm{\ensuremath{\mathrm{(target)}}}}>0)$}{

\State $\mu\leftarrow\mu^{\textrm{(init)}}$, $\bar{\tau}\leftarrow\tau^{\textrm{(init)}}$,
$\tau^{\textrm{(tmp)}}\leftarrow\tau^{\textrm{(init)}}$, $\bar{x}\leftarrow x^{\textrm{(init)}}$,
$x^{\textrm{(tmp)}}\leftarrow x^{\textrm{(init)}}$, $\bar{s}\leftarrow s^{\textrm{(init)}}$,
$s^{\textrm{(tmp)}}\leftarrow s^{\textrm{(init)}}$

\While{$\textsc{True}$}{

\State $\ox_{i}=x_{i}^{\text{(tmp)}}$ for all $i$ such that $\ox_{i}\not\approx_{\gamma/8}x_{i}^{\text{(tmp)}}$\label{line:xbar_update}
\tcp*{Ensure $\ox\approx_{\gamma/4}x$}

\State $\os_{i}=s_{i}^{\text{(tmp)}}$ for all $i$ such that $\os_{i}\not\approx_{\gamma/8}s_{i}^{\text{(tmp)}}$\label{line:sbar_update}
\tcp*{Ensure $\os\approx_{\gamma/4}s$}

\State $\overline{\tau}_{i}=\tau_{i}^{\text{(tmp)}}$ for all $i$
such that $\overline{\tau}_{i}\not\approx_{\gamma/8}\tau_{i}^{\text{(tmp)}}$\label{line:tau_update}

\State $\ow\leftarrow\omx\os$, $\ov\leftarrow\mu\cdot\omw^{-1}\overline{\tau}$
\tcp*{Ensure $\|\ov-v\|_{\infty}\leq\gamma$ where $v=\mu\cdot\tau(x,s)/w$}

\vspace{0.2 in}

\LineComment{Check termination conditions}

\State Let $\Phi(v)\defeq\exp(\lambda(v-1))+\exp(-\lambda(v-1))$
for all $v\in\R^{n}$

\lIf{$\mu=\mu^{\textrm{\ensuremath{\mathrm{(target)}}}}$ and $\Phi(\ov)\leq\frac{2^{16}n\sqrt{d}}{\alpha^{2}}$}{\State \textbf{break}}

\vspace{0.2 in}

\LineComment{Update $h=\gamma\nabla\Phi(\overline{v})^{\flat}$,
$r=\ma^{\top}\omx h$ (See Lemma~\ref{lem:smoothing:helper} for
$(\cdot){}^{\flat}$ definition)}

\State $D_{\text{\textsc{Gradient}}}.\textsc{Update}(i,\ov_{i},\overline{\tau}_{i},\ox_{i})$
for $i$ where $\ov_{i}$, $\overline{\tau}_{i}$ or $\ox_{i}$ changed

\State $h,r\leftarrow D_{\textsc{Gradient}}.\textsc{Query()}$

\vspace{0.2 in}

\LineComment{Pick any $\mh\in\R^{d\times d}$ with $\mh\approx_{(c\epsilon)/(d^{1/4}\log^{3}n)}\ma^{\top}\oms^{-1}\omx\ma$
}

\LineComment{Let $\mq=\overline{\ms}^{-1/2}\overline{\mx}^{1/2}\ma\mh^{-1}\ma^{\top}\omx^{1/2}\oms^{-1/2}$,
$\omw=\omx\oms$}

\LineComment{Update the inverse, $\Psi^{(\alpha)},\Psi\safe^{(\alpha)}\approx_{\epsilon}(\ma^{\top}\oms^{-1-2\alpha}\omx^{1-2\alpha}\ma)^{-1}$
and $\Psi\safe=\mh^{-1}\approx_{c\epsilon/(d^{1/4}\log^{3}(n))}(\ma^{\top}\oms^{-1}\omx\ma)^{-1}$}

\State $\Psi^{(\alpha)},g\leftarrow D_{\textsc{Inverse}}.\textsc{Update}(\oms^{-1-2\alpha}\ox^{1-2\alpha},\overline{\tau})$\Comment{Update
the approximate inverse}

\State $D_{\textsc{Leverage}}.\textsc{Update}(\sqrt{g})$ \tcp*{Only
scale coordinates where $g$ changed.}

\tcp{$\Psi\safe^{(\alpha)},\Psi\safe$ are implicit representations
of calls to $D_{\textsc{Inverse}}.\textsc{Solve}$}

\State $\Psi\safe^{(\alpha)}(b)\defeq D_{\textsc{Inverse}}.\textsc{Solve}(b,\oms^{-1-2\alpha}\ox^{1-2\alpha},\gamma/512\log n)$\label{line:inverse_alpha_safe}

\State $\Psi\safe(b)\defeq D_{\textsc{Inverse}}.\textsc{Solve}(b,\oms^{-1}\ox,(c\epsilon)/(d^{1/4}\log^{3}n))$\label{line:inverse_safe}

\vspace{0.2 in}

\LineComment{\textbf{$x\leftarrow x+(1+2\alpha)\omx\omw^{-1/2}(\mi-\mq)\omw^{1/2}h$},
$s\leftarrow s+(1-2\alpha)\oms\omw^{-1/2}\mq\omw^{1/2}h$}

\State $D_{\textsc{MatVec}}^{(x)}.\textsc{Scale}(\ox/\os)$ \tcp*{Only
scale coordinates where $\ox$ or $\os$ changed.}

\State $x^{\text{(tmp)}}\leftarrow D_{\textsc{MatVec}}^{(x)}.\textsc{Query}((1+2\alpha)\Psi\safe r,(1+2\alpha)\omx h)$\label{line:xtmp_formula}

\State $s^{\text{(tmp)}}\leftarrow D_{\textsc{MatVec}}^{(s)}.\textsc{Query}((1-2\alpha)\Psi\safe r,0_{n})$\label{line:stmp_formula}

\State $\tau^{\text{(tmp)}}\leftarrow D_{\textsc{Leverage}}.\textsc{Query}(\Psi^{(\alpha)},\Psi\safe^{(\alpha)})$\label{line:sigma_tmp}

\State $\delta_{\lambda}\leftarrow\textsc{MaintainFeasibility}(D_{\textsc{MatVec}}^{(x)},D_{\textsc{MatVec}}^{(s)},\mu)$
\tcp*{Algorithm \ref{alg:maintain_INfeasibility}}

\State $x^{\text{(tmp)}}\leftarrow D_{\textsc{MatVec}}^{(x)}.\textsc{Query}(\delta_{\lambda},0)$

\lIf{$\mu>\mu^{\mathrm{(target)}}$}{\State $\mu\leftarrow\max(\mu^{\textrm{(target)}},(1-\frac{\gamma\alpha}{2^{15}\sqrt{d}})\mu)$}

\lElseIf{ $\mu<\mu^{\textrm{\ensuremath{\mathrm{(target)}}}}$ }{\State $\mu\leftarrow\min(\mu^{\textrm{(target)}},(1+\frac{\gamma\alpha}{2^{15}\sqrt{d}})\mu)$}

}

\State $x\leftarrow D_{\textsc{MatVec}}^{(x)}.\textsc{ComputeExact}(1,...,n)$,
$s\leftarrow D_{\textsc{MatVec}}^{(s)}.\textsc{ComputeExact}(1,...,n)$

\State \Return $(x,s,\tau^{\text{(tmp)}},\mu)$

}

\end{algorithm2e}

The following theorem shows how to reduce solving any bounded linear
program to solving a linear program with a non-degenerate constrain
matrix and an explicit initial primal and dual interior point. This
theorem is proven in Appendix~\ref{sec:initialPoint}.

\begin{restatable}[Initial Point]{thm}{initialPoint}\label{thm:infeasible_reduction}
Consider linear program $\min_{\ma^{\top}x=b,x\geq0}c^{\top}x$ with
$n$ variables and $d$ constraints. Assume that \\
 1. Diameter of the polytope : For any $x\geq0$ with $\ma^{\top}x=b$,
we have that $\|x\|_{2}\leq R$.\\
 2. Lipschitz constant of the linear program : $\|c\|_{2}\leq L$.\\
 3. The constraint matrix $\ma$ is non-degenerate.

For any $\delta\in(0,1]$, the modified linear program $\min_{\overline{\ma}^{\top}\overline{x}=\overline{b},\overline{x}\geq0}\overline{c}^{\top}\overline{x}$
with
\begin{align*}
\overline{\ma} & =\left[\begin{array}{cc}
\ma & 1_{n}\|\ma\|_{F}\\
0 & 1\|\ma\|_{F}\\
\frac{1}{R}b^{\top}-1_{n}^{\top}\ma & 0
\end{array}\right]\in\R^{(n+2)\times(d+1)},\\
\overline{b} & =\begin{bmatrix}\frac{1}{R}b\\
(n+1)\|\ma\|_{F}
\end{bmatrix}\in\R^{d+1}\text{~and,~}\overline{c}=\begin{bmatrix}\frac{\delta}{L}\cdot c\\
0\\
1
\end{bmatrix}\in\R^{n+2}
\end{align*}
satisfies the following:\\
 1. $\overline{x}=\begin{bmatrix}1_{n}\\
1\\
1
\end{bmatrix}$, $\overline{y}=\begin{bmatrix}0_{d}\\
-1
\end{bmatrix}$ and $\overline{s}=\begin{bmatrix}1_{n}+\frac{\delta}{L}\cdot c\\
1\\
1
\end{bmatrix}$ are feasible primal dual vectors.\\
2. Let $(\overline{x},\overline{y},\overline{s})$ be primal dual
vectors of the modified LP and $\Phi_{b}\defeq\frac{1}{\mu}\cdot\|\overline{\ma}^{\top}\ox-\overline{b}\|_{(\overline{\ma}^{\top}\omx\oms^{-1}\overline{\ma})^{-1}}^{2}$,
then
\begin{align*}
\|\ox\|_{\infty} & \le(1+O(\Phi_{b}))\cdot O(n).
\end{align*}
 3. Let $(\overline{x},\overline{y},\overline{s})$ be primal dual
vectors of the modified LP with $\ox\cdot\os\approx_{0.5}\mu\cdot\tau(\ox,\os)$
for $\mu<\delta^{2}/(8d)$ and small enough $\Phi_{b}:=\frac{1}{\mu}\cdot\|\overline{\ma}^{\top}\ox-\overline{b}\|_{(\overline{\ma}^{\top}\omx\oms^{-1}\overline{\ma})^{-1}}^{2}=O(1)$
(i.e. $\ox$ does not have to be feasible). The vector $\widehat{x}\defeq R\cdot\overline{x}_{1:n}$
where $\overline{x}_{1:n}$ is the first $n$ coordinates of $\ox$
is an approximate solution to the original linear program in the following
sense 
\begin{align*}
c^{\top}\widehat{x}\leq & ~\min_{\ma^{\top}x=b,x\geq0}c^{\top}x+O(nLR)\cdot(\sqrt{\Phi_{b}}+\delta),\\
\|\ma^{\top}\widehat{x}-b\|_{2}\leq & ~O(n^{2})\cdot(\|\ma\|_{F}R+\|b\|_{2})\cdot(\sqrt{\Phi_{b}}+\delta)\Big),\\
\widehat{x} & \geq0.
\end{align*}

\end{restatable} As outlined before, the initial points given in
Theorem \ref{thm:infeasible_reduction} do not satisfy

\[
xs\approx_{\epsilon}\mu\cdot(\sigma(\ms^{-1/2-\alpha}\mx^{1/2-\alpha}\ma)+\frac{d}{n}\vones).
\]
To satisfy this condition we pick $x=\vones$ and $\mu=\vones$, and
pick the initial slack vector $s$ a to satisfy $s\approx_{\epsilon}\sigma(\ms^{-\frac{1}{2}-\alpha}\ma)+\frac{d}{n}\vones$,
which is exactly the condition for $\ell_{p}$ Lewis weight with $p=\frac{1}{1+\alpha}$.
The following theorem shows that such a vector $s$ can be found efficiently
as long as $p\in(0,4)$. This vector $s$ might not be a valid slack
vector, so as outlined before, the algorithm runs in two phases: first
using a cost vector $c^{(tmp)}$ for which the initial $s$ is feasible,
and then switching to the correct $c$.
\begin{thm}[\cite{CohenP15}]
\label{thm:Lewis_weight_compute} Given $p\in(0,4)$, $\eta>0$ and
non-degenerate $\ma\in\R^{n\times d}$ w.h.p. in $n$, we can compute
$w\in\Rn_{>0}$ with $w\approx_{\epsilon}\sigma(\mw^{\frac{1}{2}-\frac{1}{p}}\ma)+\eta\vones$
in $\otilde((\nnz(\ma)+d^{\omega})\poly(1/\epsilon))$ time.
\end{thm}

\begin{proof}
Our proof is similar to \cite{CohenP15}, which proved a variant when
$\eta=\vzero$ Consider the map $T(w)\defeq(\mw^{\frac{2}{p}-1}(\sigma(\mw^{\frac{1}{2}-\frac{1}{p}}\ma)+\eta\vones))^{p/2}$.
Further, fix any positive vectors $v,w\in\Rn$ such that $v\approx_{\alpha}w$.
We have that $\ma^{\top}\mv^{1-\frac{1}{p}}\ma\approx_{|1-\frac{2}{p}|\alpha}\ma^{\top}\mw^{1-\frac{2}{p}}\ma$
and hence
\begin{equation}
a_{i}^{\top}(\ma^{\top}\mv^{1-\frac{2}{p}}\ma)^{-1}a_{i}\approx_{|1-\frac{2}{p}|\alpha}a_{i}^{\top}(\ma^{\top}\mw^{1-\frac{2}{p}}\ma)^{-1}a_{i}.\label{eq:lewis_apx_1}
\end{equation}
Note that
\begin{align*}
T(v)_{i}^{2/p} & =\frac{v_{i}^{1-\frac{2}{p}}a_{i}^{\top}(\ma^{\top}\mv^{1-\frac{2}{p}}\ma)^{-1}a_{i}+\eta}{v_{i}^{1-\frac{2}{p}}}=a_{i}^{\top}(\ma^{\top}\mv^{1-\frac{2}{p}}\ma)^{-1}a_{i}+\eta v_{i}^{\frac{2}{p}-1}\,.
\end{align*}
Using (\ref{eq:lewis_apx_1}), we have
\begin{align*}
T(v)_{i}^{2/p} & \leq e^{|1-\frac{2}{p}|\alpha}a_{i}^{\top}(\ma^{\top}\mw^{1-\frac{2}{p}}\ma)^{-1}a_{i}+\eta v_{i}^{\frac{2}{p}-1}\\
 & \leq e^{|1-\frac{2}{p}|\alpha}(a_{i}^{\top}(\ma^{\top}\mw^{1-\frac{2}{p}}\ma)^{-1}a_{i}+\eta w_{i}^{\frac{2}{p}-1})=e^{|1-\frac{2}{p}|\alpha}T(w)_{i}^{2/p}\,.
\end{align*}
Similarly, we have $T(v)_{i}^{2/p}\geq e^{-|1-\frac{2}{p}|\alpha}T(w)_{i}^{2/p}$.
Taking $p/2$ power of both sides, we have
\[
e^{-|\frac{p}{2}-1|\alpha}T(w)_{i}\leq T(v)_{i}\leq e^{|\frac{p}{2}-1|\alpha}T(w)_{i}.
\]
Hence, $T(v)\approx_{|p/2-1|\alpha}T(w)$.

Consequently, let $w_{0}=\eta\vones$ and consider the algorithm $w_{k+1}=T(w_{k})$.
Since $\eta>0$ we have that 
\[
w_{0}=w_{0}\eta^{-p/2}\eta^{p/2}\leq T(w_{0})\leq w_{0}\eta^{-p/2}\left(1+\eta\right)^{p/2}\leq w_{0}\exp(p\eta^{-1}/2)\,.
\]
Consequently, $T(w_{0})\approx_{p\eta^{-1}/2}w_{0}$ and after $k$
steps we have that $T(w_{k})\approx_{\exp(|p/2-1|^{k}p\eta^{-1}/2)}w_{k}$.
Since for $p\in(0,4)$, we have that $|p/2-1|<1$ we see that after
$O(\log(\eta^{-1}/\epsilon))$ steps we have $w_{k}\approx_{\epsilon}T(w_{k})$.
Further, since $w_{k}\geq\eta\vones$ implies that $T(w_{k})\geq\eta\vones$
we have that $w_{k}\in\R_{>0}^{n}$ as desired. To implement the steps,
one can check, by the same proof, that it suffices to get a $\approx_{O(\epsilon)}$
multiplicative approximation to $T(w)$ in each step, which can be
done in $\otilde((\nnz(\ma)+d^{\omega})/\epsilon^{2})$ per step by
standard leverage score estimation techniques.
\end{proof}
Now, we first prove the correctness of the Algorithm \ref{alg:master}.
\begin{lem}
\label{lem:correctness}Algorithm~\ref{alg:master} outputs $x$
such that w.h.p. in $n$ 
\begin{align*}
c^{\top}x\leq & ~\min_{\ma^{\top}x=b,x\geq0}c^{\top}x+LR\cdot\delta,\\
\|\ma^{\top}x-b\|_{2}\leq & ~\delta\cdot\Big(\|\ma\|_{F}\cdot R+\|b\|_{2}\Big),\\
x & \geq0.
\end{align*}
\end{lem}

\begin{proof}
We define the primal dual point $(x,s)$ via the formula of lines
\ref{line:x_formula} and \ref{line:s_formula}. We start by showing
that our implementation of Algorithm~\ref{alg:pathfollowing} in
Algorithm~\ref{alg:master} satisfies all required conditions, i.e.
that we can apply Theorem~\ref{thm:path_following}. Throughout this
proof, states hold only w.h.p. and therefore the restatement of this
is often omitted for brevity.

\paragraph{Invariant: $\protect\ox\approx_{\gamma/4}x$, $\protect\os\approx_{\gamma/4}s$,
$\overline{\tau}\approx_{\gamma/4}\sigma(\protect\oms^{-1/2-\alpha}\protect\omx^{1/2-\alpha}\protect\ma)+\frac{d}{n}\protect\vones$
and $\Psi\protect\safe\approx_{\epsilon}(\protect\ma^{\top}\protect\oms^{-1}\protect\omx\protect\ma)^{-1}$:}

We first show that throughout $\textsc{Centering}$, we have the invariant
above, assuming the input parameter $\tau_{i}^{\text{(init)}}$ satisfied
$\tau_{i}^{\text{(init)}}\approx_{\gamma/4}\sigma(\ms^{-1/2-\alpha}\mx^{1/2-\alpha}\ma)+\frac{d}{n}\vones$.
Theorem~\ref{thm:product_sum_datastructure} shows that $x^{\text{(tmp)}}\approx_{\gamma/8}x$
and $s^{\text{(tmp)}}\approx_{\gamma/8}s$ (Line~\ref{line:xtmp_formula}
and~\ref{line:stmp_formula}). Then, by the update rule (Line~\ref{line:xbar_update}
and \ref{line:sbar_update}), we have that $x^{\text{(tmp)}}\approx_{\gamma/8}\ox$
and $s^{\text{(tmp)}}\approx_{\gamma/8}\os$. Hence, we have the desired
approximation for $\ox$ and $\os$.

Next, Theorem~\ref{thm:inverse_main} shows that
\begin{itemize}
\item $\Psi^{(\alpha)}\approx_{\gamma/(512\log n)}(\ma^{\top}\ms^{-1-2\alpha}\mx^{1-2\alpha}\ma)^{-1}$
(Line \ref{line:d_inverse_init}) and
\item $\Psi\safe^{(\alpha)}\approx_{\gamma/(512\log n)}(\ma^{\top}\ms^{-1-2\alpha}\mx^{1-2\alpha}\ma)^{-1}$
(Line \ref{line:inverse_alpha_safe}). 
\end{itemize}
Hence, we can apply Theorem~\ref{thm:leverage_score_datastructure}
and get 
\[
\tau^{\text{(tmp)}}\approx_{\gamma/8}\sigma(\oms^{-1/2-\alpha}\omx^{1/2-\alpha}\ma)+\frac{d}{n}\vones
\]
(Line~\ref{line:sigma_init} and Line~\ref{line:sigma_tmp}). Again,
by update rule (Line~\ref{line:tau_update}), we have the desired
approximation for $\overline{\tau}$.

Finally, $\Psi\safe\approx_{\epsilon}(\ma^{\top}\oms^{-1}\omx\ma)^{-1}$
follows from Theorem~\ref{thm:inverse_main} (Line \ref{line:inverse_safe}).
These invariants also imply $\|\overline{v}-v\|_{\infty}\leq\gamma$
and that $D_{\text{Gradient}}$ (Theorem~\ref{thm:gradient_maintenance})
maintains $\nabla\Phi(\ov')^{\flat(\otau)}$ for some $\|\overline{v}'-v\|_{\infty}\leq\gamma$.
Thus, in summary, $\textsc{Centering}$ of Algorithm \ref{alg:master}
behaves like $\textsc{Centering}$ of Theorem~\ref{thm:path_following}.

\paragraph{Invariant: $xs\approx\mu\cdot\tau(x,s)$:}

Initially, $x=\vones$ due to the reduction (Lemma \ref{thm:infeasible_reduction}),
$\mu=1$ and $s\approx_{\epsilon}\sigma(\ms^{-1/2-\alpha}\ma)+\frac{d}{n}\vones$
(Line \ref{line:find_lewis}). Hence, $xs\approx_{\epsilon}\mu\cdot(\sigma(\ms^{-1/2-\alpha}\mx^{1/2-\alpha}\ma)+\frac{d}{n}\vones)$
initially.

We change $x,s,\mu$ in Line~\ref{line:first_centering} by calling
$\textsc{Centering}$. By Theorem~\ref{thm:path_following} we then
have $xs\approx_{\epsilon}\mu\tau$ after this call to $\textsc{Centering}$.
Next, consider the step where we switch the cost vector on Line~\ref{line:switch_c}.
Let $s^{\new}$ be the vector $s$ after Line~\ref{line:switch_c}
and $s$ is before Line~\ref{line:switch_c}. Since $\ma y+s=c^{\text{(tmp)}}$
we have that $\ma y+s^{\new}=c^{\text{(tmp)}}-s+s^{\new}=c$, i.e.
$s$ is a valid slack vector for cost $c$. Further, we have that,
by design
\[
\frac{s^{\new}-s}{s}=\frac{c-c^{\text{(tmp)}}}{s}.
\]
First, we bound the denominator. We have that $sx\approx_{1/2}\mu\tau$,
so for all $i\in[n]$ we have 
\[
s_{i}\geq\frac{\mu}{2x_{i}}\frac{d}{n}\geq\frac{\mu}{\Omega(n^{2})}
\]
where we used $\|x\|_{\infty}\leq O(n)$ by Lemma~\ref{thm:infeasible_reduction}
as we ensure $x$ is sufficiently close to feasible for the modified
linear program by Theorem~\ref{thm:path_following}. For the numerator,
we note that $\|c\|_{\infty}\leq1$ for the modified linear program
and $\|c^{\text{(tmp)}}\|_{\infty}=\|s\|_{\infty}\leq3$ for the $s$
computed by Theorem~\ref{thm:Lewis_weight_compute} in Line~\ref{line:find_lewis}.again
by the definition of $s$ (Line~\ref{line:find_lewis}) and the modified
$\ma$ and $y$. Hence, we have that 
\[
\Big\|\frac{s^{\new}-s}{s}\Big\|_{\infty}\leq\frac{16n^{2}}{\mu}\leq\frac{\gamma\alpha^{2}}{c\cdot\sqrt{d}}.
\]
for any constant $c$ by choosing the constant in the $O(\cdot)$
in Line~\ref{line:first_centering} appropriately. Thus $xs\approx_{2\epsilon}\mu\cdot\tau(x,s)$
and $\overline{\tau}\approx_{\gamma/4}\sigma(\ms^{-1/2-\alpha}\mx^{1/2-\alpha}\ma)+\frac{d}{n}\vones$,
so when we call $\textsc{Centering}$ in Line \ref{line:second_centering},
we again obtain $xs\approx_{\epsilon}\mu\cdot\tau(x,s)$.

\paragraph{Conclusion:}

Before the algorithm ends, we have $xs\approx_{1/4}\mu\tau$. In Line~\ref{line:make_near_feasible},
we reduce $\Phi_{b}:=\|\ma^{\top}x-b\|_{(\ma^{\top}\mx\ms^{-1}\ma)^{-1}}^{2}$
to some small enough $\Phi_{b}=O(\delta/n^{2})$. By Corollary \ref{cor:tiny_phi_b}
this does not move the vector $x$ too much, i.e. we still have $x\cdot s\approx_{1/2}\mu\cdot\tau(x,s)$.
Hence, by the choice of the new $\mu$ in Line~\ref{line:second_centering},
Lemma~\ref{thm:infeasible_reduction} shows that we can output a
point $\widehat{x}$ with the desired properties. 
\end{proof}
Finally, we analyze the cost of the Algorithm~\ref{alg:master}.
\begin{lem}
Algorithm \ref{alg:master} takes $O((nd+d^{3})\log^{O(1)}n\log(n/\delta))$
time with high probability in $n$.
\end{lem}

\begin{proof}
For simplicity, we use $\otilde(\cdot)$ to suppress all terms that
are $\log^{O(1)}n$. We first note that all parameters $\alpha,\epsilon,\lambda,\gamma$
are either $\log^{O(1)}n$ or $1/\log^{O(1)}n$. The cost of Algorithm~\ref{alg:master}
is dominated by the repeated calls to $\textsc{Centering}$.

\paragraph{Number of iterations:}

By Theorem~\ref{thm:path_following} the calls to $\textsc{Centering}$
in Lines \ref{line:first_centering} and \ref{line:second_centering}
perform $\otilde(\sqrt{d}\log(1/\delta))$ many iterations in total.

\paragraph{Cost of $D_{\textsc{Inverse}}$:}

Note that $\|\tau^{-1}\delta_{\tau}\|_{\tau+\infty}=2\epsilon$, $\|\ms^{-1}\delta_{s}\|_{\tau+\infty}\leq\epsilon/2$,
and $\|\mx^{-1}\delta_{x}\|_{\tau+\infty}\leq\epsilon/2$ by Theorem
\ref{thm:path_following}. Hence, the weight of the vector $w=\ms^{-1-2\alpha}x^{1-2\alpha}$
satisfies $\|\mw^{-1}\delta_{w}\|_{\tau}=O(\epsilon)$. Note however,
that we need $D_{\textsc{Inverse}}.\textsc{Update}$ to compute the
inverse with accuracy $O(\gamma/\log n)$ and hence Theorem \ref{thm:inverse_main}
requires the relative movement of $w$ to be less than $\gamma/\log n$.
This can be fixed by splitting the step into $\otilde(1)$ pieces.
Now, Theorem \ref{thm:inverse_main} shows that the total cost is
$\otilde(d^{\omega}+\sqrt{d}\log(1/\delta)(n+d^{\omega-\frac{1}{2}}+d^{2}))=\otilde((\sqrt{d}n+d^{\omega}+d^{2.5})\log(1/\delta))$
where the term $d^{\omega}$ is the initial cost and $d^{\omega-\frac{1}{2}}+d^{2}$
is the amortized cost per step. Note that the solver runs with accuracy
$\otilde(d^{-1/4})$ so the total time for all calls to $D_{\textsc{Inverse}}.\textsc{Solve}$
will be $\otilde((\sqrt{d}n+d^{3})\log(1/\delta))$.

\paragraph{Cost of $D_{\textsc{Leverage}}$:}

By Lemma \ref{lem:totalPmovement}, the total movement of the projection
matrix is
\[
\sum_{k\in[K-1]}\Big\|\sqrt{\mw^{(k+1)}}\ma\Psi^{(k+1)}\ma^{\top}\sqrt{\mw^{(k+1)}}-\sqrt{\mw^{(k)}}\ma\Psi^{(k)}\ma^{\top}\sqrt{\mw^{(k)}}\Big\|_{F}=\otilde(K)
\]
 where $K$ is the number of steps and $w=\ms^{-1-2\alpha}x^{1-2\alpha}$.
Then, Theorem~\ref{thm:leverage_score_datastructure} shows that
the total cost is $O(\frac{n}{d}K^{2}\cdot d)=O(nK^{2})=O(nd\log^{2}(1/\delta))$.

\paragraph{Cost of $D_{\textsc{MatVec}}^{(x)}$ and $D_{\textsc{MatVec}}^{(s)}$:}

By Theorem~\ref{thm:product_sum_datastructure}, the total cost for
$D_{\textsc{MatVec}}^{(x)}$ is bounded by $\otilde(K^{2}\max\|\mx^{-1}\delta_{x}\|_{2}^{2}\cdot d+Kn)$
where $\max\|\mx^{-1}\delta_{x}\|_{2}^{2}$ is the maximum movement
in one step in $\ell_{2}$-norm. Note that $\tau\geq\frac{d}{n}$
and hence $\|\mx^{-1}\delta_{x}\|_{2}^{2}\leq\frac{n}{d}\|\mx^{-1}\delta_{x}\|_{\tau}^{2}=\otilde(\frac{n}{d})$.
Hence, we have that the cost is $\otilde(d\log^{2}(1/\delta)\cdot\frac{n}{d}\cdot d)=\otilde(nd)$.
The bound for $D_{\textsc{MatVec}}^{(s)}$ is the same.

\paragraph{Cost of maintaining $\protect\ma^{\top}\overline{\protect\mx}h$
($D_{\text{Gradient}}$):}

The cost of maintaining $\ma^{\top}\overline{\mx}h$ is exactly equals
to $d$ times the number of coordinates changes in $\ox,\os,\overline{\tau}$.
Note that we change the exact $x,s,\tau$ by at most around $(1\pm\gamma/8)$
multiplicative factor in every step. So, the total number of entry
changes performed to $\bar{x},\bar{s},\bar{\tau}$ is bounded by $\otilde\left(K^{2}\left(\max\|\mx^{-1}\delta_{x}\|_{2}^{2}+\max\|\ms^{-1}\delta_{s}\|_{2}^{2}+\max\|\tau^{-1}\delta_{\tau}\|_{2}^{2}\right)\right)$,
where $\max$ refers to the maximum movement in any step in $\ell_{2}$-norm.
As we showed before, all the movement terms are bounded by $\tilde{O}(\frac{n}{d})$
(see the paragraphs regarding $D_{\textsc{MatVec}}^{(x)}$, $D_{\textsc{MatVec}}^{(s)}$,
and $D_{\textsc{Inverse}}$). So in total there are $\tilde{O}(n\log^{2}(1/\delta))$
many entry changes we perform to $\ox,\os,\overline{\tau}$. Hence,
the total cost of maintenance is again $\tilde{O}(nd\log^{2}(1/\delta))$.

\paragraph{Cost of implementing $\textsc{MaintainFeasibility}$ (Algorithm~\ref{alg:maintain_INfeasibility}):}

By Theorem \ref{thm:feasibilityComplexity} we pay $\otilde(nd^{0.5}+d^{2.5})$
amortized time per call to $\textsc{MaintainFeasibility}$. Additionally,
we must compute $\otilde(n/\sqrt{d}+d^{1.5})$ entries of $x$ in
each iteration (i.e. call $D_{\textsc{MatVec}}^{(x)}.\textsc{ComputeExact(\ensuremath{i})}$),
which costs $\otilde(d)$ per call, so we have total cost of $\otilde(nd+d^{3})$
after $\otilde(\sqrt{d})$ iterations.

\paragraph{Removing the extra $\log(1/\delta)$ term:}

We note that all the extra $\log(1/\delta$) terms are due to running
the data structures for $\sqrt{d}\log(1/\delta)$ steps. However,
we can we reinitialize the data structures every $\sqrt{d}$ iterations.
This decreases the $K$ dependence from $K^{2}$ to $K\sqrt{d}$.

\paragraph{Independence and Adaptive Adversaries:}

Randomized data-structures often can not handle inputs that depend
on outputs of the previous iteration. For example Theorem~\ref{thm:inverse_main}
is such a case, where the input to $\textsc{Update}$ is not allowed
to depend on the output of any previous calls to $\textsc{Update}$.
In our Algorithm \ref{alg:master} the input to the data-structures
inherently depends on their previous output, so here we want to verify
that this does not cause any issues.

The data-structures $D_{\textsc{MatVec}}^{(x)}$ and $D_{\textsc{MatVec}}^{(s)}$
are given by Theorem \ref{thm:product_sum_datastructure}, which explicitly
states that the data-structure works against adaptive adversaries.
This means, that the input is allowed to depend on the output of previous
iterations. Likewise, $D_{\textsc{Leverage}}$ works against adaptive
adversaries by Theorem~\ref{thm:leverage_score_datastructure} (provided
the input $\Psi\safe^{(\alpha)}$ are chosen by randomness independent
of the randomness chosen for $\Psi^{(\alpha)}$ which is the case
by Lemma ~\ref{lem:totalPmovement}). The only issue is with $D_{\textsc{Inverse}}$
(given by Theorem \ref{thm:inverse_main}), where the input $w,\tilde{\tau}$
to $\textsc{Update}$ is not allowed to depend on the output of previous
calls to $\textsc{Update}$ (however, $w\text{ and }\tilde{\tau}$
are allowed to depend on the output of $\textsc{Solve}$ by Lemma
\ref{lem:totalPmovement}). Also note, that the inputs $\bar{w},b,\delta$
to $\textsc{Solve}$ are allowed to depend on any previous output
of $\textsc{Update}$ and $\textsc{Solve}$, as Theorem \ref{thm:inverse_main}
only has issues with the inputs $w\text{ and }\tilde{\tau}$ to $\textsc{Update}$.
So to show that our algorithm works, we are only left with verifying
that the input $w,\tilde{\tau}$ to $\textsc{Update}$ does not depend
on previous results of $\textsc{Update}$. Let us prove this by induction.

The result of $\textsc{Update}$ is $\Psi^{(\alpha)}\leftarrow D_{\textsc{Inverse}}.\textsc{Update}(\oms^{-1-2\alpha}\ox^{1-2\alpha},\overline{\tau})$,
and when executing this line for the very first time, it can obviously
not depend on a previous output yet. The matrix $\Psi^{(\alpha)}$
is only used as input to $D_{\textsc{Leverage}}.\textsc{Query}(\Psi^{(\alpha)},\Psi\safe^{(\alpha)})$,
but by Theorem \ref{thm:leverage_score_datastructure} the output
of this procedure does not depend on $\Psi^{(\alpha)}$. Hence $\Psi^{(\alpha)}$
does not affect anything else, which also means the input $\oms^{-1-2\alpha}\ox^{1-2\alpha}$
and $\overline{\tau}$ to the next call of $D_{\textsc{Inverse}}.\textsc{Update}$
do not depend on previous $\Psi^{(\alpha)}$.
\end{proof}

\section{Robust Primal Dual LS-Path Following}

\label{sec:primal_dual_path}

Here we provide our stable primal-dual $\otilde(\sqrt{d})$-iteration
primal-dual IPM for solving (\ref{eq:primal_dual}). In particular
we provide the main subroutines for making progress along the central
path as discussed in Section \ref{sec:overview_IPM} and prove Theorem~\ref{alg:pathfollowing}
given in Section~\ref{sub:ipm_restated}.

As discussed we assume throughout that $\ma$ is \emph{non-degenerate},
meaning it has full-column rank and no zero rows and consider the
primal dual central path (\ref{eq:regls_path}) induced by the weight
function $\taureg(x,s)\defeq\sigma(\ms^{-1/2-\alpha}\mx^{1/2-\alpha}\ma)+\frac{d}{n}\vones$
where $\mx=\mdiag(x)$ and $\ms=\mdiag(s)$ (See Definition~\ref{weight_function}).

Rather than precisely following the central path, we follow the approach
of \cite{cohen2019solving,LeeSZ19,br2019deterministicArxiv} and instead
follow it approximately. The primary goal of this section is to design
an efficient primal dual IPM that only maintains $x$, $s$, $\ma^{\top}\ms^{-1/2}\mx^{1/2}\ma$,
and $\taureg(x,s)$ multiplicatively and yet converges in $\otilde(\sqrt{d})$
iterations.

Our algorithms maintain a near feasible point $(x,s)\in\R_{\geq0}^{n}\times\R_{\geq0}^{n}$
(see Definition~\ref{def:feasible_point}) and attempt to improve
the quality of this point with respect to the central path. Formally,
we consider the vector $w=\mx s$, and then attempt to make $w\approx\mu\cdot\taureg(x,s)$.
We measure this quality of approximation or \emph{centrality} as follows.
\begin{defn}[$(\mu,\epsilon)$-centered point]
 We say that point $(x,s)\in\R_{\geq0}^{n}\times\R_{\geq0}^{n}$
is \emph{$(\mu,\epsilon)$-centered} if $\mu>0$ and $w\defeq\mx s$
for $\mw=\mdiag(w)$ satisfies $w\approx_{\epsilon}\mu\cdot\taureg(x,s)$.
\end{defn}

Similar to previous papers \cite{cohen2019solving,LeeSZ19,br2019deterministicArxiv}
our algorithms measure progress or the quality of this approximation
by 
\[
\Phi(x,s,\mu)=\sum_{i\in[n]}\phi\left(\mu\cdot\tau_{i}\cdot w_{i}^{-1}\right)\text{ with }\phi(v)\defeq\exp(\lambda(v-1))+\exp(-\lambda(v-1))
\]
for a setting of $\lambda=2\epsilon^{-1}\log(2^{16}n\sqrt{d}\alpha^{-2})$
we derive later. In contrast to these recent papers, here our potential
function does not depend solely on $w$ as $x$ and $s$ are used
to determine the value of $\tau$. For notational convenience, we
often let $v\in\R^{n}$ denote the vector with $v_{i}=\mu\cdot\tau_{i}\cdot w_{i}^{-1}$
for all $i\in[n]$ and overload notation letting $\Phi(v)\defeq\sum_{i\in[n]}\phi(v_{i})$.

To update our points and improve the potential we consider attempting
to move $w$ in some direction $h$ suggested by $\grad\Phi(v)$.
To do so we solve for $\delta_{s},\delta_{x},\delta_{y}$ satisfying
the following 
\[
\overline{\mx}\delta_{s}+\overline{\ms}\delta_{x}=\omw h\text{, }\ma^{\top}\delta_{x}=0\text{ , and }\ma\delta_{y}+\delta_{s}=0
\]
where $\omw=\overline{\mx}\overline{\ms}$. Solving this system of
equations, we have
\[
\delta_{x}=\omx\omw^{-1/2}(\mi-\mproj(\overline{\ms}^{-1/2}\overline{\mx}^{1/2}\ma))\omw^{1/2}h\text{ and }\delta_{s}=\oms\omw^{-1/2}\mproj(\overline{\ms}^{-1/2}\overline{\mx}^{1/2}\ma)\omw^{1/2}h.
\]
Often, our algorithms we will not have access to $\ma^{\top}\overline{\ms}^{-1}\overline{\mx}\ma$
exactly and instead we will only have a spectral approximation $\mh$.
In these cases we will instead consider $\mq\defeq\overline{\ms}^{-1/2}\overline{\mx}^{1/2}\ma\mh^{-1}\ma^{\top}\omx^{1/2}\oms^{-1/2}$
as a replacement for $\mproj$. Further, to account for the change
due to $\sigma$, we slightly change our steps size amounts. Formally,
we define our steps as follows:
\begin{defn}[Steps]
\label{def:newton_step} We call $(\delta_{x},\delta_{s})$ an \emph{$\epsilon$-Newton
direction from $(x,s)\in\R_{\geq0}^{n}\times\R_{\geq0}^{n}$} if $\ox\approx_{\epsilon}x$,
$\os\approx_{\epsilon}s$, and $\mh\approx_{\epsilon}\ma^{\top}\oms^{-1}\omx\ma$
with $\epsilon\leq1/80$, and $\delta_{x},\delta_{s}\in\R^{n}$ are
given as follows
\[
\delta_{x}=(1+2\alpha)\omx\omw^{-1/2}(\mi-\mq)\omw^{1/2}h\text{ and }\delta_{s}=(1-2\alpha)\oms\omw^{-1/2}\mq\omw^{1/2}h
\]
when
\[
\mq=\overline{\ms}^{-1/2}\overline{\mx}^{1/2}\ma\mh^{-1}\ma^{\top}\omx^{1/2}\oms^{-1/2}\,.
\]
\end{defn}

Note that, since we use $\mh$ and $\mq$ (rather than a true orthogonal
projection) it is not necessarily the case that $\ma^{\top}\delta_{x}=0$.
Consequently, we will need to take further steps to control this error
and this is analyzed in Section~\ref{sec:maintaining_infeasibility}.

Much of the remaining analysis is bounding the effect of such a step
and leveraging this analysis to tune $\phi$ and $h$. In Section~\ref{sec:weight_changes}
we bound the change in the point when we take a Newton step, in Section~\ref{sec:newton_progress}
we bound the effect of a Newton step on centrality for arbitrary $\Phi$,
and in Section~\ref{sec:central_path_follow} we analyze the particular
structure of $\Phi$ and use this to prove Theorem~\ref{alg:pathfollowing}.

\subsection{Stability of Weight Changes}

\label{sec:weight_changes}

Here we bound the change in $x$ and $s$ when we take an Newton step
from a centered point. For a point $(x,s)\in\R_{\geq0}^{n}\times\R_{\geq0}^{n}$
we will wish to analyze the movement both with respect to $\norm{\cdot}_{\infty}$
and $\norm{\cdot}_{\tau}$ for $\tau=\taureg(x,s)$. To simplify this
analysis we define the following mixed norm $\mixedNorm{\cdot}{\tau}$
by
\[
\norm x_{\tau+\infty}=\norm x_{\infty}+\cnorm\norm x_{\tau}\text{ for }\cnorm\defeq\frac{10}{\alpha}\,.
\]
Note that by the definition of $\alpha$ we have $\cnorm\geq10$.
We leverage this norm extensively in our analysis and make frequent
use of the following simple facts about such norms.
\begin{fact}[Norm Facts]
\label{lem:norm_facts} For any $d\in\R^{n}$, we have that $\|\mdiag(d)\|_{\tau+\infty}=\|d\|_{\infty}$.
For any $\mm\in\R^{n\times n}$, we have $\norm{\mm}_{\tau}=\norm{\mT^{1/2}\mm\mT^{-1/2}}_{2}$
for $\mT\defeq\mdiag(\tau)$.
\end{fact}

Next, we relate leverage scores of the projection matrix $\mproj(\ms^{-1/2}\mx^{1/2}\ma)$
considered in taking a Newton step to the leverage scores used to
measure centrality, i.e. $\sigma(\ms^{-1/2-\alpha}\mx^{1/2-\alpha}\ma)$.
\begin{lem}
\label{lem:sigma_approx}For $(\mu,\epsilon)$-centered point $(x,s)$
with $\epsilon\in[0,1/80)$ we have
\[
\frac{1}{2}\sigma(\ms^{-1/2-\alpha}\mx^{1/2-\alpha}\ma)\leq\sigma(\ms^{-1/2}\mx^{1/2}\ma)\leq2\sigma(\ms^{-1/2-\alpha}\mx^{1/2-\alpha}\ma).
\]
\end{lem}

\begin{proof}
By assumption $w=\mx s$ satisfies $w\approx_{\epsilon}\mu[\sigma(\ms^{-1/2-\alpha}\mx^{1/2-\alpha}\ma)+\frac{d}{n}]$.
Since leverage scores lie between $0$ and $1$, this implies $e^{-\epsilon}\frac{d}{n}\leq\mu^{-1}w\leq e^{\epsilon}2$.
Consequently $e^{-\alpha\epsilon}2^{-\alpha}\mi\preceq\mu^{\alpha}\mw^{-\alpha}\preceq e^{\alpha\epsilon}(\frac{d}{n})^{-\alpha}\mi$
and since $\alpha=1/(4\log(4n/d))$ we have that entrywise
\begin{align*}
\sigma(\ms^{-1/2}\mx^{1/2}\mw^{-\alpha}) & \geq\sigma(\ms^{-1/2}\mx^{1/2})\cdot\frac{e^{-2\alpha\epsilon}2^{-2\alpha}}{e^{2\alpha\epsilon}(\frac{d}{n})^{-2\alpha}}=\sigma(\ms^{-1/2}\mx^{1/2})\cdot e^{-4\alpha\epsilon}(2n/d)^{-2\alpha}\\
 & \geq\frac{\sigma(\ms^{-1/2}\mx^{1/2})}{e^{1/2+4\alpha\epsilon}}
\end{align*}
and analogous calculation shows that 
\begin{align*}
\sigma(\ms^{-1/2}\mx^{1/2}\mw^{-\alpha}) & \leq\sigma(\ms^{-1/2}\mx^{1/2})\cdot e^{4\alpha\epsilon}(2n/d)^{2\alpha}\leq\sigma(\ms^{-1/2}\mx^{1/2})e^{1/2+4\alpha\epsilon}\,.
\end{align*}
The result then follows from the assumption $\epsilon\leq1/80$.
\end{proof}
Leveraging Lemma~\ref{lem:sigma_approx} we obtain the following
bounds on $\mq$ and $\omw^{-1/2}\mq\omw^{1/2}h$ for Newton directions.
\begin{lem}
\label{lem:proj_bounds} Let $(\delta_{x},\delta_{s})$ denote an
$\epsilon$-Newton direction from $(\mu,\epsilon$)-centered $(x,s)$
with $\epsilon\in[0,1/80)$. Then $\mq\preceq e^{3\epsilon}\mproj(\ms^{-1/2}\mx^{1/2}\ma)$.
Further, for all $h\in\R^{n}$,
\[
\norm{\omw^{-1/2}\mq\omw^{1/2}h}_{\tau}\leq e^{4\epsilon}\norm h_{\tau}\text{ and }\norm{\omw^{-1/2}(\mi-\mq)\omw^{1/2}h}_{\tau}\leq e^{3\epsilon}\norm h_{\tau}\,.
\]
Further, $\norm{\omw^{-1/2}\mq\omw^{1/2}h}_{\infty}\leq2\norm h_{\tau}$
and therefore $\|\omw^{-1/2}\mq\omw^{1/2}\|_{\tau+\infty}\leq e^{4\epsilon}+2/\cnorm$.
\end{lem}

\begin{proof}
Since $\mh\approx_{\epsilon}\ma^{\top}\oms^{-1}\omx\ma$, we have
$\mh^{-1}\preceq e^{\epsilon}(\ma^{\top}\oms^{-1}\omx\ma)^{-1}$ and
\[
\mq\preceq e^{\epsilon}\overline{\ms}^{-1/2}\overline{\mx}^{1/2}\ma(\ma^{\top}\oms^{-1}\omx\ma)^{-1}\ma^{\top}\omx^{1/2}\oms^{-1/2}=e^{\epsilon}\mproj(\oms^{-1/2}\omx^{1/2}\ma)\preceq e^{\epsilon}\mi\,.
\]
Since $\mq$ is PSD this implies $\norm{\mq}_{2}\leq e^{\epsilon}$
and since $\ow\approx_{2\epsilon}w$ and $w\approx_{\epsilon}\mu\tau$
this implies 
\[
\norm{\omw^{-1/2}\mq\omw^{1/2}h}_{\tau}\leq e^{1.5\epsilon}\mu^{-1/2}\norm{\mq\omw^{1/2}h}_{2}\leq e^{2.5\epsilon}\mu^{-1/2}\norm{\omw^{1/2}h}_{2}\leq e^{4\epsilon}\norm h_{\tau}\,.
\]
Similarly, we have $0\preceq\mq\preceq e^{\epsilon}\mi$ and hence
$\|\mi-\mq\|_{2}\leq1$ (using that $\epsilon\leq1/80$). Therefore,
we have
\[
\norm{\omw^{-1/2}(\mi-\mq)\omw^{1/2}h}_{\tau}\leq e^{1.5\epsilon}\mu^{-1/2}\norm{(\mi-\mq)\omw^{1/2}h}_{2}\leq e^{1.5\epsilon}\mu^{-1/2}\norm{\omw^{1/2}h}_{2}\leq e^{3\epsilon}\norm h_{\tau}\,.
\]

Further, by Cauchy Schwarz and $\mq=\mq^{1/2}\mq^{1/2}$ we have
\begin{align*}
\norm{\omw^{-1/2}\mq\omw^{1/2}h}_{\infty}^{2} & =\max_{i\in[n]}\left[e_{i}^{\top}\omw^{-1/2}\mq\omw^{1/2}h\right]^{2}\leq\max_{i\in[n]}\left[\omw^{-1/2}\mq\omw^{-1/2}\right]_{i,i}\cdot\left[h^{\top}\omw^{1/2}\mq\omw^{1/2}h\right]\,
\end{align*}
Since $\mq\preceq e^{\epsilon}\mi$ we have 
\[
h^{\top}\omw^{1/2}\mq\omw^{1/2}h\leq e^{\epsilon}h^{\top}\omw h\preceq\mu e^{4\epsilon}\norm h_{\tau}^{2}\,.
\]
Further, by $\mq\preceq e^{\epsilon}\mproj(\oms^{-1/2}\omx^{1/2})\preceq e^{3\epsilon}\mproj(\ms^{-1/2}\mx^{1/2})$
and Lemma~\ref{lem:sigma_approx} we have that 
\begin{align*}
\max_{i\in[n]}\left[\omw^{-1/2}\mq\omw^{-1/2}\right]_{i,i} & \leq\max_{i\in[n]}\frac{e^{3\epsilon}\sigma(\ms^{-1/2}\mx^{1/2})_{i}}{\ow_{i}}\\
 & \leq\frac{e^{6\epsilon}}{\mu}\max_{i\in[n]}\frac{\sigma(\ms^{-1/2}\mx^{1/2})_{i}}{\tau_{i}}\\
 & \leq\frac{2e^{6\epsilon}}{\mu}\,.
\end{align*}
Combining these and using $\epsilon\leq1/80$ yields the desired bounds.
\end{proof}
Leveraging Lemma~\ref{lem:sigma_approx} we obtain the following
bounds on the multiplicative stability of Newton directions.
\begin{lem}
\label{lem:step_stability} Let $(\delta_{x},\delta_{s})$ denote
an $\epsilon$-Newton direction from $(\mu,\epsilon$)-centered $(x,s)$
and let $\tau\defeq\taureg(x,s)$. Then the following hold
\[
\begin{array}{rclcrcl}
\|\ms^{-1}\delta_{s}\|_{\tau} & \leq & (1-2\alpha)e^{5\epsilon}\|h\|_{\tau} & \enspace & \|\ms^{-1}\delta_{s}\|_{\infty} & \leq & (1-2\alpha)3\|h\|_{\tau}\,,\\
\|\mx^{-1}\delta_{x}\|_{\tau} & \leq & (1+2\alpha)e^{4\epsilon}\|h\|_{\tau}\,,\text{ and} & \enspace & \|\mx^{-1}\delta_{x}\|_{\infty} & \leq & \left(1+2\alpha\right)\left[e^{\epsilon}\norm h_{\infty}+3\|h\|_{\tau}\right]\,.
\end{array}
\]
\end{lem}

\begin{proof}
Note that $\ms^{-1}\delta_{s}=(1-2\alpha)\ms^{-1}\oms\omw^{-1/2}\mq\omw^{1/2}h$.
Since $s\approx_{\epsilon}\os$ and $\ow\approx_{2\epsilon}w$ the
bounds on $\|\ms^{-1}\delta_{s}\|_{\tau}$ and $\norm{\ms^{-1}\delta_{s}}_{\infty}$
follow immediately from Lemma~\ref{lem:norm_facts} and Lemma~\ref{lem:proj_bounds}.

Similarly, $\mx^{-1}\delta_{x}=(1+2\alpha)\mx^{-1}\omx\omw^{-1/2}(\mi-\mq)\omw^{1/2}h$.
Since $\ox\approx_{\epsilon}x$ and $\ow\approx_{2\epsilon}w$, we
have $\norm{\mx^{-1}\omx h}_{\tau}\leq e^{\epsilon}\norm h_{\tau}$
and $\norm{\mx^{-1}\omx h}_{\infty}\leq e^{\epsilon}\norm h_{\infty}$.
The bounds on $\|\mx^{-1}\delta_{x}\|_{\tau}$ and $\|\mx^{-1}\delta_{x}\|_{\infty}$
follow by triangle inequality and the same derivation as for $\|\ms^{-1}\delta_{s}\|_{\tau}$
and $\norm{\ms^{-1}\delta_{s}}_{\infty}$.
\end{proof}

\subsection{Newton Step Progress}

\label{sec:newton_progress}

Here we analyze the effect of a Newton step from a $(\mu,\epsilon)$
centered point on the potential $\Phi(x,s,\mu)$. We show that up
to some additive error on $h$ in the appropriate mixed norm the potential
decreases by $\grad\Phi(v)^{\top}h$. The main result of this section
is the following.
\begin{thm}
\label{thm:newton_step} Let $(\delta_{x},\delta_{s})$ denote an
$\epsilon$-Newton direction from $(\mu,\epsilon$)-centered $(x,s)$
and suppose that $\gamma\defeq\mixedNorm h{\tau}\leq\epsilon\leq\frac{\alpha}{4000}$
for $\tau\defeq\taureg(x,s)$. There are $v',e\in\R^{n}$ satisfying
\[
\|v'-\mu\mw^{-1}\tau\|_{\infty}\leq10\gamma\text{ and }\mixedNorm e{\tau}\leq\left(1-\frac{\alpha}{2}\right)\gamma
\]
such that
\[
\Phi(x+\delta_{x},s+\delta_{s},\mu)=\Phi(x,s,\mu)-\grad\Phi(v')^{\top}(h+e)
\]
where $[\grad\Phi(v')]_{i}$ is the derivative of $\phi$ at $(v')_{i}$
for all $i\in[n]$.
\end{thm}

We prove this theorem in multiple steps. First, in Lemma~\ref{lem:deriv_exact}
we directly compute the change in the potential function by chain
rule and mean value theorem. Then, in Lemma~\ref{lem:dir_lap_bound}
and Lemma~\ref{lem:stability_helper} we bound the terms in this
change of potential formula and use this to provide a formula for
the approximate the change in the potential in Lemma~\ref{deriv_approx}.
Finally, in Lemma~\ref{deriv_approx} we bound this approximate change
and use this to prove Theorem~\ref{thm:newton_step}.
\begin{lem}
\label{lem:deriv_exact} In the setting of Theorem~\ref{thm:newton_step}
for all $t\in[0,1]$ let $x_{t}=x+t\delta_{x}$ and $s_{t}=s+t\delta_{s}$.
Then for some $t_{*}$, we have that 
\[
\Phi(x+\delta_{x},s+\delta_{s},\mu)-\Phi(x_{0},s_{0},\mu)=-\grad\Phi(v_{t_{*}})^{\top}\mj_{t_{*}}h
\]
where for all $t\in[0,1]$ we let $w_{t}\defeq\mx_{t}s_{t}$,$v_{t}\defeq\mu\mw_{t}^{-1}\tau_{t}$,
$\tau_{t}\defeq\taureg(x_{t},s_{t})$, $\mT_{t}\defeq\mdiag(\tau_{t})$,
$\mLambda_{t}\defeq\mLambda(\mx_{t}^{1/2-\alpha}\ms_{t}^{-1/2-\alpha}\ma)$
and $\mj_{t}\defeq(1-2\alpha)\mj_{t}^{s}+(1+2\alpha)\mj_{t}^{x}$
with
\begin{eqnarray*}
\mj_{t}^{s} & \defeq & \left[\mu\mw_{t}^{-1}\mT_{t}\right]\left[\mi+(1+2\alpha)\mT_{t}^{-1}\mLambda_{t}\right]\ms_{t}^{-1}\oms\omw^{-1/2}\mq\omw^{1/2}\text{ , and }\\
\mj_{t}^{x} & \defeq & \left[\mu\mw_{t}^{-1}\mT_{t}\right]\left[\mi-(1-2\alpha)\mT_{t}^{-1}\mLambda_{t}\right]\mx_{t}^{-1}\omx\omw^{-1/2}(\mi-\mq)\omw^{1/2}\,.
\end{eqnarray*}
\end{lem}

\begin{proof}
By the mean value theorem, there is $t\in[0,1]$ such that
\begin{align*}
\Phi(v_{1})-\Phi(v_{0}) & =\sum_{i\in[n]}\left[\grad\Phi(\mu\cdot w_{t}^{-1}\cdot\tau_{t})\right]_{i}\cdot\mu\cdot\left(\frac{1}{[w_{t}]_{i}}\frac{d[\tau_{t}]_{i}}{dt}-\frac{[\tau_{t}]_{i}}{[w_{t}]_{i}^{2}}\frac{d[w_{t}]_{i}}{dt}\right)\\
 & =\mu\cdot\grad\Phi(v_{t})^{\top}\mw_{t}^{-1}\left(\frac{d\tau_{t}}{dt}-\mT_{t}\mw_{t}^{-1}\frac{dw_{t}}{dt}\right)\,.
\end{align*}
Now, 
\[
\mw_{t}^{-1}\frac{dw_{t}}{dt}=\mw_{t}^{-1}\left[\mx_{t}\delta_{s}+\ms_{t}\delta_{x}\right]=\ms_{t}^{-1}\delta_{s}+\mx_{t}^{-1}\delta_{x}
\]
 and Lemma~\ref{lem:deriv:proj} combined with chain rule implies
that
\begin{align*}
\frac{d\tau_{t}}{dt} & =2\mLambda_{t}\left[\ms_{t}^{-1/2-\alpha}\mx_{t}^{1/2-\alpha}\right]^{-1}\cdot\left[\frac{d}{dt}\ms_{t}^{-1/2-\alpha}x_{t}^{1/2-\alpha}\right]=\mLambda_{t}\left((1-2\alpha)\mx_{t}^{-1}\delta_{x}-(1+2\alpha)\ms_{t}^{-1}\delta_{s}\right)\,.
\end{align*}
Consequently,
\begin{align*}
\mT_{t}\mw_{t}^{-1}\frac{dw_{t}}{dt}-\frac{d\tau_{t}}{dt} & =\left[\mT_{t}+(1+2\alpha)\mLambda_{t}\right]\ms_{t}^{-1}\delta_{s}+\left[\mT_{t}-(1-2\alpha)\mLambda_{t}\right]\mx_{t}^{-1}\delta_{x}\,.
\end{align*}
Combining yields that 
\[
\Phi(v_{1})-\Phi(v_{0})=-\grad\Phi(v_{t})^{\top}\left[\mu\mw_{t}^{-1}\mT_{t}\right]\left(\left[\mi+(1+2\alpha)\mT_{t}^{-1}\mLambda_{t}\right]\ms_{t}^{-1}\delta_{s}+\left[\mi-(1-2\alpha)\mT_{t}^{-1}\mLambda_{t}\right]\mx_{t}^{-1}\delta_{x}\right)\,.
\]
The result then follows by the definition of $\delta_{s}$ and $\delta_{x}$.
\end{proof}
\begin{lem}
\label{lem:dir_lap_bound} In the setting of Lemma~\ref{lem:deriv_exact}
for all $y\in\R^{n}$, $t\in[0,1]$, and $\beta\in[0,1/2]$ we have
\[
\norm{(\mT_{t}^{-1}\mLambda_{t}-\beta\mi)y}_{\infty}\leq\left(1-\beta\right)\norm y_{\infty}+\min\{\norm y_{\infty},\norm y_{\tau_{t}}\}\text{ and }\norm{(\mT_{t}^{-1}\mLambda_{t}-\beta\mi)y}_{\tau_{t}}\leq\left(1-\beta\right)\norm y_{\tau_{t}}\,.
\]
Consequently, $\mixedNorm{\mT_{t}^{-1}\mLambda_{t}}{\tau_{t}}\leq1-\beta+(1+\cnorm)^{-1}$.
\end{lem}

\begin{proof}
Let $\sigma_{t}\defeq\sigma(\mx_{t}^{1/2-\alpha}\ms_{t}^{-1/2-\alpha}\ma)$,
$\mSigma_{t}\defeq\mdiag(\sigma_{t})$, and $\mproj_{t}^{(2)}\defeq\mproj^{(2)}(\mx_{t}^{1/2-\alpha}\ms_{t}^{-1/2-\alpha}\ma)$.
Since $\mLambda_{t}=\mSigma_{t}-\mproj_{t}^{(2)}$, Lemma~\ref{lem:tool:projection_matrices}
and $\sigma_{t}\leq\tau_{t}$ implies that 
\begin{align}
\norm{(\mT_{t}^{-1}\mLambda_{t}-\beta\mi)y}_{\infty} & \leq\norm{\left(\mT_{t}^{-1}\mSigma_{t}-\beta\mi\right)y}_{\infty}+\norm{\mT_{t}^{-1}\mproj_{t}^{(2)}y}_{\infty}\nonumber \\
 & \leq\max\left\{ 1-\beta,\beta\right\} \norm y_{\infty}+\min\{\norm y_{\infty},\norm y_{\tau_{t}}\}.\label{eq:lap_bound_1}
\end{align}
On the other hand, we know that $\mzero\preceq\mLambda_{t}\preceq\mSigma_{t}\preceq\mT_{t}$
by Lemma~\ref{lem:tool:projection_matrices} and that $\sigma_{t}\leq\tau_{t}$.
Consequently, $\mzero\preceq\mT_{t}^{-1/2}\mLambda_{t}\mT_{t}^{-1/2}\preceq\mi$
and Lemma~\ref{lem:norm_facts} yields
\begin{equation}
\norm{\mT_{t}^{-1}\mLambda_{t}-\beta\mi}_{\tau_{t}}\leq\norm{\mT_{t}^{-1/2}\mLambda_{t}\mT_{t}^{-1/2}-\beta\mi}_{2}\leq\max\left\{ 1-\beta,\beta\right\} \,.\label{eq:lap_bound_2}
\end{equation}
Since for all $c\in[0,1]$ we have $\min\{\norm y_{\infty},\norm y_{\tau_{t}}\}\leq c\norm y_{\infty}+(1-c)\norm y_{\tau_{t}}$
and $\max\{1-\beta,\beta\}\leq1-\beta$ as $\beta\leq1/2$ combining
(\ref{eq:lap_bound_1}) and (\ref{eq:lap_bound_2}) yields that for
all $c\in[0,1]$ it is the case that
\[
\mixedNorm{\left(\mT_{t}^{-1}\mLambda_{t}-\beta\mi\right)y}{\tau_{t}}\leq(1-\beta+c)\norm y_{\infty}+(1-c+\cnorm(1-\beta))\cdot\norm y_{\tau_{t}}\,.
\]
Since for $c=1/(1+\cnorm)$ we have $\cnorm c=1-c$ we have $\cnorm\left(1-\beta+c\right)=1-c+\cnorm(1-\beta)$
and the result follows.
\end{proof}
\begin{lem}
\label{lem:stability_helper} In the setting of Lemma~\ref{lem:deriv_exact}
we have that $x_{t}\approx_{\frac{3}{2}\epsilon}x_{0}$ and $s_{t}\approx_{\frac{3}{2}\epsilon}s_{0}$.
Consequently, $\tau_{t}\approx_{3\epsilon}\tau_{0}$ and for all $y\in\R^{n}$
we have
\[
\norm{\omw^{-1/2}\mq\omw^{1/2}h}_{\tau_{t}}\leq e^{10\epsilon}\norm h_{\tau_{t}}\text{ and }\norm{\omw^{-1/2}\mq\omw^{1/2}h}_{\infty}\leq2e^{3\epsilon}\norm h_{\tau_{t}}\,.
\]
For all $\beta\in[0,1/2]$ we have
\[
\mixedNorm{\omw^{-1/2}\mq\omw^{1/2}-\beta\mi}{\tau_{t}}\leq\frac{2e^{3\epsilon}}{\cnorm}+e^{6\epsilon}(e^{16\epsilon}-\beta)\le2\,.
\]
\end{lem}

\begin{proof}
By Lemma~\ref{lem:step_stability}, $\cnorm\geq3$, $\alpha\leq1/5$,
and $\|h\|_{\tau+\infty}\leq\epsilon$ we have
\[
\|\mx^{-1}\delta_{x}\|_{\infty}\leq\left(1+2\alpha\right)\left[e^{\epsilon}\norm h_{\infty}+3\|h\|_{\tau}\right]\leq\frac{7}{5}\epsilon\,.
\]
Consequently, Lemma~\ref{lem:mult_approx} and $\epsilon\in(0,1/100)$
implies that $x_{t}\approx_{\frac{3}{2}\epsilon}x_{0}$ (where we
used that $\frac{7}{5}\epsilon+(\frac{7}{5}\epsilon)^{2}\leq\frac{3}{2}\epsilon$).
Further the same reasoning implies that $s_{t}\approx_{\frac{3}{2}\epsilon}s_{0}$.
Consequently, $\tau_{t}\approx_{2\cdot(3/2)\epsilon}\tau_{0}$. As
$(x,s)$ is $(\mu,\epsilon)$-centered we have, Lemma~\ref{lem:proj_bounds},
$\cnorm\geq10$, and $\epsilon\in(0,1/100)$ imply that for all $h\in\R^{n}$
\begin{equation}
\norm{\omw^{-1/2}\mq\omw^{1/2}h}_{\tau_{t}}\leq e^{10\epsilon}\norm h_{\tau_{t}}\text{ and }\norm{\omw^{-1/2}\mq\omw^{1/2}h}_{\infty}\leq2e^{3\epsilon}\norm h_{\tau_{t}}\,.\label{eq:apx_proj_bound_1}
\end{equation}

Next, since by assumptions $w_{0}\approx_{\epsilon}\mu\cdot\tau_{0}$
and $\ow\approx_{2\epsilon}w_{0}$ we have that $\ow\approx_{6\epsilon}\mu\cdot\tau_{t}$
and therefore the above implies
\begin{align}
\norm{\omw^{-1/2}\mq\omw^{1/2}-\beta\mi}_{\tau_{t}} & \leq e^{6\epsilon}\norm{\omw^{-1/2}\mq\omw^{1/2}-\beta\mi}_{\ow}=e^{6\epsilon}\norm{\mq-\beta\mi}_{2}\nonumber \\
 & \leq e^{6\epsilon}\max\left\{ \beta,\norm{\mq}_{2}-\beta\right\} \leq e^{6\epsilon}\max\left\{ \beta,e^{16\epsilon}-\beta\right\} \label{eq:apx_proj_bound_2}
\end{align}
Combining (\ref{eq:apx_proj_bound_1}) and (\ref{eq:apx_proj_bound_2})
and using that $\beta\in[0,1/2]$ yields that for all $h\in\R^{n}$
\begin{align*}
\mixedNorm{(\omw^{-1/2}\mq\omw^{1/2}-\beta\mi)h}{\tau_{t}} & \leq\beta\norm h_{\infty}+\left(2e^{3\epsilon}+\cnorm e^{6\epsilon}(e^{16\epsilon}-\beta)\right)\norm h_{\tau_{t}}\\
 & \leq\max\left\{ \beta,\frac{2e^{3\epsilon}}{\cnorm}+e^{6\epsilon}(e^{16\epsilon}-\beta)\right\} \mixedNorm h{\tau_{t}}
\end{align*}
and the result follows from the fact that $\beta\in[0,1/2]$ and the
bounds on $\epsilon$ and $\cnorm$.
\end{proof}
\begin{lem}
\label{deriv_approx} In the setting of Lemma~\ref{lem:deriv_exact}
there is a matrix $\me_{t}$ with $\|\me_{t}\|_{\tau_{t}+\infty}\leq1000\epsilon$
such that
\begin{align*}
\mj_{t} & =\mi+\left((1-4\alpha^{2})\mT_{t}^{-1}\mLambda_{t}-2\alpha\mi\right)(2\omw^{-1/2}\mq\omw^{1/2}-\mi)+\me_{t}\,.
\end{align*}
\end{lem}

\begin{proof}
By Lemma~\ref{lem:stability_helper}, $x_{t}\approx_{\frac{3}{2}\epsilon}x_{0}$,
$s_{t}\approx_{\frac{3}{2}\epsilon}s_{0}$, and $\tau_{t}\approx_{3\epsilon}\tau_{0}$.
By definition of an $\epsilon$-Newton direction this further implies
that $x_{t}\approx_{\frac{5}{2}\epsilon}\ox$ and $s_{t}\approx_{\frac{5}{2}\epsilon}\os$.
Further, this implies that $w_{t}\approx_{3\epsilon}w_{0}$ and since
$w_{0}\approx_{\epsilon}\mu\cdot\tau_{0}$ by definition of $(\mu,\epsilon)$-centered
this implies that $w_{t}\approx_{7\epsilon}\mu\cdot\tau_{t}$. Consequently,
Fact~\ref{lem:norm_facts}, Lemma~\ref{lem:mult_approx}, and $\epsilon<1/80$
implies
\[
\mixedNorm{\mu\mw_{t}^{-1}\mT_{t}-\mi}{\tau_{t}}\leq8\epsilon\text{, }\mixedNorm{\ms_{t}^{-1}\oms-\mi}{\tau_{t}}\leq3\epsilon,\text{ and }\mixedNorm{\mx_{t}^{-1}\omx-\mi}{\tau_{t}}\leq3\epsilon\,.
\]
Further, $\mixedNorm{\omw^{-1/2}\mq\omw^{1/2}}{\tau_{t}}\leq2$ by
Lemma~\ref{lem:stability_helper} with $\beta=0$ and $\mixedNorm{\mT_{t}^{-1}\mLambda_{t}}{\tau_{t}}\leq2$
by Lemma~\ref{lem:dir_lap_bound} with $\beta=0$. Combining these
facts and applying Lemma~\ref{lem:norm_facts} we have that there
are $\me_{t}^{s}$ and $\me_{t}^{x}$ with $\norm{\me_{t}^{s}}_{\tau_{t}+\infty}\leq500\epsilon$
and $\norm{\me_{t}^{x}}_{\tau_{t}+\infty}\leq500\epsilon$ with
\begin{eqnarray*}
\mj_{t}^{s} & \defeq & \left(\mi+(1+2\alpha)\mT_{t}^{-1}\mLambda_{t}\right)\omw^{-1/2}\mq\omw^{1/2}+\me_{t}^{s}\\
\mj_{t}^{x} & \defeq & \left(\mi-(1-2\alpha)\mT_{t}^{-1}\mLambda_{t}\right)(\mi-\omw^{-1/2}\mq\omw^{1/2})+\me_{t}^{x}
\end{eqnarray*}
Letting $\me_{t}\defeq(1-2\alpha)\me_{t}^{s}+(1+2\alpha)\me_{t}^{x}$
and leveraging the definitions of $\mj_{t}$ yields that 
\begin{align*}
\mj_{t} & =(1-2\alpha)\mj_{t}^{s}+(1+2\alpha)\mj_{t}^{x}+\me_{t}\\
 & =(1+2\alpha)\mi-4\alpha\omw^{-1/2}\mq\omw^{1/2}-(1-4\alpha^{2})\mT_{t}^{-1}\mLambda_{t}(\mi-2\omw^{-1/2}\mq\omw^{1/2})+\me_{t}\,.
\end{align*}
\end{proof}

\begin{lem}
\label{lem:jacobian_approx_quality} In the setting of Lemma~\ref{deriv_approx}
we have $\mixedNorm{\mj_{t}-\me_{t}-\mi}{\tau_{t}}\leq1-\alpha$.
\end{lem}

\begin{proof}
Lemma~\ref{deriv_approx} implies that for $\beta=2\alpha(1-4\alpha^{2})^{-1}$
\[
\mj_{t}-\me_{t}-\mi=(1-4\alpha^{2})\left(\mT_{t}^{-1}\mLambda_{t}-\beta\mi\right)(2\omw^{-1/2}\mq\omw^{1/2}-\mi)\,.
\]
Further, Lemma~\ref{lem:dir_lap_bound} and $\alpha\in(0,1/5)$ implies
that $\beta\in[0,1/2]$ and therefore
\[
\norm{\mT_{t}^{-1}\mLambda_{t}-\beta\mi}_{\tau_{t}}\leq1-\beta\leq1-2\alpha\,.
\]
Further, Lemma~\ref{lem:stability_helper} implies that 
\begin{align*}
\mixedNorm{2\omw^{-1/2}\mq\omw^{1/2}-\mi}{\tau_{t}} & \leq2\mixedNorm{\omw^{-1/2}\mq\omw^{1/2}-\frac{1}{2}}{\tau_{t}}\\
 & \leq2\left(\frac{2e^{3\epsilon}}{\cnorm}+e^{6\epsilon}\left(e^{16\epsilon}-\frac{1}{2}\right)\right)
\end{align*}
Combining, and using that $\epsilon\leq\frac{\alpha}{4000}$, $\alpha\in[0,1/5]$,
and $\cnorm=10/\alpha$ yields that 
\begin{align*}
\mixedNorm{\mj_{t}-\me_{t}-\mi}{\tau_{t}} & \leq(1-4\alpha^{2})\cdot\left(1-2\alpha\right)\cdot2\left(\frac{2e^{3\epsilon}}{\cnorm}+e^{6\epsilon}\left(e^{16\epsilon}-\frac{1}{2}\right)\right)\\
 & \leq(1-2\alpha)\cdot\left((1/2)\alpha+(1+0.002\alpha)(2+.02\alpha-1\right)\leq1-\alpha\,.
\end{align*}
\end{proof}
We now have everything we need to prove the main theorem.
\begin{proof}[Proof of Theorem~\ref{thm:newton_step}]
 By Lemma~\ref{lem:deriv_exact} we know that 
\[
\Phi(x+\delta_{x},s+\delta_{s},\mu)-\Phi(x_{0},s_{0},\mu)=-\grad\Phi(v_{t_{*}})^{\top}\mj_{t_{*}}h
\]
where by Lemma~\ref{deriv_approx} and Lemma~\ref{deriv_approx}
we know that for some matrix $\me_{t_{*}}$ with $\|\me_{t_{*}}\|_{\tau_{t_{*}}+\infty}\leq1000\epsilon$
it holds that $\mixedNorm{\mj_{t_{*}}-\me_{t_{*}}-\mi}{\tau_{t_{*}}}\leq1-\alpha$.
Hence, we have $\mj_{t_{*}}h=h+e$ where 
\begin{align}
\|e\|_{\tau_{t_{*}}+\infty} & \leq\|(\mj_{t_{*}}-\me_{t_{*}}-\mi)h\|_{\tau_{*}+\infty}+\|\me_{t_{*}}h\|_{\tau_{*}+\infty}\leq(1-\alpha+1000\epsilon)\|h\|_{\tau_{*}+\infty}.\label{eq:e_bound}
\end{align}
Hence, we have
\[
\Phi(x+\delta_{x},s+\delta_{s},\mu)-\Phi(x_{0},s_{0},\mu)=-\grad\Phi(v_{t_{*}})^{\top}(h+e).
\]

Lemma~\ref{lem:step_stability} implies that $x_{t}\approx_{\frac{3}{2}\epsilon}x_{0}$,
$s_{t}\approx_{\frac{3}{2}\epsilon}s_{0}$, and $\tau_{t}\approx_{3\epsilon}\tau_{0}$.
Consequently, $w_{t}\approx_{3\epsilon}w$ and $v_{t_{*}}\approx_{6\epsilon}\mu\mw^{-1}\tau$.
The definition of $(\mu,\epsilon)$-centered implies that $\norm{\mu\mw^{-1}\tau-\vones}_{\infty}\leq\epsilon$
and therefore $\|v_{t_{*}}-\mu\mw^{-1}\tau\|_{\infty}\leq10\epsilon$.
The bound on $\|e\|_{\tau_{t_{*}}+\infty}$ follows from (\ref{eq:e_bound})
and $\|h\|_{\tau_{*}+\infty}\leq(1+5\epsilon)\|h\|_{\tau+\infty}$.
\end{proof}

\subsection{Following the Central Path}

\label{sec:central_path_follow}

Here we use the particular structure of $\Phi$ to analyze the effect
of Newton steps and changing $\mu$ on centrality. First we state
the following lemma about the potential function from \cite{lsJournal19}.
Then we use it to analyze the effect of centering for one step, Lemma~\ref{lem:newton_step_potential},
and we conclude the section by proving a simplified variant of Theorem~\ref{thm:path_following}.
\begin{lem}[Specialized Restatement of Lemma~61 of \cite{lsJournal19}]
\label{lem:smoothing:helper} For all $v\in\Rn$, we have
\begin{equation}
e^{\lambda\|v-\vones\|_{\infty}}\leq\Phi(v)\leq2n\cdot e^{\lambda\|v-\vones\|_{\infty}}\enspace\text{ and }\enspace\lambda\Phi(v)-2\lambda n\leq\|\nabla\Phi(v)\|_{1}\label{eq:smoothing:pot_prop_1}
\end{equation}
Furthermore, let $v^{\flat(\tau)}\defeq\argmax_{\mixedNorm w{\tau}\leq1}\left\langle v,w\right\rangle $
and $\|v\|_{*}=\max_{\|w\|_{\tau+\infty}\leq1}\left\langle v,w\right\rangle $.
Then, for all $v,w\in\Rn$ with $\|v-w\|_{\infty}\leq\delta\leq\frac{1}{5\lambda}$,
we have
\begin{equation}
e^{-\lambda\delta}\norm{\grad\Phi(w)}_{*}-\lambda n\leq\left\langle \grad\Phi(v),\grad\Phi(w)^{\flat}\right\rangle \leq e^{\lambda\delta}\norm{\grad\Phi(w)}_{*}+\lambda e^{\lambda\delta}n\label{eq:smoothing:pot_prop_2}
\end{equation}
and consequently,
\begin{equation}
e^{-\lambda\delta}\norm{\grad\Phi(w)}_{*}-\lambda n\leq\norm{\grad\Phi(v)}_{*}\leq e^{\lambda\delta}\norm{\grad\Phi(w)}_{*}+\lambda e^{\lambda\delta}n.\label{eq:smoothing:pot_prop_3}
\end{equation}
\end{lem}

\begin{rem*}
When $\tau$ is clear from context we also write $v^{\flat}$ for
$v^{\flat(\tau)}$.
\end{rem*}
Here we provide a general lemma bounding how much $\Phi$ can increase
for a step of bounded size in the mixed norm.
\begin{lem}[Potential Increase]
\label{lem:pot_increase}  Let $v_{0},v_{1}\in\R^{n}$ such that
$\mixedNorm{v_{1}-v_{0}}{\tau}\leq\delta\leq\frac{1}{5\lambda}$ for
some $\tau\in\R_{>0}^{n}$. Then for $\norm{\cdot}_{*}$, the dual
norm to $\mixedNorm{\cdot}{\tau}$ we have $\Phi(v_{1})\leq\Phi(v_{0})+\delta e^{\delta\lambda}(\norm{\grad\Phi(v_{0})}_{*}+\lambda n)$.
\end{lem}

\begin{proof}
For all $t\in[0,1]$ let $v_{t}\defeq v_{0}+t(v_{1}-v_{0})$. Note,
that this implies that $\mixedNorm{v_{t}-v_{0}}{\tau}\leq t\delta\leq\delta$
for all $t\in[0,1]$ and therefore, integrating and applying Cauchy
Schwarz for arbitrary norms (with $\norm{\cdot}_{*}$ as the dual
norm of $\mixedNorm{\cdot}{\tau}$) implies that 
\begin{align*}
\Phi(v_{1})-\Phi(v_{0}) & =\int_{0}^{1}\grad\Phi(v_{t})^{\top}(v_{1}-v_{0})dt\leq\int_{0}^{1}\norm{\grad\Phi(v_{t})}_{*}\norm{v_{1}-v_{0}}dt\\
 & \leq\int_{0}^{1}e^{\lambda\delta}\left(\norm{\grad\Phi(v_{0})}_{*}+\lambda n\right)\delta dt
\end{align*}
where in the second line we applied (\ref{eq:smoothing:pot_prop_3})
of Lemma~\ref{lem:smoothing:helper} as $t\delta\leq1/(5\lambda)$.
\end{proof}
\begin{lem}[Potential Parameter Stability]
\label{lem:pot_param_stab} Let $(x_{0},s)$ be a $(\mu,\epsilon)$-centered
point and for all $t\in[0,1]$ let $x_{t}\defeq x_{0}+t\delta_{x}$
where $\mixedNorm{\mx_{0}^{-1}\delta_{x}}{\tau_{0}}\leq\epsilon\leq1/100$.
If for all $t\in[0,1]$ we let $w_{t}\defeq\mx_{t}s$, $\tau_{t}\defeq\taureg(x_{t},s)$,
and $v_{t}\defeq\mu\mw_{t}^{-1}\tau_{t}$ then $\mixedNorm{v_{1}-v_{0}}{\tau_{0}}\leq16\mixedNorm{\mx_{0}^{-1}\delta_{x}}{\tau_{0}}$.
\end{lem}

\begin{proof}
Following the calculations in Lemma~\ref{lem:deriv_exact} we have
that for $\mT_{t}\defeq\mdiag(\tau_{t})$
\[
v_{1}-v_{0}=\int_{0}^{1}\mu\mw_{t}^{-1}\left(\frac{d}{dt}\tau_{t}-\mT_{t}\mw_{t}^{-1}\frac{d}{dt}w_{t}\right)dt
\]
and for $\mLambda_{t}\defeq\mLambda(\mx_{t}^{1/2-\alpha}\ms_{t}^{-1/2-\alpha}\ma)$
\[
\mw_{t}^{-1}\frac{d}{dt}=\mx_{t}^{-1}\delta_{x}\text{ and }\frac{d}{dt}\tau_{t}=(1-2\alpha)\mLambda_{t}\mx_{t}^{-1}\delta_{x}\,.
\]
Consequently,
\[
v_{1}-v_{0}=-\int_{0}^{1}\mu\mw_{t}^{-1}\mT_{t}\left(\left(1-2\alpha\right)\mT_{t}^{-1}\mLambda_{t}-\mi\right)\mx_{t}^{-1}\delta_{x}dt\,.
\]
Further, since $\norm{\mx^{-1}\delta_{x}}_{\infty}\leq\epsilon\leq1/(100)$
we have that 
\begin{align*}
\mixedNorm{v_{1}-v_{0}}{\tau_{0}} & \leq\int_{0}^{1}\mixedNorm{\mu\mw_{t}^{-1}\mT_{t}\left(\left(1-2\alpha\right)\mT_{t}^{-1}\mLambda_{t}-\mi\right)\mx_{t}^{-1}\delta_{x}}{\tau_{0}}dt\\
 & \leq2\int_{0}^{1}\mixedNorm{\mu\mw_{t}^{-1}\mT_{t}}{\tau_{t}}\mixedNorm{\left(1-2\alpha\right)\mT_{t}^{-1}\mLambda_{t}-\mi}{\tau_{t}}\mixedNorm{\mx_{t}^{-1}\mx_{0}}{\tau_{t}}\mixedNorm{\mx_{0}^{-1}\delta_{x}}{\tau_{0}}\,.
\end{align*}
Now, the proof of Lemma~\ref{lem:dir_lap_bound} shows that $\mixedNorm{\left(1-2\alpha\right)\mT_{t}^{-1}\mLambda_{t}-\mi}{\tau_{t}}\leq2$
and that $\norm{\mx^{-1}\delta_{x}}_{\infty}\leq\epsilon\leq1/100$
and $(x_{0},s)$ is $(\mu,\epsilon)$-centered shows that $\mixedNorm{\mu\mw_{t}^{-1}\mT_{t}}{\tau_{t}}\le2$
and $\mixedNorm{\mx_{t}^{-1}\mx_{0}}{\tau_{t}}\leq2$. Combining yields
the result.
\end{proof}
\begin{lem}
\label{lem:newton_step_potential} Let $(\delta_{x},\delta_{s})$
denote an $\epsilon$-Newton direction from $(\mu,\epsilon$)-centered
$(x,s)$ with $h=\gamma\nabla\Phi(\overline{v})^{\flat(\otau)}$ for
$\otau\approx_{\epsilon}\tau$ and $\|\overline{v}-v\|_{\infty}\leq\gamma$
for $v\defeq\mu\mw^{-1}\taureg(x,s)$ where $w\defeq\mx s$, $\gamma\leq\min\{\epsilon,\frac{\alpha}{50\lambda}\}$,
$\epsilon\leq\frac{\alpha}{4000}$, and $\lambda\geq\log n$. Then
\[
\Phi(x+\delta_{x}+e_{x},s+\delta_{s},\mu+\delta_{\mu})\leq\Big(1-\frac{\gamma\alpha^{2}\lambda}{640\sqrt{d}}\Big)\cdot\Phi(x,s,\mu)+4\gamma\lambda n
\]
for any $\delta_{\mu}\in\R$ with $\left|\delta_{\mu}\right|\leq\frac{\gamma\alpha}{2^{15}\sqrt{d}}\mu$
and $e_{x}\in\R^{n}$ with $\mixedNorm{\mx^{-1}e_{x}}{\tau}\leq\frac{\gamma\alpha}{2^{20}}$.
\end{lem}

\begin{proof}
Let $\Delta=\Phi(x+\delta_{x},s+\delta_{s},\mu)-\Phi(x,s,\mu)$. By
Theorem~\ref{thm:newton_step}, we have that for $e,\tau,v'\in\R^{n}$
given by this theorem it is the case that 
\begin{align*}
\Delta & =-\grad\Phi(v')^{\top}(h+e)\\
 & \leq-\gamma\left\langle \nabla\Phi(\overline{v})^{\flat(\otau)},\grad\Phi(v')\right\rangle +\|\grad\Phi(v')\|_{*}\|e\|_{\tau+\infty}\\
 & \leq-\gamma e^{-\gamma}\left\langle \nabla\Phi(\overline{v})^{\flat(\tau)},\grad\Phi(v')\right\rangle +\|\grad\Phi(v')\|_{*}\|e\|_{\tau+\infty}\\
 & \leq-\gamma e^{-11\lambda\gamma-\epsilon}\|\grad\Phi(v')\|_{*}+\gamma\lambda n+\left(1-\frac{\alpha}{2}\right)\gamma\|\grad\Phi(v')\|_{*}
\end{align*}
where we used (\ref{eq:smoothing:pot_prop_2}), $\|\overline{v}-v'\|_{\infty}\leq\|\overline{v}-v\|_{\infty}+\|v-v'\|_{\infty}\leq11\gamma$
and Theorem \ref{thm:newton_step}. Since $\gamma\leq\frac{\alpha}{50\lambda}$
and $\epsilon\le\frac{\alpha}{4000}$, (\ref{eq:smoothing:pot_prop_3})
implies that
\begin{align}
\Delta & \le-\gamma\frac{\alpha}{4}\|\grad\Phi(v')\|_{*}+\gamma\lambda n\leq-\gamma\frac{\alpha}{8}\|\grad\Phi(v)\|_{*}+2\gamma\lambda n\,.\label{eq:step_1}
\end{align}

Now, let $x_{1}\defeq x+\delta_{x}$ and $x_{2}\defeq x+\delta_{x}+e_{x}$.
Further, define $w_{1}\defeq\mx_{1}(s+\delta_{s})$, $w_{2}=\mx_{2}(s+\delta_{s})$,
$\tau_{1}=\taureg(x_{1},s+\delta_{s})$, $\tau_{2}=\taureg(x_{2},s+\delta_{s})$,
$v_{1}=\mu\mw_{1}^{-1}\tau_{1}$, and $v_{2}=\mu\mw_{2}^{-1}\tau_{2}$.
Note, that $\mixedNorm{\mx_{1}^{-1}e_{x}}{\tau_{1}}\leq2\mixedNorm{\mx^{-1}e_{x}}{\taureg(x,s)}\leq1/100$
and therefore $\mixedNorm{v_{2}-v_{1}}{\tau_{1}}\leq16\mixedNorm{\mx_{1}^{-1}e_{x}}{\tau_{1}}$
by Lemma~\ref{lem:pot_param_stab}. This implies that, $v_{3}\defeq(\mu+\delta_{\mu})\mw_{2}^{-1}\tau_{2}$
satisfies
\begin{align*}
\mixedNorm{v_{3}-v_{1}}{\tau_{1}} & \leq\mixedNorm{v_{3}-v_{2}}{\tau_{1}}+\mixedNorm{v_{2}-v_{1}}{\tau_{1}}\\
 & \leq\frac{\delta_{\mu}}{\mu}\mixedNorm{v_{2}}{\tau_{1}}+16\mixedNorm{\mx^{-1}e_{x}}{\tau_{1}}\,.
\end{align*}
Further $\norm{\tau_{1}}_{1}=2d$ and $\norm{v_{2}}_{\infty}\leq2$
imply that
\begin{align*}
\mixedNorm{v_{2}}{\tau_{1}} & \leq\norm{v_{2}}_{\infty}+\sqrt{\sum_{i\in[n]}[v_{2}]_{i}^{2}\cdot[\tau_{1}]_{i}}\\
 & \leq\left(1+\sqrt{\norm{\tau_{1}}_{1}}\right)\norm{v_{2}}_{\infty}\leq4\sqrt{d}\,.
\end{align*}
Using the bounds on $\delta_{\mu}$ and $\mixedNorm{\mx^{-1}e_{x}}{\tau_{1}}\leq\gamma\alpha2^{-15}$
this implies that $\mixedNorm{v_{3}-v_{1}}{\tau_{1}}\leq\gamma\alpha2^{-10}$.
Consequently, applying Lemma~\ref{lem:pot_increase} yields
\begin{align*}
\Phi(x+\delta_{x}+e_{x},s+\delta_{s},\mu+\delta_{\mu})-\Phi(x+\delta_{x},s+\delta_{s},\mu) & \leq\gamma\alpha2^{-10}e^{\lambda\gamma\alpha2^{-10}}(\norm{\grad\Phi(v_{1})}_{*}+\lambda n)\\
 & \leq\gamma\alpha2^{-10}e^{\lambda\gamma\alpha2^{-10}}\left(e^{10\gamma\lambda}\left(\norm{\grad\Phi(v)}_{*}+\lambda n\right)+\lambda n\right)\\
 & \leq\gamma\alpha2^{-8}\norm{\grad\Phi(v)}_{*}+\gamma\alpha2^{-6}n\,.
\end{align*}
where in the second line we applied Lemma~\ref{lem:smoothing:helper}
and $\norm{v-v_{1}}_{\infty}\leq10\lambda$. Combining with (\ref{eq:step_1})
yields
\[
\Phi(x+\delta_{x}+e_{x},s+\delta_{s},\mu+\delta_{\mu})\geq\Phi(x,s,\mu)-\frac{\gamma\alpha}{16}\norm{\grad\Phi(v)}_{*}+3\gamma\lambda n\,.
\]

Finally, for any vector $u$, we have $\|u\|_{*}\geq\left\langle u,\frac{\alpha\text{\ensuremath{\cdot}sign}(u)}{40\sqrt{d}}\right\rangle =\frac{\alpha}{40\sqrt{d}}\|u\|_{1}$
as
\begin{align*}
\left\Vert \frac{\alpha\text{\ensuremath{\cdot}sign}(u)}{40\sqrt{d}}\right\Vert _{\tau+\infty} & \leq\frac{1}{4}+\cnorm\frac{\alpha}{40\sqrt{d}}\left(\sum_{i\in[n]}\tau_{i}\right)^{1/2}\\
 & =\frac{1}{4}+\cnorm\frac{\alpha\sqrt{2d}}{40\sqrt{d}}\leq1.
\end{align*}
Hence, by Lemma~\ref{lem:smoothing:helper} we have
\begin{align*}
\norm{\grad\Phi(v)}_{*} & \geq\frac{\alpha}{40\sqrt{d}}\norm{\grad\Phi(v)}_{1}\\
 & \geq\frac{\alpha\lambda}{40\sqrt{d}}\left(\Phi(v)-2n\right)
\end{align*}
and combining yields the desired result.
\end{proof}
We now prove Theorem~\ref{thm:path_following_simplified}, which
is a simplified variant of Theorem~\ref{thm:path_following}, where
we do not claim that the result $x^{\mathrm{(final)}}$ is nearly
feasible. For Theorem~\ref{thm:path_following_simplified}, we analyze
Algorithm~\ref{alg:pathfollowing} under the assumption that it does
not execute Line~\ref{line:pf_maintainFeasibility} (i.e. the function
$\textsc{MaintainFeasibility}$ is not called). We show in Appendix~\ref{sec:maintaining_infeasibility}
that this result can be extended to return a nearly feasible solution,
i.e. how to extend Theorem~\ref{thm:path_following_simplified} to
Theorem~\ref{thm:path_following} by calling Line~\ref{line:pf_maintainFeasibility}
of Algorithm~\ref{alg:pathfollowing}.
\begin{thm}
\label{thm:path_following_simplified} Assume that the movement $e_{x}$
of $x$ induced by Line~\ref{line:pf_move_x_improve_feasible} of
Algorithm \ref{alg:pathfollowing} satisfies $\mixedNorm{\mx^{-1}e_{x}}{\tau}\leq\frac{\gamma\alpha}{2^{20}}$.
Given $x^{\textrm{\ensuremath{\mathrm{(init)}}}},s^{\textrm{\ensuremath{\mathrm{(init)}}}}\in\Rn_{>0}$,
$\mu^{\textrm{\ensuremath{\mathrm{(init)}}}}>0$, $\mu^{\textrm{\ensuremath{\mathrm{(target)}}}}>0$,
and $\epsilon\in(0,\alpha/16000)$ with $x^{\textrm{\ensuremath{\mathrm{(init)}}}}s^{\textrm{\ensuremath{\mathrm{(init)}}}}\approx_{2\epsilon}\mu^{\textrm{\ensuremath{\mathrm{(init)}}}}\cdot\tau(x^{\textrm{\ensuremath{\mathrm{(init)}}}},s^{\textrm{\ensuremath{\mathrm{(init)}}}})$
Algorithm~\ref{alg:pathfollowing} outputs $(x^{\mathrm{(final)}},s^{\mathrm{(final)}})$
such that 
\[
x^{\mathrm{(final)}}s^{\mathrm{(final)}}\approx_{\epsilon}\mu^{\mathrm{(target)}}\cdot\tau(x^{\mathrm{(final)}},s^{\mathrm{(final)}})
\]
in
\[
O\left(\sqrt{d}\log(n)\cdot\left(\log\left(\frac{\mu^{\textrm{\ensuremath{\mathrm{(target)}}}}}{\mu^{\textrm{\ensuremath{\mathrm{(init)}}}}}\right)/(\epsilon\alpha)+\frac{1}{\alpha^{3}}\right)\right)
\]
iterations.

Furthermore, during the algorithm \ref{alg:pathfollowing}, we have
\begin{itemize}
\item $xs\approx_{4\epsilon}\mu\cdot\tau(x,s)$ for some $\mu$ where $(x,s)$
is immediate points in the algorithms
\item $\|\mx^{-1}\delta_{x}\|_{\tau+\infty}\leq\frac{\epsilon}{2}$, $\|\ms^{-1}\delta_{s}\|_{\tau+\infty}\leq\frac{\epsilon}{2}$,
and $\|\diag(\tau)^{-1}\delta_{\tau}\|_{\tau+\infty}\leq2\epsilon$
where $\delta_{x}$, $\delta_{s}$ and $\delta_{\tau}$ is the change
of $x$, $s$ and $\tau$ in one iteration and $\norm a_{\tau+\infty}\defeq\norm x_{\infty}+\cnorm\norm x_{\tau}$
for $\cnorm\defeq\frac{10}{\alpha}$ (See Section~\ref{sec:weight_changes}).
\end{itemize}
\end{thm}

\begin{proof}
Initially, we have $\Phi(v)\leq2n\cdot e^{3\epsilon\lambda}$ because
of $x^{\textrm{(init)}}s^{\textrm{(init)}}\approx_{2\epsilon}\mu^{\textrm{(init)}}\cdot\tau(x^{\textrm{(init)}},s^{\textrm{(init)}})$
and (\ref{eq:smoothing:pot_prop_1}) of Lemma~\ref{lem:smoothing:helper}. 

We proceed to show that in each iteration $\Phi(v)\leq\min\{2n\cdot e^{3\epsilon\lambda},2^{16}n\sqrt{d}\alpha^{-2}\}$.
Note that, by Lemma~\ref{lem:smoothing:helper} in any iteration
this holds, this implies that $\norm{v-1}_{\infty}\leq\frac{1}{\lambda}\log\left(\Phi(v)\right)$.
Since $\lambda=\frac{2}{\epsilon}\log(2^{16}n\sqrt{d}/\alpha^{2})\geq\frac{2}{\epsilon}\log(2n)$
this implies that $\norm{v-1}_{\infty}\leq3.5\epsilon$ which by the
choice of $\epsilon$ implies that $x,s$ is $(\mu,4\epsilon)$ centered.
Further, Lemma~\ref{lem:newton_step_potential} and the design of
the algorithm then imply that if $\Phi(v)\geq2^{16}n\sqrt{d}\alpha^{-2}$
in this iteration, then 
\begin{equation}
\Phi(x+\delta_{x}+e_{x},s+\delta_{s},\mu+\delta_{\mu})\leq\Big(1-\frac{\gamma\alpha^{2}\lambda}{1280\sqrt{d}}\Big)\cdot\Phi.\label{eq:phi_decrement}
\end{equation}
Consequently, by induction the claim holds for all iterations and
$xs\approx_{4\epsilon}\mu\cdot\taureg(x,s)$ for all iterations as
desired.

Now recall that $\alpha=1/(4\log(4n/d))\geq1/(2^{8}n)$ and consequently
as $d\leq n$ we have
\begin{align*}
\lambda & =\frac{2}{\epsilon}\log(2^{16}n\sqrt{d}/\alpha^{2})\\
 & \leq\frac{2}{\epsilon}\log\left(2^{25}n^{4}\right)\\
 & \leq\frac{8}{\epsilon}\log\left(7n\right)
\end{align*}
and
\begin{align*}
\gamma & =\min\left\{ \frac{\epsilon}{4},\frac{\alpha}{50\lambda}\right\} \\
 & =\frac{\alpha}{50\lambda}\\
 & \geq\frac{\alpha\epsilon}{400\log(7n)}\,.
\end{align*}
Consequently, the reasoning in the preceding paragraph implies that
in each step $\mu$ is moved closer to $\mu^{\textrm{(target)}}$
by a $1-\frac{\gamma\alpha}{2^{15}\sqrt{d}}=1-\Omega(\frac{\epsilon\alpha}{\log(n)\sqrt{d}})$
multiplicative factor. Hence, it takes $O(\frac{\sqrt{d}\log(n)}{\epsilon\alpha}\log(\frac{\mu^{\textrm{(target)}}}{\mu^{\textrm{(init)}}}))$
iterations to arrive $\mu^{\textrm{(target)}}$. Further, (\ref{eq:phi_decrement})
shows that it takes $O(\sqrt{d}\log(n)/\alpha^{3})$ iterations to
decrease $\Phi$ to $2^{16}n\sqrt{d}/\alpha^{2}$. Hence, in total,
it takes

\[
O\left(\sqrt{d}\log(n)\cdot\left(\frac{1}{\epsilon\alpha}\cdot\log\left(\frac{\mu^{\textrm{(target)}}}{\mu^{\textrm{(init)}}}\right)+\frac{1}{\alpha^{3}}\right)\right)
\]
iterations for the algorithm to terminate. Further, when the algorithm
terminates, we have that $x^{(\textrm{final})}s^{(\textrm{final})}\approx_{\epsilon}\mu^{\textrm{(target)}}\cdot\tau(x^{(\textrm{final})},s^{(\textrm{final})})$
because $\Phi\leq2^{16}n\sqrt{d}\alpha^{-2}$, $\lambda=2\epsilon^{-1}\log(2^{16}n\sqrt{d}\alpha^{-2})$
and (\ref{eq:smoothing:pot_prop_1}).

Finally, the bounds on the movement of $x$, $s$, follow from Lemma~\ref{lem:step_stability},
the condition $\mixedNorm{\mx^{-1}e_{x}}{\tau}\leq\frac{\gamma\alpha}{2^{20}}$,
and the choice of $\gamma$. The movement of $\tau$ follows from
Lemma 14 in \cite{lee2015efficient}.
\end{proof}

\inputencoding{latin9}\global\long\def\dsquery{\textsc{Query}}%
\global\long\def\dsinit{\textsc{Initialize}}%
\global\long\def\dsmark{\textsc{MarkExact}}%

\section{Vector Maintenance\label{sec:vec_maintenance}}

In this section we prove Theorem~\ref{thm:product_sum_datastructure}
(presented in Section~\ref{sub:vector_restated} and restated below)
by providing and analyzing our vector maintenance data-structure.
Recall that vector maintenance asks to maintain an approximation of
the vector $y^{(t+1)}:=\sum_{k\in[t]}\mg^{(k)}\ma h^{(k)}+\delta^{(k)}$
for an online sequence of diagonal matrices $\mg^{(1)},\mg^{(2)},...,\mg^{(T)}\in\R^{n\times n}$
and vectors $h^{(1)},...,h^{(T)}\in\R^{d}$ and $\delta^{(1)},...,\delta^{(T)}\in\R^{n}$.
This data-structure is used in Section~\ref{sec:lp_algorithms} to
maintain approximately feasible points.

\productsumdatastructure*

We split the proof of Theorem~\ref{thm:product_sum_datastructure}
into three parts. First, in Section~\ref{sub:large_entry_datastructure},
we show how to use sparse-recovery algorithms to build a data-structure
for quickly computing the large entries (i.e. $\ell_{2}$-heavy hitters)
of the matrix-vector product $\ma h$. Second, in Section~\ref{sub:accumulated_matrix_product},
we show how to use the data-structure of Section~\ref{sub:large_entry_datastructure}
to obtain a new data-structure that can quickly approximate the partial
sum $\sum_{k=t_{0}}^{t_{1}}\mg^{(k)}\ma h^{(k)}$, in time roughly
proportional to $d\cdot\|\sum_{k=t_{0}}^{t_{1}}\mg^{(k)}\ma h^{(k)}\|_{2}^{2}$.
The naive way of computing an approximation of $y^{(t)}$, by using
this data-structure for $t_{0}=1$ and $t_{1}=t$, would be prohibitively
slow and therefore, in the last Section~\ref{sub:product_sum_datastructure},
we obtain the main result Theorem~\ref{thm:product_sum_datastructure}
by instead splitting the sum $\sum_{k=1}^{t}\mg^{(k)}\ma h^{(k)}$
into roughly $O(\log n)$ partial sums, each of which is approximated
efficiently by the data-structure of Section~\ref{sub:accumulated_matrix_product}.

\subsection{$\ell_{2}$-heavy Hitter\label{sub:large_entry_datastructure}}

Here we show how to preprocess any matrix $\ma$, such that we can
build a data structure which supports quickly computing the large
entries of the product $\mg\ma h$ for changing diagonal $\mg$.
\begin{lem}[Approximate matrix-vector product]
\label{lem:large_entry_datastructure} There exists a Monte-Carlo
data-structure (Algorithm~\ref{alg:large_entry_datastructure}),
that works against an adaptive adversary, with the following procedures:
\begin{itemize}
\item \textsc{$\dsinit(\ma,g)$:} Given a matrix $\ma\in\R^{n\times d}$
and a scaling $g\in\R^{n}$, the data-structure preprocesses in $O(\nnz(\ma)\log^{4}n)$
time.
\item \textsc{$\dsquery(h,\epsilon)$:} Given a vector $h\in\R^{d}$, w.h.p.
in $n$, the data-structure outputs the vector $v\in\R^{n}$ such
that
\[
v_{i}=\begin{cases}
(\mg\ma h)_{i}, & \text{if }|(\mg\ma h)_{i}|\ge\varepsilon\\
0, & \text{otherwise}
\end{cases}
\]
for $\mg=\mdiag(g)$ in time $O(\|\mg\ma h\|^{2}\cdot\varepsilon^{-2}\cdot d\log^{3}n)$.
\item \textsc{Scale($i,u$):} Sets $g_{i}\leftarrow u$ in $O(d\log^{4}n)$
time.
\end{itemize}
\end{lem}

\renewcommand{\Comment}[1]{\vspace{0.05 in}\tcp{#1}}

\begin{algorithm2e}[!t]

\caption{Approximate Matrix-vector Product (Lemma~\ref{lem:large_entry_datastructure})
\label{alg:large_entry_datastructure}}

\SetKwProg{Members}{members}{}{}

\SetKwProg{Proc}{procedure}{}{}

\Members{}{

$\ma\in\R^{n\times d}$, $g\in\R^{n}$ \tcp*{Input matrix $\ma$
and row scaling $g$}

$\mPhi_{j}=\sketch(2^{-j},n)\in\R^{m_{j}\times n}$ for all $j\in[\log_{2}(n)/2]$
\tcp*{Sketching matrices }

$\mm_{j}\in\R^{m_{j}\times d}$ for all $j\in[\log_{2}(n)/2]$ \tcp*{Sketched
input matrix $\mm_{j}=\mPhi_{j}\mg\ma$}

$\mr=\JL(1/100,n)\in\R^{k\times n}$ with $k=O(\log n)$\tcp*{$(1\pm\frac{1}{100})$-approximate
JL-matrix }

$\mj\in\R^{k\times d}$ \tcp*{Sketched input matrix $\mj=\mr\mg\ma$}

}

\vspace{0.1 in}

\Proc{\textsc{Initialize}$(\ma\in\R^{n\times d},g\in\R^{n})$}{

$\mPhi_{j}\leftarrow\sketch(2^{-j},n)$ for all $j\in[\log_{2}(n)/2]$
\tcp*{Lemma \ref{lem:ell_2_heavy_hitter}}

$\mr\leftarrow\JL(1/100,n)$ \tcp*{Lemma \ref{lem:JL_lemma}}

$\mm_{j}\leftarrow\mPhi_{j}\mg\ma$ for $j\in[\log_{2}(n)/2]$

$\mj\leftarrow\mr\mg\ma$, $\ma\leftarrow\ma$, $g\leftarrow g$

}

\vspace{0.1 in}

\Proc{\textsc{Scale}$(i\in[n],u\in\R)$}{

$\mm_{j}\leftarrow\mm_{j}+(u-g_{i})\mPhi_{j}\indicVec i\indicVec i^{\top}\ma$
for $j\in[\log_{2}(n)/2]$

$\mj\leftarrow\mj+(u-g_{i})\mr e_{i}e_{i}^{\top}\ma$

$g_{i}\leftarrow u$

}

\vspace{0.1 in}

\Proc{\textsc{Query}$(h\in\R^{d},\epsilon\in(0,1))$}{$r\leftarrow\|\mj h\|_{2}$,
$j\leftarrow1+\left\lceil \log_{2}(r/\epsilon)\right\rceil $ \tcp*{Estimate
$\|\mg\ma h\|_{2}$ and which $\mm_{j}$ to consider}

\Comment{Compute $\ma h$ directly if we have the budget to do so}

\lIf{$j\ge\log_{2}(n)/2$} {\Return $\mg\ma h$ }

\Comment{Otherwise compute $\ma h$ approximately using $\mm_{j}$}

Compute the list $L\subset[n]$ of possible heavy hitters of $\mg\ma h$
from $\mm_{j}h$ via Lemma \ref{lem:ell_2_heavy_hitter}

$v\leftarrow\mzero_{n}$

\For{$i\in L$}{

\lIf{$(\mg\ma h)_{i}\ge\varepsilon$}{$v_{i}\leftarrow(\mg\ma h)_{i}$}

}

\Return $v$

}

\end{algorithm2e}

The proof of Lemma~\ref{lem:large_entry_datastructure} is based
on $\ell_{2}$-heavy hitter sketches. While there exist many results
for sketching and computing heavy hitters with different trade-offs
\cite{cm04,nnw14,ccf02,knpw11,p13,ch09,lnnt16}, we only require the
following Lemma~\ref{lem:ell_2_heavy_hitter} adapted from \cite{knpw11,p13}.
\begin{lem}[$\ell_{2}$-heavy hitter, \cite{knpw11,p13}]
\label{lem:ell_2_heavy_hitter} There exists a function $\sketch(\epsilon,n)$
that given $\epsilon>0$ explicitly returns a matrix $\mPhi\in\R^{m\times n}$
with \textup{$m=O(\varepsilon^{-2}\log^{2}n)$} and column sparsity
$c=O(\log^{2}n)$ in $O(nc+m\poly\log n)$ time, and uses $O(nc)$
spaces to store the matrix $\Phi$. There further exists a function
$\hitter(y)$ that given a vector $y=\mPhi x$ in time $O(\varepsilon^{-2}\log^{2}n)$
reports a list $L\subset[n]$ of size $O(\epsilon^{-2})$ that with
probability at least $9/10$ includes all $i$ with 
\[
|x_{i}|\geq\epsilon\min_{\|y\|_{0}\leq\epsilon^{-2}}\|y-x\|_{2}
\]
\end{lem}

Johnson-Lindenstrauss lemma is a famous result about low-distortion
embeddings of points from high-dimensional into low-dimension Eculidean
space. It was named after William B. Johnson and Joram Lindenstrauss
\cite{jl84}. The most classical matrix that gives the property is
random Gaussian matrix, it is known that several other matrices also
suffice to show the guarantees, e.g. FastJL\cite{ac06}, SparseJL\cite{kn14},
Count-Sketch matrix \cite{ccf02,tz12}, subsampled randomized Hadamard/Fourier
transform \cite{ldfu13}.
\begin{lem}[Johnson--Lindenstrauss (JL) \cite{jl84}]
\label{lem:JL_lemma} There exists a function $\JL(\epsilon,n)$
that given $\epsilon>0$ returns a matrix $\mj\in\R^{m\times n}$
with $m=O(\epsilon^{-2}\log n)$ in $O(mn)$ time. For any $v\in\R^{n}$
this matrix $\mj$ satisfies with high probability in $n$ that $\|\mj v\|_{2}\approx_{\epsilon}\|v\|_{2}$.
\end{lem}

It is worth to mention that JL lemma is tight up to some constant
factor, i.e., there exists a set of points of size $n$ that needs
dimension $\Omega(\epsilon^{-2}\log n)$ in order to preserve the
distances between all pairs of points \cite{ln17}.

Our proof of Lemma~\ref{lem:large_entry_datastructure} is through
the analysis of Algorithm \ref{alg:large_entry_datastructure} which
works as follows. The data-structure creates $\log_{2}n$ sketching
matrices $\Phi_{1},...,\Phi_{\log_{2}n}$ according to Lemma~\ref{lem:ell_2_heavy_hitter},
where each $\Phi_{j}$ is constructed for $\varepsilon=2^{-j}.$The
data-structure then maintains the matrices $\mm_{j}=\Phi_{j}\mg\ma$,
whenever $\mg$ is changed by calling $\textsc{Scale}$. Whenever
$\textsc{Query}(h,\varepsilon)$ is called, the data-structure estimates
$\|\mg\ma h\|_{2}$ via a Johnson-Lindenstrass matrix, and then pick
the smallest $j$, for which $\mm_{j}h=\Phi_{j}\mg\ma h$ is guaranteed
to find all $i$ with $|(\mg\ma h)_{i}|>\varepsilon$, i.e. where
$\Phi_{j}$ was initialized for $\varepsilon'=2^{-j}\approx\|\mg\ma h\|_{2}$.
To make sure that the algorithm works against an adaptive adversary,
we compute these entries $(\mg\ma h)_{i}$ exactly and output only
those entries whose absolute value is at least $\varepsilon$ (as
we might also detect a few $i$ for which the entry is smaller).

We now give a formal proof, that this data-structure is correct.
\begin{proof}[Proof of Lemma \ref{lem:large_entry_datastructure}]
 \textbf{Failure probability: }Here we prove Algorithm~\ref{alg:large_entry_datastructure}
succeeds with constant probability. Note that Lemma~\ref{lem:large_entry_datastructure}
states that the result holds w.h.p. in $n$, i.e. probability $1-n^{-c}$
for any constant $c>0$. To achieve this better probability bound,
one can run $O(c\log n)=O(\log n)$ copies of this algorithm in parallel.
This also means that for each query we obtain many vectors $v_{1},...,v_{\ell}$
for $\ell=O(\log n)$. We combine them into a single vector by taking
the union, i.e. $v_{i}=(\mg\ma h)_{i}$ if there is some $j$ with
$(v_{j})_{i}\neq0$. This increases the time complexities for preprocessing,
query and scaling increase by an $O(\log n)$ factor. Consequently,
in what follows we show that the runtimes of the algorithm are an
$O(\log n)$ factor smaller than the runtime stated in Lemma\ref{lem:large_entry_datastructure}.

\paragraph{Initialize:}

We use Lemma~\ref{lem:ell_2_heavy_hitter} to create $\mPhi_{j}\in\R^{m_{j}\times n}$
for each $j\in[\log_{2}(n)/2]$ in $\otilde(\sqrt{n})$ time by Lemma~\ref{lem:ell_2_heavy_hitter}
as the smallest epsilon is $2^{-\log_{2}n/2}=\Omega(1/\sqrt{n})$.
These matrices are then multiplied with $\mg\ma$ to obtain $\mPhi_{j}\mg\ma$.
Computing each of these products can be implemented in $O(\nnz(\ma)\cdot\log^{2}n)$
time as each column of $\mPhi_{j}$ has only $O(\log^{2}n)$ non-zero
elements. Since $\nnz(\ma)\ge n$, as otherwise there would be some
empty rows in $\ma$ that could be removed, the initialization takes
$O(\nnz(\ma)\log^{3}n)$ time.

\paragraph{Queries:}

For query vector $h\in\R^{d}$ we want to find $v\in\R^{n}$ with
$v_{i}=(\mg\ma h)_{i}$ for all $i$ where $|(\mg\ma h)_{i}|\ge\varepsilon$
and $v_{i}=0$ otherwise. By Lemma~\ref{lem:JL_lemma}, $r=\|\mj h\|_{2}$
is a 2-approximation of $\|\mg\ma h\|_{2}$.

Now, If $r/\epsilon\geq\sqrt{n}$ then $\|\mg\ma h\|_{2}/\epsilon\geq\sqrt{n}/2$
and $O(nd)$ time is within the $O(\|\mg\ma h\|^{2}\cdot\varepsilon^{-2}\cdot d\log^{3}n)$
time-budget for the query. Hence, the algorithm simply outputs the
vector by direct calculations, which takes $O(nd)$ time.

Otherwise, the algorithm uses Lemma~\ref{lem:ell_2_heavy_hitter}
to find a list of indices $L$ which contains all $i$ such that $|(\mg\ma h)_{i}|\geq2^{-j}\|\mg\ma h\|_{2}$
with probability at least $9/10$. Since $j=1+\left\lceil \log_{2}(r/\epsilon)\right\rceil $
and $r\leq2\|\mg\ma h\|_{2}$, the heavy hitters algorithm outputs
all $i$ such that $|(\mg\ma h)_{i}|\geq\epsilon$. Hence, the vector
$v$ that we output is correct.

To bound the runtime, note that we find the list $L$ by first computing
$y:=\mPhi_{j}\mg\ma h=\mm_{j}h$ in $O(m_{j}d)=O((4^{j}\log^{2}n)d)$
time. Then, we use Lemma \ref{lem:ell_2_heavy_hitter} to decode the
list in time $O(4^{j}\log^{2}n)$ time. Hence, in total, it takes
$O(\|\mg\ma h\|_{2}^{2}\varepsilon^{-2}\cdot d\log^{2}n)$.

\paragraph{Scaling:}

When we set $g_{i}=u$, i.e by adding $u-g_{i}$ to entry $g_{i}$,
then the product $\mPhi_{j}\mg\ma$ changes by adding the outer product
$((u-g_{i})\cdot\mPhi_{j}\indicVec i)(\indicVec i^{\top}\mg\ma)$.
This outer product can be computed in $O(d\log^{2}n)$ time, because
each column of $\mPhi_{j}$ has only $O(\log^{2}n)$ non-zero entries.
Since there are $O(\log n)$ many $j$, the total time is $O(d\log^{3}n)$.
\end{proof}

\subsection{Maintaining Accumulated Sum\label{sub:accumulated_matrix_product}}

\begin{algorithm2e}[]

\caption{Maintain Accumulated Matrix-vector Product (Lemma \ref{lem:accumulated_matrix_product})}\label{alg:accumulated_matrix_product}

\SetKwProg{Members}{members}{}{}

\SetKwProg{Proc}{procedure}{}{}

\Members{}{

\State  $\ma\in\R^{n\times d}$ \tcp*{Input matrix}

\State  $t\in\N$ \tcp*{Number of updates so far}

\State  $(h\k)_{k\ge1},$ where $h\k\in\R^{d}$ \tcp*{Input vectors
since last query}

\State  $D$ \tcp*{Data-structure of Lemma~\ref{lem:large_entry_datastructure},
Algorithm \ref{alg:large_entry_datastructure}}

\State  $F\subset[n]$ \tcp*{Rows rescaled between the past $t$
queries}

\State  $g,g\old\in\R^{n}$ \tcp*{Current and old row scaling.}

\State  $(\delta\k)_{k\ge1}$where $\delta\k\in\R^{n}$ \tcp*{Relative
row scalings }

\State  $(\bar{h}\k)_{k\ge1},$ where $\bar{h}\k\in\R^{d}$ \tcp*{$\bar{h}{}^{(k)}=\sum_{j=k}^{t}h^{(j)}$
for $h^{(j)}$ and $t$ at last query}

\State  $\bar{g}\old\in\R^{n}$, $(\bar{\delta}\k)_{k\ge1}$ where
$\bar{\delta}\k\in\R^{n}$ \tcp*{Snapshots at last call to $\dsquery$}

}

\vspace{0.1 in}

\Proc{$\dsinit(\ma\in\R^{n\times d},g\in\R^{n})$}{

\State $\ma\leftarrow\ma,g\leftarrow g,g\old\leftarrow g,\bar{g}\old\leftarrow g$

\State $\delta^{(1)}\leftarrow0_{n}$, $\bar{\delta}^{(1)}\leftarrow0_{n}$,
$F\leftarrow\emptyset$, $t\leftarrow0$, $\bar{h}^{(1)}\leftarrow0_{n}$

\State $D.\dsinit(\ma,g)$\tcp*{Algorithm \ref{alg:large_entry_datastructure}}

}

\vspace{0.1 in}

\Proc{\textsc{Update}$(h\in\R^{n})$}{

\State $t\leftarrow t+1$, $h^{(t)}\leftarrow h$, $\delta^{(t+1)}\leftarrow\vzero_{n}$,

}

\vspace{0.1 in}

\Proc{\textsc{Scale}$(i\in\N,u\in\R$)}{

\State $\delta_{i}^{(t+1)}\leftarrow u-g_{i}$, $g_{i}\leftarrow u$,
$F\leftarrow F\cup\{i\}$ \label{line:alg1:update_delta}\tcp*{Save
row rescaling for next \textsc{Query}}\label{line:alg1:add_to_S}

}

\vspace{0.1 in}

\Proc{\textsc{MarkExact}$(i\in\N)$}{

\State $F\leftarrow F\cup\{i\}$, \tcp*{Mark row $i$ to be computed
exactly in the next $\textsc{Query}$}

}

\vspace{0.1 in}

\Proc{$\textsc{ComputeExact}$$(i\in[n])$}{

\State $x\leftarrow\bar{g}_{i}\old e_{i}^{\top}\ma\bar{h}^{(1)}$\tcp*{$g_{i}\old e_{i}^{\top}\ma\sum_{j=1}^{t}h^{(j)}$}
\label{line:alg3:fixed_sum}

\For{$k$ with $\bar{\delta}\k_{i}\neq0$}{ \label{line:alg3:dynamic_sum}

\State $x\leftarrow x+\bar{\delta}\k_{i}e_{i}^{\top}\ma\bar{h}^{(k)}$\tcp*{Adds
$\sum_{k\in[t]}(\delta\k)_{i}~\indicVec i^{\top}\ma\sum_{j=k}^{t}h^{(j)}$
to $x$}

}

\State \Return $x$ \tcp*{This is cached to improve efficiency in
proof of Lemma~\ref{lem:accumulated_matrix_product}}

}

\vspace{0.1 in}

\Proc{$\dsquery(\epsilon\in(0,1))$}{

\State $\bar{g}\old\leftarrow g\old$,$\bar{\delta}\k\leftarrow\delta\k$
for all $k\in[t]$ \label{line:alg1:copy_delta} \tcp*{Save for later
$\textsc{ComputeExact}$ calls}

\State Compute $\bar{h}{}^{(k)}\leftarrow\sum_{j=k}^{t}h^{(j)}$
for all $k\in[t]$. \label{line:alg1:way1start} \label{line:alg1:precompute}

\State Set $v\in\R^{n}$ with $v_{i}\leftarrow\textsc{ComputeExact}(i)$\label{line:alg1:way1end}
for all $i\in F$ and $v_{i}\leftarrow0$ for all $i\notin F$

\State $D.\textsc{Scale}(i,0)$ for all $i\in F$. \label{line:alg1:scale_to_0}
\label{line:alg1:way2start}

\State $v\leftarrow v+D.\textsc{Query}(\sum_{k=1}^{t}h\k,\epsilon)$
\label{line:alg1:query} \tcp*{Algorithm \ref{alg:large_entry_datastructure}}

\State $D.\textsc{Scale}(i,g_{i})$ for all $i\in F$.\label{line:alg1:way2end}
\tcp*{Algorithm \ref{alg:large_entry_datastructure}}

\State $g\old\leftarrow g$, $t\leftarrow0$, and $\delta\k\leftarrow0_{n}$
for all $k$

\State $F\leftarrow\emptyset$ \label{line:alg1:reset_S}

\State \Return $v$

}

\end{algorithm2e}

Here we extend Lemma~\ref{lem:large_entry_datastructure} to approximately
maintain the sum of products $\sum_{k\in[t]}\mg\k\ma h\k$. Note that
when the $g^{(k)}$ were all the same, i.e equal to some vector $g$,
then it suffices to maintain $\mg\ma\sum_{k\in[t]}h\k$, i.e. a single
matrix-vector product, and this can be approximated in $\ell_{\infty}$
norm using Lemma~\ref{lem:large_entry_datastructure}. With this
in mind, we split the task of approximating $\sum_{k\in[t]}\mg\k\ma h\k$
into two parts, corresponding to the rows of $G$ which have changed
and those which have not. Our algorithm for for this problem, Algorithm~\ref{alg:accumulated_matrix_product},
stores $\mf$, a diagonal 0-1 matrix where $\mf_{ii}=1$ indicates
that $g\k_{i}\neq g_{i}^{(k')}$ for some $k,k'$. The algorithm then
then approximates $(\mi-\mf)\sum_{k=1}^{t}\mg\k\ma h\k=(\mi-\mf)\mg\t\ma\sum_{k=1}^{t}h\k$
and $\mf\sum_{k=1}^{t}\mg\k\ma h\k$ in two different ways. From Line
\ref{line:alg1:way1start} to \ref{line:alg1:way1end}, the second
term is maintained directly. This is efficient because there are not
too many rows that are rescaled. From Line \ref{line:alg1:way2start}
to \ref{line:alg1:way2end}, the algorithm then approximates the first
term using Lemma \ref{lem:large_entry_datastructure}.
\begin{lem}[Maintain accumulated matrix-vector product]
\label{lem:accumulated_matrix_product}There exists a Monte-Carlo
data-structure (Algorithm \ref{alg:accumulated_matrix_product}),
that works against an adaptive adversary, with the following procedures:
\begin{itemize}
\item \textsc{Initialize($\ma,g$):} Given a matrix $\ma\in\R^{n\times d}$
and a scaling $g\in\R^{n}$, the data-structure preprocesses in $O(\nnz(A)\log^{4}n)$
amortized time.
\item \textsc{Update($h$):} The data-structure processes the vector $h\in\R^{d}$
in $O(d)$ amortized time.
\item \textsc{Scale($i,u$):} Sets $g_{i}=u$ in $O(d\log^{4}n)$ amortized
time.
\item \textsc{MarkExact($i$):} Marks the $i$-th row to be calculated exactly
in the next call to\textup{ $\textsc{Query}$ in }$O(d\log^{4}n)$
amortized time.
\item \textsc{Query($\epsilon$):} Let $t$ be the number of updates before
the last query, $h\k\in\R^{d}$ and $g\k\in\R^{n}$ are the corresponding
vectors and scalings during update number $k$. (That means $h^{(1)},...,h^{(t)}$,
$g^{(1)},...,g^{(t)}$ are the $t$ vectors and scalings since the
last call to $\textsc{Query}$.) W.h.p. in $n$ the data-structure
outputs the vector $v$ where $v_{i}=(\sum_{k\in[t]}\mg\k\ma h\k)_{i}$
if $|\sum_{k\in[t]}[\mg\k\ma h\k]_{i}|\geq\epsilon$, $i$ is set
to be exact or $g_{i}$ was rescaled between this and the previous
call to\textup{ $\textsc{Query}$; }otherwise $v_{i}=0$. Let $\mf$
be the 0-1-matrix with $\mf_{i,i}=1$ when $i$ is marked to be exact
or $g_{i}$ was rescaled between this and the previous call to\textup{
$\textsc{Query}$}, then the call to \textup{$\textsc{Query}$} can
be implemented in amortized time 
\[
O\left(\left\Vert (\mi-\mf)\sum_{k\in[t]}\mg\k\ma h\k\right\Vert _{2}^{2}\cdot\varepsilon^{-2}\cdot d\log^{3}n\right).
\]
\item $\textsc{ComputeExact}(i)$: Let $w=\sum_{k\in[t]}\mg\k\ma h\k$ be
the vector approximated by the last call to $\textsc{Query}$.Then
$\textsc{ComputeExact}(i)$ returns $w_{i}$ in $O(d)$ amortized
time.
\end{itemize}
\end{lem}

\begin{proof}
\textbf{Complexity}: The time complexity for the \textsc{Initialize}
procedure follows from Lemma~\ref{lem:large_entry_datastructure}.
The time complexity for \textsc{Update} is clearly $O(d)$. We note
that the amortized time for \textsc{Scale} is $\otilde(d)$ instead
of $O(1)$ because every call of \textsc{Scale} will introduce two
calls to $D.\textsc{Scale}$ in the calls to \textsc{Query}. The same
holds for \textsc{MarkExact}.

The cost of \textsc{ComputeExact$(i)$ }is $O(d+Sd)$, where $S$
is the number of times $\textsc{Scale}(i,\cdot)$ was called for that
specific $i$, between the past two calls to $\textsc{Query}$. This
is because $\bar{\delta}^{(k)}$ is set to $\delta^{(k)}$ in the
last call to $\textsc{Query}$ (Line \ref{line:alg1:copy_delta}),
and the number of $k$ for which $\delta_{i}^{(k)}\neq0$ is the number
of times $\textsc{Scale}$$(i,\cdot)$ was called, because of Line~\ref{line:alg1:update_delta}.
Hence the number of iterations of the loop in Line~\ref{line:alg3:dynamic_sum}
is bounded by $S$. Without loss of generality, we can assume that
\textsc{ComputeExact$(i)$} is called only once between two calls
to $\textsc{Query}$ as we could modify the data-structure to save
the result and return this saved result during the second call to
\textsc{ComputeExact$(i)$}. Thus we can charge the $O(Sd)$ cost
of $\textsc{ComputeExact}$ to $\textsc{Scale}$ instead. In summary
$\textsc{Scale}$ has amortized cost $\tilde{O}(d)$ and $\textsc{ComputeExact}$
has amortized cost $O(d)$.

Finally, the cost of \textsc{Query} mainly consists of the following:
(1) computation of sums in Line \ref{line:alg1:precompute}, (2) the
calculation on indices in $F$ (i.e. calling $\textsc{ComputeExact}$),
(3) the call of $D.\textsc{Query},$ and finally (4) the calls to
$D.\textsc{Scale}$.

The cost of (1) is bounded by $O(td)$, which we charge as amortized
cost to \textsc{Update}. The cost of (2) is bounded by $O(|F|d)$,
which we charge to the amortized cost of \textsc{Scale} and \textsc{MarkExact}.
The cost of (3) is exactly the runtime we claimed due to Lemma \ref{lem:large_entry_datastructure}
since we zero out some of the rows using scale (Line~\ref{line:alg1:scale_to_0})
which causes the $(\mi-\mf)$ term in the complexity.. The cost of
(4) is bounded by $O(|F|)$ times the cost of $D.\textsc{Scale}$.
We charge this cost to the \textsc{Scale} and $\textsc{MarkExact}$
procedure of the current data structure.

\subparagraph{Correctness:}

We now prove that the output of \textsc{Query} is correct. Let $v$
be the output of the algorithm. The algorithm \textsc{Query} consists
of two parts. From Line~\ref{line:alg1:way1start} to \ref{line:alg1:way1end},
we calculate $v$ on indices in $F$. From Line~\ref{line:alg1:way2start}
to \ref{line:alg1:way2end}, we calculate $v$ on indices in $F^{c}$.

For any $i\notin F$, we note that $g_{i}$ is never changed. Hence,
we have 
\[
\sum_{k\in[t]}(\mg\k\ma h\k)_{i}=\left(\mg\k\ma\sum_{k\in[t]}h\k\right)_{i}
\]
and Lemma \ref{lem:large_entry_datastructure} yields that we correctly
calculate $v_{i}$ for $i\notin F$. Now, we prove the correctness
of $v_{i}$ for $i\in F$ (i.e. the correctness of $\textsc{ComputeExact}$).
For notational convenience, let $\mg^{(0)}\defeq\mg\old$, and note
that 
\[
\left(\sum_{k\in[t]}\mg\k\ma h\k\right)_{i}=\underbrace{\left(\sum_{k\in[t]}(\mg\k-\mg^{(0)})\ma h\k\right)_{i}}_{\text{Computed by Line \ref{line:alg3:dynamic_sum}}}+\underbrace{\left(\mg^{(0)}\ma\sum_{k\in[t]}h\k\right)_{i}}_{\text{Computed by Line \ref{line:alg3:fixed_sum}}}.
\]
It is easy to see that Line~\ref{line:alg3:fixed_sum} computes the
last term, as $\bar{g}^{\old}=g^{\old}=g^{(1)}$. The more difficult
part is proving that the first term is computed by Line~\ref{line:alg3:dynamic_sum}.
However, note that by design of $\bar{\delta}_{k},\delta_{k}$we have
\[
[g^{(k)}]_{i}-[g^{(0)}]_{i}=\sum_{j\in[k]}[g^{(j)}]_{i}-[g^{(j-1)}]_{i}=\sum_{j\in[k]}\bar{\delta}_{j}
\]
and therefore we can re-write the $i$-th entry of $\sum_{k\in[t]}\left(\mg\k-\mg^{(0)}\right)\ma h\k$
as 
\begin{align*}
\indicVec i^{\top}\sum_{k\in[t]}\left(\mg\k-\mg^{(0)}\right)\ma h\k & =\sum_{k\in[t]}\sum_{j\in[k]}\bar{\delta}_{j}\indicVec i^{\top}\ma h^{(k)}=\sum_{j\in[t]}\bar{\delta}_{j}\indicVec i^{\top}\ma\left[\sum_{k=j}^{t}h^{(k)}\right]
\end{align*}
By the definition of $\bar{h}^{(k)}$ this is precisely what is computed
in Line~\ref{line:alg3:dynamic_sum}.

In summary, $\textsc{ComputeExact}(i)$ computes $(\sum_{k=1}^{t}\mg\k\ma h\k)_{i}$
and the vector $v$ computed by \textsc{Query} is indeed $\sum_{k=1}^{t}\mg\k\ma h\k$.
As the variables $\bar{g}\old,\bar{\delta}\k,\bar{h}\k$ are only
changed during $\textsc{Query}$, the function $\textsc{ComputeExact}(i)$
computes the $i$-th entry of the vector approximated by the last
call to $\textsc{Query}$.
\end{proof}

\subsection{Vector Maintenance Algorithm and Analysis\label{sub:product_sum_datastructure}}

Using the tools from the past two sections we are now ready to provide
our vector maintenance data structure, Algorithm~\ref{alg:product_sum_datastructure},
and prove Theorem~\ref{thm:product_sum_datastructure}. For simplicity,
we assume that the sequence $x\t\defeq x_{\textsc{init}}+\sum_{k\in[t]}(\mg\k\ma h\k+\delta^{(k)})$
satisfies 
\begin{equation}
\frac{8}{9}x_{i}^{(t-1)}\leq x\t_{i}\leq\frac{9}{8}x_{i}^{(t-1)}\text{ for all }i\label{eq:x_assumption}
\end{equation}
We note that this assumption is satisfied for our interior point method.
If this assumption is violated, we can detect it easily by Lemma~\ref{lem:large_entry_datastructure}
and propagate this change through the data-structure with small amortized
cost. However, this would make the data structure have more special
cases and make the code more difficult to read.

The idea of Algorithm~\ref{alg:product_sum_datastructure} is as
follows. We split the sum $\sum_{k\in[t]}\mg^{(k)}\ma h^{(k)}$ into
$\log t$ partial sums of increasing length. Each partial sum is maintained
via the algorithms of Lemma \ref{lem:accumulated_matrix_product},
where copy number $\ell$ is queried after every $2^{\ell}$ iterations.
The data-structure also stores vectors $v^{(1)},...,v^{(\lceil\log n\rceil)}$,
where $v^{(\ell)}$ is the result of these data-structures from the
last time they were queried (see Line \ref{line:assign_vi}), i.e.
each $v^{(\ell)}$ is an approximation of some partial sum of $\sum_{k\in[t]}\mg^{(k)}\ma h^{(k)}$.
Note that for Theorem~\ref{thm:product_sum_datastructure} we want
to have a multiplicative approximation for the entries, whereas Lemma
\ref{lem:accumulated_matrix_product} only yields an additive error.
So the partial sums that are maintained in Algorithm~\ref{alg:product_sum_datastructure}
by using the data-structure of Lemma \ref{lem:accumulated_matrix_product}
are actually of the form $\sum_{k=t_{0}}^{t_{1}}\mz\mg^{(k)}\ma h^{(k)}$
for some diagonal matrix $\mz$. By scaling the entries of $\mz$
turn the absolute error into a relative one. We prove that Algorithm~\ref{alg:product_sum_datastructure}
can be applied carefully to prove Theorem~\ref{thm:product_sum_datastructure}.

\begin{algorithm2e}[]

\caption{Maintain Multiplicative Approximation (Theorem~\ref{thm:product_sum_datastructure})}

\label{alg:product_sum_datastructure}

\SetKwProg{Members}{members}{}{}

\SetKwProg{Proc}{procedure}{}{}

\Members{}{

\State  $\ma\in\R^{n\times d}$, $x_{\textsc{init}}\in\Rn$, $\epsilon>0$
\tcp*{Initial input matrix, vector, and accuracy}

\State  $t\in\N,g\in\Rn$ \tcp*{Number of queries so far and current
scaling}

\State $(D^{(\ell)})_{\ell\in\{0,...,\lceil\log n\rceil\}}$ \tcp*{Data-structures
of Lemma \ref{lem:accumulated_matrix_product} (Algorithm \ref{alg:accumulated_matrix_product})}

\State $(z^{(\ell)})_{\ell\geq0\}}$ where $z^{(\ell)}\in\R^{n}$\tcp*{Scalings
used in $D^{(\ell)}$}

\State $(F^{(\ell)})_{\{0,...,\lceil\log n\rceil\}}$ where $F_{\ell}\subset[n]$
\tcp*{Changed coordinates for $D^{(\ell)}$}

\State $(v^{(\ell)})_{\{0,...,\lceil\log n\rceil\}}$ where $v_{\ell}\in\R^{n}$
\tcp*{Past results of $D^{(\ell)}$ queries}

\State $(\delta^{(t)})_{t\ge1}$ where $\delta^{(t)}\in\R^{n}$\tcp*{Additive
terms in queries}

\State $y\in\R^{n}$ \tcp*{Past result of this data-structure}

}

\vspace{0.1 in}

\Proc{\textsc{Initialize}$(\ma\in\R^{n\times d},g\in\R^{n},x_{\textsc{init}}\in\R^{n},\epsilon>0)$}{

\State  $\ma\leftarrow\ma$, $x_{\textsc{init}}\leftarrow x_{\textsc{init}}$,
$\epsilon\leftarrow\epsilon$,, $t\leftarrow1$, $g\leftarrow g$,
$y\leftarrow x_{\textsc{init}}$

\State $z^{(\ell)}\leftarrow x_{\textsc{init}}$, $v^{(\ell)}\leftarrow0_{n}$,
$F^{(\ell)}\leftarrow\emptyset$ for all $\ell\in\{0,...,\lceil\log n\rceil\}$

\State $D^{(\ell)}.\textsc{Initialize}(\ma,(\mz^{(\ell)})^{-1}g)$
for all $\ell\in\{0,...,\lceil\log n\rceil\}$ \tcp*{Algorithm \ref{alg:accumulated_matrix_product}}

}

\vspace{0.1 in}

\Proc{\textsc{Scale}$(i\in\N,u\in\R)$}{

\State $g_{i}\leftarrow u$

\State $D^{(\ell)}.\textsc{Scale}(i,u/z_{i}^{(\ell)})$ for all $\ell=0,...,\lceil\log n\rceil$
\tcp*{Algorithm \ref{alg:accumulated_matrix_product}}

}

\vspace{0.1 in}

\Proc{\textsc{Query}$(h^{(t)}\in\R^{d},\delta^{(t)}\in\R^{n})$}{

\For{$\ell=0,...,\lceil\log n\rceil$}{ \label{line:alg4:query_loop}

\State $D^{(\ell)}.\textsc{Update}(h^{(t)})$ \tcp*{Algorithm~\ref{alg:accumulated_matrix_product}}

\Comment{Process coordinates with a new multiplicative change}

\For{$i\in[n]$ with $z_{i}^{(\ell)}\notin[\frac{8}{9}y_{i},\frac{9}{8}y_{i}]$
and $i\notin F^{(\ell)}$\label{line:F_checkloop}}{

\State $v_{i}^{(\ell)}\leftarrow D^{(\ell)}.\textsc{ComputeExact}(i)$
\label{line:set_exact_retroactive} \tcp*{Set $v_{i}^{(\ell)}$ exactly
(Algorithm~\ref{alg:accumulated_matrix_product})}

\State $D^{(\ell)}.\textsc{MarkExact}(i)$ \tcp*{Compute $v_{i}^{(\ell)}$
exactly at next query.}\label{line:set_exact_1}

\State$F^{(\ell)}\leftarrow F^{(\ell)}\cup\{i\}$ \tcp*{Record that
coordinate $i$ changed.}

}

\Comment{Update sum stored in $D_{\ell}$ every $2^{\ell}$ iterations}

\If{$t=0\mod2^{\ell}$}{

\State $v^{(\ell)}\leftarrow\mz^{(\ell)}D^{(\ell)}.\textsc{Query}(\epsilon/4\lceil\log n\rceil)+\sum_{k=t-2^{\ell}+1}^{t}\delta^{(k)}$
\label{line:assign_vi} \tcp*{Algorithm~\ref{alg:accumulated_matrix_product}}

\State $z_{i}^{(\ell)}\leftarrow y_{i}$ for all $i\in F^{(\ell)}$
\label{line:set_exact_2}

\State $D^{(\ell)}.\textsc{Scale}(i,g_{i}/z_{i}^{(\ell)})$ for all
$i\in F^{(\ell)}$ \label{line:set_scale_exact} \tcp*{Algorithm~\ref{alg:accumulated_matrix_product}}

\State $F^{(\ell)}\leftarrow\emptyset$\tcp*{Reset marked coordinates}

}

}

\State $y\leftarrow x_{\textsc{init}}+\sum_{\ell=0}^{\lceil\log n\rceil}a_{\ell}v^{(\ell)}$
where $a_{\ell}$ is the $\ell$-th bit of $t$.\tcp*{$t=\sum_{\ell=0}^{\lceil\log n\rceil}a_{\ell}2^{\ell}$
}\label{line:compute_x}

\State $t\leftarrow t+1$

\State \Return $y$

}

\vspace{0.1 in}

\Proc{\tcp*[f]{Return exact $i$th entry.}}{$\textsc{ComputeExact}$$(i)$}{

\State $y\leftarrow(x_{\textsc{init}})_{i}$

\For{$\ell=0,...,\lceil\log n\rceil$}{\lIf(\tcp*[f]{Algorithm~\ref{alg:accumulated_matrix_product}}){$\ell$-th
bit of $t$ is $1$}{$y\leftarrow y+D^{(\ell)}.\textsc{ComputeExact}(i)$}}

\State \Return $y$

}

\end{algorithm2e}
\begin{proof}[Proof of Theorem~\ref{thm:product_sum_datastructure}]
 We show that Algorithm~\ref{alg:product_sum_datastructure} can
be used to obtain a data-structure with the desired running time.
Formally, we run Algorithm~\ref{alg:product_sum_datastructure} and
after every $n$ calls to query compute $y\in\R^{n}$ with $y_{i}=\textsc{ComputeExact}(i)$
for all $i\in[n]$, i.e. the exact state of the vector and then re-initialize
the datastructure with \textsc{Initialize}$(\ma,g,y,\epsilon)$ for
the current value of $g$. Note that, by Lemma~\ref{lem:accumulated_matrix_product},
the time to initialize the data structure is $O(\nnz(\ma)\log^{5}n)$
and the additional cost this contributes to query is $O(\lceil T/n\rceil\cdot\nnz(\ma)\log^{5}n)=O(Td\log^{5}n)$.
Consequently, it suffices to prove the theorem when at most $n$ queries
occur.

To analyze this data structure, recall that 
\[
x\t\defeq x_{\textsc{init}}+\sum_{k\in[t]}(\mg\k\ma h\k+\delta^{(k)})
\]
and let $y\t$ be the output of the algorithm at iteration $t$. Our
goal is to prove that the $x\t\approx_{\epsilon}y\t$. We prove this
by induction.We assume $x^{(s)}\approx_{\epsilon}y^{(s)}$ for all
$s<t$ and proceed to prove this is also true for $s=t$.

\paragraph{The vectors $v^{(0)},...,v^{(\lceil\log n\rceil)}$:}

Assume we currently perform query number $t$ and the data-structure
is currently performing the loop of Line \ref{line:alg4:query_loop}.
Further assume that $\ell$ is such that $t\mod2^{\ell}=0$, so the
data-structure will assign a new vector to $v^{(\ell)}$. Let $g\k$
be the current state of vector $g$ during query number $k\le t$.
By Lemma~\ref{lem:accumulated_matrix_product}, the vector 
\[
v^{(\ell)}\defeq\mz^{(\ell)}D^{(\ell)}.\textsc{Query}(\epsilon/8\lceil\log n\rceil)+\sum_{k=t-2^{\ell}+1}^{t}\delta^{(k)}
\]
satisfies
\[
\left|v_{i}^{(\ell)}-\sum_{k=t-2^{\ell}+1}^{t}[\mg\k\ma h\k+\delta^{(k)}]_{i}\right|\leq\frac{\epsilon}{2\lceil\log n\rceil}\cdot z_{i}^{(\ell)}
\]
where $z^{(\ell)}$ is the extra scaling data structure $D^{(\ell)}$
uses.

By Line~\ref{line:set_exact_1} we see that if $z_{i}^{\ell}\notin[\frac{8}{9}y_{i},\frac{9}{8}y_{i}]$
then $D^{(\ell)}.\textsc{MarkExact}(i)$ had been called after the
last query to $D^{(\ell)}$. Consequently, if $v_{i}^{(\ell)}\neq\sum_{k=t-2^{\ell}+1}^{t}[\mg\k\ma h\k+\delta^{(k)}]_{i}$
(i.e. the error is non-zero) then $\frac{8}{9}y_{i}^{(t-1)}\leq z^{(\ell)}\leq\frac{9}{8}y_{i}^{(t-1)}.$
By \eqref{eq:x_assumption} and the induction hypothesis $x^{(t-1)}\approx_{\epsilon}y^{(t-1)}$,
we have that $z^{(\ell)}\leq e^{\epsilon}(\frac{9}{8})^{2}x^{(t)}\leq\frac{7}{5}x^{(t)}$.
Hence, for all $i\in[n]$ we have 
\[
\left|v^{(\ell)}-\sum_{k=t-2^{\ell}+1}^{t}(\mg\k\ma h\k+\delta^{(k)})\right|_{i}\leq\frac{7}{10}\frac{\epsilon}{\lceil\log n\rceil}\cdot x_{i}^{(t)}\,.
\]

Now consider the case that $t$ mod $2^{\ell}\neq0$ and define $t'=2^{\ell}\cdot\lfloor t/2^{\ell}\rfloor$.
If the error of $v_{i}^{(\ell)}$ is not zero, then it is still bounded
by $\frac{\epsilon}{2\lceil\log n\rceil}\cdot z_{i}^{(\ell)}$, because
$z_{i}^{(\ell)}$ was not changed in Line \ref{line:set_exact_2}.
We can now repeat the previous argument: by Line \ref{line:set_exact_retroactive},
we see that $D^{(\ell)}.\textsc{ComputeExact}(i)$ is called whenever
$z_{i}^{\ell}\notin[\frac{8}{9}y_{i},\frac{9}{8}y_{i}]$, which sets
the error of $v_{i}^{(\ell)}$ to zero. Hence, if the error is non-zero,
we know that $\frac{8}{9}y_{i}^{(t-1)}\leq z^{(\ell)}\leq\frac{9}{8}y_{i}^{(t-1)}.$
With the same argument as before we thus have
\[
\left|v^{(\ell)}-\sum_{k=t'-2^{\ell}+1}^{t'}(\mg\k\ma h\k+\delta^{(k)})\right|_{i}\leq\frac{7}{10}\frac{\epsilon}{\lceil\log n\rceil}\cdot x_{i}^{(t)}
\]
for all $i$. (Note that the sum on the left is using $t'$ for its
indices while the error bound on the right uses $x^{(t)}$.)

\paragraph{Correctness of $x$:}

We compute the vector exactly every $n$ iterations and reset $t$
to 0. Hence, $t\le n$ and hence $t=\sum_{i=0}^{\lceil\log n\rceil}a_{\ell}2^{\ell}$
with each $a_{\ell}\in\{0,1\}$, i.e. the binary representation of
$t$. Then 
\[
\sum_{k=1}^{t}(\mg\k\ma h\k+\delta^{(k)})=\sum_{i=0}^{\lceil\log n\rceil}a_{\ell}\sum_{k=2^{\ell}\lfloor t/2^{\ell}\rfloor-2^{\ell}+1}^{2^{\ell}\lfloor t/2^{\ell}\rfloor}(\mg\k\ma h\k+\delta^{(k)}).
\]
Note that $\sum_{k=2^{\ell}\lfloor t/2^{\ell}\rfloor-2^{\ell}+1}^{2^{\ell}\lfloor t/2^{\ell}\rfloor}(\mg\k\ma h\k+\delta^{(k)})$
is exactly the vector approximated by $v^{(\ell)}$, because this
sum only changes whenever $t$ is a multiple of $2^{\ell}$, in which
case 
\[
\sum_{k=2^{\ell}\lfloor t/2^{\ell}\rfloor-2^{\ell}+1}^{2^{\ell}\lfloor t/2^{\ell}\rfloor}(\mg\k\ma h\k+\delta^{(k)})=\sum_{k=t-2^{i}+1}^{t}(\mg\k\ma h\k+\delta^{(k)})\approx v_{i}.
\]
Thus for $y^{(t)}:=x_{\textsc{init}}+\sum_{\ell=0}^{\lceil\log n\rceil}a_{\ell}v^{(\ell)}$,
as computed in Line \ref{line:compute_x}, $x^{(t)}\approx_{\epsilon}y^{(t)}$
due to the following 
\begin{align*}
\left|y^{(t)}-x^{(t)}\right|_{i} & =\left|\sum_{\ell=0}^{\lceil\log n\rceil}a_{\ell}v^{(\ell)}-\sum_{k=1}^{t}(\mg\k\ma h\k+\delta^{(k)})\right|_{i}\\
 & \leq\left|\sum_{\ell=0}^{\lceil\log n\rceil}a_{\ell}\left(v^{(\ell)}-\sum_{k=2^{\ell}\lfloor t/2^{\ell}\rfloor-2^{\ell}+1}^{2^{\ell}\lfloor t/2^{\ell}\rfloor}(\mg\k\ma h\k+\delta^{(k)})\right)\right|_{i}\\
 & \leq\sum_{\ell=0}^{\lceil\log n\rceil}\left|v^{(\ell)}-\sum_{k=2^{\ell}\lfloor t/2^{\ell}\rfloor-2^{\ell}+1}^{2^{\ell}\lfloor t/2^{\ell}\rfloor}(\mg\k\ma h\k+\delta^{(k)})\right|_{i}\\
 & \leq\sum_{\ell=0}^{\lceil\log n\rceil}\frac{7\epsilon}{10\lceil\log n\rceil}\cdot x_{i}^{(t)}\leq\frac{7}{10}\epsilon\cdot x_{i}^{(t)}.
\end{align*}

\paragraph{ComputeExact:}

Note that $D^{(\ell)}.\textsc{ComputeExact}(i)=\sum_{k=2^{\ell}\lfloor t/2^{\ell}\rfloor-2^{\ell}+1}^{2^{\ell}\lfloor t/2^{\ell}\rfloor}[\mg\k\ma h\k+\delta^{(k)}]_{i}$
and therefore, $\textsc{ComputeExact}$ computes the $i$th entry
of $\sum_{k\in[t]}(\mg\k\ma h\k+\delta^{(k)})$ exactly.

\paragraph{Complexity:}

The runtime of initialization and $\textsc{ComputeExact}$ are that
of Lemma~\ref{lem:accumulated_matrix_product}, but increased by
an $O(\log n)$ factor, as we run that many instances of these data-structures.

To bound the runtime of $\dsquery$, fix some $\ell\geq0$. We first
analyze the complexity of the loop in Line~\ref{line:alg4:query_loop}
for that specific $\ell$. In every iteration we spend $O(n)$ time
to determine which coordinates need to be added to $F^{(\ell)}$ (Line~\ref{line:F_checkloop}).
A coordinate $i\in[n]$ is added to $F^{(\ell)}$ at most once every
$2^{\ell}$ calls to $\dsquery$ and this happens only when $y_{i}$
has changed by some $9/8$ multiplicatively from $z_{i}^{(\ell)}$.
Further, if we let $k\geq T-2^{\ell}$ denote the last iteration on
which $k=0\mod2^{\ell}$ then as $\frac{5}{7}x^{(t-2^{\ell})}\le z_{i}^{(\ell)}\le\frac{7}{5}x^{(t-2^{\ell})}$
and $y_{i}\approx x_{i}^{(t)}$, we can bound the number of indices
that were in $F^{(\ell)}$ during the past $t-k$ iterations by the
number of indices $i$ for which $x_{i}$ changed by at least constant
factor during the past $t-k$ iterations. The number of such indices
is $O((t-k)\sum_{k=t-2^{\ell}+1}^{t}\left\Vert (\mx^{(k)})^{-1}\mg\k\ma h\k\right\Vert ^{2}+\left\Vert (\mx^{(k)})^{-1}\delta\k\right\Vert ^{2})$.
For every index that was in $F^{(\ell)}$ over the past $t-k$ iterations,
the algorithms calls $D_{\ell}.\textsc{Scale}$, $D_{\ell}\textsc{.MarkExact}$
and $D_{\ell}.\textsc{ComputeExact}$ at most once. Consequently,
applying Lemma~\ref{lem:accumulated_matrix_product} and summing
over every time $t=0\mod2^{\ell}$ held, we have that ignoring the
cost of $D^{(\ell)}.\dsquery()$ the runtime of the the loop for a
single $\ell$ is

\begin{equation}
O\left(Tn+Td\log^{4}n\cdot\sum_{k\in[t]}\left(\left\Vert (\mx^{(k)})^{-1}\mg\k\ma h\k\right\Vert ^{2}+\left\Vert (\mx^{(k)})^{-1}\delta\k\right\Vert ^{2}\right)\right)\,.\label{eq:query_free_while_loop_bound}
\end{equation}

Next, we bound the runtime of the calls to $D^{(\ell)}.\dsquery()$.
Consider such a call when $t=0\mod2^{\ell}$. By Lemma~\ref{lem:accumulated_matrix_product}
and the fact that $D_{\ell}.\dsmark(i)$ is called for all $i\in F^{(\ell)}$
the runtime of this call is bounded by

\[
O(1)\cdot\left\Vert (\mi-\mf^{(\ell)})\sum_{k=t-2^{\ell}+1}^{t}(\mz^{(\ell)})^{-1}\mg\k\ma h\k\right\Vert _{2}^{2}\cdot\varepsilon^{-2}\cdot d\log^{5}n\,,
\]
where the extra $\log^{2}n$ term (compared to Lemma \ref{lem:accumulated_matrix_product})
is due to the $\log n$ factor in $\epsilon/\log n$. Now for $i\notin F^{(\ell)}$
and $k\in[t-2^{l}+1,t]$ we have $z^{(\ell)}\approx_{O(1)}y^{(k-1)}$
by construction and $y^{(k-1)}\approx{}_{O(1)}x^{(k)}$ by the preceding
correctness proof. Consequently, 
\begin{align*}
\left\Vert (\mi-\mf)\sum_{k=t-2^{\ell}+1}^{t}(\mz^{(\ell)})^{-1}\mg\k\ma h\k\right\Vert _{2}^{2} & \leq2^{\ell}\sum_{k=t-2^{\ell}+1}^{t}\left\Vert (\mi-\mf)(\mz^{(\ell)})^{-1}\mg\k\ma h\k\right\Vert _{2}^{2}\\
 & \leq2^{\ell}\sum_{k=t-2^{\ell}+1}^{t}\left\Vert (\mi-\mf)(\mx^{(k)})^{-1}\mg\k\ma h\k\right\Vert _{2}^{2}\\
 & \leq2^{\ell}\sum_{k=t-2^{\ell}+1}^{t}\left\Vert (\mx^{(k)})^{-1}\mg\k\ma h\k\right\Vert _{2}^{2}\,.
\end{align*}
 Since, $D^{(\ell)}.\textsc{Query()}$ was called only $O(T/2^{\ell})$
times this yields that the total runtime contribution of $D^{(\ell)}.\textsc{Query()}$
for a single $\ell$ is

\[
O\left(T\sum_{k\in[t]}\left\Vert (\mx^{(k)})^{-1}\mg\k\ma h\k\right\Vert ^{2}\cdot\varepsilon^{-2}d\log^{5}n\right)\,\cdot
\]
Summing, over all $O(\log n)$ different values of $\ell$ and combining
with \eqref{eq:query_free_while_loop_bound} then yields the desired
running time for query.

\paragraph{Adaptive Adversaries:}

The algorithm works against adaptive adversaries, because the data
structures of Lemma \ref{lem:accumulated_matrix_product} work against
adaptive adversaries.
\end{proof}

\section{Leverage Score Maintenance}

\inputencoding{latin9}\label{sec:leverage_maintenance}

The IPM presented in this work requires approximate leverage scores
of a matrix of the form $\mg\ma$, where $\mg$ is a non-negative
diagonal matrix that changes slowly from one iteration to the next.
Here we provide a data-structure which exploits this stability to
maintain the leverage scores of this matrix more efficiently than
recomputing them from scratch every time $\mg$ changes. Formally,
we provide Algorithm~\ref{alg:leverage_score_datastructure_1} and
show it proves the following theorem on maintaining leverage scores:

\leveragescoredatastructure*

\begin{algorithm2e}[]

\caption{Leverage Score Maintenance (Theorem~\ref{thm:leverage_score_datastructure})}

\label{alg:leverage_score_datastructure_1}

\SetKwProg{Members}{members}{}{}

\SetKwProg{Proc}{procedure}{}{}

\Members{}{

\State $D$ \tcp*{Data-structure of Corollary~\ref{cor:large_entry_datastructure},
Lemma~\ref{lem:large_entry_datastructure}, Algorithm~\ref{alg:large_entry_datastructure}
}

\State $(g\t)_{t\ge0}\in\R^{n}$ \tcp*{Rows scaling of $\ma$ at
time $t$}

\State $(\ttau)_{t\ge0}\in\R^{n}$ \tcp*{Approximate shifted leverage
scores at time $t$}

\State $(\Psi\t)_{t\ge0},(\Psi\t\safe)_{t\ge0}\in\R^{d\times d}$
\tcp*{Approximate inverses at time $t$}

\State $\mr=\JL(1/100,d)^{\top},(\mb\t)_{t\ge0}\in\R^{d\times b}$\tcp*{JL-matrix
$\mr$,$\mb\t=\ma^{\top}\mg\t\mr$}

}

\vspace{0.1 in}

\Proc{\textsc{Initialize}$(\ma\in\R^{n\times d},g\in\R^{n},\varepsilon\in(0,1))$}{

\State $\mr\leftarrow\JL(1/100,d)^{\top}$\tcp*{$\|v^{\top}\mr\|_{2}\in(1\pm\frac{1}{100})\|v\|_{2}$
w.h.p. in $n$}

\State $D.\textsc{Initialize}(\ma,g)$ \tcp*{Algorithm~\ref{alg:large_entry_datastructure}
}

\State $t\leftarrow0$, $\varepsilon\leftarrow\varepsilon$, $g^{(t)}\leftarrow g$,
$\mb\t\leftarrow\ma^{\top}\mg\t\mr$

}

\vspace{0.1 in}

\Proc{\textsc{Scale}$(j\in[n],u\in\R^{n})$}{

\State $\mb\t\leftarrow\mb\t+(\ma^{\top}e_{j})(u-g\t_{j})(e_{j}^{\top}\mr)$
\label{line:ls:update_Phi}\tcp*{Maintain $\mb\t=\ma^{\top}\mg\t\mr$}

\State $D.\textsc{Scale}(j,u)$, $g\t_{j}\leftarrow u$ \label{line:ls:scale_Di}\tcp*{Make
$D$ represent $\mg\t\ma$. Algorithm~\ref{alg:large_entry_datastructure}
}

}

\vspace{0.1 in}

\Proc{\textsc{EstimateScore}$(J\subset[n],t\in\N,\delta\in(0,1))$\tcp*[f]{Internal
procedure}}{

\State $\widetilde{\mr}\leftarrow\JL(\delta/8,d)^{\top}$

\lIf{first step}{$\ \mw\leftarrow\Psi\t\safe\ma^{\top}\mg\t\widetilde{\mr}$}

\Else(\tcp*[f]{$\ma^{\top}\widetilde{\mg}^{2}\ma\approx_{\delta/4}\ma^{\top}(\mg\t)^{2}\ma$
w.h.p. in $n$}){

\State $\ \mw\leftarrow\Psi\t\safe\ma^{\top}\widetilde{\mg}\widetilde{\mr}$
where $\widetilde{g}_{i}=\begin{cases}
\frac{g\t_{i}}{\sqrt{p_{i}}} & \text{ with probability }p_{i}=\Theta(\frac{\log n}{\delta^{2}}\ttau_{i}\t)\\
0 & \text{otherwises}
\end{cases}.$\label{line:ls:sampling}

}

\State \Return $v\in\R^{J}$ with $v_{j}=\|e_{j}^{\top}\mg\t\ma\mw\|_{2}^{2}+\frac{d}{n}$
for $j\in J$.

}

\vspace{0.1 in}

\Proc{\textsc{FindIndices}$()$ \tcp*[f]{Internal Procedure}}{

\State $J\leftarrow\emptyset$\tcp*{list $J$ depends only on $\Psi\safe$}

\For(\tcp*[f]{Compute candidate update indices.}){$i\in\{0,...,\lfloor\log n\rfloor\}$
s.t. $t=0\mod2^{i}$}{

\State $F_{i}\leftarrow\{j:\ g_{j}^{(t)}\neq g_{j}^{(t')}\text{ for some }t-2^{i}\leq t'<t\}$.

\State $\mw\leftarrow D.\textsc{Query}(\Psi\t\mb^{(t)}-\Psi^{(t-2^{i})}\mb^{(t-2^{i})},F_{i},\frac{\epsilon}{48\log n}\sqrt{\frac{d}{nb}})$\label{line:ls:query}
\tcp*{Algorithm~\ref{alg:large_entry_datastructure} }

\State $J_{i}\leftarrow F_{i}\cup\{j:\mw_{j,\ell}\neq0\text{ for some }\ell\}$\label{line:ls:Ji}

\State $v\leftarrow\textsc{EstimateScore}(J_{i},t,\epsilon/(12\log n))$

\State $v\old\leftarrow\textsc{EstimateScore}(J_{i},t-2^{i},\epsilon/(12\log n))$\label{line:ls:est_2_score}

\State $J\leftarrow J\cup\{j\text{ such that }v_{j}\not\approx_{\frac{\epsilon}{3\log n}}v_{j}\old\}$\label{line:ls:est_2_score2}\label{line:ls:Ji2}

}

\State \Return $J$

}

\vspace{0.1 in}

\Proc{\textsc{Query}$(\Psi\in\R^{d\times d},\Psi\safe\in\R^{d\times d})$}{

\State $\Psi\t\leftarrow\Psi$, $\Psi\t\safe\leftarrow\Psi\safe$

\If{$t=0\mod2^{\lfloor\log n\rfloor}$ or first step}{

\State $\ttau^{(0)}\leftarrow\textsc{EstimateScore}([n],t,\epsilon)$,
$t\leftarrow1$

\State \Return$\ttau^{(0)}$

}

\State $J\leftarrow\textsc{FindIndices}()$ and $v\leftarrow\textsc{EstimateScore}(J,t,\epsilon)$

\State  For all $j\in[n]$ set $\ttau\t_{j}\leftarrow v_{j}$ if
$j\in J$ and $\ttau\t_{j}\leftarrow\ttau_{j}^{(t-1)}$ otherwise.

\State $\mb^{(t+1)}\leftarrow\mb\t$, $g^{(t+1)}\leftarrow g\t$,
$\ttau^{(t+1)}\leftarrow\ttau\t$, $t\leftarrow t+1$

\State \Return $\ttau^{(t-1)}$

}

\end{algorithm2e}

The high-level idea for obtaining this data-structure is that when
$\mg$ changes slowly, the leverage scores of $\mg\ma$ change slowly
as well. Thus not too many leverage scores must be recomputed per
iteration, when we are interested in some $(1+\varepsilon)$-approximation,
as the previously computed values are still a valid approximation.
The task of maintaining a valid approximation can thus be split into
two parts: (i) for some small set $J$, quickly compute the $j$-th
leverage score for $j\in J$, and (ii) detect which leverage scores
should be recomputed, i.e. which leverage scores have changed enough
so their previously computed value is no longer a valid approximation.

This section is structured according to the two tasks, (i) and (ii).
In Section~\ref{sub:compute_few_ls} we show how to solve (i), i.e.
quickly compute a small number of leverage scores. In Section~\ref{sub:detect_changed_ls}
we then show how to solve (ii) by showing how to detect which leverage
scores have changed significantly and must be recomputed. The last
Section~\ref{sub:ls_maintenance_proof} then combines these two results
to prove Theorem~\ref{thm:leverage_score_datastructure}.

Before we proceed to prove (i) and (ii), we first give a more accurate
outline of how the algorithm of Theorem~\ref{thm:leverage_score_datastructure}
works. An important observation is, that by standard dimension reduction
tricks, a large part of Theorem~\ref{thm:leverage_score_datastructure}
can be obtained by reduction to the vector data-structures of Section~\ref{sec:vec_maintenance}.
Note that the $i$-th leverage score of $\mg^{(t)}\ma$ is
\[
\sigma(\mg^{(t)}\ma)_{i}=\indicVec i^{\top}\mg\t\ma(\ma^{\top}(\mg\t)^{2}\ma)^{-1}\ma^{\top}\mg\t\indicVec i=\|e_{i}^{\top}\mg\t\ma(\ma^{\top}(\mg\t)^{2}\ma)^{-1}\ma^{\top}\mg\t\|_{2}^{2}
\]
and the term on the right-hand side can be estimated using a JL-matrix
$\mr\in\R^{n\times b}$. This allows us to greatly reduce the complexity
of estimating $\sigma(\mg^{(t)}\ma)$ as $\mg\t\ma(\ma^{\top}(\mg\t)^{2}\ma)^{-1}\ma^{\top}\mg\t\mr$
is essentially just $O(b)$ matrix-vector products, where $\mg\t\ma$
is the matrix and the columns of $(\ma^{\top}(\mg\t)^{2}\ma)^{-1}\ma^{\top}\mg\t\mr$
are the vectors. With this in mind, we use the following corollary
to prove Theorem~\ref{thm:leverage_score_datastructure}.
\begin{cor}
\label{cor:large_entry_datastructure} There exists a Monte-Carlo
data-structure, that holds against an adaptive adversary, with the
following procedures:
\begin{itemize}
\item \textsc{Initialize($\ma,g$):} Given a matrix $\ma\in\R^{n\times d}$
and a scaling $g\in\R^{n}$, the data-structure preprocesses in $O(\nnz(\ma)\log^{4}n)$
time.
\item \textsc{Query($\mh,F,\epsilon$):} Given a matrix $\mh\in\R^{d\times b}$
and a set $F\subset[n]$, w.h.p. in $n$ the data-structure outputs
the matrix 
\[
\mv_{i,j}=\begin{cases}
((\mi-\mf)\mg\ma\mh)_{i,j}, & \text{if }|(\mg\ma\mh)_{i,j}|\ge\varepsilon\\
0, & \text{otherwise}
\end{cases}
\]
 where $\mf\in\{0,1\}^{n\times n}$ is a $0$-$1$ diagonal matrix
with $\mf_{ii}=1$ if and only if $i\in F$ in time 
\[
O(\|(\mi-\mf)\mg\ma\mh\|_{F}^{2}\cdot\varepsilon^{-2}\cdot d\cdot\log^{3}n+|F|\cdot d\cdot\log^{4}n+db).
\]
\item \textsc{Scale($i,u$):} Sets $g_{i}\leftarrow u$ in $O(d\log^{4}n)$
time.
\end{itemize}
\end{cor}

\begin{proof}
The algorithm is directly obtained from Lemma~\ref{lem:large_entry_datastructure}.
Let $\mh,F,\epsilon$ be the input to \textsc{Query} and let $D$
be an instance of Lemma~\ref{lem:large_entry_datastructure}. We
perform $D.\textsc{Scale}(i,0)$ for all $i\in F$, so the data-structure
$D$ represents $(\mi-\mf)\mg\ma$. We then obtain $\mv$ by setting
the $i$-th column of $\mv$ to the result of $D.\textsc{Query}(\mh e_{i},\epsilon)$
for $i=1,...,b$. Afterward, we call $D.\textsc{Scale}$ again, to
revert back to the previous scaling. The correctness and the time
complexity of this new data-structure follow from Lemma~\ref{lem:large_entry_datastructure}.
\end{proof}
Now that we have Corollary~$\ref{cor:large_entry_datastructure}$,
we can outline how we obtain Theorem~\ref{thm:leverage_score_datastructure}
from it. Given a matrix $\Psi\approx(\ma^{\top}(\mg\t)^{2}\ma)^{-1}$,
it is easy to compute individual leverage scores (see Section~\ref{sub:compute_few_ls},
Lemma $\ref{lem:EstScore}$). However, such computation requires nearly
linear work, which is quite slow for our purposes, and therefore we
do not want to recompute every leverage score every time $\mg\t$
changes a bit. Instead, we use Corollary~$\ref{cor:large_entry_datastructure}$
to detect which scores have changed a lot. This is done in Section~\ref{sub:detect_changed_ls},
Lemma~$\ref{lem:ls:good_query_result}$ and Lemma~s$\ref{lem:ls:find_index}$.
The idea is that, if $\|\indicVec i^{\top}\mg\t\ma(\ma^{\top}(\mg\t)^{2}\ma)^{-1}\ma^{\top}\mg\t\|_{2}$
and $\|\indicVec i^{\top}\mg\k\ma(\ma^{\top}(\mg\k)^{2}\ma)^{-1}\ma^{\top}\mg\k\|_{2}$
differ substantially, for some $k<t$, then 
\[
\|\indicVec i^{\top}\left(\mg\t\ma(\ma^{\top}(\mg\t)^{2}\ma)^{-1}\ma^{\top}\mg\t-\mg\k\ma(\ma^{\top}(\mg\k)^{2}\ma)^{-1}\ma^{\top}\mg\k\right)\mr\|_{2}
\]
 should be quite large as well for some JL-matrix $\mr$, i.e. something
we can detect via Corollary~$\ref{cor:large_entry_datastructure}$.

\subsection{Computing Few Leverage Scores \label{sub:compute_few_ls}}

We start our proof of Theorem \ref{thm:leverage_score_datastructure}
by showing how to quickly compute a small set of leverage scores.
This corresponds to the function $\textsc{EstimateScore}$ in Algorithm
\ref{alg:leverage_score_datastructure_1}.
\begin{lem}[Guarantee of $\textsc{EstimateScore}$]
\label{lem:EstScore} Assume $\ttau\t\approx_{1/2}\sigma(\mg\t\ma)+\frac{d}{n}$
and $\Psi\t\safe\approx_{\delta/4}(\ma^{\top}(\mg\t)^{2}\ma)^{-1}$
for some $\delta>$0. Then, $\textsc{EstimateScore}(J,t,\delta)$
returns $v\in\R^{n}$ such that
\[
v_{i}\approx_{\delta}\sigma(\mg\t\ma)_{i}+\frac{d}{n}
\]
w.h.p. in $n$. Furthermore, the cost for the first iteration is $O((\nnz(A)+T_{\Psi\safe})\cdot(\log n)/\delta^{2})$
and the cost of other iterations is $O((T_{\Psi\safe}+d^{2}+d|J|)(\log n)/\delta^{2})$.
\end{lem}

\begin{proof}
\textbf{Correctness}: For the first step, we have 
\begin{align*}
\|e_{j}^{\top}\mg\t\ma\Psi\t\safe\ma^{\top}\mg\t\|_{2}^{2} & =e_{j}^{\top}\mg\t\ma\Psi\t\safe\ma^{\top}(\mg\t)^{2}\ma\Psi\t\safe\ma^{\top}\mg\t e_{j}\\
 & \approx_{\delta/4}e_{j}^{\top}\mg\t\ma\Psi\t\safe\ma^{\top}\mg\t e_{j}\\
 & \approx_{\delta/4}e_{j}^{\top}\mg\t\ma(\ma^{\top}(\mg\t)^{2}\ma)^{-1}\ma^{\top}\mg\t e_{j}=\sigma(\mg\t\ma)_{j}.
\end{align*}
where both approximations follows from the assumption $\Psi\t\safe\approx_{\delta/4}(\ma^{\top}(\mg\t)^{2}\ma)^{-1}$.
Now, by JL, Lemma \ref{lem:JL_lemma}, and the choice of $\widetilde{\mr}$,
we have 
\begin{align*}
\|e_{j}^{\top}\mg\t\ma\Psi\t\safe\ma^{\top}\mg\t\widetilde{\mr}\|_{2}^{2} & \approx_{\delta/4}\|e_{j}^{\top}\mg\t\ma\Psi\t\safe\ma^{\top}\mg\t\|_{2}^{2}\approx_{\delta/2}\sigma(\mg\t\ma)_{j}\,.
\end{align*}
Hence, we have the result.

For the other case, at Line~\ref{line:ls:sampling}, we sample row
$i$ with probability at least $\Theta(1)\frac{\log n}{\delta^{2}}\ttau_{i}\t\geq\Theta(1)\frac{\log n}{\delta^{2}}\sigma(\mg\t\ma)_{i}$
(using the assumption on $\ttau_{i}\t$). Hence, the leverage score
sampling guarantee shows that $\ma^{\top}\widetilde{\mg}^{2}\ma\approx_{\delta/4}\ma^{\top}(\mg\t)^{2}\ma$
by choosing large enough constant in the $\Theta(1)$ factor in the
probability. Hence, by the same argument as above, we have the result.

\paragraph{Complexity}

For the first step, it takes $O(nd(\log n)/\delta^{2})$ time to compute
$\ma^{\top}\mg\t\widetilde{\mr}$ and takes $O(d^{2}(\log n)/\delta^{2})$
time to compute $\Psi\t\safe\ma^{\top}\mg\t\widetilde{\mr}$. After
computing $\Psi\t\safe\ma^{\top}\mg\t\widetilde{\mr}$, it takes $|J|\log n/\delta^{2}$
times to compute the norm for each $j$. Since both $d$ and $|J|$
is smaller than $n$, the cost is dominated by $O(nd(\log n)/\delta^{2})$.

For the other case, we note that $\sum_{i}\ttau\t_{i}=O(d)$ and hence
$\|\widetilde{g}\|_{0}=O(d\log n/\delta^{2})$. Therefore, it takes
$O(d^{2}\log n/\delta^{2})$ time to compute $\ma^{\top}\widetilde{\mg}\widetilde{\mr}$.
Similar as above, the total cost is $O((T_{\Psi\safe}+d^{2}+d|J|)\log n/\delta^{2})$.
\end{proof}

\subsection{Detecting Leverage Scores Change\label{sub:detect_changed_ls}}

Now that we know how to quickly compute a small set of leverage scores,
we must figure out which scores to compute. We only want to recompute
the scores, that have changed substantially within some time-span.
More accurately, we want to show that $\textsc{FindIndices}$ (Algorithm~\ref{alg:leverage_score_datastructure_1})
does indeed find all indices that changed a lot. The proof is split
into two parts: We first show that we can detect large changes that
happen within a sequence of updates of length $2^{i}$ in Lemma~\ref{lem:ls:good_query_result}.
We then extend the result to all large leverage score changes in Lemma~\ref{lem:ls:find_index}.
\begin{lem}
\label{lem:ls:good_query_result} Consider an execution of Line~\ref{line:ls:Ji}
(Algorithm~\ref{alg:leverage_score_datastructure_1}) at iteration
$t$. Suppose $\Psi\t\approx_{\epsilon/12\log n}(\ma^{\top}(\mg\t)^{2}\ma)^{-1}$
and $\Psi^{(t-2^{i})}\approx_{\epsilon/36\log n}(\ma^{\top}(\mg^{(t-2^{i})})^{2}\ma)^{-1}$.
Then $j\in J$ w.h.p. in $n$ if 
\begin{equation}
\sigma(\mg^{(t-2^{i})}\ma)_{j}+\frac{d}{n}\not\approx_{\epsilon/2\log n}\sigma(\mg^{(t)}\ma)_{j}+\frac{d}{n}\,.\label{eq:sigma_j_apx}
\end{equation}
Further, the set $J$ remains the same w.h.p. in $n$ if we replace
Line~\ref{line:ls:Ji} by $J_{i}\leftarrow[n]$.
\end{lem}

\begin{proof}
First, we prove that for any $j\in[n]$ with
\begin{equation}
\sigma(\mg^{(t)}\ma)_{j}+\frac{d}{n}\not\approx_{\frac{\epsilon}{6\log n}}\sigma(\mg^{(t-2^{i})}\ma)_{j}+\frac{d}{n},\label{eq:sigma_j_apx2}
\end{equation}
we have $j\in J_{i}$ (see Line \ref{line:ls:Ji}). To do this, consider
any $j\notin J_{i}$, i.e. $j$ is not in $F_{i}$ and not picked
by the function $D.\textsc{Query}$. By Line~\ref{line:ls:Ji} and
Corollary~\ref{cor:large_entry_datastructure}, $j$ not picked by
the function $D.\textsc{Query}$ implies that 
\[
\left|((\mi-\mf_{i})\mg\t\ma(\Psi\t\mb^{(t)}-\Psi^{(t-2^{i})}\mb^{(t-2^{i})}))_{j,\ell}\right|\leq\epsilon_{\Psi}\defeq\frac{\epsilon}{48\log n}\sqrt{\frac{d}{nb}}
\]
for all $\ell$. Since $j$ is not in $F_{i}$, it means that $g_{j}$
does not change in the last $2^{i}$ iterations. Hence, we have
\[
((\mi-\mf_{i})\mg\t\ma(\Psi\t\mb^{(t)}-\Psi^{(t-2^{i})}\mb^{(t-2^{i})}))_{j,\ell}=((\mg\t\ma\Psi\t\ma\mg\t-\mg^{(t-2^{i})}\ma\Psi^{(t-2^{i})}\ma\mg^{(t-2^{i})})\mr)_{j,\ell}
\]
Using this and that there are only $b$ many columns in $\mr$ (namely,
there are only $b$ many $\ell$), we have that
\[
\|\indicVec j^{\top}(\mg\t\ma\Psi\t\ma\mg\t-\mg^{(t-2^{i})}\ma\Psi^{(t-2^{i})}\ma\mg^{(t-2^{i})})\mr\|_{2}^{2}\leq\epsilon_{\Psi}^{2}b.
\]
Hence, we have that w.h.p. in $n$
\begin{align*}
 & \left(\|e_{j}^{\top}\mg\t\ma\Psi\t\ma\mg\t\|_{2}^{2}+\frac{d}{n}\right)^{1/2}-\left(\|e_{j}^{\top}\mg^{(t-2^{i})}\ma\Psi^{(t-2^{i})}\ma\mg^{(t-2^{i})})\|_{2}^{2}+\frac{d}{n}\right)^{1/2}\\
\leq & \|e_{j}^{\top}\mg\t\ma\Psi\t\ma\mg\t-e_{j}^{\top}\mg^{(t-2^{i})}\ma\Psi^{(t-2^{i})}\ma\mg^{(t-2^{i})}\|_{2}\\
\leq & 2\|e_{j}^{\top}(\mg\t\ma\Psi\t\ma\mg\t-\mg^{(t-2^{i})}\ma\Psi^{(t-2^{i})}\ma\mg^{(t-2^{i})})\mr\|_{2}\leq2\epsilon_{\Psi}\sqrt{b}
\end{align*}
where the second inequality follows from JL (Lemma~\ref{lem:JL_lemma})and
the choice of $\mr$. By the definition of $\approx$ and that $\epsilon_{\Psi}$
is small enough, we have
\[
\left(\|e_{j}^{\top}\mg\t\ma\Psi\t\ma\mg\t\|_{2}^{2}+\frac{d}{n}\right)^{1/2}\approx_{4\epsilon_{\Psi}\sqrt{b}/\sqrt{\frac{d}{n}}}\left(\|e_{j}^{\top}\mg^{(t-2^{i})}\ma\Psi^{(t-2^{i})}\ma\mg^{(t-2^{i})})\|_{2}^{2}+\frac{d}{n}\right)^{1/2}.
\]
Squaring both sides, using the fact that 
\begin{align*}
\|e_{j}^{\top}\mg\t\ma\Psi\t\ma^{\top}\mg\t\|_{2}^{2} & =e_{j}^{\top}\mg\t\ma\Psi\t\ma^{\top}(\mg\t)^{2}\ma\Psi\t\ma^{\top}\mg\t e_{j}\\
 & \approx_{\epsilon/(36\log n)}e_{j}^{\top}\mg\t\ma\Psi\t\ma^{\top}\mg\t e_{j}\approx_{\epsilon/(36\log n)}\sigma(\mg^{(t)}\ma)_{j}
\end{align*}
and similarly the fact that $\|e_{j}^{\top}\mg^{(t-2^{i})}\ma\Psi^{(t-2^{i})}\ma\mg^{(t-2^{i})})\|_{2}^{2}\approx_{\epsilon/(18\log n)}\sigma(\mg^{(t-2^{i})}\ma)_{j}$,
we have
\[
\sigma(\mg^{(t)}\ma)_{j}+\frac{d}{n}\approx_{8\epsilon_{\Psi}\sqrt{\frac{nb}{d}}+\frac{2\epsilon}{18\log n}}\sigma(\mg^{(t-2^{i})}\ma)_{j}+\frac{d}{n}.
\]
Using $\epsilon_{\Psi}=\frac{\epsilon}{144\log n}\sqrt{\frac{d}{nb}}$,
we showed that $j\notin J_{i}$ implies \eqref{eq:sigma_j_apx2} is
false. This proves the claim that for any $j$ such that \eqref{eq:sigma_j_apx2}
holds, we have that $j\in J_{i}$.

Now, we use the claim to prove the statement of this lemma. Fix any
$j$ such that $\sigma(\mg\t\ma)_{j}+\frac{d}{n}\not\approx_{\epsilon/(2\log n)}\sigma(\mg^{(t-2^{i})}\ma)_{j}+\frac{d}{n}$.
By the claim, we have that $j\in J_{i}$, namely, $j$ will be detected
by $D.\textsc{Query}$. Now, Lemma~\ref{lem:EstScore} shows that
$v_{j}\not\approx_{\epsilon/(3\log n)}v_{j}\old$ and hence $j$ will
be put into $J$ by Line \ref{line:ls:Ji2}. This proves that $j\in J$.

Finally, we prove that the output remains the same without using $D.\textsc{Query}$.
Note that $j\in J$ implies $v_{j}\not\approx_{\epsilon/(3\log n)}v_{j}\old$
by Line~\ref{line:ls:Ji2}. Again, by Lemma~\ref{lem:EstScore},
we see that this implies $\sigma(\mg\t\ma)_{j}+\frac{d}{n}\not\approx_{\epsilon/(6\log n)}\sigma(\mg^{(t-2^{i})}\ma)_{j}+\frac{d}{n}$
w.h.p. in $n$. Hence, the claim in the beginning shows that $j\in J_{i}$.
Therefore, regardless of whether we use the current definition of
$J_{i}$ or the alternative definition that $J_{i}=[n]$, the output
$J$ remains the same w.h.p. in $n$.
\end{proof}
The previous lemma only showed how to detect leverage scores that
changed within some time-span of length $2^{i}$, we now show that
this is enough to detect all changes, no matter if they happen slowly
(over some long time-span) or radically (within one iteration), or
something in-between.
\begin{lem}[Guarantee of $\textsc{FindIndices}$]
\label{lem:ls:find_index} Under the assumption of Lemma~\ref{lem:EstScore}
and Lemma~\ref{lem:ls:good_query_result} let $J\t$ be the output
of $\textsc{FindIndices}$ (Algorithm~\ref{alg:leverage_score_datastructure_1})
at $t$-th iteration and 
\[
J^{*(t)}=\left\{ j:\ttau\t_{j}\not\approx_{\epsilon}\sigma(\mg\t\ma)_{j}+\frac{d}{n}\right\} .
\]
Then, we have that $J^{*(t)}\subset J\t$ w.h.p. in $n$.
\end{lem}

\begin{proof}
To show that $J^{*(t)}\subset J\t$, we consider any $j\notin J\t$.
Let $t$ be the current iteration and $\overline{t}$ be the last
time $\ttau_{j}^{(\overline{t})}$ is updated. Since the algorithm
detects changes of shifted leverage scores at $2^{\ell}$ iterations,
we consider a sequence 
\[
t_{0}=t\ge t_{1}\ge\cdots\ge t_{u}=\overline{t}
\]
such that $t_{z}\mod2^{\ell_{z}}=0$ and $t_{z+1}=t_{z}-2^{\ell_{z}}$
for all $z$ (namely, the algorithm is trying to detect the change
between $t_{z}$ and $t_{z+1}$ in Line~\ref{line:ls:query}). We
can obtain such a sequence of length at most $2\log t$ as follows:
Let $t'=\lfloor t/2^{i}\rfloor\cdot2^{i}$, for the largest $i$ such
that $1\le2^{i}\le t-k$. Then $k\le t'\le t$. We now create a sequence
where we first flip the 1s to 0s in the binary representation of $t$,
from least significant to greatest significant bit. This is a sequence
$t_{0}=t\ge t_{1}\ge\cdots\ge t_{v}=t'$ of length $v\le\log t$.
Next, we continue the sequence $t'=t_{v}\ge t_{v+1}\ge\cdots\ge t_{u}=\bar{t}$
as follows: Consider the sequence in reverse, then we always add $2^{\ell_{z-1}}$
to $t_{z}$ to obtain $t_{z-1}$, where $\ell_{z-1}$ is the position
of the least significant 1-bit in $t_{z}$. By adding this $2^{\ell_{z-1}}$,
we flip the 1-bit of $t_{z}$ to a 0, but might flip a few more significant
1-bits via the carrier. This creates a sequence that will reach $t'$
after at most $\log t'$ steps. Since we recompute all shifted leverage
score every $n$ (or less) iterations, we have $t\leq n$ and hence
$u\leq2\log n$.

Next, we note that $\ttau_{j}^{(t_{z})}$ is not updated implies $j\notin J^{(t_{z})}$.
By Lemma~\ref{lem:ls:good_query_result}, we see that $j\notin J^{(t_{z})}$,
$t_{z}\mod2^{\ell_{z}}=0$ and $t_{z+1}=t_{z}-2^{\ell_{z}}$ implies
\begin{equation}
\sigma(\mg^{(t_{z})}\ma)_{j}+\frac{d}{n}\approx_{\epsilon/2\log n}\sigma(\mg^{(t_{z+1})}\ma)_{j}+\frac{d}{n}\label{eq:tau_chain_claim}
\end{equation}
 for $z=0,1,\cdots,u-1$. Since there are only $2\log n$ steps, \eqref{eq:tau_chain_claim}
shows that
\begin{equation}
\sigma(\mg^{(\overline{t})}\ma)_{j}+\frac{d}{n}\approx_{\epsilon/2}\sigma(\mg^{(t)}\ma)_{j}+\frac{d}{n}.\label{eq:tau_chain_1}
\end{equation}
Finally, we note that 
\begin{equation}
\ttau\t_{j}=\ttau{}_{j}^{(\overline{t})}\label{eq:tau_chain_2}
\end{equation}
because $\ttau$ has not been updated and that Lemma~\ref{lem:EstScore}
shows 
\begin{equation}
\ttau{}_{j}^{(\overline{t})}\approx_{\epsilon/2}\sigma(\mg^{(\overline{t})}\ma)_{j}+\frac{d}{n}\label{eq:tau_chain_3}
\end{equation}
for the last update. Combining \eqref{eq:tau_chain_1}, \eqref{eq:tau_chain_2}
and \eqref{eq:tau_chain_3}, we see that $\ttau\t_{j}\approx_{\epsilon}\sigma(\mg^{(t)}\ma)_{j}+\frac{d}{n}$
and hence $j\notin J^{*(t)}$. This shows that $j\notin J\t$ implies
$j\notin J^{*(t)}$. Hence, $J^{*(t)}\subset J\t$.
\end{proof}

\subsection{Leverage Score Maintenance\label{sub:ls_maintenance_proof}}

We can now combine the previous results Lemma~$\text{\ref{lem:EstScore}, Lemma�\ref{lem:ls:good_query_result}}$
and Lemma~$\text{\ref{lem:ls:find_index}}$ to show that Algorithm~\ref{alg:leverage_score_datastructure_1}
yields a proof of Theorem~
\begin{lem}[Correctness of Theorem \ref{thm:leverage_score_datastructure}]
\label{lem:ls:approx_sigma}Assume $\epsilon\in[0,1/4]$, $\Psi\t\approx_{\epsilon/(12\log n)}(\ma^{\top}(\mg\t)^{2}\ma)^{-1}$
and $\Psi\t\safe\approx_{\epsilon/(24\log n)}(\ma^{\top}(\mg\t)^{2}\ma)^{-1}$
for all $t$. After each call to \textsc{Query}, the returned vector
$\ttau{}_{j}^{(t)}$ satisfies
\[
\ttau{}^{(t)}\approx_{\varepsilon}\sigma(\mg\t\ma)+\frac{d}{n}.
\]
\end{lem}

\begin{proof}
We prove this statement by induction. At the start of the algorithm
the vector $\ttau{}^{(0)}$ is a good enough approximation of the
shifted leverage scores because of Lemma \ref{lem:EstScore}. We are
left with proving that with each call to \textsc{Query}, $\ttau{}^{(t)}$
is updated appropriately.

By the induction, we have that $\ttau{}^{(t-1)}\approx_{\varepsilon}\sigma(\mg^{(t-1)}\ma)+\frac{d}{n}$.
Furthermore, we know that $g\t\approx_{1/16}g^{(t-1)}$. Hence, we
have $\sigma(\mg^{(t)}\ma)\approx_{1/4}\sigma(\mg^{(t-1)}\ma)$. Using
$\epsilon\in[0,1/4]$, we have that $\ttau{}^{(t-1)}\approx_{\varepsilon}\sigma(\mg^{(t)}\ma)+\frac{d}{n}$.
Hence, the $\ttau{}^{(t-1)}$ is a good enough approximation and satisfies
the assumption of Lemma~\ref{lem:EstScore} and hence the assumption
of Lemma~\ref{lem:ls:find_index}. Lemma \ref{lem:ls:find_index}
shows that $\textsc{FindIndices}()$ finds all indices that is $\epsilon$
far away from the shifted leverage score. Lemma~\ref{lem:EstScore}
then shows that those $\ttau$ will be correctly updated. Hence, all
$\ttau$ is at most $\epsilon$ far away from the target. This finishes
the induction.
\end{proof}
\begin{proof}[Proof of Theorem~\ref{thm:leverage_score_datastructure}]

The correctness for the approximation ratio of Theorem \ref{thm:leverage_score_datastructure}
was already proven in Lemma~\ref{lem:ls:approx_sigma}. Here we bound
the time complexity and prove that the data-structure works against
an adaptive adversary.

\paragraph{Adaptive Adversary:}

Lemma~\ref{lem:ls:good_query_result} shows that the output $\ttau{}^{(t)}$
of the data-structure is independent to $\textsc{FindIndices}$ and
$\textsc{D.Query}$ w.h.p. in $n$. Hence, the output $\ttau{}^{(t)}$
of the data-structure depends only on $\textsc{EstimateScore}$. Note
that $\textsc{EstimateScore}$ is calculated using $\Psi\safe$ and
some independently sampled $\widetilde{\mr}$. So, it cannot leak
information about the random sketch $\mr$ nor the input $\Psi$.
Hence, the data structure works against an adaptive adversary.

To bound the complexity, we first bound how many leverage scores the
data structure estimate, namely the size of $J$.

\paragraph{Size of all $J$:}

An index $j$ may be added to $J$ in Line \ref{line:ls:Ji2}. These
$j$ must satisfy $\mw_{j,\ell}\neq0$ or $g\t_{j}\neq g_{j}^{(t-2^{i})}$,
where $\mw$ is the result of Line \ref{line:ls:query} ($D.\textsc{Query}$).
Let $\mw^{(i,t)}$ be the matrix $\mw$ during the $t$th call to
\textsc{FindIndices} on scale $i$, then the sum of the size of all
$J$ outputs by \textsc{FindIndices} over $T$ iterations is bounded
by
\begin{equation}
\sum_{i\in\log T}\sum_{k\in T/2^{i}}\|\mw^{(i,k2^{i})}\|_{0}+\sum_{i\in\log T}\sum_{k\in T/2^{i}}\|g^{(k2^{i})}-g^{(k2^{i}-2^{i})}\|_{0}\label{eq:size_J}
\end{equation}

For the first term in \eqref{eq:size_J}, Corollary~\ref{cor:large_entry_datastructure}
and the definition of $\mb$ shows that $\mw_{j,\ell}\neq0$ if 
\[
|(\mg^{(k2^{i})}\ma\Psi^{(k2^{i})}\ma^{\top}\mg^{(k2^{i})}\mr-\mg^{(k2^{i}-2^{i})}\ma\Psi^{(k2^{i}-2^{i})}\ma^{\top}\mg^{(k2^{i}-2^{i})}\mr)_{j,\ell}|=\Omega\Big(\frac{\epsilon}{\log n}\sqrt{\frac{d}{nb}}\Big).
\]
Hence, we have that 
\begin{align*}
 & \sum_{k=1}^{T/2^{i}}\|\mw^{(i,k2^{i})}\|_{0}\\
 & \leq O\Big(\frac{nb\log^{2}n}{d\epsilon^{2}}\Big)\sum_{k=1}^{T/2^{i}}\|\mg^{(k2^{i})}\ma\Psi^{(k2^{i})}\ma^{\top}\mg^{(k2^{i})}\mr-\mg^{(k2^{i}-2^{i})}\ma\Psi^{(k2^{i}-2^{i})}\ma^{\top}\mg^{(k2^{i}-2^{i})}\mr\|_{F}^{2}\\
 & =O\Big(\frac{nb\log^{2}n}{d\epsilon^{2}}\Big)\left(\sum_{k=1}^{T/2^{i}}\|\mg^{(k2^{i})}\ma\Psi^{(k2^{i})}\ma^{\top}\mg^{(k2^{i})}\mr-\mg^{(k2^{i}-2^{i})}\ma\Psi^{(k2^{i}-2^{i})}\ma^{\top}\mg^{(k2^{i}-2^{i})}\mr\|_{F}\right)^{2}\\
 & =O\Big(\frac{nb\log^{2}n}{d\epsilon^{2}}\Big)\left(\sum_{k=1}^{T}\|\mg^{(t)}\ma\Psi^{(t)}\ma^{\top}\mg^{(t)}\mr-\mg^{(t-1)}\ma\Psi^{(t-1)}\ma^{\top}\mg^{(t-1)}\mr\|_{F}\right)^{2}\\
 & =O\Big(\frac{nb\log^{2}n}{d\epsilon^{2}}\Big)\left(\sum_{k=1}^{T}\|\mg^{(t)}\ma\Psi^{(t)}\ma^{\top}\mg^{(t)}-\mg^{(t-1)}\ma\Psi^{(t-1)}\ma^{\top}\mg^{(t-1)}\|_{F}\right)^{2}
\end{align*}
where we used Johnson Lindenstrauss lemma \ref{lem:JL_lemma} at the
end. Now, summing up over $\log T$ scales and using that we restart
every $O(n)$ iterations, we have
\[
\sum_{i=1}^{\log T}\sum_{k=1}^{T/2^{i}}\|\mw^{(i,k2^{i})}\|_{0}=O\Big(\frac{nb}{d\epsilon^{2}}\log^{3}n\Big)\cdot\left(\sum_{k=1}^{T}\|\mg^{(t)}\ma\Psi^{(t)}\ma^{\top}\mg^{(t)}-\mg^{(t-1)}\ma\Psi^{(t-1)}\ma^{\top}\mg^{(t-1)}\|_{F}\right)^{2}.
\]

For the second term in \eqref{eq:size_J}, we bound it simply by $O(\log n)\cdot\sum_{t=1}^{T}\|g^{(t)}-g^{(t-1)}\|_{0}$.
Hence, the sum of the size of all $J$ is bounded by 
\begin{align}
\sum_{t=1}^{T}|J\t|\leq & O\Big(\frac{n}{d\epsilon^{2}}\log^{4}n\Big)\cdot\left(\sum_{t=1}^{T}\|\mg^{(t)}\ma\Psi^{(t)}\ma^{\top}\mg^{(t)}-\mg^{(t-1)}\ma\Psi^{(t-1)}\ma^{\top}\mg^{(t-1)}\|_{F}\right)^{2}\nonumber \\
 & +O(\log n)\cdot\sum_{t=1}^{T}\|g^{(t)}-g^{(t-1)}\|_{0}.\label{eq:size_J_ans}
\end{align}

\paragraph{Complexity}

The cost of $\textsc{Query}$ is dominated by the cost of $\textsc{EstimateScore}$,
the cost of computing $\Psi\t\mb\t$ and the cost of $\textsc{D.Query}$.
Lemma \ref{lem:EstScore} shows that the total cost of $\textsc{EstimateScore}$
is bounded by
\begin{equation}
\underbrace{nd\delta^{-2}\log n}_{\text{First Step}}+\underbrace{T\cdot(T_{\Psi}+d^{2})\cdot\delta^{-2}\log n}_{\text{Cost of \ensuremath{T} Step}}+\underbrace{\sum_{t=1}^{T}|J\t|}_{\text{Size of }|J|}\cdot\underbrace{d\delta^{-2}\log n}_{\text{Cost per score}}.\label{eq:est_cost}
\end{equation}
The cost of computing $\Psi\t\mb\t$ is simply $O(T_{\Psi}\log n)$
because by assumption we can multiply $\Psi\t$ with a vector in $T_{\Psi}$
time, and $\mb\t$ is a $d\times O(\log n)$ matrix. Finally, by Corollary
\ref{cor:large_entry_datastructure}, the total cost of $\textsc{D.Query}$
is bounded by
\begin{align*}
 & \sum_{i=1}^{\log T}\sum_{k=1}^{T/2^{i}}\left\Vert \mg^{(k2^{i})}\ma\Psi^{(k2^{i})}\ma^{\top}\mg^{(k2^{i})}\mr-\mg^{(k2^{i}-2^{i})}\ma\Psi^{(k2^{i}-2^{i})}\ma^{\top}\mg^{(k2^{i}-2^{i})}\mr\right\Vert _{F}^{2}\cdot\left(\frac{\log n}{\epsilon}\sqrt{\frac{nb}{d}}\right)^{2}\cdot d\log^{3}n\\
+ & \sum_{i=1}^{\log T}\sum_{k=1}^{T/2^{i}}\|g^{(k2^{i})}-g^{(k2^{i}-2^{i})}\|_{0}\cdot d\log^{4}n
\end{align*}
Simplify the term by the similar calculation on bounding $|J|$, we
have
\begin{equation}
\left(\sum_{t=1}^{T}\|\mg^{(t)}\ma\Psi^{(t)}\ma^{\top}\mg^{(t)}-\mg^{(t-1)}\ma\Psi^{(t-1)}\ma^{\top}\mg^{(t-1)}\|_{F}\right)^{2}\cdot\epsilon^{-2}n\log^{7}n+\sum_{t=1}^{T}\|g^{(t)}-g^{(t-1)}\|_{0}\cdot d\log^{5}n.\label{eq:query_cost}
\end{equation}
Combining \eqref{eq:est_cost} and \eqref{eq:query_cost} and using
\eqref{eq:size_J_ans} and $\delta=\frac{\epsilon}{6\log n}$, we
have the total cost bounded by 
\begin{align*}
 & nd\epsilon^{-2}\log^{3}n+T((T_{\Psi}+d^{2})\epsilon^{-2}\log^{3}n+T_{\Psi}\log n)\\
+ & \left(\sum_{t=1}^{T}\|\mg^{(t)}\ma\Psi^{(t)}\ma^{\top}\mg^{(t)}-\mg^{(t-1)}\ma\Psi^{(t-1)}\ma^{\top}\mg^{(t-1)}\|_{F}\right)^{2}\cdot n\epsilon^{-4}\log^{7}n\\
+ & \left(\sum_{t=1}^{T}\|g^{(t)}-g^{(t-1)}\|_{0}\right)\cdot(d\epsilon^{-2}\log^{5}n)
\end{align*}
We charge the first term to $\textsc{Initalize}$ and the last term
to $\textsc{Scale}$. This completes the proof for the time complexity.
\end{proof}

\section{Inverse Maintenance with Leverage Score Hints}

\label{sec:inverse_maintenance}

In this section, we show how to approximately maintain the inverse
$(\ma^{\top}\mw\ma)^{-1}$ where $\ma\in\R^{n\times d}$ with $n\gg d$
in amortized $\otilde(d^{\omega-(1/2)}+d^{2})$ per step under changes
to $\mw$ using estimates of the leverage scores of $\mw^{1/2}\ma$.
When $n\approx d$, this is a well-studied problem that has been solved
in \cite{lee2015efficient,adil2019iterative,cohen2019solving,BrandNS19,Sankowski04};
no information of leverage score is required by these algorithms and
it is possible to achieve a runtime of roughly $n^{\omega-(1/2)}+n^{2}$.
However, when $n\gg d$, we need leverage scores to avoid reading
the whole matrix $\ma$ every step and achieve the (possibly sublinear)
time bound of $\otilde(d^{\omega-(1/2)}+d^{2})$.

One approach to obtain this result is to first approximate $\ma^{\top}\mw\ma$
by $\ma^{\top}\mv\ma$ for some sparse diagonal matrix $\mv$ with
$\otilde(d)$ non-zeros using leverage score sampling each step and
then maintain $\ma^{\top}\mv\ma$ via  Woodbury matrix identity. Unfortunately,
this does not immediately work as the diagonal matrices $\mv$ may
then differ widely in every step, inhibiting our ability to perform
smaller low-rank updates per step.

To handle this issue of changing $\mv$, \cite{lee2015efficient}
resampled rows of $\ma$ only when the corresponding leverage score
changed signficantly. With this trick, the diagonal matrix $\mv$
changes slowly enough to maintain. However, this approach only works
if the input of the algorithm is independent on the output. In \cite{lee2015efficient},
this issue was addressed by using the maintained inverse as a pre-conditioner
to solve linear systems in $\ma^{\top}\mw\ma$ to high precision and
then added noise to the output. However, we cannot afford to read
the whole matrix $\ma$ each step and therefore need to modify this
technique.

To avoid reading the entire matrix $\ma$, we instead make further
use of the given given leverage score to sample the matrix $\ma^{\top}\mathbf{U}\ma\approx\ma^{\top}\mw\ma$.
Instead of solving $\ma^{\top}\mw\ma x=b$, we use the maintained
inverse $(\ma^{\top}\mv\ma)^{-1}$ as a pre-conditioner to solve $\ma^{\top}\mathbf{U}\ma x=b$.
By solving $\ma^{\top}\mathbf{U}\ma x=b$ very accurately and by adding
a small amount of noise, we can hide any information about $\mv$.
This is formalized by Algorithm~\ref{alg:inverse_maintence} (and
its second half Algorithm~\ref{alg:inverse_maintence-part2}).

\begin{algorithm2e}[!t]

\caption{Inverse Maintenance with Leverage Score Hints (Theorem~\ref{thm:inverse_main})}

\label{alg:inverse_maintence}

\SetKwProg{Members}{members}{}{}

\SetKwProg{Proc}{procedure}{}{}

\SetKwProg{PrivateProc}{private procedure}{}{}

\Members{}{

\State $\ma\in\R^{n\times d}$, $\epsilon\in(0,1/4)$\tcp*{ Input
matrix and error tolerance}

\State $\gamma\defeq c_{1}\log n$ for some large enough constant
$c_{1}$

\State $\tw^{\alg},\ttau^{\alg}\in\R^{n}$\tcp*{Current approximate
weight and leverage score }

\State $v\in\R^{n}$\tcp*{Current maintained sparse weights with
$\ma^{\top}\mv\ma\approx\ma^{\top}\mw^{\alg}\ma$}

\State $\Psi\in\R^{\otilde(d)\times\otilde(d)}$\tcp*{Current maintained
inverse matrix with $\Psi=(\ma^{\top}\mv\ma)^{-1}$}

}

\vspace{0.1 in}

\Proc{$\textsc{Initialize}$$(\ma\in\R^{n\times d},\tw\in\R_{>0}^{n},\ttau\in\R_{>0}^{n},\epsilon)$
\tcp*[f]{\textbf{Assume} $\ttau\in(1\pm\frac{1}{8\left\lceil \log d\right\rceil })\tau(w)$}}{

\State $\ma\leftarrow\ma$, $\tw^{\alg}\leftarrow\tw$, $\ttau^{\alg}\leftarrow\ttau$,
$\epsilon\leftarrow\epsilon$

\State $v_{i}\leftarrow\begin{cases}
\tw_{i}^{\alg}/\min\{1,\gamma\epsilon^{-2}\cdot\ttau_{i}^{\alg}\} & \text{with probability }\min\{1,\gamma\epsilon^{-2}\cdot\ttau_{i}^{\alg}\}\\
0 & \text{otherwise}
\end{cases}\text{ for }i\in[n]$

\State $\Psi\leftarrow(\ma^{\top}\mv\ma)^{-1}$

}

\vspace{0.1 in}

\Proc{$\textsc{Update}$$(\tw\in\R_{>0}^{n},\ttau\in\R_{>0}^{n})$
\tcp*[f]{\textbf{Assume} $\ttau\in(1\pm\frac{1}{8\left\lceil \log d\right\rceil })\tau(w)$}}{

\State $y_{i}\leftarrow\frac{8}{\epsilon}(\tw_{i}/\tw_{i}^{\alg}-1)$
and $y_{i+n}\leftarrow2(\ttau_{i}/\ttau_{i}^{\alg}-1)$ for $i\in[n]$

\State Let $\pi:[2n]\rightarrow[2n]$ be a sorting permutation such
that $|y_{\pi(i)}|\geq|y_{\pi(i+1)}|$

\State For each integer $\ell$, we define $i_{\ell}$ be the smallest
integer such that $\sum_{j\in[i_{\ell}]}\ttau_{\pi(j)}\geq2^{\ell}$.

\State Let $k$ be the smallest integer such that $|y_{\pi(i_{k})}|\leq1-\frac{k}{2\left\lceil \log d\right\rceil }$

\For{each coordinate $j\in[i_{k}]$}{\label{line:low_rank_update}

\State  Set $i=\pi(j)$ if $\pi(j)\leq n$ and set $i=\pi(j)-n$
otherwise

\State $\tw_{i}^{\alg}\leftarrow\tw_{i}$ and $\ttau_{i}^{\alg}\leftarrow\ttau_{i}$

\State $v_{i}\leftarrow\begin{cases}
\tw_{i}^{\alg}/\min\{1,\gamma\epsilon^{-2}\cdot\ttau_{i}^{\alg}\} & \text{with probability }\min\{1,\gamma\epsilon^{-2}\cdot\ttau_{i}^{\alg}\}\\
0 & \text{otherwise}
\end{cases}$

\State Update $\Psi$ to $(\ma^{\top}\mv\ma)^{-1}$ using Woodbury
matrix identity

\State \Return $\Psi$, $\tw^{\alg}$

}

}

\algstore{inversemaintence}

\end{algorithm2e}

\begin{algorithm2e}[!t]

\caption{Inverse Maintenance with Leverage Score Hints (Continuation
of Algorithm~\ref{alg:inverse_maintence})}

\label{alg:inverse_maintence-part2}

\SetKwProg{Proc}{procedure}{}{}

\SetKwProg{PrivateProc}{private procedure}{}{}

\algrestore{inversemaintence}

\Proc{$\textsc{Solve}$$(b\in\R^{d},\ow\in\R_{>0}^{n},\delta>0)$
\tcp*[f]{Secure and unbiased version}}{

\State  $y\leftarrow\textsc{SecureSolve}(b,\ow,\delta/2^{20})$.
$y^{(0)}\leftarrow y$.

\State  Let $X\in\R_{\geq0}$ be geometrically distributed such that
$\P[X\ge k]=2^{-k}$.

\For{$k=1,\cdots,X$}{

\State $u_{i}\leftarrow\begin{cases}
\ow_{i}/\min\{1,\gamma\delta^{-2}\cdot\ttau_{i}^{\alg}\} & \text{with probability }\min\{1,\gamma\delta^{-2}\cdot\ttau_{i}^{\alg}\}\\
0 & \text{otherwise}
\end{cases}\text{ for }i\in[n]$

\State $y^{(k)}\leftarrow2(y^{(k-1)}-\frac{1}{2}y^{(0)}+\Psi(b-\ma^{\top}\mathbf{U}\ma y))$
\tcp*{$\Psi b\defeq\textsc{SecureSolve}(b,\ow,\frac{1}{8})$}\label{line:inverse_gd}

}

\Return $y^{(X)}$

}

\vspace{0.1 in}

\PrivateProc{$\textsc{SecureSolve}$$(b\in\R^{d},\ow\in\R_{>0}^{n},\delta>0)$
\tcp*[f]{Used by Lemma \ref{lem:totalPmovement}}}{

\State $u_{i}\leftarrow\begin{cases}
\ow_{i}/\min\{1,\gamma\delta^{-2}\cdot\ttau_{i}^{\alg}\} & \text{with probability }\min\{1,\gamma\delta^{-2}\cdot\ttau_{i}^{\alg}\}\\
0 & \text{otherwise}
\end{cases}\text{ for }i\in[n]$

\State $c_{1}\leftarrow b$, $c_{2}\leftarrow\ma^{\top}\mathbf{U}^{1/2}\eta$
where $\eta\sim N(0,\mi_{n})$

\State Initialize $y_{1}$ and $y_{2}$ as $0$ vectors of length
$d$

\For{$k=1,\cdots,c_{1}\log(n/\delta)$ for some large enough constant
$c_{1}$}{

\State $y_{j}\leftarrow y_{j}+\frac{1}{10}\Psi(c_{j}-\ma^{\top}\mathbf{U}\ma y_{j})$
for $j\in\{1,2\}$\label{line:inverse_gd-2}

}

\State \textbf{return} $y_{1}+\alpha y_{2}$ where $\alpha=c_{3}\sqrt{\frac{\delta}{d}}\|y_{1}\|_{\ma^{\top}\mathbf{U}\ma}$
for some constant $c_{3}$.

}

\vspace{0.1 in}

\PrivateProc{$\textsc{IdealSolve}$$(b\in\R^{d},\ow\in\R_{>0}^{n},\delta>0)$
\tcp*[f]{Used for proof only}}{

\State $u_{i}\leftarrow\begin{cases}
\ow_{i}/\min\{1,\gamma\delta^{-2}\cdot\ttau_{i}^{\alg}\} & \text{with probability }\min\{1,\gamma\delta^{-2}\cdot\ttau_{i}^{\alg}\}\\
0 & \text{otherwise}
\end{cases}\text{ for }i\in[n]$

\State $c_{1}\leftarrow b$, $c_{2}\leftarrow\ma^{\top}\mathbf{U}^{1/2}\eta$
where $\eta\sim N(0,\mi_{n})$

\Return  $(\ma^{\top}\mathbf{\mathbf{U}}\ma)^{-1}(c_{1}+\alpha c_{2})$
where $\alpha=c_{3}\sqrt{\frac{\delta}{d}}\|y_{1}\|_{\ma^{\top}\mathbf{U}\ma}$
for some constant $c_{3}$.

}

\end{algorithm2e}

\inversemain*

Before we prove Theorem \ref{thm:inverse_main} we prove two helper
lemmas. First, note that Theorem \ref{thm:inverse_main} asks for
the solver to return a solution $y=\Psi b$ where $\Psi^{-1}\approx_{\delta}\ma^{\top}\omw\ma$
and $\E[\Psi b]=(\ma^{\top}\omw\ma)^{-1}b$. The following lemma shows
how we can satisfy this condition on the expectation while having
only a small additive error to our solution. The later Lemma~\ref{lem:perturb_system}
then shows that this small additive error can be interpreted as a
small spectral error of the inverse used to solve the system.
\begin{lem}
\label{lem:preconditioner} Let $\mm\in\R^{d\times d}$ be a positive
definite matrix and let $\mm^{(1)},\mm^{(2)},...$ be a sequence of
independent random $d\times d$ matrices, each with $\E[\mm\k]=\mm$
that satisfy with high probability $(1-\epsilon)\mm\prec\mm\k\prec(1+\epsilon)\mm$
for some $\frac{1}{8}\geq\epsilon>0$. Further let $\mn^{(1)},\mn^{(2)},...$
be a sequence of independent identical distributed random $d\times d$
matrices, each satisfying with high probability $\frac{7}{8}\mn\k\prec\mm^{-1}\prec\frac{8}{7}\mn\k$.
Let $X$ be a geometric distributed random variable with $\P[X\ge k]=2^{-k}$.

For any vector $b\in\R^{d}$ and any initial vector $y^{(0)}$, define
$y^{(i+1)}=2(y^{(i)}-\frac{1}{2}y^{(0)}+\mn^{(i+1)}(b-\mm^{(i+1)}y^{(i)}))$
for all $i\geq0$. Then $\E[y^{(X)}]=\mm^{-1}b$, and with high probability
we have 
\begin{equation}
\|y^{(X)}-y^{(0)}\|_{\mm}\leq30\|b-\mm y^{(0)}\|_{\mm^{-1}}+30\epsilon\|b\|_{\mm^{-1}}.\label{eq:x_err}
\end{equation}
\end{lem}

\begin{proof}
We first show that $\E[y^{(X)}]$ exists. To show that, we first bound
$\|y^{(i)}-y^{(0)}\|_{\mm}$ for all $i$. Note that
\begin{align*}
y^{(i)}-y^{(0)} & =2y^{(i-1)}-2y^{(0)}+2\mn^{(i)}(b-\mm^{(i)}y^{(i-1)})\\
 & =2(\mi-\mn^{(i)}\mm^{(i)})(y^{(i-1)}-y^{(0)})+2\mn^{(i)}(b-\mm^{(i)}y^{(0)}).
\end{align*}
Hence, we have
\begin{equation}
\|y^{(i)}-y^{(0)}\|_{\mm}\leq2\|\mi-\mn^{(i)}\mm^{(i)}\|_{\mm}\|y^{(i-1)}-y^{(0)}\|_{\mm}+2\|\mn^{(i)}(b-\mm^{(i)}y^{(0)})\|_{\mm}.\label{eq:y_norm_bound}
\end{equation}
Note that 
\[
(\frac{7}{8})^{2}(\mm^{(i)})^{-1}\preceq\frac{7}{8}\mm^{-1}\preceq\mn^{(i)}\preceq\frac{8}{7}\mm^{-1}\preceq(\frac{8}{7})^{2}(\mm^{(i)})^{-1}.
\]
Hence, we have
\[
-0.31\cdot\mi\preceq(1-(\frac{8}{7})^{2})\mi\preceq\mi-(\mm^{(i)})^{1/2}\mn^{(i)}(\mm^{(i)})^{1/2}\preceq(1-(\frac{7}{8})^{2})\mi\preceq0.24\cdot\mi.
\]
Using $\frac{7}{8}\mm\preceq\mm^{(i)}\preceq\frac{8}{7}\mm$, we have
\begin{align*}
\|\mi-\mn^{(i)}\mm^{(i)}\|_{\mm} & =\|\mi-\mm^{1/2}\mn^{(i)}\mm^{(i)}\mm^{-1/2}\|_{2}\\
 & =\|\mi-\mm^{1/2}(\mm^{(i)})^{-1/2}(\mm^{(i)})^{1/2}\mn^{(i)}(\mm^{(i)})^{-1/2}(\mm^{(i)})^{1/2}\mm^{-1/2}\|_{2}\\
 & \leq\|\mm^{1/2}(\mm^{(i)})^{-1/2}\|_{2}\|\mi-(\mm^{(i)})^{1/2}\mn^{(i)}(\mm^{(i)})^{-1/2}\|_{2}\|(\mm^{(i)})^{1/2}\mm^{-1/2}\|_{2}\\
 & \leq\sqrt{\frac{7}{8}}\cdot0.31\cdot\sqrt{\frac{8}{7}}\leq\frac{2}{5}.
\end{align*}
Substituting into (\ref{eq:y_norm_bound}), we have
\begin{align}
\|y^{(i)}-y^{(0)}\|_{\mm} & \leq\frac{4}{5}\|y^{(i-1)}-y^{(0)}\|_{\mm}+2\|\mn^{(i)}(b-\mm^{(i)}y^{(0)})\|_{\mm}\nonumber \\
 & \leq\frac{4}{5}\|y^{(i-1)}-y^{(0)}\|_{\mm}+3\|b-\mm^{(i)}y^{(0)}\|_{\mm^{-1}}\label{eq:y_diff}
\end{align}
where we used that $\frac{7}{8}\mm^{-1}\preceq\mn^{(i)}\preceq\frac{8}{7}\mm^{-1}$
at the end. Using $(1-\epsilon)\mm\prec\mm\k\prec(1+\epsilon)\mm$,
we have
\begin{align*}
\|b-\mm^{(i)}y^{(0)}\|_{\mm^{-1}} & =\|\mm^{(i)}\mm^{-1}\mm(\mm^{(i)})^{-1}b-\mm^{(i)}\mm^{-1}\mm y^{(0)}\|_{\mm^{-1}}\\
 & \leq\|\mm^{(i)}\mm^{-1}\|_{\mm^{-1}}\|\mm(\mm^{(i)})^{-1}b-\mm y^{(0)}\|_{\mm^{-1}}\\
 & \leq(1+\epsilon)\left(\|b-\mm y^{(0)}\|_{\mm^{-1}}+\|\mm(\mm^{(i)})^{-1}b-b\|_{\mm^{-1}}\right)\\
 & \leq2\|b-\mm y^{(0)}\|_{\mm^{-1}}+2\epsilon\|b\|_{\mm^{-1}}.
\end{align*}
Putting this into (\ref{eq:y_diff}), we have
\[
\|y^{(i)}-y^{(0)}\|_{\mm}\leq\frac{4}{5}\|y^{(i-1)}-y^{(0)}\|_{\mm}+6\|b-\mm y^{(0)}\|_{\mm^{-1}}+6\epsilon\|b\|_{\mm^{-1}}.
\]
Solving it, we have
\[
\|y^{(i)}-y^{(0)}\|_{\mm}\leq30\|b-\mm y^{(0)}\|_{\mm^{-1}}+30\epsilon\|b\|_{\mm^{-1}}
\]
for all $i$. Since $\|y^{(i)}\|_{\mm}$ is bounded uniformly, $\E y^{(X)}$
exists. This also gives the bound for (\ref{eq:x_err}).

Now, we compute $\E y^{(X)}$ and note that 
\begin{align*}
\E[y^{(X)}] & =\frac{1}{2}y^{(0)}+\frac{1}{2}\E[2(y^{(X-1)}-\frac{1}{2}y^{(0)}+\mn^{(X-1)}(b-\mm^{(X-1)}y^{(X-1)}))|X\geq1]\\
 & =\frac{1}{2}y^{(0)}+\E[y^{(X-1)}-\frac{1}{2}y^{(0)}+\mn^{(X-1)}(b-\mm^{(X-1)}y^{(X-1)}))|X\geq1]\\
 & =\frac{1}{2}y^{(0)}+\E[y^{(X)}-\frac{1}{2}y^{(0)}+\mn^{(X)}(b-\mm^{(X)}y^{(X)})]\\
 & =\E[y^{(X)}]+\E_{i}\left[\E_{\mm^{(i)},\mn^{(i)},y^{(i)}}\left[\mn^{(i)}(b-\mm^{(i)}y^{(i)})|X=i\right]\right]\\
 & =\E[y^{(X)}]+\E_{i}\left[\E\left[\mn^{(i)}|X=i\right]\left(b-\E\left[\mm^{(i)}|X=i\right]\E\left[y^{(i)}|X=i\right]\right)\right]\\
 & =\E[y^{(X)}]+\E_{i}\left[\E[\mn]\left(b-\mm\E\left[y^{(i)}|X=i\right]\right)\right]\\
 & =\E[y^{(X)}]+\E[\mn](b-\mm\E[y^{(X)}])
\end{align*}
where we used that $X-1$ subjects to $X\geq1$ has the same distribution
as $X$ in the third equality, $\mn^{(i)},\mm^{(i)}$ and $y^{(i)}$
are independent on the fifth equality, we used that $\E\mm^{(i)}=\mm$
and that $\mn^{(i)}$ are identically distributed on the sixth equality. 

As $\E[\mn^{(X)}]$ is positive definite, this implies $\E y^{(X)}=\mm^{-1}b$.
\end{proof}
Next we provide a general stability lemma about solving linear systems
and then use this to prove Theorem~\ref{thm:inverse_main} and Theorem~\ref{lem:totalPmovement}.
This lemma is mainly for convenience purpose as the vector maintenance
and leverage score maintenance assume we solve the linear system $\mh y=b$
by some $y=\Psi b$ where $\Psi\approx\mh^{-1}$. The following lemma
shows that this definition is same as $\|y-\mh^{-1}b\|_{\mh}$ is
small compared to $\|b\|_{\mh^{-1}}$.
\begin{lem}
\label{lem:perturb_system} Let $y=\mh^{-1}(b+v)\in\R^{d}$ such that
$\|v\|_{\mh^{-1}}\leq\epsilon\|b\|_{\mh^{-1}}$ with $0\leq\epsilon\leq\frac{1}{80}$
for symmetric positive definite $\mh\in\R^{d\times d}$ and vector
$b\in\R^{d}$. There is a symmetric matrix $\mDelta\in\R^{d\times d}$
such that $y=(\mh+\mDelta)^{-1}b\in\R^{d}$ and $\mh+\mDelta\approx_{20\epsilon}\mh$.
\end{lem}

\begin{proof}
Note that the claim is trivial if $\norm v_{\mh^{-1}}=0$ and thus
we assume this is not the case. Consequently, by the assumption on
$v$ we have $\norm{b+v}_{\mh^{-1}}>0$ and can define the vectors
$\overline{v}=\frac{\mh^{-1/2}v}{\|v\|_{\mh^{-1}}}\in\R^{d}$ and
$\overline{z}=\frac{\mh^{-1/2}(b+v)}{\|b+v\|_{\mh^{-1}}}\in\R^{d}$
where $s=\sign(b^{\top}\mh^{-1}v)$. We define 
\[
\mDelta=s\frac{\|v\|_{\mh^{-1}}}{\|b+v\|_{\mh^{-1}}}\mh^{1/2}\Big(\overline{z}\overline{z}^{\top}-\frac{(\overline{z}+s\overline{v})(\overline{z}+s\overline{v})^{\top}}{(\overline{z}+s\overline{v})^{\top}\overline{z}}\Big)\mh^{1/2}
\]
Note that
\begin{align*}
\mDelta y= & \mDelta\mh^{-1}(b+v)\\
= & s\frac{\|v\|_{\mh^{-1}}}{\|b+v\|_{\mh^{-1}}}\mh^{1/2}\Big(\overline{z}\overline{z}^{\top}-\frac{(\overline{z}+s\overline{v})(\overline{z}+s\overline{v})^{\top}}{(\overline{z}+s\overline{v})^{\top}\overline{z}}\Big)\mh^{-1/2}(b+v)\\
= & s\|v\|_{\mh^{-1}}\mh^{1/2}\Big(\overline{z}\overline{z}^{\top}-\frac{(\overline{z}+s\overline{v})(\overline{z}+s\overline{v})^{\top}}{(\overline{z}+s\overline{v})^{\top}\overline{z}}\Big)\overline{z}\\
= & s\|v\|_{\mh^{-1}}\mh^{1/2}(\overline{z}-(\overline{z}+s\overline{v}))=-v.
\end{align*}
Hence, we have $(\mh+\Delta)y=b+v-v=b$. Hence, we have $y=(\mh+\mDelta)^{-1}b$
if $\mh+\mDelta\succ0$. In fact, we now prove that $\mh+\mDelta\approx_{20\epsilon}\mh$.

First, we note that $\|v\|_{\mh^{-1}}\leq\frac{1}{2}\|b\|_{\mh^{-1}}$
and hence $\frac{\|v\|_{\mh^{-1}}}{\|b+v\|_{\mh^{-1}}}\leq\frac{2\|v\|_{\mh^{-1}}}{\|b\|_{\mh^{-1}}}\leq2\epsilon$.
Hence, we have
\begin{align*}
\mh^{-1/2}\mDelta\mh^{-1/2} & \preceq2\epsilon\Big(\overline{z}\overline{z}^{\top}+\frac{(\overline{z}+s\overline{v})(\overline{z}+s\overline{v})^{\top}}{(\overline{z}+s\overline{v})^{\top}\overline{z}}\Big).
\end{align*}
The term $\overline{z}\overline{z}^{\top}$ is simply bounded by $\mi$.
We note that $(\overline{z}+s\overline{v})(\overline{z}+s\overline{v})^{\top}\preceq\|\overline{z}+s\overline{v}\|^{2}\cdot\mi\preceq4\cdot\mi$.
Finally, we note that $(\overline{z}+s\overline{v})^{\top}\overline{z}\geq1$
because $s=\sign(\ov^{\top}\overline{z})$. Hence, we have
\[
\mh^{-1/2}\mDelta\mh^{-1/2}\preceq2\epsilon\cdot5\cdot\mi=10\epsilon\cdot\mi.
\]
Similarly, we have $\mh^{-1/2}\Delta\mh^{-1/2}\succeq-10\epsilon\cdot\mi$.
Hence, we have $\mh+\mDelta\approx_{20\epsilon}\mh$.
\end{proof}
We now have all tools available to prove Theorem \ref{thm:inverse_main}.
\begin{proof}[Proof of Theorem \ref{thm:inverse_main}]
 Since the input $\tw^{(i)}\approx w^{(i)}$ of the algorithm is
independent to the output of the algorithm in the previous iterations,
every constraint $i\in[n]$ has $\Theta(\tau_{i}(w)\epsilon^{-2}\gamma)$
probability of being picked into $\ma^{\top}\mv\ma$. Therefore, the
expected number of rank changes in Line~\ref{line:low_rank_update}
of $\textsc{Update}(\tw)$ is given by
\[
\Theta(\gamma\epsilon^{-2})\sum_{j\in[i_{\ell}]}\tau_{j}(w)=\Theta(2^{\ell}\epsilon^{-2}\log(n)).
\]
For simplicity, we call such step, a \emph{rank $2^{\ell}$ update}.
Hence, the total complexity of $\textsc{Update}(w)$ over $K$ iterations
is
\[
O(Kn)+\sum_{\ell=0}^{\left\lceil \log n\right\rceil }(\text{\# of rank }2^{\ell}\text{ updates)}\cdot(\text{cost of rank }2^{\ell}\text{ update)}.
\]
Note that a rank $2^{\ell}$ update can be implemented in $O(d^{2}(2^{\ell}\cdot\epsilon^{-2}\cdot\log(n))^{\omega-2})$
time. Now, we bound the number of such updates over $T$ iterations.

The algorithm defines an error vector $y_{i}\leftarrow\frac{8}{\epsilon}(\tw_{i}/\tw_{i}^{\alg}-1)$
and $y_{i+n}\leftarrow2(\ttau_{i}/\ttau_{i}^{\alg}-1)$. We now also
define the true error vector $y_{i}^{*}\leftarrow\frac{8}{\epsilon}(w_{i}/w_{i}^{\alg}-1)$
and $y_{i+n}^{*}\leftarrow2(\tau_{i}/\tau_{i}^{\alg}-1)$ for $i\in[n]$
where $\tau_{i}^{\alg}$ is updated to $\tau_{i}$ whenever $\ttau_{i}^{\alg}$
is updated to $\ttau_{i}$, and likewise $w_{i}^{\alg}$ is updated
to $w_{i}$ whenever $\tw_{i}^{\alg}$ is updated to $\tw_{i}$. Let
$y^{*(0)}\in\R^{d}$ be the true error vector defined by the algorithm
directly \emph{after} a rank $2^{\ell}$ update. Assume there is another
rank $2^{\ell}$ update after $b$ iterations. Let $y^{*(1)},y^{*(2)},\cdots,y^{*(j)},\cdots,y^{*(b)}\in\R^{d}$
be the true error vector \emph{after} the $j$-th iteration for $j=1,...,b$.
We order the entries of each $y^{*(j)}$ and $y^{(j)}$ in the same
way such that the entries of $|y^{(b)}|$ (the error vector of the
algorithm) are decreasing.  Let $j_{i}$ be the last iteration $w_{i}^{\alg}$
(if $i\leq n$) or $\tau_{i-n}^{\alg}$ (if $i>n$) is updated after
this first update. We set $j_{i}=0$ if it is never updated after
the first update. If it is never updated, then we have $|y_{i}^{*(j_{i})}|=|y_{i}^{*(0)}|\le|y_{i}^{(0)}|+\frac{1}{8\left\lceil \log d\right\rceil }\leq1-\frac{\ell}{2\left\lceil \log d\right\rceil }+\frac{1}{8\left\lceil \log d\right\rceil }$
where $1-\frac{\ell}{2\left\lceil \log d\right\rceil }$ comes from
the threshold for $2^{\ell}$ update and $\frac{1}{8\left\lceil \log d\right\rceil }$
comes from the error of $\ttau$ and $\tw$ compared to $\tau$ and
$w$. If it is updated during the interval, we have $|y_{i}^{*(j_{i})}|=\frac{1}{8\left\lceil \log d\right\rceil }$
(due to the error of $\ttau$). In both cases, we have that 
\[
|y_{i}^{*(j_{i})}|\leq1-\frac{\ell}{2\left\lceil \log d\right\rceil }+\frac{1}{8\left\lceil \log d\right\rceil }\text{ for all }i\in[n].
\]
On the other hand, that we are making a $2^{\ell}$ update at step
$b$ implies it does not pass the threshold for a $2^{\ell-1}$ update.
Hence, we have
\begin{align*}
|y_{i}^{*(b-1)}| & \geq1-\frac{\ell-1}{2\left\lceil \log d\right\rceil }-\frac{1}{8\left\lceil \log d\right\rceil }\text{ for all }i\leq i_{\ell-1},
\end{align*}
where $i_{\ell}$ is the one defined in the $b$-th iteration. In
particular, we have that $|y_{i}^{*(b-1)}-y_{i}^{*(j_{i})}|\geq\frac{1}{4\left\lceil \log d\right\rceil }$
for all $i\leq i_{\ell-1}$. Let $\tau_{i}^{(j)}$ be the shifted
leverage score $\tau(w)$ before the $j$-th iteration. Then, by the
definition of $i_{\ell-1}$, we have
\begin{equation}
\Omega\Big(\frac{1}{\log^{2}d}\Big)\cdot2^{\ell-1}\leq\sum_{i\in[n]}\tau_{i}^{(b-1)}(y_{i}^{*(b-1)}-y_{i}^{*(j_{i})})^{2}.\label{eq:lower_change_maintain}
\end{equation}
Next, we note that the denominator of $y_{i}^{*(j)}$ never changed
between $j_{i}$ and $b$. Hence, $y_{i}^{*(j+1)}-y_{i}^{*(j)}$ is
exactly equals to the relative change of $w$ or $\tau$. Hence, we
have
\begin{align}
\sum_{i\in[n]}\tau_{i}^{(b-1)}(y_{i}^{*(b-1)}-y_{i}^{*(j_{i})})^{2} & \leq O(b)\sum_{i\in[n]}\tau_{i}^{(b-1)}\sum_{j=j_{i}}^{b-1}(y_{i}^{*(j+1)}-y_{i}^{*(j)})^{2}\nonumber \\
 & \leq O(b)\sum_{i\in[n]}\sum_{j=j_{i}}^{b-1}\tau_{i}^{(j)}(y_{i}^{*(j+1)}-y_{i}^{*(j)})^{2}\nonumber \\
 & =O(b^{2})\label{eq:upper_change_maintain}
\end{align}
where we used that $\ttau_{i}^{(j)}$ (and thus also $\tau_{i}^{(j)}$)
does not change more than a constant from iteration $j_{i}$ to iteration
$b-1$ in the second inequality and the assumption (\ref{eq:input_ass})
at the end. From (\ref{eq:lower_change_maintain}) and (\ref{eq:upper_change_maintain}),
we see that there is at least $b\geq\Omega(\frac{2^{\ell/2}}{\log d})$
many iterations between two $2^{\ell}$ rank updates. Hence, the total
complexity of $\textsc{Update}(w)$ over $K$ iterations is
\[
O(Kn)+\sum_{\ell=0}^{\left\lceil \log d\right\rceil }(\frac{K\log d}{2^{\ell/2}})\cdot d^{2}(\epsilon^{-2}2^{\ell}\log(n))^{\omega-2}=O(Kn+\epsilon^{-2}K\cdot(d^{\omega-(1/2)}+d^{2})\cdot\log^{3/2}(n)).
\]

The cost of $\textsc{Solve}$ is clear from the description.

For the correctness of $\textsc{Update}$, we note that the internal
$\tw^{\alg}\in\R^{n}$ and the input $\tw\in\R^{n}$ satisfies $\left|\frac{\tw^{\alg}-\tw}{\tw}\right|\leq\frac{\epsilon}{4}$
and that $\ma^{\top}\mv\ma\approx_{\epsilon/4}\ma^{\top}\tmw^{\alg}\ma$
with probability $1-1/\poly(nK)$ due to the matrix Chernoff bound
(see e.g. the proof of Lemma 4 in \cite{cohen2015uniform}). Here,
we used the fact that the input $\tw^{\new}\in\R^{n}$ is independent
of the algorithm output and hence each sample of $v_{i}$ are independent.
Hence, we have $\ma^{\top}\mv\ma\approx_{\epsilon}\ma^{\top}\tmw\ma$
with probability $1-1/\poly(n)$.

For the correctness of $\textsc{Solve}$, assume for now that $\Psi$
in Line \ref{line:inverse_gd} satisfies $\Psi^{-1}=(1\pm1/8)\ma^{\top}\omw\ma$
and that $y^{(0)}$ is a good approximation of the solution with
\begin{align*}
\|y^{(0)}-(\ma^{\top}\mathbf{\omw}\ma)^{-1}b\|_{(\ma^{\top}\mathbf{\omw}\ma)}\le & (\delta/1800)\cdot\|b\|_{(\ma^{\top}\mathbf{\omw}\ma)^{-1}}.
\end{align*}
 (This is shown in the proof of Lemma~\ref{lem:totalPmovement}.)
So Lemma~\ref{lem:preconditioner} shows that the final result $y$
satisfies
\begin{align*}
\|b-(\ma^{\top}\mathbf{\omw}\ma)^{-1}y\|_{(\ma^{\top}\mathbf{\omw}\ma)^{-1}}\le & 30\|b-(\ma^{\top}\mathbf{\omw}\ma)^{-1}y^{(0)}\|_{(\ma^{\top}\mathbf{\omw}\ma)^{-1}}+(\delta/60)\|b\|_{(\ma^{\top}\mathbf{\omw}\ma)^{-1}}\\
\le & \delta/30\|b\|_{(\ma^{\top}\mathbf{\omw}\ma)^{-1}}.
\end{align*}
Finally, Lemma~\ref{lem:perturb_system} shows that we can view the
output of the algorithm as $\Psi b$ for some spectral approximation
$\Psi\approx_{\delta}\ma^{\top}\mathbf{\omw}\ma$.
\end{proof}
\totalPmovement*
\begin{proof}
The output $y\in\R^{d}$ of the function $\textsc{SecureSolve}$ solves
satisfies
\begin{equation}
y=\Psi_{\textsc{Solve}}(b+\alpha\ma^{\top}\mathbf{U}^{1/2}\eta)\label{eq:y_formula-1}
\end{equation}
with $\alpha=\frac{c_{3}\epsilon}{\sqrt{d\log(n/\delta)}}\|y_{1}\|_{\ma^{\top}\mathbf{U}\ma}$
for some $\Psi_{\textsc{Solve}}$ satisfying $\Psi_{\textsc{Solve}}\approx_{\poly(\delta/n)}(\ma^{\top}\mathbf{U}\ma)^{-1}$
as the algorithm performs iterative refinement on our preconditioner
$\Psi$ maintained by $\textsc{Update}$. Note that 
\begin{align*}
\alpha^{2}\|\ma^{\top}\mathbf{U}^{1/2}\eta\|_{\Psi_{\textsc{Solve}}}^{2} & \leq2\alpha^{2}\eta^{\top}\mathbf{U}^{1/2}\ma(\ma^{\top}\mathbf{U}\ma)^{-1}\ma^{\top}\mathbf{U}^{1/2}\eta=O(\alpha^{2}\log(n)d)
\end{align*}
w.h.p. in $n$ where we used that $\mathbf{U}^{1/2}\ma(\ma^{\top}\mathbf{U}\ma)^{-1}\ma^{\top}\mathbf{U}^{1/2}$
is a rank $d$ orthogonal projection matrix and $\eta\sim N(0,\mi_{n})$.
By choosing small enough $c_{3}$, we have
\[
\alpha\|\ma^{\top}\mathbf{U}^{1/2}\eta\|_{\Psi_{\textsc{Solve}}}\leq\frac{\delta}{200}\|y_{1}\|_{\ma^{\top}\mathbf{U}\ma}=\frac{\delta}{200}\|\Psi_{\textsc{Solve}}b\|_{\ma^{\top}\mathbf{U}\ma}\leq\frac{\delta}{100}\|b\|_{\Psi_{\textsc{Solve}}}.
\]
By picking small enough $c_{3}$, Lemma~\ref{lem:perturb_system}
shows that we have $y=(\Psi_{\textsc{Solve}}^{-1}+\mDelta)^{-1}b$
for $\Psi_{\textsc{Solve}}^{-1}+\mDelta\approx_{\delta/2}\Psi_{\textsc{Solve}}^{-1}\approx_{\delta}\ma^{\top}\omw\ma$.

Note that we also have
\[
y\sim N(\Psi_{\textsc{Solve}}b,\alpha^{2}\Psi_{\textsc{Solve}}\ma^{\top}\mathbf{U}\ma\Psi_{\textsc{Solve}}).
\]
On the other hand, the output $y\in\R^{d}$ of the function $\textsc{IdealSolve}$
satisfies
\[
y^{(\mathrm{ideal})}\sim N((\ma^{\top}\mathbf{U}\ma)^{-1}b,\alpha^{2}(\ma^{\top}\mathbf{U}\ma)^{-1}\ma^{\top}\mathbf{U}\ma(\ma^{\top}\mathbf{U}\ma)^{-1}).
\]
Since $\Psi_{\textsc{Solve}}^{-1}\approx_{\poly(\epsilon/(nK))}\ma^{\top}\mathbf{\mathbf{U}}\ma$,
the total variation between $y$ and $y^{(\mathrm{ideal})}\in\R^{d}$
is $\poly(\epsilon/(nK))$ small.

Now consider what happens when we replace the matrix $\Psi$ in Lines
\ref{line:inverse_gd} with the function $\textsc{IdealSolve}$. Note
that Theorem~\ref{thm:inverse_main} then holds when the input $w\in\R^{n}$
and $\ttau\in\R^{n}$ of the algorithm depends on the output of $\textsc{Solve}$
because $\textsc{Solve}$ and $\textsc{IdealSolve}$ do not use $v$
and hence cannot leak any information about $v$. If Theorem \ref{thm:inverse_main}
does not hold with input depending on $\textsc{Solve}$, when we replace
the matrix $\Psi$ in Lines \ref{line:inverse_gd} with the function
$\textsc{SecureSolve}$ instead, then this would give an algorithm
to distinguish the functions $\textsc{IdealSolve}$ and $\textsc{SecureSolve}$
with probability at least $\poly(\epsilon/(nK))$. This contradicts
with the total variation between $y$ and $y^{(\mathrm{ideal})}\in\R^{d}$.

For the Frobenius bound we will write for simplicity $\mw^{(k)}$
for $\tmw^{\alg}$. We note that
\begin{align}
E_{k}\defeq & \left\Vert \sqrt{\mw^{(k+1)}}\ma\Psi^{(k+1)}\ma^{\top}\sqrt{\mw^{(k+1)}}-\sqrt{\mw^{(k)}}\ma\Psi^{(k)}\ma^{\top}\sqrt{\mw^{(k)}}\right\Vert _{F}^{2}\nonumber \\
= & 3\left\Vert \sqrt{\mw^{(k+1)}}\ma\Psi^{(k+1)}\ma^{\top}\sqrt{\mw^{(k+1)}}-\sqrt{\mw^{(k)}}\ma\Psi^{(k+1)}\ma^{\top}\sqrt{\mw^{(k+1)}}\right\Vert _{F}^{2}\nonumber \\
 & +3\left\Vert \sqrt{\mw^{(k)}}\ma\Psi^{(k+1)}\ma^{\top}\sqrt{\mw^{(k+1)}}-\sqrt{\mw^{(k)}}\ma\Psi^{(k)}\ma^{\top}\sqrt{\mw^{(k+1)}}\right\Vert _{F}^{2}\nonumber \\
 & +3\left\Vert \sqrt{\mw^{(k)}}\ma\Psi^{(k)}\ma^{\top}\sqrt{\mw^{(k+1)}}-\sqrt{\mw^{(k)}}\ma\Psi^{(k)}\ma^{\top}\sqrt{\mw^{(k)}}\right\Vert _{F}^{2}\nonumber \\
= & 3\cdot\tr\left[(\sqrt{\mw^{(k+1)}}-\sqrt{\mw^{(k)}})^{2}\ma\Psi^{(k+1)}\ma^{\top}\mw^{(k+1)}\ma\Psi^{(k+1)}\ma\right]\nonumber \\
 & +3\cdot\tr\left[\ma^{\top}\mw^{(k)}\ma(\Psi^{(k+1)}-\Psi^{(k)})\ma^{\top}\mw^{(k+1)}\ma(\Psi^{(k+1)}-\Psi^{(k)})\right]\nonumber \\
 & +3\cdot\tr\left[(\sqrt{\mw^{(k+1)}}-\sqrt{\mw^{(k)}})^{2}\ma\Psi^{(k)}\ma^{\top}\mw^{(k)}\ma\Psi^{(k)}\ma\right]\nonumber \\
\leq & 8\cdot\tr\left[(\sqrt{\mw^{(k+1)}/\mw^{(k)}}-\mi)^{2}\sqrt{\mw^{(k)}}\ma\Psi^{(k)}\ma^{\top}\mw^{(k)}\ma\Psi^{(k)}\ma\sqrt{\mw^{(k)}}\right]\nonumber \\
 & +4\cdot\tr\left[\ma^{\top}\mw^{(k)}\ma(\Psi^{(k+1)}-\Psi^{(k)})\ma^{\top}\mw^{(k)}\ma(\Psi^{(k+1)}-\Psi^{(k)})\right]\nonumber \\
\defeq & E_{k}^{(1)}+E_{k}^{(2)}\label{eq:E_E1E2}
\end{align}
where we used $\Psi^{(k+1)}\preceq\frac{9}{8}\Psi^{(k)}$ and $\mw^{(k+1)}\preceq\frac{9}{8}\mw^{(k)}$
in the first inequality.

For the first term, we use $\Psi^{(k)}\preceq\frac{9}{8}(\ma^{\top}\mw^{(k)}\ma)^{-1}$,
$\mproj^{(k)}\defeq\sqrt{\mw^{(k)}}\ma(\ma^{\top}\mw^{(k)}\ma)^{-1}\ma^{\top}\sqrt{\mw^{(k)}}$
and $\left|\sqrt{x}-1\right|\leq|x-1|$ for $\frac{1}{2}\leq x\leq2$
and get
\begin{equation}
E_{k}^{(1)}\leq30\cdot\tr\left[\left(\frac{\mw^{(k+1)}-\mw^{(k)}}{\mw^{(k)}}\right)^{2}\mproj^{(k)}\mproj^{(k)}\right]=30\cdot\sum_{i\in[n]}\sigma_{i}^{(k)}\left(\frac{w_{i}^{(k+1)}-w_{i}^{(k)}}{w_{i}^{(k)}}\right)^{2}.\label{eq:E1}
\end{equation}
For the second term, we note that
\begin{align*}
\Psi^{(k+1)}-\Psi^{(k)} & =(\ma^{\top}\mv^{(k+1)}\ma)^{-1}-(\ma^{\top}\mv^{(k)}\ma)^{-1}\\
 & =\int_{0}^{1}\mh_{s}^{-1}\ma^{\top}(\mv^{(k+1)}-\mv^{(k)})\ma\mh_{s}^{-1}ds
\end{align*}
where $\mh_{s}=\ma^{\top}\left((1-s)\mv^{(k)}+s\mv^{(k+1)}\right)\ma\in\R^{d\times d}$.
Let $\md_{v}^{(k)}=\frac{\mv^{(k+1)}-\mv^{(k)}}{\mw^{(k)}}\in\R^{n\times n}$,
we have
\begin{align*}
E_{k}^{(2)} & =4\left\Vert \sqrt{\mw^{(k)}}\ma\int_{0}^{1}\mh_{s}^{-1}\ma^{\top}(\mv^{(k+1)}-\mv^{(k)})\ma\mh_{s}^{-1}ds\ma^{\top}\sqrt{\mw^{(k)}}\right\Vert _{F}^{2}\\
 & \leq4\left(\int_{0}^{1}\left\Vert \sqrt{\mw^{(k)}}\ma\mh_{s}^{-1}\ma^{\top}(\mv^{(k+1)}-\mv^{(k)})\ma\mh_{s}^{-1}\ma^{\top}\sqrt{\mw^{(k)}}\right\Vert _{F}ds\right)^{2}\\
 & =4\left(\int_{0}^{1}\left\Vert \sqrt{\mw^{(k)}}\ma\mh_{s}^{-1}\ma^{\top}\sqrt{\mw^{(k)}}\left(\frac{\mv^{(k+1)}-\mv^{(k)}}{\mw^{(k)}}\right)\sqrt{\mw^{(k)}}\ma\mh_{s}^{-1}\ma^{\top}\sqrt{\mw^{(k)}}\right\Vert _{F}ds\right)^{2}\\
 & \leq8\normFull{\mproj^{(k)}\left(\frac{\mv^{(k+1)}-\mv^{(k)}}{\mw^{(k)}}\right)\mproj^{(k)}}_{F}^{2}
\end{align*}
where we used $\mh_{s}^{-1}\preceq\frac{9}{8}(\ma^{\top}\mw^{(k)}\ma)^{-1}$
at the end. Note that

\begin{align*}
\mathbb{E}[E^{(2)}] & =8\E\left[\tr\left[\mproj^{(k)}\left(\frac{\mv^{(k+1)}-\mv^{(k)}}{\mw^{(k)}}\right)\mproj^{(k)}\left(\frac{\mv^{(k+1)}-\mv^{(k)}}{\mw^{(k)}}\right)\right]\right]\\
 & =8\E\left[\sum_{i\in[n]}\sum_{j\in[n]}(\mproj^{(k)})_{i,j}^{2}\cdot\frac{v_{i}^{(k+1)}-v_{i}^{(k)}}{w_{i}^{(k)}}\cdot\frac{v_{j}^{(k+1)}-v_{j}^{(k)}}{w_{j}^{(k)}}\right]
\end{align*}
and for each $i$, $\P(v_{i}^{(k)}\neq0)=\min\{1,\gamma\epsilon^{-2}\cdot\ttau_{i}^{(\mathrm{alg},k)}\}$.
Also, this is independent to whether $v_{j}^{(k)}\neq0$. Using the
formula of $v$, we have
\begin{align*}
\E\left[\frac{v_{i}^{(k+1)}-v_{i}^{(k)}}{w_{i}^{(k)}}\cdot\frac{v_{j}^{(k+1)}-v_{j}^{(k)}}{w_{j}^{(k)}}\right] & \leq\begin{cases}
\frac{2}{\min\{1,\gamma\epsilon^{-2}\cdot\ttau_{i}^{(\mathrm{alg},k)}\}}\cdot\frac{(w_{i}^{(\mathrm{alg},k)})^{2}}{(w_{i}^{(k)})^{2}} & \text{if }i=j\\
4\cdot\frac{w_{i}^{(\mathrm{alg},k)}}{w_{i}^{(k)}}\cdot\frac{w_{j}^{(\mathrm{alg},k)}}{w_{j}^{(k)}} & \text{if }i\neq j
\end{cases}\\
 & \leq\begin{cases}
\frac{16}{\gamma\epsilon^{-2}\cdot\tau_{i}^{(k)}} & \text{if }i=j\\
16 & \text{if }i\neq j
\end{cases}
\end{align*}
where we used that $w_{i}^{(\mathrm{alg},k)}\leq2w_{i}^{(k)}$ and
$\tau_{i}^{(k)}\geq\frac{1}{2}\ttau_{i}^{(\mathrm{alg},k)}$.

Let $\mathcal{I}_{k}=\{i:\ w_{i}^{(\mathrm{alg},k+1)}\neq w_{i}^{(\mathrm{alg},k)}\}$
be the indices that the algorithm resampled at iteration $k$. Then,
we have
\begin{align}
\mathbb{E}[E_{k}^{(2)}] & \leq16\sum_{i\in\mathcal{I}_{k},j\in\mathcal{I}_{k}}(\mproj^{(k)})_{i,j}^{2}+16\sum_{i\in\mathcal{I}_{k}}\frac{(\mproj^{(k)})_{i,i}^{2}}{\gamma\epsilon^{-2}\cdot\sigma_{i}^{(k)}}\nonumber \\
 & \leq16\sum_{i\in\mathcal{I}_{k}}\sum_{j\in[n]}(\mproj^{(k)})_{i,j}^{2}+\frac{16\epsilon^{2}}{\gamma}\sum_{i\in\mathcal{I}_{k}}\mproj_{i,i}^{(k)}\nonumber \\
 & =16\sum_{i\in\mathcal{I}_{k}}\sigma_{i}^{(k)}+\frac{16\epsilon^{2}}{\gamma}\sum_{i\in\mathcal{I}_{k}}\sigma_{i}^{(k)}\leq32\sum_{i\in\mathcal{I}_{k}}\sigma_{i}^{(k)}\label{eq:E2a}
\end{align}
where we used the fact $\sum_{j\in[n]}(\mproj^{(k)})_{i,j}^{2}=\sigma_{i}^{(k)}$
for all $i$. Finally, we note that by the choice of the indices to
update, we have
\begin{equation}
\sum_{i\in\mathcal{I}_{k}}\sigma_{i}^{(k)}\leq2^{\ell_{k}}+1\label{eq:E2b}
\end{equation}
where $\ell_{k}$ is defined in the $k$-th step of $\textsc{Update}$.
Hence, we have $\mathbb{E}[E_{k}^{(2)}]\leq64\cdot2^{\ell_{k}}.$
Using this, (\ref{eq:E_E1E2}) and (\ref{eq:E1}), we have
\begin{align*}
 & \E\left[\sum_{k=0}^{K-1}\Big\|\sqrt{\mw^{(k+1)}}\ma\Psi^{(k+1)}\ma^{\top}\sqrt{\mw^{(k+1)}}-\sqrt{\mw^{(k)}}\ma\Psi^{(k)}\ma^{\top}\sqrt{\mw^{(k)}}\Big\|_{F}\right]\\
\leq & \E\left[\sum_{k=0}^{K-1}\sqrt{E_{k}^{(1)}}+\sqrt{E_{k}^{(2)}}\right]\\
\leq & 6\sum_{k=0}^{K-1}\left(\sum_{i=1}^{n}\sigma_{i}^{(k)}\left(\frac{w_{i}^{(k+1)}-w_{i}^{(k)}}{w_{i}^{(k)}}\right)^{2}\right)^{1/2}+8\sum_{k=0}^{K-1}2^{\ell_{k}/2}\\
\leq & 6\sum_{k=0}^{K-1}\left(\sum_{i=1}^{n}\sigma_{i}^{(k)}\left(\frac{w_{i}^{(k+1)}-w_{i}^{(k)}}{w_{i}^{(k)}}\right)^{2}\right)^{1/2}+8\sum_{k=0}^{K-1}2^{\ell_{k}/2}\\
\leq & 6K+8K\log^{2}n\leq16K\log^{2}n
\end{align*}
where we used the fact that a rank $2^{\ell}$ update happens at most
once every $\frac{2^{\ell/2}}{\log n}$ steps (proved in Theorem \ref{thm:inverse_main}).
\end{proof}

\section{Open Problems}

Our main result is a linear program solver that runs in expected $\otilde(nd+d^{3})$
time. For dense constraint matrices $\ma\in\R^{n\times d}$ this is
optimal among algorithms that do not use fast matrix multiplication,
barring a major improvement in solving linear systems. The fastest
linear system solvers run in $\otilde(\nnz(A)+d^{\omega})$ time,
where $\nnz(A)$ is the number of nonzero entries in $\ma$ and $\omega$
is the matrix exponent. This leads to two open questions: (i) Can
the $nd$ term in our complexity be improved to $\nnz(A)$? (ii) Can
the $d^{3}$ term be improved to $d^{\omega}$ by exploiting fast
matrix multiplication?

A major bottleneck for question (i) is how to detect large entries
of the product $\ma h$ for $\ma\in\R^{n\times d}$, $h\in\R^{d}$
(i.e. solving the problem of Lemma~\ref{lem:large_entry_datastructure}).
Currently the complexity of Lemma~\ref{lem:large_entry_datastructure}
is $\text{\ensuremath{\otilde}}(\|\mg\ma h\|^{2}\cdot\varepsilon^{-2}\cdot d)$,
which can be interpreted as ``$d$ times the number of entries larger
than $\varepsilon$''. Note that verifying the answer (that is, for
a list of indices $I\subset[n]$, check that $(\ma h)_{i}$ is indeed
larger than $\varepsilon$) requires the same complexity for dense
matrices $\ma$. However, for matrices with $z\le d$ entries per
row, the verification complexity is just $z\cdot|I|$. This suggests
that it might be possible to improve Lemma~\ref{lem:large_entry_datastructure}
to run in $\text{\ensuremath{\otilde}}(\|\mg\ma h\|^{2}\cdot\varepsilon^{-2}\cdot z)$
for matrices $\ma$ with $z$ nonzero entries per row. If such a data
structure could be found, it would result in a faster linear programming
algorithm.

So far question (ii) has only been answered for the case $d=\Omega(n)$
\cite{cohen2019solving,br2019deterministicArxiv} and no progress
has been made for the more general case $d=o(n)$, even when allowing
for a worse dependency on $n$ than our $\otilde(nd+d^{3})$ bound.
So far the best dependency on $d$ is $d^{2.5}\gg d^{\omega}$\cite{leeS14,lee2015efficient},
so a first step could be to improve our complexity to $\otilde(nd+d^{2.5})$
time. Currently the $d^{3}$ bottleneck comes from maintaining the
feasibility of the primal solution $x$ in Section~\ref{sec:maintaining_infeasibility},
so a key step would be to improve the complexity of maintaining the
feasibility.

\section{Acknowledgements}

We thank Sébastien Bubeck, Ofer Dekel, Jerry Li, Ilya Razenshteyn,
and Microsoft Research for facilitating conversations between and
hosting researchers involved in this collaboration. We thank Vasileios
Nakos, Jelani Nelson, and Zhengyu Wang for very helpful discussions
about sparse recovery literature. We thank Jonathan Kelner, Richard
Peng, and Sam Chiu-wai Wong for helpful conversations. This project
has received funding from the European Research Council (ERC) under
the European Unions Horizon 2020 research and innovation programme
under grant agreement No 715672. This project was supported in part
by NSF awards CCF-1749609, CCF-1740551, DMS-1839116, CCF-1844855,
and Microsoft Research Faculty Fellowship. This project was supported
in part by Special Year on Optimization, Statistics, and Theoretical
Machine Learning (being led by Sanjeev Arora) at Institute for Advanced
Study.

\pagebreak{}

\bibliographystyle{alpha}
\bibliography{main}

\appendix

\section{Misc Technical Lemmas}

Here we state various technical lemmas from prior work that we use
throughout our paper.
\begin{lem}[Projection Matrix Facts (Lemma 47 of \cite{lsJournal19})]
\label{lem:tool:projection_matrices} Let $\mproj\in\R^{n\times n}$
be an arbitrary orthogonal projection matrix and let $\mSigma=\mdiag(\mproj)$
and $\mproj^{(2)}=\mproj\circ\mproj$. For all $i,j\in[m]$, $x,y\in\R^{m}$,
and $\mx=\mdiag(x)$ we have
\[
\begin{array}{lcl}
(1)\,\,\mSigma_{ii}=\sum_{j\in[m]}\mproj_{ij}^{(2)} & \enspace & (5)\,\,\norm{\mSigma^{-1}\mproj^{(2)}x}_{\infty}\leq\norm x_{\infty}\\
(2)\,\,\mzero\preceq\mproj^{(2)}\preceq\mSigma\preceq\mi,(\text{in particular,}0\leq\mSigma_{ii}\leq1) & \enspace & (6)\,\,\sum_{i\in[n]}\mSigma_{ii}=\rank(\mproj)\\
(3)\,\,\mproj_{ij}^{(2)}\leq\mSigma_{ii}\mSigma_{jj} & \enspace & (7)\,\,\left|y^{\top}\mx\mproj^{(2)}y\right|\leq\norm y_{\mSigma}^{2}\cdot\norm x_{\mSigma}\\
(4)\,\,\norm{\mSigma^{-1}\mproj^{(2)}x}_{\infty}\leq\norm x_{\mSigma} & \enspace & (8)\,\,\left|y^{\top}\left(\mproj\circ\mproj\mx\mproj\right)y\right|\leq\norm y_{\mSigma}^{2}\cdot\norm x_{\mSigma}\,.
\end{array}
\]
\end{lem}

\begin{lem}[Derivative of Projection Matrix (Lemma 49 of \cite{lsJournal19})]
\label{lem:deriv:proj} Given full rank $\ma\in\Rnd$ and $w\in\R_{>0}^{n}$
we have
\[
D_{w}\mproj(\mw\ma)[h]=\mDelta\mproj(\mw\ma)+\mproj(\mw\ma)\mDelta-2\mproj(\mw\ma)\mDelta\mproj(\mw\ma)
\]
where $\mw=\mdiag(w)$, $\mDelta=\mdiag(h/w)$, and $D_{w}f(w)[h]$
denote the directional derivative of $f$ with respect to $w$ in
direction $h$. In particular, we have that 
\[
D_{w}\sigma(\mw)[h]=2\mLambda(\mw)\mw^{-1}h\,.
\]
\end{lem}

The next lemma gives a variety of frequently used relationships between
different types of multiplicative approximations.
\begin{lem}
\label{lem:mult_approx} Let $a,b\in\R_{>0}^{n}$ and $\ma\defeq\mdiag(a)$
and $\mb\defeq\mdiag(b)$. If $a\approx_{\epsilon}b$ for $\epsilon\in(0,1/2)$
then
\[
\norm{\ma^{-1}(a-b)}_{\infty}\leq\epsilon+\epsilon^{2}\text{ and }\norm{\mb^{-1}(a-b)}_{\infty}\leq\epsilon+\epsilon^{2}\,.
\]
Further, if for $\epsilon\in(0,1/2)$ and either
\[
\norm{\ma^{-1}(a-b)}_{\infty}\leq\epsilon\text{ or }\norm{\mb^{-1}(a-b)}_{\infty}\leq\epsilon
\]
then $a\approx_{\epsilon+\epsilon^{2}}b$.
\end{lem}

\begin{proof}
These follow near immediately by Taylor expansion.

\end{proof}
Here we provide some helpful expectation bounds. This lemma just
shows that the relative change of a matrix in Frobenius norm can be
bounded by the change in the projection Schur product norm.
\begin{lem}[Matrix Stability]
\label{lem:matrix_stability} For full column rank $\ma\in\R^{n\times d}$
and all $\delta\in\R^{n}$ we have
\[
\norm{\left(\ma^{\top}\ma\right)^{-1/2}\ma^{\top}\mDelta\ma\left(\ma^{\top}\ma\right)^{-1/2}}_{F}=\norm{\delta}_{\mproj^{(2)}}
\]
where $\mDelta\defeq\mdiag(\delta)$ and $\mproj\defeq\mproj(\ma)$.
\end{lem}

\begin{proof}
Direct calculation yields that 
\begin{align*}
\normFull{\left(\ma^{\top}\ma\right)^{-1/2}\ma^{\top}\mDelta\ma\left(\ma^{\top}\ma\right)^{-1/2}}_{F}^{2} & =\tr\left[\left(\ma^{\top}\ma\right)^{-1/2}\ma^{\top}\mDelta\mproj\mDelta\ma^{\top}\left(\ma^{\top}\ma\right)^{-1/2}\right]\\
 & =\tr\left[\mproj\mDelta\mproj\mDelta\right]\\
 & =\sum_{i\in[n]}\sum_{j\in[n]}\left[\mproj\mDelta\right]_{ij}\left[\mproj\mDelta\right]_{ji}\\
 & =\sum_{i,j\in[n]}\mproj_{ij}^{2}\delta_{i}\delta_{j}\,
\end{align*}
\end{proof}
The following lemma (coupled with the one above) shows that we can
leverage score sample to have small change in Frobenius norm.
\begin{lem}[Expected Frobenius]
\label{lem:matrix_stability_expectation} Let $\ma\in\R^{n\times d}$
be non-degenerate and $\delta\in\R^{n}$. Further let $\tau\geq\sigma(\ma)$
and let $\tilde{\delta}\in\R^{n}$ be chosen randomly where for all
$i\in[n]$ independently $\tilde{\delta}_{i}=\delta_{i}/p_{i}$ with
probability $p_{i}$ and $\tilde{\delta}_{i}=0$ otherwise where $p_{i}=\min\{1,\tau_{i}\cdot k\}$
for some $k\geq0$. Then
\[
\E[\tilde{\delta}]=\delta\text{ and }\E[\norm{\tilde{\delta}}_{\mproj^{(2)}(\ma)}^{2}]\leq\left(1+(1/k)\right)\cdot\norm{\delta}_{\sigma}^{2}.
\]
\end{lem}

\begin{proof}
Clearly $\E[\tilde{\delta}]=\delta$ by construction. Further, since
the different entries of $\tilde{\delta}_{i}$ are set independently
we have that $\E[\tilde{\delta}_{i}\tilde{\delta}_{j}]=\delta_{i}\delta_{j}$
when $i\neq j$. Consequently, letting $\mproj$ therefore
\begin{align*}
\E[\norm{\tilde{\delta}}_{\mproj^{(2)}(\ma)}^{2}] & =\E\left[\sum_{i,j\in[n]}\mproj_{i,j}^{2}\delta_{i}\delta_{j}\right]\\
 & =\sum_{i\in[n]}\mproj_{i,i}^{2}\E[\tilde{\delta}_{i}^{2}]+\sum_{i,j\in[n]:i\neq j}\mproj_{i,j}^{2}\E[\delta_{i}\delta_{j}]\\
 & =\sum_{i\in[n]}\sigma_{i}^{2}\cdot(\E[\tilde{\delta}_{i}^{2}]-(\E[\tilde{\delta}_{i}])^{2})+\norm{\delta}_{\mproj^{(2)}}^{2}\,.
\end{align*}
Now, either $p_{i}=1$ and
\[
\sigma_{i}^{2}\cdot(\E[\tilde{\delta}_{i}^{2}]-(\E[\tilde{\delta}_{i}])^{2})=\sigma_{i}\cdot(\delta_{i}^{2}-\delta_{i}^{2})=0
\]
or $p_{i}<1$ and $p_{i}=k\cdot\tau_{i}\geq k\cdot\sigma_{i}$ in
which case
\[
\sigma_{i}^{2}\cdot(\E[\tilde{\delta}_{i}^{2}]-(\E[\tilde{\delta}_{i}])^{2})\leq\frac{\sigma_{i}^{2}\delta_{i}^{2}}{p_{i}}\leq\frac{\sigma_{i}}{k}\cdot\delta_{i}^{2}\,.
\]
Combining these facts and leveraging that $\mproj^{(2)}\preceq\mSigma$
(Lemma~\ref{lem:tool:projection_matrices}) then yields the result.
\end{proof}

\section{Maintaining Near Feasibility}

\label{sec:maintaining_infeasibility}

In this section we prove the remainder of Theorem~\ref{thm:path_following}.
Note that when assuming the interior point $x\in\R^{n}$ stays feasible,
we have proven Theorem~\ref{thm:path_following} by Theorem~\ref{thm:path_following_simplified}.
Thus we are only left with analyzing how to make sure that $x$ stays
approximately feasible. This is required because the step $\delta_{x}$
performed by our IPM does not necessarily satisfy $\ma^{\top}\delta_{x}=0$,
due to the spectral approximation we used for $(\ma^{\top}\omx\oms^{-1}\ma)^{-1}$
in our steps.

Given some infeasible $x$ we can obtain a feasible $x'$ via
\begin{align*}
x'\defeq & x-\mx\ms^{-1}\ma(\ma^{\top}\mx\ms^{-1}\ma)^{-1}(\ma^{\top}x-b),
\end{align*}
We show in Lemma \ref{lem:move_x_correction} that these $x$ and
$x'$ are close multiplicatively
\begin{align*}
\|\mx^{-1}(x'-x)\|_{\tau+\infty}^{2}\approx & \mu^{-1}\norm{x-x'}_{\left(\ma^{\top}\mx\ms^{-1}\ma\right)^{-1}}^{2}
\end{align*}
which motivates that we want to minimize the term on the right. Hence
we measure the infeasibility of a point $x\in\R_{\geq0}^{n}$ by the
potential $\Phi_{b}:\R_{\geq0}^{n}\times\R_{\geq0}^{n}\times\R_{\geq0}^{n}\times\R_{\geq0}\rightarrow\R$,
where for any $x',s'\in\R_{\geq0}^{n}$
\begin{equation}
\Phi_{b}(x,x',s',\mu)=\mu^{-1}\|\ma^{\top}x-b\|_{(\ma^{\top}\mx'\ms'{}^{-1}\ma)^{-1}}^{2}.\label{eq:Phi_b_x_barx_bars_mu}
\end{equation}
Here we prove that throughout the algorithm, we maintain $\Phi_{b}\leq\frac{\zeta\epsilon^{2}}{\log^{\power}n}$,
where $\epsilon$ is the parameter of Theorem~\ref{thm:path_following}
and $\zeta>0$ is a sufficiently small constant.

The extra parameters $x',s'\in\R_{\geq0}^{n}$ of our potential function
$\Phi_{b}$ are motivated by the fact, that we analyze this norm in
several steps. Our proof is split into three parts. First, we analyze
how much the potential $\Phi_{b}(x+\delta_{x},x',s',\mu)$ increases
by performing the IPM step $\delta_{x}$ in Section~\ref{subsec:fes_x}.
Second, we analyze how much the potential increases by moving $x'$,
$s'$ to $x'+\delta_{x}$, $s'+\delta_{s}$ in Section~\ref{subsec:fes_xbar}.
In general, this could incur a constant factor increase of the potential,
so we modify the IPM to also perform a small corrective step $\delta'_{x}$
and analyze $\Phi_{b}(x+\delta_{x}+\delta'_{x},x'+\delta_{x'},s'+\delta_{s'},\mu)$
instead. In Section~\ref{subsec:improve_infeasibility} we show that
it takes $\tilde{\Omega}(\sqrt{d})$ iterations of the modified IPM
until $\Phi_{b}$ increases by a constant factor. After such $\tilde{\Omega}(\sqrt{d})$
iterations, the algorithm then performs another more expensive correction
step, which decreases the potential again. This way the potential
will stay small throughout the entire IPM, which then concludes the
proof of Theorem~\ref{thm:path_following}.

\subsection{Increase of Infeasibility Due to $x$\label{subsec:fes_x}}

First, we show that $\Phi_{b}$ increases slowly when the step $\delta_{x}$
of the IPM is correct in expectation (that is, $\E[\ma\delta_{x}]=0$).
In particular, it takes $\sqrt{d}$ iterations to increase $\Phi_{b}$
by 1.
\begin{lem}[Infeasibility change due to $x$]
\label{lem:fes_change_x} Assume that $\omx\os\approx_{1}\mu\overline{\tau}\in\R^{n}$,
$x'\approx_{1}x\approx_{1}\ox$, $s'\approx_{1}s\approx_{1}\os$ and
$\overline{\tau}\approx_{1}\tau(\ox,\os)\in\R^{n}$. Let $\overline{\mq}=\ma^{\top}\omx\oms^{-1}\ma\in\R^{d\times d}$.
Let $\mh\in\R^{d\times d}$ be a random symmetric matrix satisfying
\[
\mh\approx_{\epsilon_{H}}\overline{\mq}\text{ and }\E[\mh^{-1}]=\overline{\mq}^{-1}\,.
\]
for some $\epsilon_{H}\in(0,1)$. Further, let $x^{\new}=x+\delta_{x}$
where $\delta_{x}=\omx\omw^{-1/2}(\mi-\mproj)\omw^{1/2}h\in\R^{n}$
with $\omw=\omx\oms\in\R^{n\times n}$ and $\mproj=\overline{\ms}^{-1/2}\overline{\mx}^{1/2}\ma\mh^{-1}\ma^{\top}\omx^{1/2}\oms^{-1/2}$.
Then
\[
\E_{\delta_{x}}[\Phi_{b}(x^{\new},x',s',\mu)]\leq\Phi_{b}(x,x',s',\mu)+O(\epsilon_{H}^{2})\cdot\|h\|_{\overline{\tau}}^{2}.
\]
\end{lem}

\begin{proof}
Note that
\begin{align}
\E[\ma^{\top}\delta_{x}] & =\ma^{\top}\omx\omw^{-1/2}(\mi-\overline{\ms}^{-1/2}\overline{\mx}^{1/2}\ma(\ma^{\top}\oms^{-1}\omx\ma)^{-1}\ma^{\top}\omx^{1/2}\oms^{-1/2})\omw^{1/2}h=0\,.\label{eq:ker_in_expect}
\end{align}
For easy of analysis, we define matrix $\mq'\in\R^{d\times d}$ as
$\mq'\defeq\ma^{\top}\mx'\ms'^{-1}\ma$. Hence, we have
\begin{align*}
\E[\Phi_{b}(x^{\new},x',s',\mu)]= & \E\left[\mu^{-1}\|\ma^{\top}(x+\delta_{x})-b\|_{(\ma^{\top}\mx'\ms'^{-1}\ma)^{-1}}^{2}\right]\\
= & \E\left[\mu^{-1}\|\ma^{\top}(x+\delta_{x})-b\|_{(\mq')^{-1}}^{2}\right]\\
= & \mu^{-1}\|\ma^{\top}(x+\E[\delta_{x}])-b\|_{(\mq')^{-1}}^{2}+\E\left[\mu^{-1}\|\ma^{\top}(\delta_{x}-\E[\delta_{x}])\|_{(\mq')^{-1}}^{2}\right]\\
= & \Phi_{b}(x,x',s',\mu)+\E\left[\mu^{-1}\|\ma^{\top}(\delta_{x}-\E[\delta_{x}])\|_{(\mq')^{-1}}^{2}\right].
\end{align*}
where the second step follows from $\E[x^{2}]=(\E[x])^{2}+\E[(x-\E[x])^{2}]$,
and the last step follows from (\ref{eq:Phi_b_x_barx_bars_mu}) (the
definition of $\Phi_{b}$) and (\ref{eq:ker_in_expect}).

Next, we note that
\begin{align*}
\delta_{x}-\E[\delta_{x}] & =\omx\overline{\ms}^{-1}\ma(\mh^{-1}-(\ma^{\top}\omx\oms^{-1}\ma)^{-1})\ma^{\top}\omx h\\
 & =\omx\overline{\ms}^{-1}\ma(\mh^{-1}-(\overline{\mq})^{-1})\ma^{\top}\omx h\,.
\end{align*}
We define matrix $\mm\in\R^{d\times d}$ as follows
\begin{align*}
\mm & \defeq(\ma^{\top}\omx\overline{\ms}^{-1}\ma)^{1/2}\mh^{-1}(\ma^{\top}\omx\overline{\ms}^{-1}\ma)^{1/2}=\overline{\mq}^{1/2}\mh^{-1}\overline{\mq}^{1/2}\,.
\end{align*}
Since $\mh\approx_{\epsilon_{H}}\ma^{\top}\oms^{-1}\omx\ma$, we have
$\|\mm-\mi\|_{2}=O(\epsilon_{H})$. Further $(\mq')^{-1}=(\ma^{\top}\mx'\ms'^{-1}\ma)^{-1}\approx_{O(1)}(\ma^{\top}\omx\oms^{-1}\ma)^{-1}=(\overline{\mq})^{-1}$,
hence, we have
\begin{align*}
\|\ma^{\top}(\delta_{x}-\E[\delta_{x}])\|_{(\mq')^{-1}}^{2} & \leq O(1)\cdot\|\ma^{\top}(\delta_{x}-\E[\delta_{x}])\|_{(\overline{\mq})^{-1}}^{2}\\
 & =O(1)\cdot\|\overline{\mq}(\mh^{-1}-(\overline{\mq})^{-1})\ma^{\top}\omx h\|_{(\overline{\mq})^{-1}}^{2}\\
 & =O(1)\cdot\|(\mm-\mi)(\overline{\mq})^{-1/2}\ma^{\top}\omx h\|_{2}^{2}\\
 & \leq\|\mm-\mi\|_{2}^{2}\cdot\|(\overline{\mq})^{-1/2}\ma^{\top}\omx h\|_{2}^{2}\\
 & \leq O(\epsilon_{H}^{2})\cdot\|(\overline{\mq})^{-1/2}\ma^{\top}\omx h\|_{2}^{2}\\
 & \leq O(\epsilon_{H}^{2})\cdot\|\ma^{\top}\omx h\|_{(\overline{\mq})^{-1}}^{2}
\end{align*}
where the third step follows from definition of matrix $\mm$, the
forth step follows from $\|\ma u\|_{2}\le\|\ma\|_{2}\cdot\|u\|_{2}$
the fifth step follows from $\|\mm-\mi\|_{2}=O(\epsilon_{H})$. Hence,
we have 
\begin{align*}
\mu^{-1}\cdot\|\ma^{\top}(\delta_{x}-\E[\delta_{x}])\|_{(\mq')^{-1}}^{2} & \leq O(\epsilon_{H}^{2}/\mu)\cdot\|\ma^{\top}\omx h\|_{(\overline{\mq})^{-1}}^{2}\\
 & =O(\epsilon_{H}^{2}/\mu)\cdot\left\Vert \ma^{\top}\sqrt{\oms^{-1}\omx}\sqrt{\oms\omx}h\right\Vert _{(\overline{\mq})^{-1}}^{2}\\
 & =O(\epsilon_{H}^{2}/\mu)\cdot\normFull{\sqrt{\oms\omx}h}_{\mproj(\oms^{-1/2}\omx^{1/2}\ma)}^{2}\\
 & \leq O(\epsilon_{H}^{2}/\mu)\cdot\normFull{\sqrt{\oms\omx}h}_{2}^{2}
\end{align*}
where the last step follows from properties of projection matrices.
As $\omx\os\approx_{1}\mu\overline{\tau}$ and $\overline{\tau}\approx_{1}\tau(\ox,\os)$
the result follows.
\end{proof}

\subsection{Increase of Infeasibility Due to $x'$ and $s'$\label{subsec:fes_xbar}}

When $x'\in\R_{\geq0}^{n}$ or $s'\in\R_{\geq0}^{n}$ changes by a
multiplicative constant, $\Phi_{b}(x,x',s',\mu)$ changes by a multiplicative
constant. As we want that $x',s'$ are close to $x,s$ and $x,s$
change in every step of the IPM, this will blow up $\Phi_{b}$ too
quickly.

Our key idea is to move $x$ to some nearby $x^{\new}$ as we replace
$x'$ and $s'$ by $x$ and $s$, such that (\ref{eq:Phi_b_x_barx_bars_mu})
remains unchanged, i.e. $\Phi_{b}(x^{\new},x,s,\mu)\approx\Phi_{b}(x,x',s',\mu)$.
Algorithm~\ref{alg:maintain_INfeasibility} describes how we compute
this $x^{\new}$ and we show in this subsection (Lemma~\ref{lem:change_xbar_sbar})
that $\Phi_{b}$ does indeed not change by much when replacing $x,x',s'$
by $x^{\new},x,s$.

\begin{algorithm2e}[!t]

\caption{Maintaining Infeasibility}

\label{alg:maintain_INfeasibility}

\SetKwProg{Proc}{procedure}{}{}

\Proc{\textsc{MaintainInfeasibility}$(x\in\R^{n},s\in\R^{n},\overline{\tau}\in\R^{n})$}{

\State Let $x'\in\R^{n},s'\in\R^{n}$ be the previous input to $\textsc{MaintainInfeasibility}$.

\State Let $\mq=\ma^{\top}\mx\ms^{-1}\ma$ and $\mq'=\ma^{\top}\mx'\ms'{}^{-1}\ma$

\State Generate independent random $\mh_{1}\in\R^{d\times d}$ such
that 
\[
\E[\mh_{1}^{-1}]=\mq^{-1}\text{ and }\mh_{1}\approx_{\epsilon_{H}}\mq\,.
\]

\State Generate independent random $\mh_{3}\in\R^{d\times d}$ such
that 
\[
\E[\mh_{3}^{-1}]=\mq'^{-1}\text{ and }\mh_{3}\approx_{\epsilon_{H}}\mq'.
\]

\State Generate independent random $\mh_{2}\in\R^{d\times d}$ such
that \tcp*{Lemma \ref{lem:generate_H2_H4}} 
\begin{align*}
\E[\mh_{2}] & =\ma^{\top}\mx{}^{1/2}\ms{}^{-1/2}(\mx{}^{1/2}\ms{}^{-1/2}-\mx'^{1/2}\ms'^{-1/2})\ma,\\
\E[\|\mq'^{-1/2}\mh_{2}\mq'^{-1/2}\|_{F}^{2}] & =O(\|\ln x'-\ln x\|_{\tau(x,s)}^{2}+\|\ln s'-\ln s\|_{\tau(x,s)}^{2})\text{ and }\\
\mh_{2} & =\E[\mh_{2}]\pm\epsilon_{H}\cdot\mq'\,.
\end{align*}

\State Generate independent random $\mh_{4}\in\R^{d\times d}$ such
that \tcp*{Lemma~\ref{lem:generate_H2_H4}}
\begin{align*}
\E[\mh_{4}] & =\mq^{-1}-\mq'^{-1},\\
\E[\|\mq'^{1/2}\mh_{4}\mq'^{1/2}\|_{F}^{2}] & =O(\|\ln x'-\ln x\|_{\tau(x,s)}^{2}+\|\ln s'-\ln s\|_{\tau(x,s)}^{2})\text{ and }\\
\mh_{4} & =\E[\mh_{4}]\pm\epsilon_{H}\cdot\mq'^{-1}\,.
\end{align*}

\tcp{These matrices can be generated by leverage score sampling with
$\overline{\tau}$.}

\State Generate independent random $\delta_{b}^{(1)}$ and $\delta_{b}^{(2)}$
such that \tcp*{ Lemma~\ref{lem:random_delta_b}} .
\[
\E[\delta_{b}^{(1)}]=b-\ma^{\top}x,\mathrm{and\,}\delta_{b}^{(1)}\approx\E[\delta_{b}^{(1)}]
\]
where the approximation error should be as defined in Lemma~\ref{lem:random_delta_b}.

\State $\delta_{\lambda}\leftarrow\mh_{1}^{-1}\mh_{2}\mh_{3}^{-1}\delta_{b}^{(1)}+\mh_{4}\delta_{b}^{(2)}$
\tcp*{This is computed and returned }

\State $\delta_{x}\leftarrow\mx\ms{}^{-1}\ma\delta_{\lambda}$, $x^{\new}\leftarrow x+\delta_{x}$
\tcp*{Defined for analysis and not computed }

\State \Return $\delta_{\lambda}$

}

\end{algorithm2e}

To compute this new $x^{\new}$, we need to estimate $b-\ma^{\top}x$.
The next lemma shows how to find $\delta_{b}$ such that $\E[\delta_{b}]=b-\ma^{\top}x$
and $\mu^{-1}\|\delta_{b}\|_{(\ma^{\top}\ms'^{-1}\mx'\ma)^{-1}}^{2}$
are small.
\begin{lem}[Generate $\delta_{b}^{(1)}$ and $\delta_{b}^{(2)}$]
\label{lem:random_delta_b} Let $\rho\in(0,1/10)$ denote the failure
probability, let $\epsilon_{b}\in(0,1/10)$ denote an accuracy parameter.
Given $\overline{\tau}\approx_{1}\tau(\ox,\os)\in\R^{n}$ where $x'\approx_{1}x\approx_{1}\overline{x}$
and $\omx\os\approx_{1}\mu\overline{\tau}$. Let $\mq'=\ma^{\top}\mx'\ms'^{-1}\ma\in\R^{d\times d}$.
There is an algorithm that runs in $\widetilde{O}(d^{2}/\epsilon_{b})$
time and outputs vector $\delta_{b}\in\R^{d}$ such that $\E[\delta_{b}]=b-\ma^{\top}x$
. Furthermore, for any fixed vector $v\in\R^{d}$ with $\mu^{1/2}\|v\|_{\mq'}=1$,
we have 
\begin{equation}
\left|v^{\top}(\delta_{b}-\E[\delta_{b}])\right|\leq O(\sqrt{\epsilon_{b}}\log(1/\rho))\label{eq:lem_random_delta_b_v}
\end{equation}
with probability at least $1-\rho$. Furthermore, for any fixed matrix
$\mm\in\R^{d\times d}$, we have
\[
\mu^{-1/2}\|\mm(\mq')^{-\frac{1}{2}}(\delta_{b}-\E[\delta_{b}])\|_{2}=\|\mm\|_{F}\cdot O(\sqrt{\epsilon_{b}}\log(d/\rho))
\]
with probability at least $1-\rho$.
\end{lem}

\begin{proof}
We construct $\delta_{b}\in\R^{d}$ by the following random sampling
scheme:
\[
\delta_{b}=b-\sum_{i\in[n]}a_{i}\cdot\widetilde{x}_{i}
\]
where $a_{i}$ is the $i$-th row of $\ma$ and
\[
\widetilde{x}_{i}=\begin{cases}
\frac{x_{i}}{p_{i}} & \text{with probability }p_{i}=\min\{1,\overline{\tau}_{i}/\epsilon_{b}\}\\
0 & \text{otherwises}
\end{cases}.
\]
Clearly, we have $\E[\delta_{b}]=b-\ma^{\top}x$.

Fix any vector $v\in\R^{d}$ such that $\mu^{1/2}\|v\|_{\mq'}=1$
. If $p_{i}=1$, then $\widetilde{x}_{i}=x_{i}$ and clearly $|v^{\top}a_{i}\cdot(\widetilde{x}_{i}-x_{i})|=0=O(\epsilon_{b}^{2})$.
On the other hand, if $p_{i}<1$ then 
\begin{align*}
\left|v^{\top}a_{i}\cdot\widetilde{x}_{i}\right|^{2} & \leq\frac{1}{p_{i}^{2}}\cdot\left(v^{\top}\ma^{\top}\mx e_{i}\right)^{2}\\
 & =\frac{1}{p_{i}^{2}}\cdot\left(v^{\top}(\mq')^{1/2}(\mq')^{-1/2}\ma^{\top}\mx e_{i}\right)^{2}\\
 & \leq\frac{1}{p_{i}^{2}}\cdot\|v\|_{\mq'}^{2}\cdot e_{i}^{\top}\mx\ma(\mq')^{-1}\ma^{\top}\mx e_{i}\\
 & =\frac{x_{i}^{2}}{p_{i}^{2}}\cdot\frac{s'_{i}}{x'_{i}}\cdot\|v\|_{\mq'}^{2}\cdot e_{i}^{\top}\ms'^{-1/2}\mx'^{1/2}\ma(\mq')^{-1}\ma^{\top}\mx'^{1/2}\ms'^{-1/2}e_{i}\\
 & =\frac{x_{i}^{2}}{p_{i}^{2}}\cdot\frac{s'_{i}}{x'_{i}}\cdot\|v\|_{\mq'}^{2}\cdot\sigma(\ms'^{-1/2}\mx'^{1/2}\ma)_{i},
\end{align*}
where the third step follows from Cauchy-Schwarz.

Using that $x'\approx_{1}x$, $\mx's'\approx_{1}\mu\overline{\tau}$,
and $\overline{\tau}\approx_{1}\tau(x',s')$, we have 
\[
\left|v^{\top}a_{i}\cdot\widetilde{x}_{i}\right|^{2}=O(\overline{\tau}_{i}^{2}/p_{i}^{2})=O(\epsilon_{b}^{2}).
\]
where the last step follows from $p_{i}=\overline{\tau}_{i}/\epsilon_{b}$
(implied by $p_{i}<1$). Further, as in this case, we have

\[
\left|v^{\top}a_{i}\cdot x_{i}\right|^{2}\leq p_{i}^{2}\left|v^{\top}a_{i}\cdot\widetilde{x}_{i}\right|^{2}=O(\epsilon_{b}^{2}).
\]
We see that in either case we have 
\begin{equation}
\left|v^{\top}a_{i}\cdot(\widetilde{x}_{i}-x_{i})\right|^{2}=O(\epsilon_{b}^{2}).\label{eq:vax_sup}
\end{equation}

Next, we note that
\begin{align*}
\sum_{i\in[n]}\E\left[\left|v^{\top}a_{i}\cdot(\widetilde{x}_{i}-x_{i})\right|^{2}\right] & =\sum_{p_{i}<1}p_{i}\cdot\left|v^{\top}a_{i}\cdot(x_{i}/p_{i}-x_{i})\right|^{2}\\
 & \leq\sum_{p_{i}<1}\frac{1}{p_{i}}\cdot\left|v^{\top}a_{i}x_{i}\right|^{2}\\
 & =\sum_{p_{i}<1}\frac{1}{p_{i}}\cdot(v^{\top}\ma^{\top}\mx e_{i})^{2}\\
 & =\epsilon_{b}\sum_{p_{i}<1}\left(v^{\top}\ma^{\top}\mx\frac{e_{i}}{\sqrt{\overline{\tau}_{i}}}\right)^{2}
\end{align*}
where the second step follows from $p_{i}\in[0,1]$, and the last
step follows from $\overline{\tau}_{i}=p_{i}\epsilon_{b}$ as we only
consider $p_{i}<1$. Using $x'\approx_{1}x$ and $\mx's'\approx_{1}\mu\overline{\tau}$
this implies
\begin{align}
\sum_{i\in[n]}\E\left[\left|v^{\top}a_{i}\cdot(\widetilde{x}_{i}-x_{i})\right|^{2}\right] & =O(\epsilon_{b})\sum_{p_{i}<1}(v^{\top}\ma^{\top}\sqrt{\mu\mx'\ms'^{-1}}e_{i})^{2}\nonumber \\
 & =O(\epsilon_{b})\cdot\mu\cdot v^{\top}\ma^{\top}\mx'\ms'^{-1}\ma v=O(\epsilon_{b}).\label{eq:vax_var}
\end{align}
Using (\ref{eq:vax_sup}) and (\ref{eq:vax_var}), Bernstein inequality
shows
\[
\left|v^{\top}(\delta_{b}-\E[\delta_{b}])\right|=\left|\sum_{i\in[n]}v^{\top}a_{i}\cdot(\widetilde{x}_{i}-x_{i})\right|=O(\sqrt{\epsilon_{b}\log(1/\rho)}+\epsilon_{b}\log(1/\rho))=O(\sqrt{\epsilon_{b}}\log(1/\rho))
\]
with probability at least $1-\rho$.

For the last conclusion, we consider the SVD of $\mm\in\R^{d\times d}$
and let $v_{i}\in\R^{d}$ such that $v_{i}^{\top}\mm=\lambda_{i}u_{i}^{\top}\in\R^{d}$
for some orthonormal $v_{i}$ and $u_{i}\in\R^{d}$. Note that $\|(\mq')^{-\frac{1}{2}}u_{i}\|_{\mq'}=1$
and (\ref{eq:lem_random_delta_b_v}) shows
\begin{equation}
\left|\mu^{-1/2}\cdot u_{i}^{\top}(\mq')^{-\frac{1}{2}}(\delta_{b}-\E[\delta_{b}])\right|=O(\sqrt{\epsilon_{b}}\log(d/\rho))\label{eq:upper_bound_on_delta_b_under_AXSA_norm}
\end{equation}
for all $i$. Hence, we have
\begin{align*}
 & \|\mu^{-1/2}\cdot\mm(\mq')^{-\frac{1}{2}}(\delta_{b}-\E[\delta_{b}])\|_{2}^{2}\\
= & \sum_{i\in[d]}(\mu^{-1/2}\cdot v_{i}^{\top}\mm(\mq')^{-\frac{1}{2}}(\delta_{b}-\E[\delta_{b}]))^{2}\\
= & \sum_{i\in[d]}\lambda_{i}^{2}\cdot(\mu^{-1/2}\cdot u_{i}^{\top}(\mq')^{-\frac{1}{2}}(\delta_{b}-\E[\delta_{b}]))^{2}\\
= & \sum_{i\in[d]}^{d}\lambda_{i}^{2}\cdot O(\epsilon_{b}\log^{2}(d/\rho))=\|\mm\|_{F}^{2}\cdot O(\epsilon_{b}\log^{2}(d/\rho))
\end{align*}
where the third step follows from (\ref{eq:upper_bound_on_delta_b_under_AXSA_norm}),
and the last step follows from $\sum_{i\in[d]}\lambda_{i}^{2}=\|\mm\|_{F}^{2}$
\end{proof}
Now, we show that $\Phi_{b}$ changes slowly if we move $x$ to some
nearby $x^{\new}$ when replacing $x',s'$ by $x,s$, according to
the method $\textsc{MaintainInfeasibility}$ of Algorithm \ref{alg:maintain_Feasibility}.
\begin{lem}[Generate $\mh_{2}$ and $\mh_{4}$]
\label{lem:generate_H2_H4} For $x\approx_{0.1}x'$ and $s\approx_{0.1}s'$,
we can generate $\mh_{2}$ and $\mh_{4}$ with properties as defined
in Algorithm~\ref{alg:maintain_INfeasibility} in such a way that
\begin{align*}
\E[\|\mq'^{-1/2}\mh_{2}\mq'^{-1/2}\|_{F}^{2}]\le & O(\|\ln x'-\ln x\|_{\tau(x,s)}^{2}+\|\ln s'-\ln s\|_{\tau(x,s)}^{2})\\
\E[\|\mq'^{1/2}\mh_{4}\mq'^{1/2}\|_{F}^{2}]\le & O(\|\ln x'-\ln x\|_{\tau(x,s)}^{2}+\|\ln s'-\ln s\|_{\tau(x,s)}^{2}).
\end{align*}
With high probability, the time for computing $\mh_{2}b$ for any
$b\in\R^{d}$ is $O(d^{2}\epsilon_{H}^{-2}\log n)$ and the time for
computing $\mh_{4}b$ is $O(T\log n+d^{2}\epsilon_{H}^{-2}\log^{2}n)$
where $T$ is the time to compute $\mm b$ for $\mm$ satisfying $\E\mm=\mq'^{-1}$
and $\mm\approx_{\epsilon_{h}/4}\mq'$.
\end{lem}

\begin{proof}
The statement for $\|\mq'^{-1/2}\mh_{2}\mq'^{-1/2}\|_{F}^{2}$ follows
directly from Lemma~\ref{lem:matrix_stability} and Lemma~\ref{lem:matrix_stability_expectation}
when we perform the same leverage score sampling on both $\ma^{\top}\mx\ms{}^{-1}\ma$
and $\ma^{\top}(\mx\mx'){}^{\frac{1}{2}}(\ms\ms'){}^{-\frac{1}{2}}\ma$
(i.e. both matrices sample the same entries of their diagonal).

For the other statement, we let $\md=\mx'\ms'^{-1}$ and $\mDelta=\mx\ms^{-1}-\mx'\ms'^{-1}$.
By the assumption, we have $\max_{i}|\md_{ii}^{-1}\mDelta_{ii}|\leq\frac{1}{4}$.
By Taylor expansion, we have 
\begin{align*}
(\ma^{\top}(\md+\Delta)\ma)^{-1}-(\ma^{\top}\md\ma)^{-1} & =(\ma^{\top}\md\ma)^{-1}\sum_{k\ge1}(-(\ma^{\top}\Delta\ma)(\ma^{\top}\md\ma)^{-1})^{k}.
\end{align*}
Similar to Lemma~\ref{lem:preconditioner}, we can use this to generate
a spectral approximation with good expectation by truncating this
series by some random $X$ with $\P[X\ge k]=2^{-k}$ and scaling the
terms with $2^{k}$ so that the expectation stays the same. More precisely,
we let $\mm^{(k)}$ be independent solvers satisfying $\mm^{(k)}\approx_{\epsilon_{H}/4}(\ma^{\top}\md\ma)^{-1}$
with $\E[\mm^{(k)}]=(\ma^{\top}\md\ma)^{-1}$, we let $\mn^{(k)}$
are i.i.d random leverage score samplings $\ma^{\top}\Delta^{(k)}\ma$
of $\ma^{\top}\Delta\ma$ such that 
\[
\|(\ma^{\top}\md\ma)^{\frac{1}{2}}\mn^{(i)}(\ma^{\top}\md\ma)^{\frac{1}{2}}\|_{2}\leq e^{\epsilon_{H}/4}\max_{i}|\md_{ii}^{-1}\mDelta_{ii}|
\]
 and $\E[\mn^{(k)}]=\ma^{\top}\Delta\ma$. Now, we define
\[
\mh_{4}\defeq\mm^{(0)}\sum_{k=1}^{X}(-2)^{k}\prod_{i=1}^{k}(\mn^{(i)}\mm^{(i)}).
\]
Since $\mm^{(i)}$, $\mn^{(i)}$ are independent, we have 
\begin{align*}
\E_{X,\mm,\mn}[\mh_{4}] & =(\ma^{\top}\md\ma)^{-1}\E_{X}\left[\sum_{k=1}^{X}2^{k}(-(\ma^{\top}\Delta\ma)(\ma^{\top}\md\ma)^{-1})^{k}\right]\\
 & =(\ma^{\top}\md\ma)^{-1}\sum_{k=1}^{\infty}\P[X\ge k]2^{k}(-(\ma^{\top}\Delta\ma)(\ma^{\top}\md\ma)^{-1})^{k}\\
 & =(\ma^{\top}\md\ma)^{-1}\sum_{k=1}^{\infty}(-(\ma^{\top}\Delta\ma)(\ma^{\top}\md\ma)^{-1})^{k}\\
 & =(\ma^{\top}(\md+\Delta)\ma)^{-1}-(\ma^{\top}\md\ma)^{-1}.
\end{align*}
To bound the Frobenius norm, we let $\my\defeq(\ma^{\top}\md\ma)^{\frac{1}{2}}$
and note that
\[
\|\my^{-1}\mh_{4}\my^{-1}\|_{F}^{2}\leq\left(\sum_{k=1}^{X}(\frac{2}{3})^{k}\right)\cdot\sum_{k=1}^{X}(\frac{3}{2})^{k}\|\my^{-1}\mm^{(0)}(-2)^{k}\prod_{i=1}^{k}(\mn^{(i)}\mm^{(i)})\my^{-1}\|_{F}^{2}
\]
where we used that $\|\sum_{i=1}^{X}\ma_{i}\|_{F}^{2}\leq(\sum_{i=1}^{X}\|\ma_{i}\|_{F})^{2}\leq(\sum_{k=1}^{X}(\frac{2}{3})^{k})\cdot(\sum_{i=1}^{X}(\frac{3}{2})^{k}\|\ma_{i}\|_{F}^{2})$.

Taking expectation, we have
\begin{align*}
\E_{X,\mm,\mn}\left[\|\my^{-1}\mh_{4}\my^{-1}\|_{F}^{2}\right] & \leq\left(\sum_{k=1}^{\infty}(\frac{2}{3})^{k}\right)\cdot\sum_{k=1}^{\infty}3^{k}\E_{\mm,\mn}\left[\|\my^{-1}\mm^{(0)}\prod_{i=1}^{k}(\mn^{(i)}\mm^{(i)})\my^{-1}\|_{F}^{2}\right]\\
 & \leq3\sum_{k=1}^{\infty}3^{k}\E_{\mm,\mn}\left[\|\my^{-1}\mm^{(0)}\my^{-1}\prod_{i=1}^{k}(\my\mn^{(i)}\my\my^{-1}\mm^{(i)}\my^{-1})\|_{F}^{2}\right]
\end{align*}
Using that $\|\my^{-1}\mm^{(i)}\my^{-1}\|_{2}\leq e^{\epsilon_{H}/4}$
and $\|\my\mn^{(i)}\my\|_{2}\leq e^{\epsilon_{H}/4}\max_{i}|\md_{ii}^{-1}\mDelta_{ii}|\leq e^{\epsilon_{H}/4}$
as $x\approx_{0.1}x'$ and $s\approx_{0.1}s'$. Using $\epsilon_{H}\leq\frac{1}{10}$,
we have that $\|\my^{-1}\mm^{(i)}\my^{-1}\|_{2}\leq e^{1/40}$ and
$\|\my\mn^{(i)}\my\|_{2}\leq e^{1/40}\cdot\frac{1}{4}$. Hence, we
have
\begin{align*}
\E_{X,\mm,\mn}\left[\|\my^{-1}\mh_{4}\my^{-1}\|_{F}^{2}\right] & \leq3\sum_{k=1}^{\infty}3^{k}e^{1/40(2k-2)}\frac{1}{4^{k-1}}\E\left[\|\my\mn^{(1)}\my\|_{F}^{2}\right]\\
 & =O(1)\cdot\E\left[\|\my\mn^{(1)}\my\|_{F}^{2}\right]
\end{align*}
Finally, Lemma~\ref{lem:matrix_stability} and Lemma \ref{lem:matrix_stability_expectation}
shows that
\begin{align*}
\E\left[\|\my\mn^{(1)}\my\|_{F}^{2}\right] & =\E\left[\|(\ma^{\top}\md\ma)^{-\frac{1}{2}}\mn^{(1)}(\ma^{\top}\md\ma)^{-\frac{1}{2}}\|_{F}^{2}\right]\\
 & =O(1)\cdot\|\diag(\md^{-\frac{1}{2}}\Delta\md^{-\frac{1}{2}})\|_{\sigma(\md^{1/2}\ma)}^{2}\\
 & =O(1)\cdot O(\|\ln x'-\ln x\|_{\tau(x,s)}^{2}+\|\ln s'-\ln s\|_{\tau(x,s)}^{2})
\end{align*}
\end{proof}
\begin{lem}[Infeasibility Change Due to Norm Change]
\label{lem:change_xbar_sbar} Assume that $\mx s\approx_{1}\mu\tau$,
$x'\approx_{0.1}x$, $s'\approx_{0.1}s\in\R^{n}$. Consider $x^{\new}=x+\delta_{x}\in\R^{n}$
defined in Algorithm \ref{alg:maintain_INfeasibility}. Then, we have
that
\begin{align*}
\E_{\delta_{x}}[\Phi_{b}(x^{\new},x,s,\mu)]= & (1+O(\epsilon_{H}^{2}))\cdot\Phi_{b}(x,x',s',\mu)\\
 & +O(\epsilon_{b}\log^{2}d)\cdot(\|\ln x'-\ln x\|_{\tau(x,s)}^{2}+\|\ln s'-\ln s\|_{\tau(x,s)}^{2}).
\end{align*}
Furthermore, 
\begin{align*}
\|\mx^{-1}(x^{\new}-x)\|_{\overline{\tau}}^{2}\le & O(\epsilon_{b}\log^{2}(d/\rho)+\Phi_{b}(x,x',s',\mu))\cdot(\|\ln x'-\ln x\|_{\tau(x,s)}^{2}+\|\ln s'-\ln s\|_{\tau(x,s)}^{2})\\
 & +O(\epsilon_{b}\epsilon_{H}^{2}d\log^{2}(d/\rho))\\
\|\mx^{-1}(x^{\new}-x)\|_{\infty}^{2} & \le O(\epsilon_{b}\log^{2}(d/\rho)+\Phi_{b}(x,x',s',\mu))\cdot(\|\ln x'-\ln x\|_{\tau(x,s)}^{2}+\|\ln s'-\ln s\|_{\tau(x,s)}^{2})\\
 & +O(\epsilon_{b}\epsilon_{H}^{2}d\log^{2}(d/\rho))
\end{align*}
with probability at least $1-\rho$.
\end{lem}

\begin{proof}
We define the following three matrices
\begin{align*}
\mDelta_{(\ma^{\top}\mx\ms^{-1}\ma)^{-1}} & =(\ma^{\top}\mx\ms{}^{-1}\ma)^{-1}-(\ma^{\top}\mx'\ms'^{-1}\ma)^{-1}\in\R^{d\times d},\\
\mDelta_{\mx^{1/2}\ms^{-1/2}} & =\mx{}^{1/2}\ms{}^{-1/2}-\mx'^{1/2}\ms'^{-1/2}\in\R^{n\times n},\\
\mDelta_{\mx^{1/2}\ms^{-1/2}\ma(\ma^{\top}\mx\ms^{-1}\ma)^{-1}} & =\mx{}^{1/2}\ms{}^{-1/2}\ma(\ma^{\top}\mx\ms{}^{-1}\ma)^{-1}-\mx'^{1/2}\ms'^{-1/2}\ma(\ma^{\top}\mx'\ms'^{-1}\ma)^{-1}\in\R^{n\times n}.
\end{align*}
and the following two matrices
\[
\mq=\ma^{\top}\mx\ms{}^{-1}\ma\text{ and }\mq'=\ma^{\top}\mx'\ms'{}^{-1}\ma\,.
\]
Let $\delta_{\lambda}\in\R^{d}$ be defined as Algorithm \ref{alg:maintain_INfeasibility}.
Then, we have 
\begin{align*}
\E[\delta_{\lambda}] & =\E[\mh_{1}^{-1}\mh_{2}\mh_{3}^{-1}\delta_{b}^{(1)}+\mh_{4}\delta_{b}^{(2)}]\\
 & =\E[\mh_{1}^{-1}]\cdot\E[\mh_{2}]\cdot\E[\mh_{3}^{-1}]\cdot\E[\delta_{b}^{(1)}]+\E[\mh_{4}]\cdot\E[\delta_{b}^{(2)}]\\
 & =\mq^{-1}\ma^{\top}\mx^{1/2}\ms^{-1/2}(\mx^{1/2}\ms^{-1/2}-\mx'^{1/2}\ms'^{-1/2})\ma\mq'^{-1}(b-\ma^{\top}x)+(\mq^{-1}-\mq'^{-1})(b-\ma^{\top}x)\\
 & =\mq{}^{-1}\ma^{\top}\mx{}^{1/2}\ms{}^{-1/2}\mDelta_{\mx^{1/2}\ms^{-1/2}}\ma\mq^{'-1}(b-\ma^{\top}x)+\mDelta_{(\ma^{\top}\mx\ms^{-1}\ma)^{-1}}(b-\ma^{\top}x)\\
 & =-\mq{}^{-1}\ma^{\top}\mx{}^{1/2}\ms{}^{-1/2}\mDelta_{\mx^{1/2}\ms^{-1/2}\ma(\ma^{\top}\mx\ms^{-1}\ma)^{-1}}(\ma^{\top}x-b).
\end{align*}
Hence, we have
\begin{align}
 & \mx{}^{1/2}\ms{}^{-1/2}\ma\mq{}^{-1}(\ma^{\top}\E[x^{\new}]-b)\nonumber \\
= & \mx{}^{1/2}\ms{}^{-1/2}\ma\mq^{-1}(\ma^{\top}x-b)+\mx{}^{1/2}\ms{}^{-1/2}\ma\mq^{-1}\ma^{\top}\mx\ms{}^{-1}\ma\E[\delta_{\lambda}]\nonumber \\
= & \mx{}^{1/2}\ms{}^{-1/2}\ma\mq^{-1}(\ma^{\top}x-b)-\mx{}^{1/2}\ms{}^{-1/2}\ma\mq^{-1}\ma^{\top}\mx{}^{1/2}\ms{}^{-1/2}\mDelta_{\mx^{1/2}\ms^{-1/2}\ma(\ma^{\top}\mx\ms^{-1}\ma)^{-1}}(\ma^{\top}x-b)\nonumber \\
= & \mx{}^{1/2}\ms{}^{-1/2}\ma\mq^{-1}(\ma^{\top}x-b)\nonumber \\
 & -\mx{}^{1/2}\ms{}^{-1/2}\ma\mq^{-1}\ma^{\top}\mx{}^{1/2}\ms{}^{-1/2}\mx{}^{1/2}\ms{}^{-1/2}\ma\mq^{-1}(\ma^{\top}x-b)\nonumber \\
 & +\mx{}^{1/2}\ms{}^{-1/2}\ma\mq^{-1}\ma^{\top}\mx{}^{1/2}\ms{}^{-1/2}\mx'{}^{1/2}\ms'{}^{-1/2}\ma\mq'{}^{-1}(\ma^{\top}x-b)\nonumber \\
= & \mx{}^{1/2}\ms{}^{-1/2}\ma\mq^{-1}\ma^{\top}\mx{}^{1/2}\ms{}^{-1/2}\mx'^{1/2}\ms'^{-1/2}\ma\mq'{}^{-1}(\ma^{\top}x-b)\label{eq:rewrite_XSA_E_x_new_b}
\end{align}
Hence, we have
\begin{align*}
 & \Phi_{b}(\E[x^{\new}],x,s,\mu)\\
= & \|\mx{}^{\frac{1}{2}}\ms{}^{-\frac{1}{2}}\ma(\ma^{\top}\mx\ms{}^{-1}\ma)^{-1}(\ma^{\top}\E[x^{\new}]-b)\|_{2}^{2}\\
= & \|\mx{}^{\frac{1}{2}}\ms{}^{-\frac{1}{2}}\ma(\ma^{\top}\mx\ms{}^{-1}\ma)^{-1}\ma^{\top}\mx{}^{\frac{1}{2}}\ms{}^{-\frac{1}{2}}\mx'^{\frac{1}{2}}\ms'^{-\frac{1}{2}}\ma(\ma^{\top}\mx'\ms'^{-1}\ma)^{-1}(\ma^{\top}x-b)\|_{2}^{2}\\
\leq & \|\mx'^{\frac{1}{2}}\ms'^{-\frac{1}{2}}\ma(\ma^{\top}\mx'\ms'^{-1}\ma)^{-1}(\ma^{\top}x-b)\|_{2}^{2}\\
= & \Phi_{b}(x,x',s',\mu).
\end{align*}
where the first step follows from definition of $\Phi_{b}$ (see (\ref{eq:Phi_b_x_barx_bars_mu})),
the second step follows from (\ref{eq:rewrite_XSA_E_x_new_b}), the
third step follows from property of a projection matrix, the last
step follows from definition of $\Phi_{b}$ (see (\ref{eq:Phi_b_x_barx_bars_mu})).

Hence, we have
\begin{align*}
\Phi_{b}(x^{\new},x,s,\mu)= & \Phi_{b}(\E[x^{\new}],x,s,\mu)+\mu^{-1}\|\ma^{\top}(x^{\new}-\E[x^{\new}])\|_{\mq^{-1}}^{2}\\
\le & \Phi_{b}(x,x',s',\mu)+\mu^{-1}\|\delta_{\lambda}-\E[\delta_{\lambda}]\|_{\mq}^{2}
\end{align*}
Note that
\begin{align}
\E\left[\|\delta_{\lambda}-\E[\delta_{\lambda}]\|_{\mq}^{2}\right]\le & \E\left[\|\mh_{1}^{-1}\mh_{2}\mh_{3}^{-1}\delta_{b}^{(1)}-\E[\mh_{1}^{-1}\mh_{2}\mh_{3}^{-1}\delta_{b}^{(1)}]\|_{\mq}^{2}\right]\nonumber \\
 & +\E\left[\|\mh_{4}\delta_{b}^{(2)}-\E[\mh_{4}\delta_{b}^{(2)}]\|_{\mq}^{2}\right].\label{eq:delta_change_infes}
\end{align}
For the first term in (\ref{eq:delta_change_infes}), and $\delta_{b}=b-\ma^{\top}x$
we have
\begin{align}
 & \|\mh_{1}^{-1}\mh_{2}\mh_{3}^{-1}\delta_{b}^{(1)}-\E[\mh_{1}^{-1}\mh_{2}\mh_{3}^{-1}\delta_{b}^{(1)}]\|_{\mq}\nonumber \\
\le & \|(\mh_{1}^{-1}\mh_{2}\mh_{3}^{-1}-\mq^{-1}\E[\mh_{2}]\mq'^{-1})\delta_{b}\|_{\mq}\nonumber \\
 & +\|\mh_{1}^{-1}\mh_{2}\mh_{3}^{-1}(\delta_{b}^{(1)}-\delta_{b})\|_{\mq}\nonumber \\
\le & \|(\mh_{1}^{-1}\mh_{2}\mh_{3}^{-1}-\mq^{-1}\mh_{2}\mq'^{-1})\delta_{b}\|_{\mq}\nonumber \\
 & +\|\mq^{-1}(\mh_{2}-\E[\mh_{2}])\mq'^{-1}\delta_{b}\|_{\mq}\nonumber \\
 & +\|\mh_{1}^{-1}\mh_{2}\mh_{3}^{-1}(\delta_{b}^{(1)}-\delta_{b})\|_{\mq}\nonumber \\
\le & \|\mh_{1}^{-1}\mh_{2}(\mh_{3}^{-1}-\mq'^{-1})\delta_{b}\|_{\mq}\label{eq:delta_change_infeas_new}\\
 & +\|(\mh_{1}^{-1}-\mq^{-1})\mh_{2}\mq'^{-1}\delta_{b}\|_{\mq}\nonumber \\
 & +\|\mq^{-1}(\mh_{2}-\E[\mh_{2}])\mq'^{-1}\delta_{b}\|_{\mq}\nonumber \\
 & +\|\mh_{1}^{-1}\mh_{2}\mh_{3}^{-1}(\delta_{b}^{(1)}-\delta_{b})\|_{\mq}\nonumber 
\end{align}
where we used $\E_{\mh_{1}}[\mh_{1}^{-1}]=\mq^{-1}$, $\E_{\mh_{3}}[\mh_{3}^{-1}]=\mq'^{-1}$
and $\E_{\delta_{b}^{(1)}}[\delta_{b}^{(1)}]=\delta_{b}$. Note that
$-\mq'\preceq\E_{\mh_{2}}[\mh_{2}]\preceq\mq'$, $\mh_{1}^{-1}\approx_{\epsilon_{H}}\mq^{-1}$,
$\mh_{2}=\E_{\mh_{2}}[\mh_{2}]\pm\epsilon_{H}\cdot\mq'$, $\mh_{3}^{-1}\approx_{\epsilon_{H}}\mq^{-1}$,
and $\mq\approx_{1}\mq'$ so the first three terms of (\ref{eq:delta_change_infeas_new})
are all upper bounded by 
\[
O(\epsilon_{H}^{2}\|\delta_{b}\|_{\mq'^{-1}}^{2})=\mu\cdot O(\epsilon_{H}^{2}\cdot\Phi_{b}(x,x',s',\mu)).
\]
For the last term of (\ref{eq:delta_change_infeas_new}) note that
\begin{align*}
\|\mh_{1}^{-1}\mh_{2}\mh_{3}^{-1}(\delta_{b}^{(1)}-\delta_{b})\|_{\mq}^{2}\le & 2\|\mq'^{-1/2}\mh_{2}\mh_{3}^{-1}(\delta_{b}^{(1)}-\delta_{b})\|^{2}\\
= & 2\|\mq'^{-1/2}\mh_{2}\mh_{3}^{-1}\mq'^{1/2}\mq'^{-1/2}(\delta_{b}^{(1)}-\delta_{b})\|^{2}\\
= & \|\mq'^{-1/2}\mh_{2}\mh_{3}^{-1}\mq'^{1/2}\|_{F}^{2}\cdot O(\mu\epsilon_{b}\log^{2}d)\\
= & \|\mq'^{-1/2}\mh_{2}\mq'^{-1/2}\|_{F}^{2}\cdot\|\mq'^{1/2}\mh_{3}^{-1}\mq'^{1/2}\|_{2}^{2}\cdot O(\mu\epsilon_{b}\log^{2}d)\\
\le & \|\mq'^{-1/2}\mh_{2}\mq'^{-1/2}\|_{F}^{2}\cdot O(\mu\epsilon_{b}\log^{2}d),
\end{align*}
where we used Lemma \ref{lem:random_delta_b}. In summary this means
the first term in (\ref{eq:delta_change_infes}) can be bounded via
\begin{align*}
\E_{\mh_{1}^{-1},\mh_{2},\mh_{3},\delta_{b}^{(1)}}\left[\|\mh_{1}^{-1}\mh_{2}\mh_{3}^{-1}\delta_{b}^{(1)}-\E[\mh_{1}^{-1}\mh_{2}\mh_{3}^{-1}\delta_{b}^{(1)}]\|_{\mq}^{2}\right]\le & O(\mu\epsilon_{H}^{2}\cdot\Phi_{b}(x,x',s',\mu))\\
 & +\E_{\mh_{2}}\left[\|\mq'^{-1/2}\mh_{2}\mq'^{-1/2}\|_{F}^{2}\right]\cdot O(\mu\epsilon_{b}\log^{2}d).
\end{align*}
Similarly, we have that
\begin{align*}
\E_{\mh_{4},\delta_{b}^{(2)}}\left[\|\mh_{4}\delta_{b}^{(2)}-\E_{\mh_{4},\delta_{b}^{(2)}}[\mh_{4}\delta_{b}^{(2)}]\|_{\mq}^{2}\right] & \le O(\mu\epsilon_{H}^{2}\Phi_{b}(x,x',s',\mu))\\
 & +\E_{\mh_{4}}\left[\|\mq'^{1/2}\mh_{4}\mq'^{1/2}\|_{F}^{2}\right]\cdot O(\mu\epsilon_{b}\log^{2}d)
\end{align*}
which via (\ref{eq:delta_change_infes}) leads to
\begin{align*}
\E[\Phi_{b}(x^{\new},x,s,\mu)]= & (1+O(\epsilon_{H}^{2}))\cdot\Phi_{b}(x,x',s',\mu)\\
 & +\E_{\mh_{2}}\left[\|\mq'^{-1/2}\mh_{2}\mq'^{-1/2}\|_{F}^{2}\right]\cdot O(\epsilon_{b}\log^{2}d)\\
 & +\E_{\mh_{4}}\left[\|\mq'{}^{1/2}\mh_{4}\mq'{}^{1/2}\|_{F}^{2}\right]\cdot O(\epsilon_{b}\log^{2}d)\\
= & (1+O(\epsilon_{H}^{2}))\cdot\Phi_{b}(x,x',s',\mu)\\
 & +O((\|\ln x'-\ln x\|_{\tau(x,s)}^{2}+\|\ln s'-\ln s\|_{\tau(x,s)}^{2})\epsilon_{b}\log^{2}d),
\end{align*}
where we used Lemma \ref{lem:generate_H2_H4}.

For the movement of the step, note that
\begin{equation}
\|\mx^{-1}(x^{\new}-x)\|_{\tau(x,s)}^{2}=\|\mx^{-1}\mx\ms{}^{-1}\ma\delta_{\lambda}\|_{\tau(x,s)}^{2}=O(\mu^{-1})\cdot\|\delta_{\lambda}\|_{\ma^{\top}\mx'\ms'^{-1}\ma}^{2}\label{eq:x_change}
\end{equation}
where we used that $\mx s\approx_{1}\mu\tau(x,s),\ox'\approx_{1}\ox\approx_{1}x\in\R^{n}$.
Let $\delta_{\lambda}^{(1)}=\mh_{1}^{-1}\mh_{2}\mh_{3}^{-1}\delta_{b}^{(1)}\in\R^{d}$.
Using Lemma \ref{lem:random_delta_b}, we have that
\begin{align*}
\|\delta_{\lambda}^{(1)}\|_{\ma^{\top}\mx'\ms'^{-1}\ma}^{2}= & O(1)\cdot\|\mq'^{-\frac{1}{2}}\mh_{2}\mh_{3}^{-1}\mq'^{\frac{1}{2}}\mq'^{-\frac{1}{2}}\delta_{b}^{(1)}\|_{2}^{2}\\
= & O(1)\cdot\|\mq'^{-\frac{1}{2}}\mh_{2}\mh_{3}^{-1}\mq'^{\frac{1}{2}}\mq'^{-\frac{1}{2}}(\delta_{b}^{(1)}-\E[\delta_{b}^{(1)}])\|_{2}^{2}\\
 & +O(1)\cdot\|\mq'^{-\frac{1}{2}}\mh_{2}\mh_{3}^{-1}\mq'^{\frac{1}{2}}\mq'^{-\frac{1}{2}}(b-\ma^{\top}x)\|_{2}^{2}\\
= & O(\mu)\cdot\|\mq'^{-\frac{1}{2}}\mh_{2}\mh_{3}^{-1}\mq'^{\frac{1}{2}}\|_{F}^{2}\cdot(\epsilon_{b}\log^{2}(d/\rho))\\
 & +O(\mu)\cdot\|\mq'^{-\frac{1}{2}}\mh_{2}\mh_{3}^{-1}\mq'^{\frac{1}{2}}\|_{2}^{2}\cdot\Phi_{b}(x,x',s',\mu)
\end{align*}
with probability $1-\rho$. Using the definition of $\mh_{2}$ and
Lemma \ref{lem:generate_H2_H4}, we have
\begin{align*}
\|\mq'^{-\frac{1}{2}}\mh_{2}\mh_{3}^{-1}\mq'^{\frac{1}{2}}\|_{2}^{2} & =\|\mq'^{-\frac{1}{2}}\mh_{2}\mq'^{-\frac{1}{2}}\mq'^{\frac{1}{2}}\mh_{3}^{-1}\mq'^{\frac{1}{2}}\|_{2}^{2}\\
 & \leq\|\mq'^{-\frac{1}{2}}\mh_{2}\mq'^{-\frac{1}{2}}\|_{2}^{2}\cdot\|\mq'^{\frac{1}{2}}\mh_{3}^{-1}\mq'{}^{\frac{1}{2}}\|_{2}^{2}\\
 & =O(1)\cdot(\|\mq'^{-\frac{1}{2}}\E[\mh_{2}]\mq'^{-\frac{1}{2}}\|_{2}^{2}+\|\mq'^{-\frac{1}{2}}(\mh_{2}-\E[\mh_{2}])\mq'^{-\frac{1}{2}}\|_{2}^{2})\\
 & =O(1)\cdot\|\mq'^{-\frac{1}{2}}\E[\mh_{2}]\mq'^{-\frac{1}{2}}\|_{2}^{2}+O(\epsilon_{H}^{2})\\
 & \le O(\E[\|\mq'^{-\frac{1}{2}}\mh_{2}\mq'^{-\frac{1}{2}}\|_{2}^{2}]+\epsilon_{H}^{2})\\
 & =O(\|\ln x'-\ln x\|_{\tau(x,s)}^{2}+\|\ln s'-\ln s\|_{\tau(x,s)}^{2}+\epsilon_{H}^{2})
\end{align*}
where we used $\|\mq'^{\frac{1}{2}}\mh_{3}^{-1}\mq'{}^{\frac{1}{2}}\|_{2}^{2}=O(1)$,
$\E[x^{2}]\ge\E[x]^{2}$, $\|\ma\|_{2}^{2}\leq\|\ma\|_{F}^{2}$, and
Lemma (\ref{lem:generate_H2_H4}).

Similarly, we have
\[
\|\mq'^{-\frac{1}{2}}\mh_{2}\mh_{3}^{-1}\mq'^{\frac{1}{2}}\|_{F}^{2}\leq O(\|\ln x'-\ln x\|_{\tau(x,s)}^{2}+\|\ln s'-\ln s\|_{\tau(x,s)}^{2}+d\epsilon_{H}^{2}).
\]
Hence, we have
\begin{align*}
\|\delta_{\lambda}^{(1)}\|_{\mq'}^{2}= & O(\mu(\epsilon_{b}\log^{2}(d/\rho)+\Phi_{b}(x,x',s',\mu)))\cdot(\|\ln x'-\ln x\|_{\tau(x,s)}^{2}+\|\ln s'-\ln s\|_{\tau(x,s)}^{2}).\\
 & +O(\mu\epsilon_{b}\epsilon_{H}^{2}d\log^{2}(d/\rho))
\end{align*}
Similarly, we have the same bound for $\|\delta_{\lambda}^{(2)}\|_{\mq'}^{2}$.
Putting these two into (\ref{eq:x_change}) gives the result.

Finally for the $\ell_{\infty}$-norm we have
\begin{align*}
\|\mx^{-1}\delta_{x}\|_{\infty}= & \max_{i\in[n]}|e_{i}^{\top}\mx^{-1}\delta_{x}|
\end{align*}
where for any $i\in[n]$ we have 
\begin{align*}
|e_{i}^{\top}\mx^{-1}\delta_{x}|^{2}= & |e_{i}^{\top}\ms^{-1}\ma\delta_{\lambda}|^{2}\\
= & |e_{i}^{\top}\ms^{-1}\ma\mq'^{-1/2}\mq'^{1/2}\delta_{\lambda}|^{2}\\
\le & \|\mq'^{1/2}\ma\ms e_{i}\|_{2}^{2}\cdot\|\mq'^{1/2}\delta_{\lambda}\|_{2}^{2}\\
= & O(1)\cdot\|\mq^{1/2}\ma\mx^{1/2}\ms^{-1/2}(\mx\ms)^{-1/2}e_{i}\|_{2}^{2}\cdot\|\delta_{\lambda}\|_{\ma\mx'\ms'^{-1}\ma}^{2}\\
= & O((\mu\tau(x,s)_{i})^{-1})\cdot\|\mq^{1/2}\ma\mx^{1/2}\ms^{-1/2}e_{i}\|_{2}^{2}\cdot\|\delta_{\lambda}\|_{\ma\omx\oms^{-1}\ma}^{2}\\
= & O((\mu\tau(x,s)_{i})^{-1})\cdot\sigma(x,s)_{i}\cdot\|\delta_{\lambda}\|_{\ma\omx\oms^{-1}\ma}^{2}\\
= & O(\epsilon_{b}\epsilon_{H}d\log^{2}(d/\rho)+(\epsilon_{b}\log^{2}(d/\rho)+\Phi_{b}(x,x',s',\mu)))\cdot(\|\ln x'-\ln x\|_{\tau}^{2}+\|\ln s'-\ln s\|_{\tau}^{2}).
\end{align*}
where the first step follows from definition of $\delta_{x}$, the
fifth step follows from $xs\approx\mu\tau(x,s)$, the sixth step follows
from the definition of leverage scores, and the last step follows
from our previously proven bound on $\|\delta_{\lambda}\|_{\ma\omx\oms^{-1}\ma}$
and Lemma \ref{lem:sigma_approx} to bound $\sigma(x,s)\le\tau(x,s)$.
\end{proof}

\subsection{Improving Infeasibility\label{subsec:improve_infeasibility}}

From Section \ref{subsec:fes_x} and \ref{subsec:fes_xbar}, we see
that $\Phi_{b}$ increases slowly over time. Here we show in Lemma
\ref{lem:Phi_b_move_slowly}, that over $\sqrt{d}/\log^{\power}n$
iterations of the IPM, the potential $\Phi_{b}$ changes only by a
constant factor. Intuitively, this can be seen by the IPM calling
$\textsc{MaintainFeasibility}$ (Algorithm~\ref{alg:maintain_Feasibility})
which in turn calls Algorithm~\ref{alg:maintain_INfeasibility}.
The previous subsection showed that each call to Algorithm~\ref{alg:maintain_INfeasibility}
increases $\Phi_{b}$ by only a small amount. As seen in Line \ref{line:correction_step}
of Algorithm~\ref{alg:maintain_Feasibility}, after $\sqrt{d}/\log^{\power}n$
iterations we decrease the potential again. Lemma \ref{lem:decrease_Phib}
shows that Line \ref{line:correction_step} does indeed decrease the
potential $\Phi_{b}$sufficiently.

\begin{algorithm2e}[!t]

\caption{Maintaining Feasibility}

\label{alg:maintain_Feasibility}

\SetKwProg{Proc}{procedure}{}{}

\Proc{\textsc{MaintainFeasibility}$(x\in\R^{n},s\in\R^{n},\overline{\tau}\in\R^{n})$}{

\State $\delta_{\lambda}\leftarrow\textsc{MaintainInfeasibility}(x,s,\overline{\tau})$
\tcp*{Make $\Phi_{b}(x+\mx\ms{}^{-1}\ma\delta_{\lambda},x,s,\mu)$
unbiased, so it stays small with constant probability between any
two execution of the following if-branch:}

\State \If{it has been $\sqrt{d}/\log^{\power}n$ iterations since
this branch was executed}{ \Comment{Reduce $\Phi_{b}$}

\State $\delta_{\lambda}=\delta_{\lambda}+\mh^{-1}(b-\ma^{\top}x)\in\R^{n}$
with $\mh\approx_{\epsilon_{H}}\mq\in\R^{d\times d}$\Comment{$\mq=\ma^{\top}\mx\ms^{-1}\ma$}\label{line:correction_step}

}

\LineComment{With constant probability $\Phi_{b}(x+\mx\ms{}^{-1}\ma\delta_{\lambda},x,s,\mu)$
stays small. If it failed, re-try the past $\otilde(\sqrt{d})$ iterations
of the IPM since the last time this branch was executed.}

\State \Return $\delta_{\lambda}$

}

\end{algorithm2e}
\begin{lem}
\label{lem:Phi_b_move_slowly} Consider $T\leq\sqrt{d}/\log^{\power}n$
iterations of the algorithm \ref{alg:pathfollowing} and let $x^{(k)},s^{(k)}$
be the input to the $k$-th call to Algorithm~\ref{alg:maintain_INfeasibility}.
Suppose that $\Phi_{b}(x^{(1)},x^{(1)},s^{(1)},\mu^{(1)})\leq\frac{\zeta\epsilon^{2}}{\log^{\power}n}$,
$\mu^{(k+1)}=(1-\frac{\epsilon_{\mu}}{\sqrt{d}})\mu^{(k)},\forall k\in[T]$
and $\epsilon_{b}\leq\frac{c\zeta\epsilon^{2}}{\sqrt{d}\log^{2}n}$,
$\epsilon_{H}\leq\frac{c\zeta\epsilon}{d^{1/4}}$ for some small enough
constants $\zeta,c>0$, where $\epsilon_{b}$ is the accuracy parameter
used in Lemma \ref{lem:random_delta_b}. Suppose that we update $x$
using an unbiased linear system solver with accuracy $\epsilon_{H}$
as defined in Lemma \ref{lem:fes_change_x} during the algorithm \ref{alg:pathfollowing}.
Assume further 
\[
\|\ln x^{(k)}-\ln x^{(k-1)}\|_{\tau(x^{(k)},s^{(k)})+\infty}+\|\ln s^{(k)}-\ln s^{(k-1)}\|_{\tau(x^{(k)},s^{(k)})+\infty}\le0.1\text{ for all }1<k\le T.
\]
Then, we have 
\[
\Phi_{b}(x^{(k)},x^{(k-1)},s^{(k-1)},\mu^{(k)})\leq\frac{5\zeta\epsilon^{2}}{\log^{\power}n}\text{ for all }1<k\leq T
\]
with probability at least $\frac{1}{2}$.
\end{lem}

\begin{proof}
Let $\delta_{x}^{(k)}$ be the vector $\delta_{x}$ as defined in
Algorithm~\ref{alg:maintain_INfeasibility}, when we currently perform
the $k$-th call to that function. Then we have that

\begin{align*}
 & \E[\Phi_{b}(x^{(k+1)},x^{(k)},s^{(k)},\mu^{(k+1)})]\\
\leq & \Phi_{b}(x^{(k)}+\delta_{x}^{(k)},x^{(k)},s^{(k)},\mu^{(k+1)})+O(\epsilon_{H}^{2})\cdot\|h\|_{\tau(x^{(k)},s^{(k)})}^{2}\\
\leq & (1+O(\epsilon_{H}^{2}))\cdot\Phi_{b}(x^{(k)},x^{(k-1)},s^{(k-1)},\mu^{(k+1)})+O(\epsilon_{H}^{2})\cdot\|h\|_{\tau(x^{(k)},s^{(k)})}^{2}\\
 & +O(\epsilon_{b}\log^{2}d)\cdot\left(\|\ln x^{(k)}-\ln x^{(k-1)}\|_{\tau(x^{(k)},s^{(k)})}+\|\ln s^{(k)}-\ln s^{(k-1)}\|_{\tau(x^{(k)},s^{(k)})}\right)\\
= & (1+O(\epsilon_{H}^{2}))\cdot(1+2\epsilon_{\mu}d^{-1/2})\cdot\Phi_{b}(x^{(k)},x^{(k-1)},s^{(k-1)},\mu^{(k)})+O(\epsilon_{H}^{2})\cdot\|h\|_{\tau(x^{(k)},s^{(k)})}^{2}\\
 & +O(\epsilon_{b}\log^{2}d)\\
= & (1+O(\epsilon_{H}^{2}))\cdot(1+2\epsilon_{\mu}d^{-1/2})\cdot\Phi_{b}(x^{(k)},x^{(k-1)},s^{(k-1)},\mu^{(k)})+O(\epsilon_{H}^{2})+O(\epsilon_{b}\log^{2}d)\\
= & (1+2\epsilon_{\mu}d^{-1/2}+O(c\zeta\epsilon^{2}d^{-1/2}))\cdot\Phi_{b}(x^{(k)},x^{(k-1)},s^{(k-1)},\mu^{(k)})+O(c\zeta\epsilon^{2}d^{-1/2}),
\end{align*}
where the first step follows from Lemma \ref{lem:fes_change_x}, the
second step follows from Lemma \ref{lem:change_xbar_sbar}, the third
step follows from $\mu^{(k+1)}\leq(1-\epsilon_{\mu}d^{-1/2}\mu^{(k)})$
and $\|\ln x^{(k)}-\ln x^{(k-1)}\|_{\tau(x^{(k)},s^{(k)})}+\|\ln s^{(k)}-\ln s^{(k-1)}\|_{\tau(x^{(k)},s^{(k)})}\le0.1$,
the fourth step follows from the step of the IPM with $\|h\|_{\tau}^{2}=O(1)$.

Let $\Phi^{(k)}=\Phi_{b}(x^{(k)},x^{(k-1)},s^{(k-1)},\mu^{(k)})$.
Let 
\begin{equation}
\Psi^{(k)}=\frac{\Phi^{(k)}+c\zeta\epsilon^{2}\log^{-\power}n}{\frac{8\zeta\epsilon^{2}}{\log^{\power}n}+k\zeta\epsilon^{2}d^{-1/2}}.\label{eq:write_Psi_in_terms_of_Phi}
\end{equation}
Note that $\Psi^{(k)}\ge1/9$ for $k\le\sqrt{d}/\log^{\power}n$ (since
$\Phi^{(k)}\geq0$) and further there exists some $c'=O(c)$ such
that
\begin{align*}
\E[\Psi^{(k+1)}]\le & \frac{\Phi^{(k)}+c\zeta\epsilon^{2}\log^{-\power}n+c'\epsilon^{2}\zeta d^{-1/2}+(2\epsilon_{\mu}+c'\epsilon^{2}\zeta)d^{-1/2}\cdot\Phi^{(k)}}{\frac{8\zeta\epsilon^{2}}{\log^{\power}n}+k\epsilon^{2}\zeta d^{-1/2}+\zeta\epsilon^{2}d^{-1/2}}.
\end{align*}
We want to construct a supermartingale, so we want that this expectation
is less than $\Psi^{(k)}$. This is the case if
\begin{align*}
\frac{c'\epsilon^{2}\zeta d^{-1/2}+(2\epsilon_{\mu}+c'\epsilon^{2}\zeta)d^{-1/2}\cdot\Phi^{(k)}}{\zeta\epsilon^{2}d^{-1/2}}\le & \frac{\Phi^{(k)}+c\zeta\epsilon^{2}\log^{-\power}n}{\frac{8\zeta\epsilon^{2}}{\log^{\power}n}+k\zeta\epsilon^{2}d^{-1/2}}=\Psi^{(k)},
\end{align*}
so let us analyze which other conditions are required to satisfy this
inequality. For now assume that $\Psi^{(k)}\le1$, and $k\leq\frac{\sqrt{d}}{\log^{\power}n}$,
then 
\begin{align*}
\Phi^{(k)}\le & \frac{8\zeta\epsilon^{2}}{\log^{\power}n}+k\zeta\epsilon d^{-1/2}<\frac{9\zeta\epsilon^{2}}{\log^{\power}n}.
\end{align*}
So if we choose $c$ small enough such that $2\epsilon_{\mu}+c'\epsilon^{2}\zeta\le1/100$
(note that $\epsilon_{\mu}\ll1/16000$ by Theorem~\ref{thm:path_following})
and $\Phi^{(k)}\leq\frac{9\zeta\epsilon^{2}}{\log^{\power}n}$, then

\begin{align*}
\frac{c'\epsilon^{2}\zeta d^{-1/2}+(2\epsilon_{\mu}+c'\epsilon^{2}\zeta)d^{-1/2}\cdot\Phi^{(k)}}{\zeta\epsilon^{2}d^{-1/2}} & \le\frac{\epsilon^{2}\zeta d^{-1/2}/100+d^{-1/2}\Phi^{(k)}/100}{\zeta\epsilon^{2}d^{-1/2}}\le1/9\le\Psi^{(k)},
\end{align*}
which in turn shows that $\E[\Psi^{(k+1)}]\le\Psi^{(k)}$. In general,
without the assumption $\Psi^{(k)}\le1$, we have
\[
\E\left[\min\Big\{\Psi^{(k+1)},1\Big\}\right]\leq\min\Big\{\Psi^{(k)},1\Big\}.
\]
Hence, for $k\le\sqrt{d}/\log^{\power}n=T$ we have that $\min\Big\{\Psi^{(k)},1\Big\}$
is a non-negative supermartingale.

By Ville\textquoteright s maximal inequality \cite{ville1939etude}
for supermartingales, we have that
\[
\text{\ensuremath{\P}}\left[\max_{k\in[T]}\min\Big\{\Psi^{(k)},1\Big\}>\frac{1}{2}\right]\leq\frac{\min\Big\{\Psi^{(1)},1\Big\}}{\frac{1}{2}}\leq\frac{\frac{1}{4}}{\frac{1}{2}}=\frac{1}{2}.
\]
where the second step follows from $\Phi^{(1)}\leq\frac{1}{4}$. Hence,
$\max_{k\in[T]}\Psi^{(k)}\leq\frac{1}{2}$ with probability at least
$\frac{1}{2}$. Under this event, we have for small enough $c$ that
\begin{align*}
\max_{k\in[T]}\Phi^{(k)} & \leq\max_{k\in[T]}\Psi^{(k)}(\frac{8\zeta\epsilon^{2}}{\log^{\power}n}+k\zeta\epsilon^{2}d^{-1/2})\\
 & \leq\frac{1}{2}\cdot\frac{9\zeta\epsilon^{2}}{\log^{\power}n}\\
 & <\frac{5\zeta\epsilon^{2}}{\log^{\power}n}.
\end{align*}
\end{proof}
Now, to move $\ma^{\top}x-b$ closer to $0$, we solve the equation
\begin{align*}
\omx\delta_{s}+\oms\delta_{x} & =0,\in\R^{n}\\
\ma^{\top}\delta_{x} & =\delta_{b},\in\R^{d}\\
\ma\delta_{y}+\delta_{s} & =0\in\R^{n}
\end{align*}
with $\delta_{b}=b-\ma^{\top}x$. This gives the formula 
\begin{align*}
\delta_{s} & =-\ma(\ma^{\top}\oms^{-1}\omx\ma)^{-1}\delta_{b},\\
\delta_{x} & =\omx\oms^{-1}\ma(\ma^{\top}\oms^{-1}\omx\ma)^{-1}\delta_{b}.
\end{align*}
We see in Algorithm~\ref{alg:maintain_Feasibility} that every $\otilde(\sqrt{d})$
iterations we will try to decrease the potential $\Phi_{b}$. The
following lemma shows that such a step decreases $\Phi_{b}$ sufficiently.
\begin{lem}
\label{lem:decrease_Phib}Let $\widehat{x},x,s\in\R_{>0}^{n}$ and
consider $x^{\new}=\widehat{x}+\delta_{x}\in\R^{n}$ where $\delta_{x}=\mx\ms{}^{-1}\ma\mh^{-1}(b-\ma^{\top}\widehat{x})\in\R^{n}$
with $\mh\approx_{\epsilon_{H}}\ma^{\top}\mx\ms^{-1}\ma\in\R^{d\times d}$
for $\epsilon_{H}\in(0,1/20]$. Then, we have that
\[
\Phi_{b}(x^{\new},x,s,\mu)\leq5\epsilon_{H}\cdot\Phi_{b}(\widehat{x},x,s,\mu).
\]
\end{lem}

\begin{proof}
For any $p>0$, we have
\begin{align*}
\Phi_{b}(x^{\new},x,s,\mu)= & \mu^{-1}\cdot\|\ma^{\top}x^{\new}-b\|_{(\ma^{\top}\mx\ms^{-1}\ma)^{-1}}^{2}\\
= & \mu^{-1}\cdot\|\ma^{\top}(\widehat{x}+\delta_{x})-b\|_{(\ma^{\top}\mx\ms^{-1}\ma)^{-1}}^{2}\\
= & \mu^{-1}\cdot\|(\mi-\ma^{\top}\mx\ms{}^{-1}\ma\mh^{-1})(\ma^{\top}\widehat{x}-b)\|_{(\ma^{\top}\mx\ms^{-1}\ma)^{-1}}^{2},
\end{align*}
where the first step follows from definition of $\Phi_{b}$ (see (\ref{eq:Phi_b_x_barx_bars_mu})),
the second step follows from $x^{\new}=\widehat{x}+\delta_{x}$, and
the last step follows from definition of $\delta_{x}.$

Using $\mh\approx_{\epsilon_{H}}\ma^{\top}\mx\ms^{-1}\ma$, we have
\begin{align*}
 & (\mi-\mh^{-1}\ma^{\top}\mx\ms{}^{-1}\ma)(\ma^{\top}\mx\ms^{-1}\ma)^{-1}(\mi-\ma^{\top}\mx\ms{}^{-1}\ma\mh^{-1})\\
= & (\ma^{\top}\mx\ms^{-1}\ma)^{-1}-2\mh^{-1}+\mh^{-1}\ma^{\top}\mx\ms{}^{-1}\ma\mh^{-1}\\
\preceq & (1-2e^{-\epsilon_{H}}+e^{2\epsilon_{H}})\cdot(\ma^{\top}\mx\ms^{-1}\ma)^{-1}\\
\preceq & 5\epsilon_{H}\cdot(\ma^{\top}\mx\ms^{-1}\ma)^{-1}.
\end{align*}
where the last step follows from $e^{2\epsilon_{H}}\leq1+3\epsilon_{H}$
and $e^{-\epsilon_{H}}\geq1-\epsilon_{H}$ for $\epsilon_{H}\in(0,1/20]$.

Hence, we have
\[
\Phi_{b}(x^{\new},x,s,\mu)\leq5\epsilon_{H}\mu^{-1}\cdot\|\ma^{\top}\widehat{x}-b\|_{(\ma^{\top}\mx\ms^{-1}\ma)^{-1}}^{2}=5\epsilon_{H}\cdot\Phi_{b}(\widehat{x},x,s,\mu).
\]
\end{proof}
Here we show that the correction step of Lemma~\ref{lem:decrease_Phib}
does not change the solution $x$ by much, provided that $\Phi_{b}$
is small.
\begin{lem}
\label{lem:move_x_correction}Let $\widehat{x},x,s\in\R_{>0}^{n}$
and consider $\delta_{x}=\mx\ms{}^{-1}\ma\mh^{-1}(b-\ma^{\top}\widehat{x})\in\R^{n}$
with $\mh\approx_{\epsilon_{H}}\ma^{\top}\ms{}^{-1}\mx\ma\in\R^{d\times d}$
with $\epsilon_{H}\in(0,1/20]$ as in Lemma \ref{lem:decrease_Phib}.
Then, we have that
\begin{align*}
\|\mx{}^{-1}\delta_{x}\|_{\tau(x,s)}\le & O(1)\cdot\Phi_{b}(\widehat{x},x,s,\mu)^{1/2}\\
\|\mx{}^{-1}\delta_{x}\|_{\infty}\le & O(1)\cdot\Phi_{b}(\widehat{x},x,s,\mu)^{1/2}
\end{align*}
\end{lem}

\begin{proof}
For the bound on $\|\mx{}^{-1}\delta_{x}\|_{\infty}$ consider the
following
\begin{align*}
\|\mx{}^{-1}\delta_{x}\|_{\infty}= & O(1)\cdot\max_{i\in[n]}\|\indicVec i^{\top}\ms{}^{-1}\ma\mh^{-1}(b-\ma^{\top}\widehat{x})\|_{\infty}\\
\le & O(1)\cdot\max_{i\in[n]}\|\indicVec i^{\top}\ms{}^{-1}\ma\mh^{-1}(b-\ma^{\top}\widehat{x})\|_{2}\\
\le & O(1)\cdot\max_{i\in[n]}\|\indicVec i^{\top}\ms{}^{-1}\ma\mh^{-1/2}\|_{2}\|\mh^{-1/2}(b-\ma^{\top}\widehat{x})\|_{2}.
\end{align*}
Here the first factor can be bounded via 
\begin{align*}
\|e_{i}^{\top}\ms{}^{-1}\ma\mh^{-1/2}\|_{2}^{2}= & O(1)\cdot\indicVec i^{\top}\frac{1}{(\mx\ms)^{1/2}}\frac{\mx{}^{1/2}}{\ms{}^{1/2}}\ma(\ma^{\top}\mx\ms{}^{-1}\ma)^{-1}\ma^{\top}\frac{\mx{}^{1/2}}{\ms{}^{1/2}}\frac{1}{(\mx\ms')}\indicVec i\\
= & O(1)\cdot\frac{\sigma(\mx^{1/2}\ms^{-1/2}\ma)_{i}}{x_{i}s_{i}}=O(\mu^{-1}),
\end{align*}
where at the end we used $\sigma(\mx^{1/2}\ms^{-1/2}\ma)\le\sigma(\mx^{1/2-\alpha}\ms^{-1/2-\alpha}\ma)\le\tau(x,s)$
via Lemma \ref{lem:sigma_approx} and $xs\approx\mu\tau$. The last
term can be bounded by 
\begin{align*}
\|\mh^{-1/2}(b-\ma^{\top}\widehat{x})\|_{2}= & O(1)\cdot\|b-\ma^{\top}\widehat{x}\|_{(\ma^{\top}\mx\ms^{-1}\ma)^{-1}}\\
= & O(\mu^{1/2})\cdot\sqrt{\Phi_{b}(\widehat{x},x,s,\mu)}.
\end{align*}
Thus in summary we have $\|\mx{}^{-1}\delta_{x}\|_{\infty}\le O(1)\cdot\sqrt{\Phi_{b}(\widehat{x},x,s,\mu)}$.
For the $\|\cdot\|_{\tau}$norm we have because of $xs=\mu\tau$,
and the definition of $\Phi_{b}$ that
\begin{align*}
\|\mx{}^{-1}\delta_{x}\|_{\tau(x,s)}\le & O(\mu^{-1/2})\cdot\|(\mx\ms)^{1/2}\mx{}^{-1}\delta_{x}\|_{2}\\
\le & O(\mu^{-1/2})\cdot\|\mx{}^{1/2}\ms{}^{-1/2}\ma\mh^{-1}(b-\ma^{\top}\widehat{x})\|_{2}\\
= & O(\mu^{-1/2})\cdot\|b-\ma\widehat{x}\|_{(\ma^{\top}\mx\ms{}^{-1}\ma)^{-1}}\\
= & O(1)\cdot\Phi_{b}(\widehat{x},x,s,\mu)^{1/2}.
\end{align*}
\end{proof}
To obtain the final solution of our LP, $\Phi_{b}=\Omega(1/\log^{\power}n)$
is still too large. Here we show that iterative application of Lemma~\ref{lem:decrease_Phib}
yields a very accuracte solution.
\begin{cor}
\label{cor:tiny_phi_b}There exists some small enough $\zeta=O(1)$,
such that given a primal dual pair $(x,s)$ with $\Phi_{b}(x,x,s,\mu)\le\frac{4\zeta\epsilon}{\log n}$
and $xs\approx_{1/4}\mu\tau(x,s)$ and any $\delta>$0, we can compute
an $x'$ with 
\begin{align*}
\Phi_{b}(x',x',s,\mu) & \le\delta,\\
\|\mx^{-1}(x-x')\|_{\tau(x,s)+\infty} & \le O\left(\Phi_{b}^{1/2}(x,x,s,\mu)\right),\\
x's & \approx_{1/2}\mu\tau(x',s)
\end{align*}
 in $\otilde((nd+d^{3})\log(1/\delta))$ time.
\end{cor}

\begin{proof}
This follows by repeatedly applying Lemma \ref{lem:decrease_Phib}
for small enough $\epsilon_{H}=O(1)$. We need $O(\log\delta^{-1})$
repetitions to decrease $\Phi_{b}(x',x,s,\mu)$ down to $c\delta$
for some small enough $c=O(1)$. Note that by Lemma \ref{lem:move_x_correction}
the total movement is bounded by $O\left(\Phi_{b}^{1/2}(x,x,s,\mu)\right)=O(1)$
as the movement per iteration is exponentially decaying. At last,
going from $\Phi_{b}^{1/2}(x',x,s,\mu)$ to $\Phi_{b}^{1/2}(x',x',s,\mu)$
increases the potential by at most some $O(1)$ factor given that
$x$ and $x'$ differ by at most a constant factor. Thus for small
enough $c=O(1)$ we have $\Phi_{b}^{1/2}(x',x,s,\mu)\le\delta$. Likewise
we have $x's\approx_{1/2}\mu\tau(x',s)$ when the constant factor
difference between $x$ and $x'$ is small enough which can be guaranteed
by choosing small enough $\zeta>0$.
\end{proof}
Throughout the IPM we move $x$ several times. First, we perform the
classic IPM step and then we perform the corrective steps of Algorithm~\ref{alg:maintain_INfeasibility}
and~\ref{alg:maintain_Feasibility}. Note that by Lemma~\ref{lem:change_xbar_sbar}
the extra movement of $x$ depends on how much $x$ moved in the previous
iteration. Here we show that this does not create an amplifying feedback
loop, i.e. as long as $\Phi_{b}\le5\zeta\epsilon^{2}/\log^{6}n$,
the total movement $\|\mx^{-1}(x-x')\|_{\tau+\infty}$ is always bounded
by $\frac{\epsilon}{2}$. By induction over the number of iterations
Lemma~\ref{lem:change_xbar_sbar} and Lemma~\ref{lem:move_x_total}
then imply that we always have $\Phi_{b}\le5\zeta\epsilon^{2}/\log^{6}n$
and $\|\mx^{-1}(x-x')\|_{\tau+\infty}\le\epsilon/2$.
\begin{lem}
\label{lem:move_x_total}For some small enough constants $\zeta,c>0$
let $\epsilon_{b}\leq\frac{c\zeta\epsilon^{2}}{\sqrt{d}\log^{2}n}$,
$\epsilon_{H}\leq\frac{c\zeta\epsilon}{d^{1/4}\log^{3}n}$ in $\textsc{MaintainFeasibility}$
and let $\epsilon$ be the parameter of Theorem~\ref{thm:path_following}.
Let $x^{(k)},s^{(k)}$ the inputs of the $k$-th call to $\textsc{MaintainFeasibility}$.
Assume 
\begin{align*}
\Phi_{b}(x^{(k-1)},x^{(k-2)},s^{(k-2)},\mu^{(k-1)})\le & 5\zeta\epsilon^{2}/\log^{\power}n\\
\|\mx^{(k-2)-1}(x^{(k-1)}-x^{(k-2)})\|_{\tau+\infty}\le & \frac{\epsilon}{2}
\end{align*}
 for all $k'<k$, then we have with high probability 
\begin{align}
\|\mx^{(k)-1}(x^{(k)}-x^{(k-1)})\|_{\infty+\tau}^{2}\le & \epsilon^{2}/2.\label{eq:total_x_movement}
\end{align}
Further let $e_{x}\in\R^{n}$ be the movement of $x$ induced by calling
$\textsc{MaintainFeasibility}$ (Algorithm~\ref{alg:maintain_Feasibility}),
then for small enough constant $\zeta>0$ we have with high probability
that
\begin{align*}
\mixedNorm{\mx^{-1}e_{x}}{\tau}\leq & \frac{\gamma\alpha}{2^{20}}
\end{align*}
\end{lem}

\begin{proof}
Theorem \ref{thm:path_following_simplified} yields the bound (\ref{eq:total_x_movement})
if the extra movement $e_{x}$ caused by $\textsc{MaintainFeasibility}$
satisfies $\mixedNorm{\mx^{-1}e_{x}}{\tau}\leq\frac{\gamma\alpha}{2^{20}}$.

We define three errors : 1. the movement due to $\textsc{MaintainInfeasibility}$,
say $\mathrm{Err_{MIN}}$ ; 2, the movement due to IF-condition, say
$\mathrm{Err_{IFC}}$. Then we know that

\begin{align*}
\mathrm{Err_{MIN}^{2}}= & O(\epsilon_{b}\log^{2}(d)+\Phi_{b}(x^{(k-1)},x^{(k-2)},s^{(k-2)},\mu^{(k-1)}))\\
 & \cdot(\|\ln x^{(k-1)}-\ln x^{(k-2)}\|_{\tau(x^{(k-1)},s^{(k-1)})}^{2}+\|\ln s^{(k-1)}-\ln s^{(k-2)}\|_{\tau(x^{(k-1)},s^{(k-1)})}^{2})\\
 & +O(\epsilon_{H}^{2}\epsilon_{b}d\log^{2}n)\\
\mathrm{Err_{IFC}^{2}}= & O(\Phi_{b}(x^{(k-1)},x^{(k-2)},s^{(k-2)},\mu^{(k-2)}))
\end{align*}
by Lemma \ref{lem:change_xbar_sbar} and Lemma \ref{lem:move_x_correction}.
These can be bounded as follows
\begin{align*}
\mathrm{Err_{MIN}^{2}}\le & O(\zeta\epsilon^{2}/\log^{\power}n)\cdot\frac{\epsilon}{2}+O(\zeta\epsilon^{2}/\log^{6}n)\\
\mathrm{Err_{IFC}^{2}}\le & O(\zeta\epsilon^{2}/\log^{\power}n)
\end{align*}
So for small enough $\zeta>0$ we get
\[
\|\mx^{-1}e_{x}\|_{\tau+\infty}\leq\mathrm{Err_{MIN}+\mathrm{Err_{IFC}}}\leq\frac{\gamma\alpha}{2^{20}}
\]
where we use $\gamma\alpha=\Omega(\epsilon/\log^{3}n)$.
\end{proof}

\begin{proof}[Proof of Theorem~\ref{thm:path_following}]
Note that, when ignoring feasibility, Theorem~\ref{thm:path_following}
was proven as Theorem~\ref{thm:path_following_simplified}. So we
are only left with showing that $\Phi_{b}$ stays small and that the
condition of Theorem~\ref{thm:path_following_simplified} is true,
i.e. that the extra movement $e_{x}$ of $x$ by calling $\textsc{MaintainFeasibility}$
in Algorithm~\ref{alg:pathfollowing} satisfies

\[
\mixedNorm{\mx^{-1}e_{x}}{\tau}\leq\frac{\gamma\alpha}{2^{20}}.
\]

On one hand, the latter claim was proven in Lemma~\ref{lem:move_x_total},
assuming the infeasibility potential $\Phi_{b}$ is small. On the
other hand, Lemma~\ref{lem:Phi_b_move_slowly} shows that $\Phi_{b}$
does not change more than a multiplicative factor within $\sqrt{d}/\log^{\power}n$
iterations as long as $x$ and $s$ do not change to much in each
iteration. Thus by induction we have that $\Phi_{b}$ is small and
that $x$ and $s$ do not change much. After $\sqrt{d}/\log^{\power}n$
iterations the potential is decreased again by Lemma \ref{lem:decrease_Phib},
so $\Phi_{b}(x,x,s,\mu)$ stays small even after $\sqrt{d}\log^{\power}n$
iterations.

Note that the claim, that $\Phi_{b}$ stays small throughout a sequence
of $\otilde(\sqrt{d})$ iterations, hold only with probability $0.5$
by Lemma~\ref{lem:Phi_b_move_slowly}. Thus after every $\otilde(\sqrt{d})$
iterations we may have to revert the changes of the past $\otilde(\sqrt{d})$
iterations and retry. This increases the total runtime by another
$O(\log n)$ factor.

At last, note that during the last iteration of the IPM, we can apply
Lemma \ref{lem:decrease_Phib} once more to reduce
\[
\Phi_{b}(x,x,s,\mu)\le\frac{\zeta\epsilon^{2}}{\sqrt{d}\log^{\power}n}.
\]
\end{proof}
We are left with analyzing the complexity of Algorithm~\ref{alg:maintain_Feasibility}.
\begin{proof}[Proof of Theorem~\ref{thm:feasibilityComplexity}]
We choose two accuracy parameters as follows:
\[
\epsilon_{b}=\Omega(d^{-1/2}\epsilon^{-2}\log^{-2}n)\,\text{ and }\,\epsilon_{H}=\Omega(d^{-1/4}\epsilon^{-1}\log^{-3}n)
\]
then Lemma \ref{lem:random_delta_b} shows that computing $\delta_{b}$
in $\textsc{MaintainInfeasibility}$ takes $\otilde(d^{5/2}\epsilon^{-2})$
time. The products with $\mh_{1},\mh_{2},\mh_{3},\mh_{4}$ either
need $\otilde(d^{2}/\epsilon_{h}^{2})=\otilde(d^{5/2}\epsilon^{-2})$
time, or (in case of inverse matrices) require the same time as solving
the systems in Lines \ref{line:pf_3} and \ref{line:pf_4} of Algorithm
\ref{alg:pathfollowing}. So these costs are subsumed and we do not
have to count them. The algorithm used in Lemma \ref{lem:decrease_Phib}
involves computing $\ma^{\top}x-b$ which takes $\otilde(nd)$ time
and solving a linear system, which is subsumed by Lines \ref{line:pf_3}
and \ref{line:pf_4} of Algorithm \ref{alg:pathfollowing}. Since
this happens every $\tilde{\Omega}(\sqrt{d})$ iterations, the amortized
cost is $\otilde(n\sqrt{d})$. 

Thus the total additional amortized cost is 

\[
\otilde(n\sqrt{d}+d^{5/2}\epsilon^{-2}).
\]
\end{proof}

\section{Constructing the Initial Point\label{sec:initialPoint}}

Here we want to prove Lemma \ref{thm:infeasible_reduction} which
is used to quickly find an initial point for our IPM. Lemma \ref{thm:infeasible_reduction}
is an extension of the following known reduction. We modify this reduction
so that we no longer require our final solution to be feasible.
\begin{lem}[\cite{YTM94,cohen2019solving}]
\label{lem:original_reduction}Consider a linear program $\min_{\ma^{\top}x=b,x\geq0}c^{\top}x$
with $n$ variables and $d$ constraints. Assume that \\
 1. Diameter of the polytope : For any $x\geq0$ with $\ma^{\top}x=b$,
we have that $\|x\|_{2}\leq R$.\\
 2. Lipschitz constant of the linear program : $\|c\|_{2}\leq L$.

For any $\delta\in(0,1]$, the modified linear program $\min_{\overline{\ma}^{\top}\overline{x}=\overline{b},\overline{x}\geq0}\overline{c}^{\top}\overline{x}$
with
\[
\overline{\ma}=\left[\begin{array}{cc}
\ma & 1_{n}\|\ma\|_{F}\\
0 & 1\|\ma\|_{F}\\
\frac{1}{R}b^{\top}-1_{n}^{\top}\ma & 0
\end{array}\right]\in\R^{(n+2)\times(d+1)},\overline{b}=\begin{bmatrix}\frac{1}{R}b\\
(n+1)\|\ma\|_{F}
\end{bmatrix}\in\R^{d+1}\text{~and,~}\overline{c}=\begin{bmatrix}\frac{\delta}{L}\cdot c\\
0\\
1
\end{bmatrix}\in\R^{n+2}
\]
satisfies the following:\\
 1. $\overline{x}=\begin{bmatrix}1_{n}\\
1\\
1
\end{bmatrix}$, $\overline{y}=\begin{bmatrix}0_{d}\\
-1
\end{bmatrix}$ and $\overline{s}=\begin{bmatrix}1_{n}+\frac{\delta}{L}\cdot c\\
1\\
1
\end{bmatrix}$ are feasible primal dual vectors.\\
 2. For any feasible primal dual vectors $(\overline{x},\overline{y},\overline{s})$
with duality gap $\leq\delta^{2}$, consider the vector $\widehat{x}=R\cdot\overline{x}_{1:n}$
($\overline{x}_{1:n}$ is the first $n$ coordinates of $\overline{x}\in\R^{n+2}$)
is an approximate solution to the original linear program in the following
sense 
\begin{align*}
c^{\top}\widehat{x}\leq & ~\min_{\ma^{\top}x=b,x\geq0}c^{\top}x+LR\cdot\delta,\\
\|\ma^{\top}\widehat{x}-b\|_{2}\leq & ~4n^{2}\delta\cdot\Big(\|A\|_{F}R+\|b\|_{2}\Big),\\
\widehat{x} & \geq0.
\end{align*}
\end{lem}

\initialPoint*
\begin{proof}
Assume for simplicity that the original input LP satisfies $\|\ma^{\top}1_{n}-b/R\|_{\infty}\ge0.5(\|\ma\|_{F}+\|b\|/R)$.
If this assumption is not satisfied, then we can simply add a variable
$x_{n+1}$ and the constraint $(\|\ma\|_{F}+\|b\|/R)x_{n+1}=0$. Note
that this does not increase the value of $\norm b_{2}$,$\norm c_{2}$,
or $R$ for this new LP. Further, the constraint matrix, $\ma'\in\R^{(d+1)\times(n+1)}$
of this new LP satisfies $\|\ma'\|_{F}\le2\|\ma\|_{F}+\|b\|/R$ and
the new constraint vector $b'$ satisfies
\[
\norm{\left(\ma'\right)^{\top}\vones_{n}-b'/R}_{\infty}\geq\left(\norm{\ma}_{F}+\norm b_{2}/R\right)\geq\frac{1}{2}\left(\norm{\ma'}_{F}+\norm b_{2}/R\right)
\]
as desired. Consequently, if we take this new LP as input to our reduction
we only need to decrease $\delta$ by a factor of $2$ more than before
to obtain the same result. So for now assume $\|\ma^{\top}1_{n}-b\|_{\infty}\ge0.5(\|\ma\|_{F}+\|b\|/R)$.

Given an (infeasible) solution $(\overline{x},\overline{y},\overline{s})\in\R^{n+2}\times\R^{d+1}\times\R^{n+2}$
to the modified LP, let $(\overline{x}',\overline{y},\overline{s})$
be the feasible solution we get via projecting
\begin{align*}
\overline{x}'= & \overline{x}+\omx\oms{}^{-1}\overline{\ma}(\overline{\ma}^{\top}\omx\oms^{-1}\overline{\ma})^{-1}(b-\overline{\ma}^{\top}\overline{x}).
\end{align*}
Then this new point is feasible because 
\begin{align*}
\overline{\ma}^{\top}x'= & \overline{\ma}^{\top}\overline{x}+\overline{\ma}^{\top}\omx\oms{}^{-1}\overline{\ma}(\overline{\ma}^{\top}\omx\oms^{-1}\overline{\ma})^{-1}(b-\overline{\ma}^{\top}\overline{x})\\
= & \overline{\ma}^{\top}\ox+(b-\overline{\ma}^{\top}\overline{x})=b,
\end{align*}
and by Lemma~\ref{lem:move_x_correction} we have
\begin{align}
\|\omx^{-1}(\overline{x}'-\overline{x})\|_{\tau+\infty}^{2} & =O\left(\Phi_{b}\right).\label{eq:infeasible_x_close}
\end{align}
As $\ox'$ is feasible, we have $\sum_{i\in[n+1]}\ox'_{i}\le n+1$
by the constraints of the modified LP and thus $\ox'_{i}\le n+1$
for $i\le n+1$. To bound $\theta\defeq\ox'_{n+2}$ note that for
all $i$ with $(\ma^{\top}1_{n}-b/R)_{i}\neq0$ we have by the constraint
for the LP
\begin{align*}
\theta= & \frac{(\ma^{\top}\ox'_{1:n}-b/R)_{i}}{(\ma^{\top}1_{n}-b/R)_{i}}=\frac{\left|(\ma^{\top}\ox'_{1:n}-b/R)_{i}\right|}{\left|(\ma^{\top}1_{n}-b/R)_{i}\right|}.
\end{align*}
By choosing $i$ to be the maximizer of the absolute value of the
denominator we have
\[
\theta\leq\frac{\norm{\ma^{\top}\ox'_{1:n}-b/R}_{\infty}}{\norm{(\ma^{\top}1_{n}-b/R)_{i}}_{\infty}}\leq\frac{(\|\ma\|_{F}+\|b\|_{2}/R)(n+1)}{\norm{(\ma^{\top}1_{n}-b/R)_{i}}_{\infty}}=O(n)
\]
where in the second step we used that $\norm{\ox'_{1:n}}_{2}\leq\sum_{j\in[n]}\ox'_{j}\leq n+1$
and in the third step we used the assumption that $\|\ma^{\top}1_{n}-b/R\|_{\infty}\ge0.5(\|\ma\|_{F}+\|b\|_{2}/R)$.
Thus in summary we have $\|\ox'\|_{\infty}=O(n)$ and $\|\ox\|_{\infty}\le(1+O(\Phi_{b}))\cdot O(n)$
by (\ref{eq:infeasible_x_close}). This concludes the proof on $\|\ox\|_{\infty}$
and we are left with proving the last claim of Theorem~\ref{thm:infeasible_reduction}.

For our feasible solution $(\ox',\overline{y},\os)$ we can bound
the duality gap as follows
\begin{align*}
\sum_{i\in[n+2]}\ox'_{i}\cdot\os_{i}= & \sum_{i\in[n+2]}\ox_{i}\cdot\os_{i}(1+\ox_{i}^{-1}(\overline{x}'_{i}-\overline{x}_{i}))\\
\le & \sum_{i\in[n+2]}2\mu\tau_{i}(1+O(\sqrt{\Phi_{b}}))\\
= & 4\mu d(1+O(\sqrt{\Phi_{b}}))
\end{align*}
where we used $\ox_{i}\cdot\os_{i}\le2\mu\tau_{i}$, $\tau_{i}=\sigma_{i}+n/d$,
$\sum_{i}\sigma_{i}=d$ and the previous bound on $\|\omx^{-1}(\overline{x}'-\overline{x})\|_{\tau+\infty}^{2}$.
So for $\mu\le\delta^{2}/(8d)$ and small enough $\Phi_{b}=O(1)$,
this duality gap is bounded by $\delta^{2}$. This allows us to apply
Lemma~\ref{lem:original_reduction}, so we can construct a good solution
$\widehat{x}'\in\R^{n}$ for the original LP from $\ox'\in\R^{n+2}$
by taking the first $n$ coordinates and scaling them by $R$. We
further construct an $\widehat{x}\in\R^{n}$ from our infeasible solution
$\ox\in\R^{n+2}$ in the same way. For this infeasible $\widehat{x}\in\R^{n}$
we then have 
\begin{align*}
c^{\top}\widehat{x}= & c^{\top}\widehat{x}'+c^{\top}(\widehat{x}-\widehat{x}')\\
\le & \min_{\ma^{\top}x=b,x\geq0}c^{\top}x+LR\cdot\delta+c^{\top}(\widehat{x}-\widehat{x}').
\end{align*}
The impact of the last error term can be bounded by
\begin{align*}
c^{\top}(\widehat{x}-\widehat{x}') & \le\|c^{\top}\|_{2}\cdot\|\widehat{x}-\widehat{x}'\|_{2}\le L\|\widehat{x}'\cdot\widehat{x}'^{-1}(\widehat{x}-\widehat{x}')\|_{2}\\
 & \le L\|\widehat{x}'\|_{2}\cdot\|\widehat{x}'^{-1}(\widehat{x}-\widehat{x}')\|_{\infty}\\
 & \le LR\cdot O\left(n\sqrt{\Phi_{b}}\right)\,,
\end{align*}
where we used that $\overline{x}\in\R^{n+2}$ is feasible, so $\sum_{i\in[n+1]}\overline{x}'_{i}=n+1$,
which means $\|\widehat{x}'\|_{2}\le\|\widehat{x}'\|_{1}\le R(n+1)$.
The last step uses that the entries of $\overline{x}'\in\R^{n+2}$
and $\overline{x}\in\R^{n+2}$ are close multiplicatively, so the
same is also true for $\widehat{x}\in\R^{n}$ and $\widehat{x}'\in\R^{n}$.
In summary we thus obtain
\begin{align*}
c^{\top}\widehat{x}\le & \min_{\ma^{\top}x=b,x\geq0}c^{\top}x+LR\cdot(\delta+O(n\sqrt{\Phi_{b}})).
\end{align*}

We are left with proving the bound on $\|\ma^{\top}\widehat{x}-b\|_{2}$.
For this note that
\begin{align*}
\|\ma^{\top}\widehat{x}-b\|_{2}\le & \|\ma^{\top}\widehat{x}'-b\|_{2}+\|\ma^{\top}(\widehat{x}-\widehat{x}')\|_{2}\\
\le & 4n^{2}\delta\cdot\Big(\|A\|_{F}R+\|b\|_{2}\Big)+\|\ma^{\top}(\widehat{x}-\widehat{x}')\|_{2}
\end{align*}
via triangle inequality and Lemma \ref{lem:original_reduction}. So
it remains to bound $\|\ma^{\top}(\widehat{x}-\widehat{x}')\|_{2}$.
Let $x\in\R^{n+1}$ be the first $n+1$ entries of $\ox\in\R^{n+2}$
and $\theta\in\R$ be the last entry of $\ox\in\R^{n+2}$ (and likewise
$x'$ and $\theta'$ for the feasible $\ox'$), then
\begin{align*}
\|\overline{\ma}^{\top}(\ox-\ox')\|_{2}^{2} & =\frac{1}{R^{2}}\|\ma^{\top}(\widehat{x}-\widehat{x}')+(b-R\ma^{\top}1_{n})(\theta-\theta')\|_{2}^{2}+(1_{n+1}^{\top}(x-x'))^{2}\\
 & \geq\frac{1}{R^{2}}\|\ma^{\top}(\widehat{x}-\widehat{x}')+(b-R\ma^{\top}1_{n})(\theta-\theta')\|_{2}^{2}.
\end{align*}
Hence, we have
\begin{equation}
\|\ma^{\top}(\widehat{x}-\widehat{x}')\|_{2}\leq\|\overline{\ma}^{\top}(\ox-\ox')\|_{2}R+\|b-R\ma^{\top}1_{n}\|_{2}|\theta-\theta'|.\label{eq:initial_pt_A_err}
\end{equation}

We first bound the term $\|\overline{\ma}^{\top}(\ox-\ox')\|_{2}$
in (\ref{eq:initial_pt_A_err}). Note that $\Phi_{b}=\mu^{-1}\|\overline{\ma}^{\top}(\ox-\ox')\|_{(\overline{\ma}^{\top}\omx\oms{}^{-1}\overline{\ma})^{-1}}^{2}$.
So we bound the $\ell_{2}$-norm via this local norm.
\begin{align*}
\|\overline{\ma}^{\top}(\ox-\ox')\|_{2}^{2}\le & \frac{1}{\lambda_{\min}((\overline{\ma}^{\top}\omx\oms{}^{-1}\overline{\ma})^{-1})}\|\overline{\ma}^{\top}(\ox-\ox')\|_{(\overline{\ma}^{\top}\omx\oms{}^{-1}\overline{\ma})^{-1}}^{2}\\
= & \lambda_{\max}(\overline{\ma}^{\top}\omx\oms{}^{-1}\overline{\ma})\mu\Phi_{b}\\
\leq & (\max_{i\in[n]}\ox{}_{i}/\os{}_{i})\lambda_{\max}(\overline{\ma}^{\top}\overline{\ma})\mu\Phi_{b}\\
\leq & (\max_{i\in[n]}\ox{}_{i}/\os{}_{i})\|\overline{\ma}\|_{F}^{2}\mu\Phi_{b}\\
\leq & O(1)\cdot(\max_{i\in[n]}x'{}_{i}/\os{}_{i})\|\overline{\ma}\|_{F}^{2}\mu\Phi_{b}
\end{align*}
where in the last step we used that $x$ and $x'$ differ by an $1\pm O(\sqrt{\Phi_{b}})=O(1)$
factor. We already argued $\|x'\|_{\infty}=O(n)$ and we have $x's\ge0.5\tau\mu$,
so 
\begin{align*}
1/s_{i} & \le2x'_{i}/(\tau\mu)\le2nx'_{i}/(d\mu).
\end{align*}
With the previous bounds on $x$ this leads to
\begin{align*}
\max_{i\in[n]}x'_{i}/s{}_{i} & =O(n^{3}/(d\mu))=O(n^{3}/(d\mu)),
\end{align*}
which means 
\begin{equation}
\|\overline{\ma}^{\top}(\ox-\ox')\|_{2}^{2}\le O((n^{3}/d)\|\overline{\ma}\|_{F}^{2}\Phi_{b}).\label{eq:initial_pt_A_err1}
\end{equation}

Now, we bound the term $\|b-\ma^{\top}1_{n}\|_{2}|\theta-\theta'|$
in (\ref{eq:initial_pt_A_err}). Since we know that $|\theta|$ and
$|\theta'|$ are bounded by $O(n\delta)$ (this is proven in \cite{cohen2019solving}
where they proved Lemma~\ref{lem:original_reduction}), we have the
bound
\begin{equation}
\|b-R\ma^{\top}1_{n}\|_{2}^{2}\cdot|\theta-\theta'|^{2}\leq O(n\|\ma\|_{F}^{2}R^{2}+\|b\|_{2}^{2})\cdot n^{2}\delta^{2}.\label{eq:initial_pt_A_err2}
\end{equation}
Putting (\ref{eq:initial_pt_A_err1}) and (\ref{eq:initial_pt_A_err2})
into (\ref{eq:initial_pt_A_err}), we have
\[
\|\ma^{\top}(\widehat{x}-\widehat{x}')\|_{2}^{2}\leq O((n^{3}/d)\|\overline{\ma}\|_{F}^{2}\Phi_{b}R^{2})+O(n\|\ma\|_{F}^{2}R^{2}+\|b\|_{2}^{2})\cdot n^{2}\delta^{2}.
\]
At last we can bound 
\begin{align*}
\|\overline{\ma}\|_{F}\le & \|\ma\|_{F}+(n+1)^{0.5}\|\ma\|_{F}+\|\frac{1}{R}b-\ma^{\top}1_{n}\|_{2}\\
\le & O(n^{0.5}\|\ma\|_{F}+\frac{1}{R}\|b\|_{2}),
\end{align*}
So, in total, we obtain
\[
\|\ma^{\top}(\widehat{x}-\widehat{x}')\|_{2}\leq O(\|\ma\|_{F}Rn^{2}+\|b\|_{2}n)\cdot(\sqrt{\Phi_{b}}+\delta).
\]
and hence
\begin{align*}
\|\ma^{\top}\widehat{x}-b\|_{2}\le & \|\ma^{\top}\widehat{x}'-b\|_{2}+\|\ma^{\top}(\widehat{x}-\widehat{x}')\|_{2}\\
\le & 4n^{2}\delta\cdot\Big(\|\ma\|_{F}R+\|b\|_{2}\Big)+\|\ma^{\top}(\widehat{x}-\widehat{x}')\|_{2}\\
\le & 4n^{2}\delta\cdot\Big(\|\ma\|_{F}R+\|b\|_{2}\Big)+O(\|\ma\|_{F}Rn^{2}+\|b\|_{2}n)\cdot(\sqrt{\Phi_{b}}+\delta)\\
\le & O(n^{2})\cdot(\|\ma\|_{F}R+\|b\|_{2})\cdot(\sqrt{\Phi_{b}}+\delta).
\end{align*}
\end{proof}

\section{Gradient Maintenance}

\label{sec:grad_maintenance}

In this section we provide a data structure for efficiently maintaining
$\nabla\Phi(\overline{v})^{\flat}$ and $\ma^{\top}\mx\nabla\Phi(\overline{v})^{\flat}$
in our IPM, i.e. Algorithm~\ref{alg:pathfollowing}.

\begin{algorithm2e}[t]

\caption{Algorithm for Maintaining $\nabla\Phi(\ov)^{\flat}$ and
$\ma^{\top}\mx\nabla\Phi(\ov)^{\flat}$ (Theorem \ref{thm:gradient_maintenance})}\label{alg:gradient_maintenance}

\SetKwProg{Members}{members}{}{}

\SetKwProg{Proc}{procedure}{}{}

\Members{}{

\State $\ov_{i}^{(k,\ell)}\in\{0,1\}^{n}$ \tcp*{Indicator vectors
for a partition of $[n]$}

\State $w^{(k,\ell)}\in\R^{d}$ \tcp*{Maintained to be $\ma^{\top}\mx\ov^{(k,\ell)}$}

\State $x,v,\overline{\tau}\in\R^{n}$ 

}

\vspace{0.1 in}

\Proc{$\dsinit(\ma\in\R^{n\times d},v\in\R^{n},\tau\in\R^{n},x\in\R^{n},\epsilon>0)$}{

\State $\ma\leftarrow\ma$, $x\leftarrow x$, $v\leftarrow v$

\For{$k=1,...,\log^{-1}(1-\epsilon)$ and $\ell=0,...,(1.5/\epsilon)$}{

\State $\ov^{(k,\ell)}\leftarrow0_{n}$

\State $\ov_{i}^{(k,\ell)}\leftarrow1$ and $\overline{\tau}_{i}\leftarrow(1-\epsilon)^{k+1}$
if $0.5+\ell\epsilon/2\le v_{i}<0.5+(\ell+1)\epsilon/2$ and $(1-\epsilon)^{k+1}\le\tau_{i}\le(1-\epsilon)^{k}.$

\State $w^{(k,\ell)}\leftarrow\ma^{\top}\mx\cdot\ov^{(k,\ell)}$

}

}

\vspace{0.1 in}

\Proc{$\textsc{Update}(i\in[n],a\in\R,b\in\R,c\in\R)$}{

\State Find $k,\ell$ such that $\ov_{i}^{(k,\ell)}=1$, then $\ov_{i}^{(k,\ell)}\leftarrow0$.

\State $w^{(k,\ell)}\leftarrow w^{(k,\ell)}-\ma^{\top}x_{i}\indicVec i$

\State Find $k,\ell$ such that $0.5+\ell\epsilon/2\le a<0.5+(\ell+1)\epsilon/2$
and $(1-\epsilon)^{k+1}\le b\le(1-\epsilon)^{k}$, then $\ov_{i}^{(k,\ell)}\leftarrow0$.

\State $w^{(k,\ell)}\leftarrow w^{(k,\ell)}+\ma^{\top}c\indicVec i$

\State $v_{i}\leftarrow a$, $x_{i}\leftarrow c$, $\overline{\tau}_{i}\leftarrow(1-\epsilon)^{k+1}$

}

\vspace{0.1 in}

\Proc{$\textsc{Query}()$}{

\State Find $s^{(k,l)}$ with $\argmax_{\|w\|_{2}+\|w/\sqrt{\overline{\tau}}\|_{\infty}\leq1}\left\langle \nabla\Phi(\overline{v})/\sqrt{\overline{\tau}},w\right\rangle =\sum_{k,l}s^{(k,l)}\overline{v}^{(k,\ell)}$
via Algorithm 8 from \cite{lsJournal19}.

\State \Return$\sum_{k,l}s^{(k,l)}\overline{v}^{(k,\ell)}$ and
$\sum_{k,l}s^{(k,l)}\ma^{\top}\mx\overline{v}^{(k,\ell)}$

}

\end{algorithm2e}
\begin{thm}
\label{thm:gradient_maintenance} There exists a deterministic data-structure
that supports the following operations
\begin{itemize}
\item $\textsc{Initialize }(\ma\in\R^{n\times d},v\in\R^{n},\tau\in\R^{n},x\in\R^{n},\epsilon>0)$:
The data-structure preprocesses the given matrix $\ma\in\R^{n\times d}$,
vectors $v,\tau,x\in\R^{n}$, and accuracy parameter $\epsilon>0$
in $O(nd)$ time. The data-structure assumes $0.5\le v\le2$ and $d/n\le\tau\le2$.
\item $\textsc{Update}(i,a,b,c)$: Sets $v_{i}=a$, $\tau_{i}=b$ and $x_{i}=c$
in $O(d)$ time. The data-structure assumes $0.5\le v\le2$ and $d/n\le\tau\le2$.
\item $\textsc{Query}()$: Returns $\nabla\Phi(\ov)^{\flat}$ and $\ma^{\top}\mx\nabla\Phi(\ov)^{\flat}$
for $\|\ov-v\|_{\infty}\le\epsilon$ in $\otilde(n)$ time. See Lemma~\ref{lem:smoothing:helper}
for $(\cdot){}^{\flat}$.
\end{itemize}
\end{thm}

\begin{proof}
We start by explaining the algorithm and analyzing its complexity.
Afterwards we prove the correctness.

\paragraph*{Algorithm:}

During preprocessing and updates, the algorithm maintains a partition
$V^{(k,\ell)}\subset[n]$ where $i\in V^{(k,\ell)}$ if 
\begin{align*}
0.5+\ell\epsilon/2\le & v_{i}<0.5+(\ell+1)\epsilon/2\\
(1-\epsilon)^{k+1}\le & \tau_{i}\le(1-\epsilon)^{k}.
\end{align*}
Additionally the data-structure maintains the products $\ma^{\top}\mx\ov^{(k,\ell)}$
where $\ov_{i}^{(k,\ell)}=1$ if $i\in V^{(k,\ell)}$ and $v_{i}^{(k,\ell)}=0$
otherwise. 

During queries the data-structure computes the following: Define $\phi(x):=\lambda(\exp(\lambda(x-1))-\exp(-\lambda(x-1)))=(\nabla\Phi(x))_{i}$,
then we apply Algorithm 8 from \cite{lsJournal19} to compute
\begin{align*}
\sum_{k,\ell}s^{(k,\ell)}\overline{v}^{(k,\ell)} & =\argmax_{\|w\|_{2}+\|w/\sqrt{\overline{\tau}}\|_{\infty}\leq1}\left\langle \sum_{k,\ell}\phi(0.5+\ell\epsilon/2)\cdot\overline{v}^{(k,\ell)}/\sqrt{\overline{\tau}},w\right\rangle 
\end{align*}
in $O(n\log n)$ time and then return $\sum_{k,\ell}s^{(k,\ell)}\overline{v}^{(k,\ell)}$
and $\sum_{k,\ell}s^{(k,\ell)}\ma^{\top}\mx\overline{v}^{(k,\ell)}$.
The solution can be represented by $\sum_{k,\ell}s^{(k,\ell)}\overline{v}^{(k,\ell)}$
because all coordinates with same value of $\tau_{i}$ and $x_{i}$
has the same $w_{i}$ in the solution. This can be seen easily from
the description of Algorithm 8 from \cite{lsJournal19}.

\paragraph{Complexity:}

Every entry change to $v,\tau$ might move an index $i$ from some
$V^{(k,\ell)}$ to another $V^{(k,\ell)}$ which corresponds to changing
one entry of $\overline{v}^{(k,\ell)}$ and $\overline{v}^{(k',\ell')}$.
So every update to $v,\tau,x$ costs only $O(d)$ time to update $\ma^{\top}\mx\overline{v}^{(k,\ell)}$.

\paragraph{Correctness:}

For the given vector $v$, we can split $v$ into groups by grouping
the entries to multiples of $\epsilon/2$. By assumption we have $0.5\le v\le2$,
so we have at most $O(1/\epsilon)$ many groups. By rounding down
$v$ on each of these groups we obtain
\begin{align*}
\ov & :=\sum_{k}(0.5+k\epsilon/2)\cdot\sum_{i:0.5+k\epsilon/2\le v_{i}<0.5+(k+1)\epsilon/2}\indicVec i
\end{align*}
which satisfies $\|\ov-v\|_{\infty}\leq\epsilon/2$. For notational
simplicity define 
\begin{align*}
\ov^{(k)}:= & \sum_{i:0.5+k\epsilon/2\le v_{i}<0.5+(k+1)\epsilon/2}\indicVec i
\end{align*}
then $\ov=(0.5+k\epsilon/2)\sum_{k}\ov^{(k)}$ and for $\phi(x):=\lambda(\exp(\lambda(x-1))-\exp(-\lambda(x-1)))=(\nabla\Phi(x))_{i}$
we have
\begin{align*}
\nabla\Phi(\overline{v})= & \sum_{k}\nabla\Phi((0.5+k\epsilon/2)\cdot\overline{v}^{(k)})\\
= & \sum_{k}\phi(0.5+k\epsilon/2)\cdot\overline{v}^{(k)}.
\end{align*}
Next we split $[n]$ into $O(\epsilon^{-1}\log(n/d))$ many groups
based on multiplicative approximations of $\tau$, in other words
we have that
\begin{align*}
\overline{\tau}:= & \sum_{k}(1-\epsilon)^{k}\sum_{i:(1-\epsilon)^{k+1}\le\tau_{i}<(1-\epsilon)^{k}}\indicVec i
\end{align*}
is an entry-wise $(1\pm\epsilon)$-approximation of $\tau$. Then
for $\overline{v}^{(k,\ell)}:=\overline{v}^{(k)}\cdot\sum_{i:(1-\epsilon)^{k+1}\le\tau_{i}<(1-\epsilon)^{k}}\indicVec i$,
we have
\begin{align*}
\nabla\Phi(\overline{v})/\sqrt{\overline{\tau}}= & \sum_{k,\ell}(1-\epsilon)^{k}\cdot\phi(0.5+\ell\epsilon/2)\cdot\overline{v}^{(k,\ell)}.
\end{align*}
Next, Algorithm 8 from \cite{lsJournal19} shows how to compute scalars
$s^{(k,l)}$ in $O(n\log n)$ time with
\begin{align*}
\argmax_{\|w\|_{2}+\|w/\sqrt{\overline{\tau}}\|_{\infty}\leq1}\left\langle \nabla\Phi(\overline{v})/\sqrt{\overline{\tau}},w\right\rangle = & \sum_{k,l}s^{(k,l)}\overline{v}^{(k,\ell)}.
\end{align*}
Note that there is $\|\ov'-\ov\|_{\infty}\le\epsilon/\lambda$ with
\begin{align*}
 & \argmax_{\|w\|_{2}+\|w/\sqrt{\overline{\tau}}\|_{\infty}\leq1}\left\langle \nabla\Phi(\overline{v})/\sqrt{\overline{\tau}},w\right\rangle \\
= & \argmax_{\|w\sqrt{\overline{\tau}/\tau}\|_{2}+\|w/\sqrt{\tau}\|_{\infty}\leq\sqrt{\overline{\tau}/\tau}}\left\langle \nabla\Phi(\overline{v})/\sqrt{\overline{\tau}},w\right\rangle \\
= & \argmax_{\|w\sqrt{\overline{\tau}/\tau}\|_{2}+\|w/\sqrt{\tau}\|_{\infty}\leq1}\left\langle \nabla\Phi(\overline{v})/\sqrt{\overline{\tau}},w\cdot\sqrt{\overline{\tau}/\tau}\right\rangle \\
= & \argmax_{\|w\sqrt{\overline{\tau}/\tau}\|_{2}+\|w/\sqrt{\tau}\|_{\infty}\leq1}\left\langle \nabla\Phi(\overline{v})/\sqrt{\tau},w\right\rangle \\
= & (1\pm\epsilon)\argmax_{\|w\|_{2}+\|w/\sqrt{\tau}\|_{\infty}\leq1}\left\langle \nabla\Phi(\overline{v})/\sqrt{\tau},w\right\rangle \\
= & \argmax_{\|w\|_{2}+\|w/\sqrt{\tau}\|_{\infty}\leq1}\left\langle \nabla\Phi(\overline{v}')/\sqrt{\tau},w\right\rangle \\
= & \nabla\Phi(\overline{v}')^{\flat}
\end{align*}
So $\sum_{k,l}s^{(k,l)}\overline{v}^{(k,\ell)}=\nabla\Phi(\overline{v}')^{\flat}$
for some $\|\ov'-v\|_{\infty}\le\epsilon$. Thus the algorithm returns
$\nabla\Phi(\overline{v}')^{\flat}$ and $\ma^{\top}\mx\nabla\Phi(\overline{v}')^{\flat}$.
\end{proof}

\end{document}